%% file: main.tex
\documentclass[final]{scrartcl}

\input{headers}

\usepackage{fullpage}

\ifdraft{
  \usepackage[scrtime]{prelim2e}
  \usepackage[inline]{showlabels}

  \ifdefined\fastmode
    \usepackage{comment}
    \excludecomment{figure}
    
  \fi
}{}

\DeclareSymbolFont{greek}     {OML}{cmm}{m}{it}
\DeclareMathSymbol{\alpha}{\mathord}{greek}{"0B}
\DeclareMathSymbol{\beta}{\mathord}{greek}{"0C}
\DeclareMathSymbol{\gamma}{\mathord}{greek}{"0D}
\DeclareMathSymbol{\delta}{\mathord}{greek}{"0E}
\DeclareMathSymbol{\epsilon}{\mathord}{greek}{"22}
\DeclareMathSymbol{\varepsilon}{\mathord}{greek}{"22}
\DeclareMathSymbol{\zeta}{\mathord}{greek}{"10}
\DeclareMathSymbol{\eta}{\mathord}{greek}{"11}
\DeclareMathSymbol{\theta}{\mathord}{greek}{"12}
\DeclareMathSymbol{\iota}{\mathord}{greek}{"13}
\DeclareMathSymbol{\kappa}{\mathord}{greek}{"14}
\DeclareMathSymbol{\lambda}{\mathord}{greek}{"15}
\DeclareMathSymbol{\mu}{\mathord}{greek}{"16}
\DeclareMathSymbol{\nu}{\mathord}{greek}{"17}
\DeclareMathSymbol{\xi}{\mathord}{greek}{"18}
\DeclareMathSymbol{\pi}{\mathord}{greek}{"19}
\DeclareMathSymbol{\rho}{\mathord}{greek}{"1A}
\DeclareMathSymbol{\sigma}{\mathord}{greek}{"1B}
\DeclareMathSymbol{\tau}{\mathord}{greek}{"1C}
\DeclareMathSymbol{\upsilon}{\mathord}{greek}{"1D}
\DeclareMathSymbol{\phi}{\mathord}{greek}{"27}
\DeclareMathSymbol{\chi}{\mathord}{greek}{"1F}
\DeclareMathSymbol{\psi}{\mathord}{greek}{"20}
\DeclareMathSymbol{\omega}{\mathord}{greek}{"21}
\DeclareMathSymbol{\vartheta}{\mathord}{greek}{"23}
\DeclareMathSymbol{\varpi}{\mathord}{greek}{"24}
\DeclareMathSymbol{\varrho}{\mathord}{greek}{"25}
\DeclareMathSymbol{\varsigma}{\mathord}{greek}{"26}
\DeclareMathSymbol{\varphi}{\mathord}{greek}{"27}

\begin{document}
\hyphenation{Brow-serID}
\hyphenation{in-fra-struc-ture}
\hyphenation{brow-ser}
\hyphenation{doc-u-ment}
\hyphenation{Chro-mi-um}
\hyphenation{meth-od}
\hyphenation{sec-ond-ary}
\hyphenation{Java-Script}
\hyphenation{Mo-zil-la}
\hyphenation{post-Mes-sage}

\title{A Comprehensive Formal Security Analysis of OAuth 2.0\footnote{\hspace{1ex}An abridged version appears
    in CCS 2016 \cite{FettKuestersSchmitz-CCS-2016}.}}
\date{}
\author{Daniel Fett\\University of Trier, Germany\\\texttt{fett@uni-trier.de} \and Ralf K{\"u}sters\\University of Trier, Germany\\\texttt{kuesters@uni-trier.de} \and Guido Schmitz\\University of Trier, Germany\\\texttt{schmitzg@uni-trier.de}}

\maketitle

\input{section-abstract}

\newpage{}

\ifdraft{
\listoftodos
}{ }

\tableofcontents

\newpage

\listoffigures
\listofalgorithms

\newpage

\input{section-intro}

\input{section-oauth}

\input{section-attacks}

\input{section-model}

\input{section-analysis}

\input{section-related-work}

\input{section-conclusion}

\section{Acknowledgements} 
This work was partially supported by \textit{Deu\-tsche
  Forschungsgemeinschaft} (DFG) through  Grant KU\ 1434/10-1.

\input{bibl}
\appendix

\let\stdsection\section
\renewcommand\section{\clearpage\stdsection}  %

\input{appendix-oauth-description}

\input{appendix-oauth-attack}

\input{appendix-openid-connect}

\input{appendix-webmodel}

\input{appendix-message-formats}

\input{appendix-browsermodel}
\input{appendix-oauth-model}

\input{appendix-oauth-model-webattacker}

\input{appendix-oauth-secproperties}

\input{appendix-oauth-proof}

\end{document}

%% file: headers.tex
\usepackage[utf8]{inputenc}
\usepackage[T1]{fontenc}

\usepackage{mathptmx}

\usepackage{keyval}
\usepackage{footmisc}
\usepackage{ragged2e}
\usepackage{xspace}
\usepackage[usenames,dvipsnames,svgnames,table]{xcolor}
\usepackage[final,stretch=10,shrink=10]{microtype}
\usepackage{ifdraft}
\usepackage[inline]{enumitem}
\usepackage{footnote}
\usepackage{tikz}
\usetikzlibrary{backgrounds,snakes,arrows,decorations.markings,calc}
\usepackage{varwidth}
\usepackage{algpseudocode}
\usepackage{algorithm} %
\usepackage{soul} %
\usepackage{adjustbox}
\usepackage{multirow}
\usepackage{calc}

\usepackage{captcont}
\usepackage{caption}
\DeclareCaptionLabelSeparator{dot}{.\enspace}
\clearcaptionsetup{figure}
\usepackage[nottoc,notlof,section]{tocbibind}

\usepackage[thmmarks]{ntheorem}
\usepackage{macros}
\usepackage{cite}
\usepackage{amsmath}
\usepackage{amssymb}

\setlistdepth{9}
\setlist[itemize,1]{label=$\bullet$}
\setlist[itemize,2]{label=$\bullet$}
\setlist[itemize,3]{label=$\bullet$}
\setlist[itemize,4]{label=$\bullet$}
\setlist[itemize,5]{label=$\bullet$}
\setlist[itemize,6]{label=$\bullet$}
\setlist[itemize,7]{label=$\bullet$}
\setlist[itemize,8]{label=$\bullet$}
\setlist[itemize,9]{label=$\bullet$}

\renewlist{itemize}{itemize}{9}

\usepackage[english]{babel}
\usepackage[final]{hyperref}
\usepackage{color}
\definecolor{linkcolor}{rgb}{0,0,0.5}
\hypersetup{
 colorlinks=true,
  linkcolor=linkcolor,
  anchorcolor=linkcolor,
  citecolor=linkcolor,
  filecolor=linkcolor,
  menucolor=linkcolor,
  runcolor=linkcolor,
  urlcolor=linkcolor,
}

\usepackage{todonotes}

\usepackage{fancybox}
\usepackage{wrapfig}

%% file: section-abstract.tex
The OAuth~2.0 protocol 
is one of the most widely deployed authorization/single sign-on (SSO)
protocols and also serves as the foundation for the new SSO standard
OpenID Connect. Despite the popularity of OAuth, so far analysis
efforts were mostly targeted at finding bugs in specific
implementations and were based on formal models which abstract from
many web features or did not provide a formal treatment at all.

In this paper, we carry out the first extensive formal analysis of the
OAuth 2.0 standard in an expressive web model. Our analysis aims at
establishing strong authorization, authentication, and session
integrity guarantees, for which we provide formal definitions. In our
formal analysis, all four OAuth grant types (authorization code grant,
implicit grant, resource owner password credentials grant, and the
client credentials grant) are covered. They may even run
simultaneously in the same and different relying parties and identity
providers, where malicious relying parties, identity providers, and
browsers are considered as well. Our modeling and
analysis of the OAuth 2.0 standard assumes that security
recommendations and best practices are
followed in order to avoid obvious and known attacks.

When proving the security of OAuth in our model, we discovered four attacks which
break the security of OAuth. The vulnerabilities can be exploited in
practice and are present also in OpenID Connect.

We propose fixes for the identified vulnerabilities, and then, for the
first time, actually prove the security of OAuth in an expressive web
model. In particular, we show that the fixed version of OAuth (with
security recommendations and best practices in place) provides the
authorization, authentication, and session integrity properties we
specify.

%% file: section-intro.tex
\section{Introduction}
\label{sec:introduction}

The OAuth~2.0 authorization framework \cite{rfc6749-oauth2} defines a
web-based protocol that allows a user to grant web sites access to her
resources (data or services) at other web sites
(\emph{authorization}).  The former web sites are called relying
parties (RP) and the latter are called identity providers
(IdP).\footnote{Following the OAuth~2.0 terminology, IdPs are called
  \emph{authorization servers} and \emph{resource servers}, RPs are
  called \emph{clients}, and users are called \emph{resource
    owners}. Here, however, we stick to the more common terms
  mentioned above.} In practice, OAuth~2.0 is often used for
\emph{authentication} as well. That is, a user can log in at an RP
using her identity managed by an IdP (single sign-on, SSO).

Authorization and SSO solutions have found widespread adoption in the
web over the last years, with OAuth 2.0 being one of the most popular
frameworks. OAuth~2.0, in the following often simply called
\emph{OAuth},\footnote{Note that in this document, we consider only
  OAuth~2.0, which is very different to its predecessor, OAuth~1.0(a).}
is used by identity providers such as Amazon, Facebook, Google,
Microsoft, Yahoo, GitHub, LinkedIn, StackExchange, and Dropbox. This
enables billions of users to log in at millions of RPs or share their
data with these \cite{similartech-fb-connect-stats-2015}, making OAuth
one of the most used single sign-on systems on the web.

OAuth is also the foundation for the new single sign-on protocol
OpenID Connect, which is already in use and actively supported by
PayPal (``Log In with PayPal''), Google, and Microsoft, among others.
Considering the broad industry support for OpenID Connect, a
widespread adoption of OpenID Connect in the next years seems
likely. OpenID Connect builds upon OAuth and provides clearly defined
interfaces for user authentication and additional (optional) features,
such as dynamic identity provider discovery and relying party
registration, signing and encryption of messages, and logout.

In OAuth, the interactions between the user and her
browser, the RP, and the IdP can be performed in four different flows,
or \emph{grant types}: authorization code grant, implicit grant,
resource owner password credentials grant, and the client credentials
grant (we refer to these as \emph{modes} in the following). In
addition, all of these modes provide further options.

The goal of this work is to provide an in-depth security analysis of
OAuth. Analyzing the security of OAuth is a challenging task, on the
one hand due to the various modes and options that OAuth provides, and
on the other hand due to the inherent complexity of the web.

So far, most analysis efforts regarding the security of OAuth were
targeted towards finding errors in specific
implementations~\cite{BansalBhargavanetal-JCS-2014,LiMitchell-ISC-2014,SantsaiBeznosov-CCS-2012-OAuth,Chenetal-2014,ShehabMohsen-2014,Yangetal-AsiaCCS-2016,Shernanetal-DIMVA-2015},
rather than the comprehensive analysis of the standard itself.
Probably the most detailed formal analysis carried out on OAuth so far
is the one in \cite{BansalBhargavanetal-JCS-2014}. However, none of
the existing analysis efforts of OAuth account for all modes of OAuth
running simultaneously, which may potentially introduce new security
risks. In fact, many existing approaches analyze only the
authorization code mode and the implicit mode of OAuth. Also,
importantly, there are no analysis efforts that are based on a
comprehensive formal web model (see below), which, however, is
essential to rule out security risks that arise when running the
protocol in the context of common web technologies (see
Section~\ref{sec:related-work} for a more detailed discussion of
related work).

\subsubsection{Contributions of this Paper} We perform the first
extensive formal analysis of the OAuth~2.0 standard for all four
modes, which can even run simultaneously within the same and different
RPs and IdPs, based on a comprehensive web model which covers large
parts of how browsers and servers interact in real-world setups. Our
analysis also covers the case of malicious IdPs, RPs, and
browsers/users.

\paragraph{Formal model of OAuth} Our formal analysis of OAuth uses an
expressive Dolev-Yao style model of the web
infrastructure~\cite{FettKuestersSchmitz-SP-2014} proposed by Fett,
K{\"u}sters, and Schmitz (FKS). The FKS model has already been used to
analyze the security of the BrowserID single sign-on
system~\cite{FettKuestersSchmitz-SP-2014,FettKuestersSchmitz-ESORICS-BrowserID-Primary-2015}
as well as the security and privacy of the SPRESSO single sign-on
system~\cite{FettKuestersSchmitz-CCS-2015}. This web model is designed
independently of a specific web application and closely mimics
published (de-facto) standards and specifications for the web, for
instance, the HTTP/1.1 and HTML5 standards and associated (proposed)
standards. It is the most comprehensive web model to date. Among
others, HTTP(S) requests and responses, including several headers,
such as cookie, location, strict transport security (STS), and origin
headers, are modeled. The model of web browsers captures the concepts
of windows, documents, and iframes, including the complex navigation
rules, as well as new technologies, such as web storage and
web messaging (via postMessage). JavaScript is modeled in an
abstract way by so-called \emph{scripts} which can be sent around
and, among others, can create iframes and initiate XMLHTTPRequests (XHRs).
Browsers may be corrupted dynamically by the adversary.

Using the generic FKS model, we build a formal model of OAuth, closely
following the OAuth~2.0 standard (RFC6749~\cite{rfc6749-oauth2}).
Since this RFC does not fix all aspects of the protocol and in order
to avoid known implementation attacks, we use the OAuth~2.0 security
recommendations (RFC6819~\cite{rfc6819-oauth2-security}), additional
RFCs and OAuth Working Group drafts
(e.g., RFC7662~\cite{rfc7662-oauth-token-introspection},
\cite{rfc-draft-oauth-jwt-encoded-state}) and current web best
practices (e.g., regarding session handling) to obtain a model of
OAuth with state-of-the-art security features in place, while
making as few assumptions as possible. Moreover, as mentioned above,
our model includes RPs and IdPs that (simultaneously) support all four
modes and can be dynamically corrupted by the adversary. Also, we
model all configuration options of OAuth (see
Section~\ref{sec:oauth}).

\paragraph{Formalization of security properties} Based on this model
of OAuth, we provide three central security properties of OAuth:
authorization, authentication, and session integrity, where session
integrity in turn is concerned with both authorization and
authentication.

\paragraph{Attacks on OAuth 2.0 and fixes} While trying to prove these
properties, we discovered four attacks on OAuth. In the first attack,
which breaks the authorization and authentication properties, IdPs
inadvertently forward user credentials (i.e., username and password)
to the RP or the attacker. In the second attack (IdP mix-up), a
network attacker playing the role of an IdP can impersonate any
victim. This severe attack, which again breaks the authorization and
authentication properties, is caused by a logical flaw in the
OAuth~2.0 protocol. Two further attacks allow an attacker to force a
browser to be logged in under the attacker’s name at an RP or force an
RP to use a resource of the attacker instead of a resource of the
user, breaking the session integrity property. We have verified all
four attacks on actual implementations of OAuth and OpenID Connect. We
present our attacks on OAuth in detail in Section~\ref{sec:attacks}.
In Appendix~\ref{app:openid-connect} we show how the attacks can be exploited in OpenID
Connect. We also show how the attacks can be fixed by changes that are easy to
implement in new and existing deployments of OAuth and OpenID Connect.

We notified the respective working groups, who confirmed the
attacks and that changes to the standards/recommendations are
needed. The IdP mix-up attack already resulted in a draft of a new
RFC~\cite{rfc-draft-ietf-oauth-mix-up-mitigation-01}.

\paragraph{Formal analysis of OAuth 2.0} Using our model of OAuth with
the fixes in place, we then were able to prove that OAuth satisfies
the mentioned security properties. This is the first proof which
establishes central security properties of OAuth in a comprehensive
and expressive web model  (see also Section~\ref{sec:related-work}).

We emphasize that, as mentioned before, we model OAuth with security
recommendations and best practices in place. As discussed in
Section~\ref{sec:analysis}, implementations not following these
recommendations and best practices may be vulnerable to attacks. In
fact, many such attacks on specific implementations have been pointed
out in the literature (e.g.,
\cite{BansalBhargavanetal-JCS-2014,rfc6749-oauth2,rfc6819-oauth2-security,Wangetal-USENIX-Explicating-SDKs-2013,SantsaiBeznosov-CCS-2012-OAuth,LiMitchell-ISC-2014,Chenetal-2014}).
Hence, our results also provide guidelines for secure OAuth
implementations.

We moreover note that, while these results provide strong security guarantees
for OAuth, they do not directly imply security of OpenID Connect
because OpenID Connect adds specific details on top of OAuth. We leave
a formal analysis of OpenID Connect to future work. The results
obtained here can serve as a good foundation for such an analysis.

\subsubsection{Structure of this Paper} In Section~\ref{sec:oauth}, we
provide a detailed description of OAuth~2.0 using the
authorization code mode as an example. In Section~\ref{sec:attacks}, we
present the attacks that we found during our analysis. An overview of
the FKS model we build upon in our analysis is provided in
Section~\ref{sec:model}, with the formal analysis of OAuth presented
in Section~\ref{sec:analysis}. Related work is discussed in
Section~\ref{sec:related-work}. We conclude in
Section~\ref{sec:conclusion}. Full details, including how the attacks
can be applied to OpenID Connect, further details on our model of
OAuth, and our security proof, can be found in the appendix.

%% file: section-oauth.tex
\section{OAuth 2.0}
\label{sec:oauth}

In this section, we provide a description of the OAuth authorization
code mode, with the other three modes explained only briefly. In
Appendix~\ref{sec:oauth-all-modes}, we provide a detailed description
of the remaining three modes (grant types).

OAuth was first intended for \emph{authorization}, i.e.,  users
authorize RPs to access user data (called \emph{protected resources})
at IdPs. For example, a user can use OAuth to authorize services such
as IFTTT\footnote{IFTTT (\emph{If This Then That}) is a web service which can be used to automate
  actions: IFTTT is triggered by user-defined events (e.g., Twitter
  messages) and carries out user-defined tasks (e.g., posting on the
  user's Facebook wall).} to access her (private) timeline on
Facebook. In this case, IFTTT is the RP and Facebook the IdP.

Roughly speaking, in the most common modes, OAuth works as follows: If
a user wants to authorize an RP to access some of the user's data at
an IdP, the RP redirects the user (i.e., the user's browser) to the IdP,
where the user authenticates and agrees to grant the RP access to some
of her user data at the IdP. Then, along with some token (an
\emph{authorization code} or an \emph{access token}) issued by the IdP, the user is
redirected back to the RP. The RP can then use the token as a
credential at the IdP to access the user's data at the IdP.

OAuth is also commonly used for \emph{authentication}, although it was
not designed with authentication in mind. A user can, for example, use
her Facebook account, with Facebook being the IdP, to log in at the
social network Pinterest (the RP). Typically, in order to log in, the
user authorizes the RP to access a unique user identifier at the IdP.
The RP then retrieves this identifier and considers this user to be
logged in.

Before an RP can interact with an IdP, the RP needs to be registered
at the IdP. The details of the registration process are out of the
scope of the OAuth protocol. In practice, this process is usually a
manual task. During the registration process, the IdP assigns credentials to the
RP: a public OAuth client id and (optionally) a client
secret. (Recall that in the terminology of the OAuth standard the term
``client'' stands for RP.)
The RP may later use the client secret (if issued) to authenticate to
the IdP.

Also, an RP registers one or more \emph{redirection endpoint} URIs (located at the RP) at an IdP. As
we will see below, in some OAuth modes, the IdP redirects the user's
browser to one of these URIs. Note that
(depending on the implementation of an IdP) an RP may also register a
pattern as a redirect URI and then specify the exact redirect URI
during the OAuth run.

In all modes, OAuth provides several options, such as those
mentioned above. For brevity of presentation (and in contrast to our analysis), in the following
descriptions, we consider only a specific set of options. For example,
we assume that an RP always provides a redirect URI and shares an
OAuth client secret with the IdP.

\subsubsection{Authorization Code Mode}
\begin{figure}[t!]
  \centering
  \input{figure-oauth-auth-code-grant}
  \caption[OAuth~2.0 authorization code mode]{OAuth~2.0 authorization code mode. Note that data depicted
    below the arrows is either transferred in URI parameters, HTTP headers, or
    POST bodies.}
  \label{fig:oauth-auth-code-grant}
  \vspace{-2ex}
\end{figure}
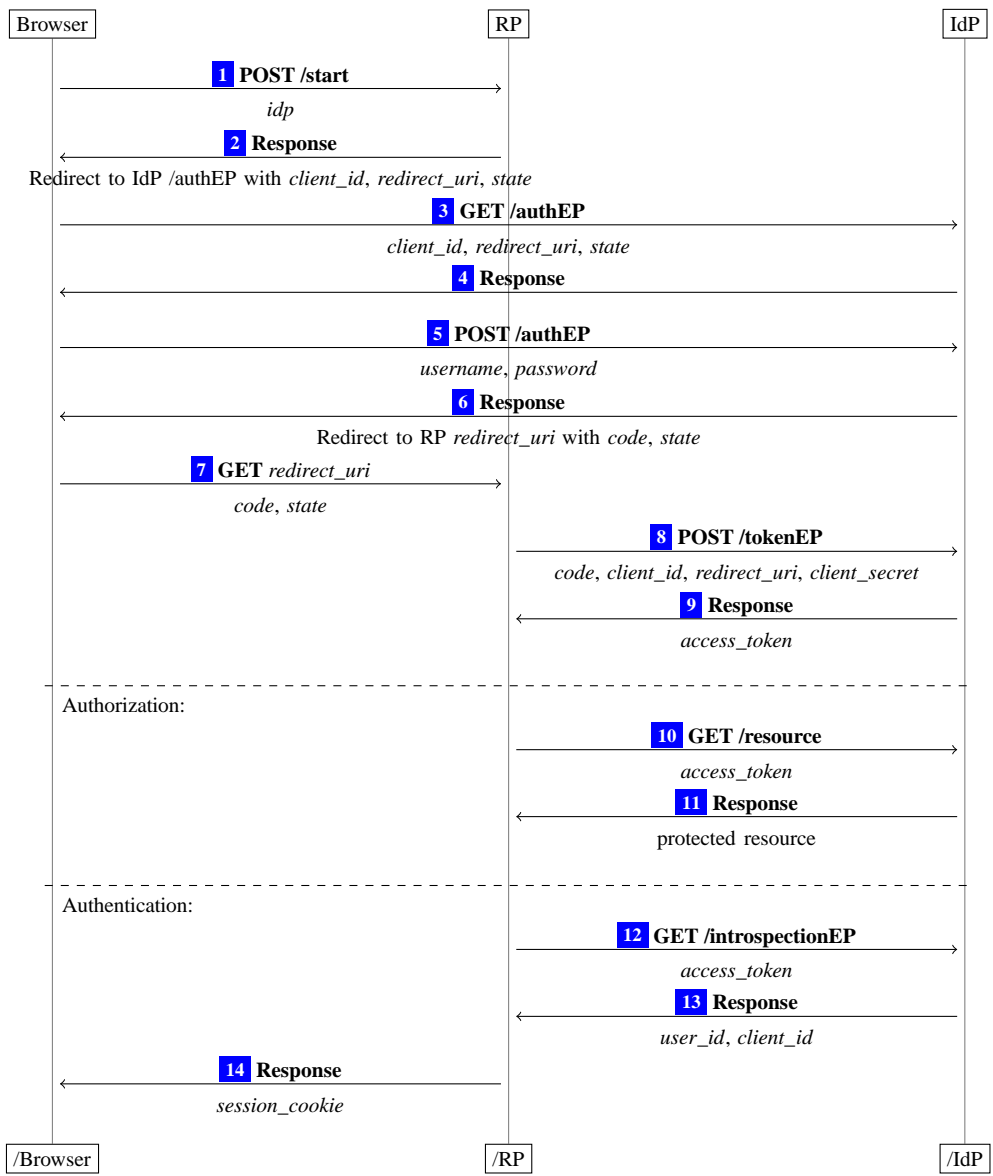
When the user tries to authorize an RP to access her data at an IdP or
to log in at an RP, the RP first redirects the user's browser to the
IdP. The user then authenticates to the IdP, e.g., by providing her
user name and password, and finally is redirected back to the RP along
with an \emph{authorization code} generated by the IdP. The RP can now
contact the IdP with this authorization code (along with the client id
and client secret) and receive an \emph{access token}, which the RP in turn
can use as a credential to access the user's protected
resources at the IdP.

\paragraph{Step-by-Step Protocol Flow} In what follows, we describe
the protocol flow of the authorization code mode step-by-step (see
also Figure~\ref{fig:oauth-auth-code-grant}). First, the user starts
the OAuth flow, e.g., by clicking on a button to select an IdP,
resulting in request~\refprotostep{acg-start-req} being sent to the
RP. The RP selects one of its redirection endpoint URIs $\mi{redirect\_uri}$
(which will be used later in~\refprotostep{acg-redir-ep-req}) and a
value $\mi{state}$
(which will serve as a token to prevent CSRF attacks). The RP then redirects the browser to the
so-called \emph{authorization endpoint} URI at the IdP in~\refprotostep{acg-start-resp}
and~\refprotostep{acg-idp-auth-req-1} with its $\mi{client\_id}$,
$\mi{redirect\_uri}$,
and $\mi{state}$
appended as parameters to the URI. The IdP then prompts the user
to provide her username and password
in~\refprotostep{acg-idp-auth-resp-1}. The user's browser sends this
information to the IdP in~\refprotostep{acg-idp-auth-req-2}. If the
credentials are correct, the IdP creates a nonce $\mi{code}$
(the authorization code) and redirects the user's browser to RP's
redirection endpoint URI $\mi{redirect\_uri}$
in~\refprotostep{acg-idp-auth-resp-2}
and~\refprotostep{acg-redir-ep-req} with $\mi{code}$ and $\mi{state}$
appended as parameters to the URI. If $\mi{state}$ is the same as
above, the RP contacts the IdP in~\refprotostep{acg-token-req} and
provides $\mi{code}$, $\mi{client\_id}$, $\mi{client\_secret}$, and
$\mi{redirect\_uri}$.  Then the IdP checks whether this information is
correct, i.e., it checks that $\mi{code}$ was issued for the RP
identified by $\mi{client\_id}$, that $\mi{client\_secret}$ is the
secret for $\mi{client\_id}$, that $\mi{redirect\_uri}$
coincides with the one in Step~\refprotostep{acg-start-resp}, and that
$\mi{code}$ has not been redeemed before. If these checks are
successful, the IdP issues an access token $\mi{access\_token}$
in~\refprotostep{acg-token-resp}. Now, the RP can use
$\mi{access\_token}$ to access the user's protected resources at the
IdP (authorization) or log in the user (authentication), as described
next.

When OAuth is used for \emph{authorization}, the RP uses the access
token to view or manipulate the protected resource at the IdP
(illustrated in Steps~\refprotostep{acg-resource-req}
and~\refprotostep{acg-resource-resp}).

For \emph{authentication}, the RP fetches a user id (which uniquely
identifies the user at the IdP) using the access token,
Steps~\refprotostep{acg-introspect-req}
and~\refprotostep{acg-introspect-resp}.  The RP then
issues a session cookie to the user's browser as shown
in~\refprotostep{acg-redir-ep-resp}.\footnote{Authentication is not part of RFC6749, but this method
  for authentication is commonly used in practice, for example by
  Amazon, Facebook, LinkedIn, and StackExchange, and is also defined
  in OpenID Connect~\cite{openid-connect-core-1.0}.}

\paragraph{Tracking User Intention} Note that in order for an RP which
supports multiple IdPs to process
Step~\refprotostep{acg-redir-ep-req}, the RP must know which IdP a
user wanted to use for authorization. There are two different
approaches to this used in practice: First, the RP can use different
redirection URIs to distinguish different IdPs. We call this
\emph{na\"ive user intention tracking}. Second, the RP can store the
user intention in a session after Step~\refprotostep{acg-start-req}
and use this information later. We call this \emph{explicit user
  intention tracking}. The same applies to the implicit mode of OAuth
presented below.

\subsubsection{Implicit Mode}
This mode is similar to the authorization code mode, but
instead of providing an authorization code, the IdP directly
delivers an access token to the RP via the user's browser.

More specifically, in the implicit mode,
Steps~\refprotostep{acg-start-req}--\refprotostep{acg-idp-auth-req-2}
(see Figure~\ref{fig:oauth-auth-code-grant}) are the same as in the
authorization code mode. Instead of creating an authorization code, the IdP 
issues an access token right away and redirects the user's
browser to RP's redirection endpoint with the access token contained
in the fragment of the URI. (Recall that a fragment is a
special part of a URI indicated by the `\#' symbol.) 

As fragments are not sent in HTTP requests, the access token is not
immediately transferred when the browser contacts the RP. Instead, the RP
needs to use a JavaScript to retrieve the contents of the fragment.
Typically, such a JavaScript is sent in RP's answer at the redirection
endpoint. Just as in the authorization code mode, the RP can now use
the access token for authorization or authentication (analogously to
Steps~\refprotostep{acg-resource-req}--\refprotostep{acg-redir-ep-resp}
of Figure~\ref{fig:oauth-auth-code-grant}).\footnote{The response from
  the IdP in Step~\refprotostep{acg-introspect-resp} includes the RP's
  OAuth client id, which is checked by the RP when
  \emph{authenticating} a user
  (cf.~RFC7662~\cite{rfc7662-oauth-token-introspection}). This check
  prevents re-use of access tokens across RPs in the OAuth implicit
  mode, as explained in~\cite{Wangetal-USENIX-Explicating-SDKs-2013}.
  This check is not needed for authorization.}

\subsubsection{Resource Owner Password Credentials Mode}
In this mode, the user gives her credentials for an IdP directly to an
RP. The RP can then authenticate to the IdP on the user's behalf and
retrieve an access token.  This mode is intended for highly-trusted
RPs, such as the operating system of the user's device or
highly-privileged applications, or if the previous two modes are not
possible to perform (e.g., for applications without a web browser).

\subsubsection{Client Credentials Mode}
In contrast to the modes shown above, this mode works without the
user's interaction. Instead, it is started by an RP in order to fetch
an access token to access the resources of RP at an IdP. For example,
Facebook allows RPs to use the client credentials mode to obtain an
access token to access reports of their advertisements' performance.

%% file: figure-oauth-auth-code-grant.tex
 \scriptsize{ \newlength\blockExtraHeightAIcHFJEedAFceEbcFJdcbBbDdIGfEedFB
\settototalheight\blockExtraHeightAIcHFJEedAFceEbcFJdcbBbDdIGfEedFB{\parbox{0.4\linewidth}{$\mi{idp}$}}
\setlength\blockExtraHeightAIcHFJEedAFceEbcFJdcbBbDdIGfEedFB{\dimexpr \blockExtraHeightAIcHFJEedAFceEbcFJdcbBbDdIGfEedFB - 3ex/4}
\newlength\blockExtraHeightBIcHFJEedAFceEbcFJdcbBbDdIGfEedFB
\settototalheight\blockExtraHeightBIcHFJEedAFceEbcFJdcbBbDdIGfEedFB{\parbox{0.4\linewidth}{Redirect to IdP /authEP with $\mi{client\_id}$, $\mi{redirect\_uri}$, $\mi{state}$}}
\setlength\blockExtraHeightBIcHFJEedAFceEbcFJdcbBbDdIGfEedFB{\dimexpr \blockExtraHeightBIcHFJEedAFceEbcFJdcbBbDdIGfEedFB - 3ex/4}
\newlength\blockExtraHeightCIcHFJEedAFceEbcFJdcbBbDdIGfEedFB
\settototalheight\blockExtraHeightCIcHFJEedAFceEbcFJdcbBbDdIGfEedFB{\parbox{0.4\linewidth}{$\mi{client\_id}$, $\mi{redirect\_uri}$, $\mi{state}$}}
\setlength\blockExtraHeightCIcHFJEedAFceEbcFJdcbBbDdIGfEedFB{\dimexpr \blockExtraHeightCIcHFJEedAFceEbcFJdcbBbDdIGfEedFB - 3ex/4}
\newlength\blockExtraHeightDIcHFJEedAFceEbcFJdcbBbDdIGfEedFB
\settototalheight\blockExtraHeightDIcHFJEedAFceEbcFJdcbBbDdIGfEedFB{\parbox{0.4\linewidth}{}}
\setlength\blockExtraHeightDIcHFJEedAFceEbcFJdcbBbDdIGfEedFB{\dimexpr \blockExtraHeightDIcHFJEedAFceEbcFJdcbBbDdIGfEedFB - 3ex/4}
\newlength\blockExtraHeightEIcHFJEedAFceEbcFJdcbBbDdIGfEedFB
\settototalheight\blockExtraHeightEIcHFJEedAFceEbcFJdcbBbDdIGfEedFB{\parbox{0.4\linewidth}{$\mi{username}$, $\mi{password}$}}
\setlength\blockExtraHeightEIcHFJEedAFceEbcFJdcbBbDdIGfEedFB{\dimexpr \blockExtraHeightEIcHFJEedAFceEbcFJdcbBbDdIGfEedFB - 3ex/4}
\newlength\blockExtraHeightFIcHFJEedAFceEbcFJdcbBbDdIGfEedFB
\settototalheight\blockExtraHeightFIcHFJEedAFceEbcFJdcbBbDdIGfEedFB{\parbox{0.4\linewidth}{Redirect to RP $\mi{redirect\_uri}$ with $\mi{code}$, $\mi{state}$}}
\setlength\blockExtraHeightFIcHFJEedAFceEbcFJdcbBbDdIGfEedFB{\dimexpr \blockExtraHeightFIcHFJEedAFceEbcFJdcbBbDdIGfEedFB - 3ex/4}
\newlength\blockExtraHeightGIcHFJEedAFceEbcFJdcbBbDdIGfEedFB
\settototalheight\blockExtraHeightGIcHFJEedAFceEbcFJdcbBbDdIGfEedFB{\parbox{0.4\linewidth}{$\mi{code}$, $\mi{state}$}}
\setlength\blockExtraHeightGIcHFJEedAFceEbcFJdcbBbDdIGfEedFB{\dimexpr \blockExtraHeightGIcHFJEedAFceEbcFJdcbBbDdIGfEedFB - 3ex/4}
\newlength\blockExtraHeightHIcHFJEedAFceEbcFJdcbBbDdIGfEedFB
\settototalheight\blockExtraHeightHIcHFJEedAFceEbcFJdcbBbDdIGfEedFB{\parbox{0.4\linewidth}{$\mi{code}$, $\mi{client\_id}$, $\mi{redirect\_uri}$, $\mi{client\_secret}$}}
\setlength\blockExtraHeightHIcHFJEedAFceEbcFJdcbBbDdIGfEedFB{\dimexpr \blockExtraHeightHIcHFJEedAFceEbcFJdcbBbDdIGfEedFB - 3ex/4}
\newlength\blockExtraHeightIIcHFJEedAFceEbcFJdcbBbDdIGfEedFB
\settototalheight\blockExtraHeightIIcHFJEedAFceEbcFJdcbBbDdIGfEedFB{\parbox{0.4\linewidth}{$\mi{access\_token}$}}
\setlength\blockExtraHeightIIcHFJEedAFceEbcFJdcbBbDdIGfEedFB{\dimexpr \blockExtraHeightIIcHFJEedAFceEbcFJdcbBbDdIGfEedFB - 3ex/4}
\newlength\blockExtraHeightJIcHFJEedAFceEbcFJdcbBbDdIGfEedFB
\settototalheight\blockExtraHeightJIcHFJEedAFceEbcFJdcbBbDdIGfEedFB{\parbox{0.4\linewidth}{None}}
\setlength\blockExtraHeightJIcHFJEedAFceEbcFJdcbBbDdIGfEedFB{\dimexpr \blockExtraHeightJIcHFJEedAFceEbcFJdcbBbDdIGfEedFB - 3ex/4}
\newlength\blockExtraHeightBAIcHFJEedAFceEbcFJdcbBbDdIGfEedFB
\settototalheight\blockExtraHeightBAIcHFJEedAFceEbcFJdcbBbDdIGfEedFB{\parbox{0.4\linewidth}{$\mi{access\_token}$}}
\setlength\blockExtraHeightBAIcHFJEedAFceEbcFJdcbBbDdIGfEedFB{\dimexpr \blockExtraHeightBAIcHFJEedAFceEbcFJdcbBbDdIGfEedFB - 3ex/4}
\newlength\blockExtraHeightBBIcHFJEedAFceEbcFJdcbBbDdIGfEedFB
\settototalheight\blockExtraHeightBBIcHFJEedAFceEbcFJdcbBbDdIGfEedFB{\parbox{0.4\linewidth}{protected resource}}
\setlength\blockExtraHeightBBIcHFJEedAFceEbcFJdcbBbDdIGfEedFB{\dimexpr \blockExtraHeightBBIcHFJEedAFceEbcFJdcbBbDdIGfEedFB - 3ex/4}
\newlength\blockExtraHeightBCIcHFJEedAFceEbcFJdcbBbDdIGfEedFB
\settototalheight\blockExtraHeightBCIcHFJEedAFceEbcFJdcbBbDdIGfEedFB{\parbox{0.4\linewidth}{None}}
\setlength\blockExtraHeightBCIcHFJEedAFceEbcFJdcbBbDdIGfEedFB{\dimexpr \blockExtraHeightBCIcHFJEedAFceEbcFJdcbBbDdIGfEedFB - 3ex/4}
\newlength\blockExtraHeightBDIcHFJEedAFceEbcFJdcbBbDdIGfEedFB
\settototalheight\blockExtraHeightBDIcHFJEedAFceEbcFJdcbBbDdIGfEedFB{\parbox{0.4\linewidth}{$\mi{access\_token}$}}
\setlength\blockExtraHeightBDIcHFJEedAFceEbcFJdcbBbDdIGfEedFB{\dimexpr \blockExtraHeightBDIcHFJEedAFceEbcFJdcbBbDdIGfEedFB - 3ex/4}
\newlength\blockExtraHeightBEIcHFJEedAFceEbcFJdcbBbDdIGfEedFB
\settototalheight\blockExtraHeightBEIcHFJEedAFceEbcFJdcbBbDdIGfEedFB{\parbox{0.4\linewidth}{$\mi{user\_id}$, $\mi{client\_id}$}}
\setlength\blockExtraHeightBEIcHFJEedAFceEbcFJdcbBbDdIGfEedFB{\dimexpr \blockExtraHeightBEIcHFJEedAFceEbcFJdcbBbDdIGfEedFB - 3ex/4}
\newlength\blockExtraHeightBFIcHFJEedAFceEbcFJdcbBbDdIGfEedFB
\settototalheight\blockExtraHeightBFIcHFJEedAFceEbcFJdcbBbDdIGfEedFB{\parbox{0.4\linewidth}{$\mi{session\_cookie}$}}
\setlength\blockExtraHeightBFIcHFJEedAFceEbcFJdcbBbDdIGfEedFB{\dimexpr \blockExtraHeightBFIcHFJEedAFceEbcFJdcbBbDdIGfEedFB - 3ex/4}

 \begin{tikzpicture}
   \tikzstyle{xhrArrow} = [color=blue,decoration={markings, mark=at
    position 1 with {\arrow[color=blue]{triangle 45}}}, preaction
  = {decorate}]

    \matrix [column sep={6cm,between origins}, row sep=4.5ex]
  {

    \node[draw,anchor=base](Browser-start-0){Browser}; & \node[draw,anchor=base](RP-start-0){RP}; & \node[draw,anchor=base](IdP-start-0){IdP};\\
\node(Browser-0){}; & \node(RP-0){}; & \node(IdP-0){};\\[\blockExtraHeightAIcHFJEedAFceEbcFJdcbBbDdIGfEedFB]
\node(Browser-1){}; & \node(RP-1){}; & \node(IdP-1){};\\[\blockExtraHeightBIcHFJEedAFceEbcFJdcbBbDdIGfEedFB]
\node(Browser-2){}; & \node(RP-2){}; & \node(IdP-2){};\\[\blockExtraHeightCIcHFJEedAFceEbcFJdcbBbDdIGfEedFB]
\node(Browser-3){}; & \node(RP-3){}; & \node(IdP-3){};\\[\blockExtraHeightDIcHFJEedAFceEbcFJdcbBbDdIGfEedFB]
\node(Browser-4){}; & \node(RP-4){}; & \node(IdP-4){};\\[\blockExtraHeightEIcHFJEedAFceEbcFJdcbBbDdIGfEedFB]
\node(Browser-5){}; & \node(RP-5){}; & \node(IdP-5){};\\[\blockExtraHeightFIcHFJEedAFceEbcFJdcbBbDdIGfEedFB]
\node(Browser-6){}; & \node(RP-6){}; & \node(IdP-6){};\\[\blockExtraHeightGIcHFJEedAFceEbcFJdcbBbDdIGfEedFB]
\node(Browser-7){}; & \node(RP-7){}; & \node(IdP-7){};\\[\blockExtraHeightHIcHFJEedAFceEbcFJdcbBbDdIGfEedFB]
\node(Browser-8){}; & \node(RP-8){}; & \node(IdP-8){};\\[\blockExtraHeightIIcHFJEedAFceEbcFJdcbBbDdIGfEedFB]
\node(Browser-9){}; & \node(RP-9){}; & \node(IdP-9){};\\[\blockExtraHeightJIcHFJEedAFceEbcFJdcbBbDdIGfEedFB]
\node(Browser-10){}; & \node(RP-10){}; & \node(IdP-10){};\\[\blockExtraHeightBAIcHFJEedAFceEbcFJdcbBbDdIGfEedFB]
\node(Browser-11){}; & \node(RP-11){}; & \node(IdP-11){};\\[\blockExtraHeightBBIcHFJEedAFceEbcFJdcbBbDdIGfEedFB]
\node(Browser-12){}; & \node(RP-12){}; & \node(IdP-12){};\\[\blockExtraHeightBCIcHFJEedAFceEbcFJdcbBbDdIGfEedFB]
\node(Browser-13){}; & \node(RP-13){}; & \node(IdP-13){};\\[\blockExtraHeightBDIcHFJEedAFceEbcFJdcbBbDdIGfEedFB]
\node(Browser-14){}; & \node(RP-14){}; & \node(IdP-14){};\\[\blockExtraHeightBEIcHFJEedAFceEbcFJdcbBbDdIGfEedFB]
\node(Browser-15){}; & \node(RP-15){}; & \node(IdP-15){};\\[\blockExtraHeightBFIcHFJEedAFceEbcFJdcbBbDdIGfEedFB]
\node[draw,anchor=base](Browser-end-1){/Browser}; & \node[draw,anchor=base](RP-end-1){/RP}; & \node[draw,anchor=base](IdP-end-1){/IdP};\\
};
\draw[->] (Browser-0) to node [above=2.6pt, anchor=base]{\protostep{acg-start-req} \textbf{POST /start}} node [below=-8pt, text width=0.5\linewidth, anchor=base]{\begin{center} $\mi{idp}$\end{center}} (RP-0); 

\draw[->] (RP-1) to node [above=2.6pt, anchor=base]{\protostep{acg-start-resp} \textbf{Response}} node [below=-8pt, text width=0.5\linewidth, anchor=base]{\begin{center} Redirect to IdP /authEP with $\mi{client\_id}$, $\mi{redirect\_uri}$, $\mi{state}$\end{center}} (Browser-1); 

\draw[->] (Browser-2) to node [above=2.6pt, anchor=base]{\protostep{acg-idp-auth-req-1} \textbf{GET /authEP}} node [below=-8pt, text width=0.5\linewidth, anchor=base]{\begin{center} $\mi{client\_id}$, $\mi{redirect\_uri}$, $\mi{state}$\end{center}} (IdP-2); 

\draw[->] (IdP-3) to node [above=2.6pt, anchor=base]{\protostep{acg-idp-auth-resp-1} \textbf{Response}} node [below=-8pt, text width=0.5\linewidth, anchor=base]{\begin{center} \end{center}} (Browser-3); 

\draw[->] (Browser-4) to node [above=2.6pt, anchor=base]{\protostep{acg-idp-auth-req-2} \textbf{POST /authEP}} node [below=-8pt, text width=0.5\linewidth, anchor=base]{\begin{center} $\mi{username}$, $\mi{password}$\end{center}} (IdP-4); 

\draw[->] (IdP-5) to node [above=2.6pt, anchor=base]{\protostep{acg-idp-auth-resp-2} \textbf{Response}} node [below=-8pt, text width=0.5\linewidth, anchor=base]{\begin{center} Redirect to RP $\mi{redirect\_uri}$ with $\mi{code}$, $\mi{state}$\end{center}} (Browser-5); 

\draw[->] (Browser-6) to node [above=2.6pt, anchor=base]{\protostep{acg-redir-ep-req} \textbf{GET $\mi{redirect\_uri}$}} node [below=-8pt, text width=0.5\linewidth, anchor=base]{\begin{center} $\mi{code}$, $\mi{state}$\end{center}} (RP-6); 

\draw[->] (RP-7) to node [above=2.6pt, anchor=base]{\protostep{acg-token-req} \textbf{POST /tokenEP}} node [below=-8pt, text width=0.5\linewidth, anchor=base]{\begin{center} $\mi{code}$, $\mi{client\_id}$, $\mi{redirect\_uri}$, $\mi{client\_secret}$\end{center}} (IdP-7); 

\draw[->] (IdP-8) to node [above=2.6pt, anchor=base]{\protostep{acg-token-resp} \textbf{Response}} node [below=-8pt, text width=0.5\linewidth, anchor=base]{\begin{center} $\mi{access\_token}$\end{center}} (RP-8); 

\draw [dashed] (Browser-9.west) -- (IdP-9.east);
\node[draw=none,anchor=northwest,below=2ex,right=1ex] at (Browser-9.west) {Authorization:};

\draw[->] (RP-10) to node [above=2.6pt, anchor=base]{\protostep{acg-resource-req} \textbf{GET /resource}} node [below=-8pt, text width=0.5\linewidth, anchor=base]{\begin{center} $\mi{access\_token}$\end{center}} (IdP-10); 

\draw[->] (IdP-11) to node [above=2.6pt, anchor=base]{\protostep{acg-resource-resp} \textbf{Response}} node [below=-8pt, text width=0.5\linewidth, anchor=base]{\begin{center} protected resource\end{center}} (RP-11); 

\draw [dashed] (Browser-12.west) -- (IdP-12.east);
\node[draw=none,anchor=northwest,below=2ex,right=1ex] at (Browser-12.west) {Authentication:};

\draw[->] (RP-13) to node [above=2.6pt, anchor=base]{\protostep{acg-introspect-req} \textbf{GET /introspectionEP}} node [below=-8pt, text width=0.5\linewidth, anchor=base]{\begin{center} $\mi{access\_token}$\end{center}} (IdP-13); 

\draw[->] (IdP-14) to node [above=2.6pt, anchor=base]{\protostep{acg-introspect-resp} \textbf{Response}} node [below=-8pt, text width=0.5\linewidth, anchor=base]{\begin{center} $\mi{user\_id}$, $\mi{client\_id}$\end{center}} (RP-14); 

\draw[->] (RP-15) to node [above=2.6pt, anchor=base]{\protostep{acg-redir-ep-resp} \textbf{Response}} node [below=-8pt, text width=0.5\linewidth, anchor=base]{\begin{center} $\mi{session\_cookie}$\end{center}} (Browser-15); 

\begin{pgfonlayer}{background}
\draw [color=gray] (Browser-start-0) -- (Browser-end-1);
\draw [color=gray] (RP-start-0) -- (RP-end-1);
\draw [color=gray] (IdP-start-0) -- (IdP-end-1);
\end{pgfonlayer}
\end{tikzpicture}}

%% file: section-attacks.tex
\section{Attacks}
\label{sec:attacks}

As mentioned in the introduction, while trying to prove the security
of OAuth based on the FKS web model and our OAuth model, we found four
attacks on OAuth, which we call \emph{307 redirect attack}, \emph{IdP
  mix-up attack}, \emph{state leak attack}, and \emph{na\"ive RP
  session integrity attack}, respectively. In this section, we provide
detailed descriptions of these attacks along with easily implementable
fixes. Our formal analysis of OAuth (see Section~\ref{sec:analysis})
then shows that these fixes are indeed sufficient to establish the
security of OAuth. The attacks also apply to OpenID Connect (see
Section~\ref{sec:impl-open-conn}). Figure~\ref{fig:oauth-att-table}
provides an overview of where the attacks apply. We have verified our
attacks on actual implementations of OAuth and OpenID Connect and
reported the attacks to the respective working groups who confirmed
the attacks (see Section~\ref{sec:verification}).

\begin{figure*}
  \centering
  \footnotesize{
  \begin{tabular}[h]{|l|c|c|c|c|c|}
    \cline{2-6}
    \multicolumn{1}{c|}{}& \multicolumn{2}{c|}{attack on OAuth}      & \multicolumn{3}{c|}{applicable to OpenID Connect} \\
    \multicolumn{1}{c|}{}& auth code mode & implicit mode   & auth code mode & implicit mode   & hybrid mode  \\\hline
    307 Redirect Attack   & az + an  & az + an   & az + an   & az + an  & az + an  \\ \hline
    IdP Mix-Up Attack     & az* + an & az + an   & az* + an  & --     & az + an** \\ \hline
    State Leak Attack     & si       & si        & si        & si     & si \\ \hline
    Na\"ive RP Session Integrity Att.     & si & si   & si  & si   & si \\ \hline
  \end{tabular}\\[0.7em]
    \textbf{az:} breaks authorization. \textbf{an:} breaks authentication. \textbf{si:} breaks session integrity. \textbf{--:} not applicable. \textbf{*}~if client secrets are not used. \textbf{**}~restriction: if client secrets are used, either authorization or authentication is broken, depending on implementation details.
}
  \caption{Overview of attacks on OAuth 2.0 and OpenID Connect}

  \label{fig:oauth-att-table}
\end{figure*}

\subsection{307 Redirect Attack}
\label{sec:307-redirect}

In this attack, which breaks our authorization and authentication
properties (see Section~\ref{sec:secur-auth-prop}), the attacker
(running a malicious RP) learns the user's credentials when the user
logs in at an IdP that uses the wrong HTTP redirection status
code. While the attack itself is based on a simple error, to the best
of our knowledge, this is the first description of an attack of this
kind.

\subsubsection{Assumptions} The main assumptions are that (1) the IdP
that is used for the login chooses the 307 HTTP status code when
redirecting the user's browser back to the RP
(Step~\refprotostep{acg-idp-auth-resp-2}
in Figure~\ref{fig:oauth-auth-code-grant}), and (2) the IdP redirects
the user immediately after the user entered her credentials (i.e., in
the response to the HTTP POST request that contains the form data sent
by the user's browser).

\paragraph{Assumption~(1)}
This assumption is reasonable because neither the OAuth
standard~\cite{rfc6749-oauth2} nor the OAuth security
considerations~\cite{rfc6819-oauth2-security} (nor the OpenID Connect
standard~\cite{openid-connect-core-1.0}) specify the exact method of
how to redirect. The OAuth standard rather explicitly permits any HTTP
redirect: 
\begin{quotation}
\noindent
  While the examples in this specification show the use of the HTTP
  302 status code, any other method available via the user-agent to
  accomplish this redirection is allowed and is considered to be an
  implementation detail.
\end{quotation}

\paragraph{Assumption~(2)}
This assumption is reasonable as many examples for redirects
immediately after entering the user credentials can be found in
practice, for example at \nolinkurl{github.com} (where, however,
assumption~(1) is not satisfied.)

\subsubsection{Attack} When a user uses the authorization code or
implicit mode of OAuth to log in at a \emph{malicious} RP, then she is
redirected to the IdP and prompted to enter her credentials. The IdP
then receives these credentials from the user's browser in a POST
request. It checks the credentials and redirects the user's browser to
the RP's redirection endpoint in the response to the POST request.
Since the 307 status code is used for this redirection, the user's
browser will send a POST request to RP that contains all form data
from the previous request, including the user credentials. Since the
RP is run by the attacker, he can use these credentials to impersonate
the user.

\subsubsection{Fix}
Contrary to the current wording in the OAuth standard, the exact
method of the redirect is not an implementation detail but essential
for the security of OAuth. In the HTTP
standard~\cite{rfc7231-http-semantics}, only the 303 redirect is
defined unambiguously to drop the body of an HTTP POST request.
Therefore, the OAuth standard should require 303 redirects for the
steps mentioned above in order to fix this problem.

\subsection{IdP Mix-Up Attack}
\label{sec:malicious-idp-mitm}

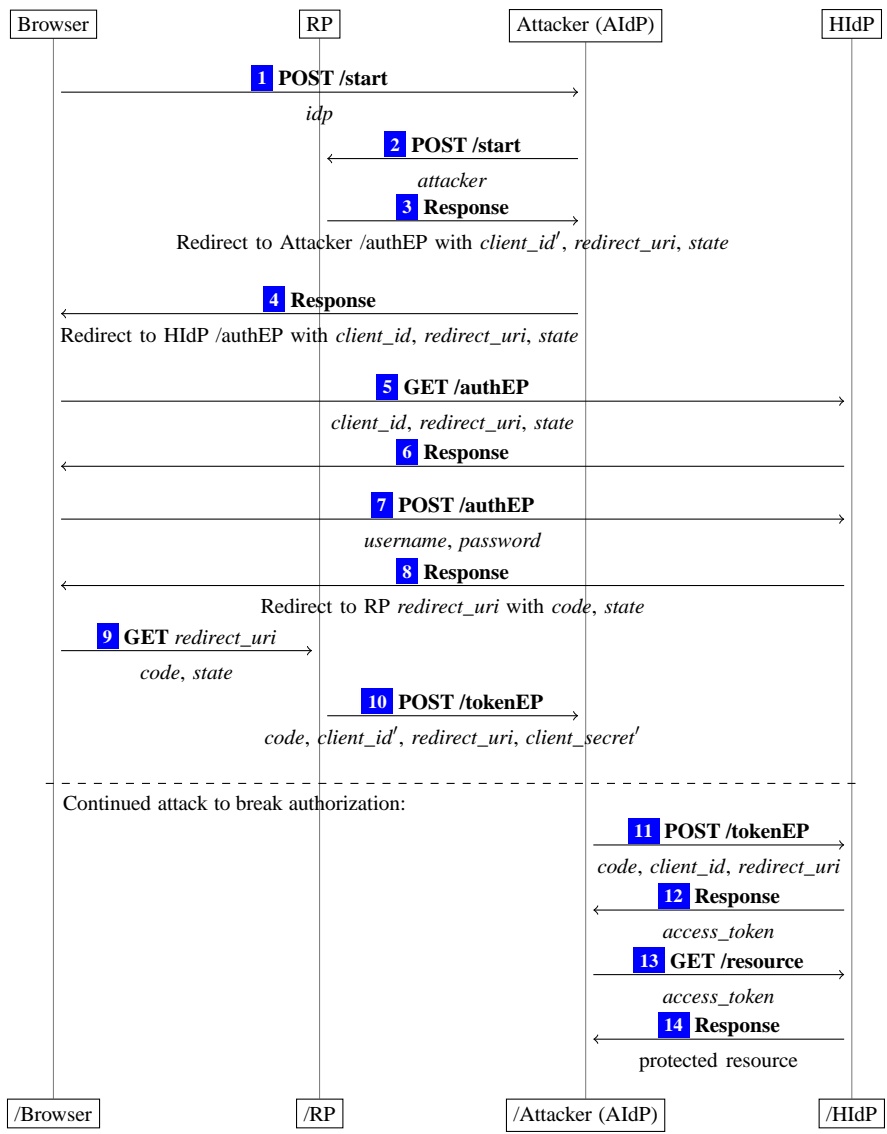
\begin{figure}[t!]
  \centering
  \input{figure-oauth-auth-code-grant-attack}
  \caption{Attack on OAuth 2.0 authorization code mode}
  \label{fig:oauth-acg-att}
  \vspace{-3ex}
\end{figure}

In this attack, which breaks our authorization and
authentication properties (see Section~\ref{sec:secur-auth-prop}), the
attacker confuses an RP about which IdP the user chose at the
beginning of the login/authorization process in order to acquire an
authentication code or access token which can be used to impersonate
the user or access user data.

This attack applies to the authorization code mode and the implicit
mode of OAuth when explicit user intention tracking\footnote{Recall the meaning of ``user intention tracking'' from Section~\ref{sec:oauth}.} is used by the RP.
To launch the attack, the attacker manipulates the first request of
the user such that the RP thinks that the user wants to use an
identity managed by an IdP of the attacker (AIdP) while the user
instead wishes to use her identity managed by an honest IdP (HIdP). As
a result, the RP sends the authorization code or the access token
issued by HIdP to the attacker. The attacker then can use
this information to login at the RP under the user's identity (managed
by HIdP) or access the user's protected resources at HIdP.

We here present the attack in the authorization code mode. In the implicit
mode, the attack is very similar and is shown in detail
in Appendix~\ref{sec:attack-mixup-in-implicit}.

\subsubsection{Assumptions}
\label{sec:attack-model-assumpt}
For the IdP mix-up attack to work, we need three assumptions that we
further discuss below: (1) the presence of a network attacker who can
manipulate the request in which the user sends her identity to the RP
as well as the corresponding response to this request (see
Steps~\refprotostep{acg-start-req} and~\refprotostep{acg-start-resp}
in Figure~\ref{fig:oauth-auth-code-grant}), (2) an
RP which allows users to log in with identities provided by (some)
HIdP and identities provided by AIdP, and (3) an RP that uses explicit
user intention tracking and issues the same redirection URI to all
IdPs.\footnote{Alternatively, the attack would work if the RP issues
  different redirection URIs to different IdPs, but treats them as the
  same URI.} We emphasize that we do not assume that the user sends
any secret (such as passwords) over an unencrypted channel.

\paragraph{Assumption~(1)} It would be unrealistic to assume that a
network attacker can never manipulate
Steps~\refprotostep{acg-start-req} and~\refprotostep{acg-start-resp}
in Figure~\ref{fig:oauth-auth-code-grant}.

First, these messages are sent between the user and the RP, i.e., the
attacker does not need to intercept server-to-server communication. He
could, e.g., use ARP spoofing in a wifi network to mount the attack.

Second, the need for HTTPS for these steps is not obvious to users or
RPs, and the use of HTTPS is not suggested by the OAuth security
recommendations, since the user only selects an IdP at this point;
credentials are not transferred.

Third, even if an RP intends to use HTTPS also for the first request
(as in our model), it has to protect itself against TLS stripping by
adding the RP domain to a browser preloaded Strict Transport Security
(STS) list~\cite{hsts-preloading-form}. Other mitigations, such as the
STS header, can be circumvented (see~\cite{Selvi-Blackhat-2014}), and
do not work on the very first connection between the user's browser
and RP. For example, when a user enters the address of an RP into her
browser, browsers by default try unencrypted connections. It is
therefore unrealistic to assume that all RPs are always protected
against TLS stripping.

Our formal analysis presented in Section~\ref{sec:analysis} shows that
OAuth can be operated securely even if no HTTPS is used for the
initial request (given that our fix, presented below, is applied).

\paragraph{Assumption~(2)} RPs may use different IdPs, some of which
might be malicious, and hence, OAuth should provide security in this
case. Using a technique called dynamic client registration, OAuth RPs can even allow the
ad-hoc use of any IdP, including malicious ones. This is particularly
relevant in OpenID Connect, where this technique was first
implemented. 

\paragraph{Assumption~(3)} 
Typically, RPs that use explicit user intention tracking do not
register different redirection URIs for different IdPs, as in this
case the RP records the IdP a user wants to authenticate with. In
particular, for RPs that allow for dynamic registration, using the
same URI is an obvious implementation choice. This is for
example the case in the OAuth/OpenID Connect implementations
\emph{mod\_auth\_openidc} and \emph{pyoidc} (see below).

\subsubsection{Attack on Authorization Code Mode}
We now describe the IdP Mix-Up attack on the OAuth authorization code
mode. As mentioned, a very similar attack also applies to the 
implicit mode. Both attacks also work if IdP supports just one of
these two modes.

The IdP mix-up attack for the authorization code mode is depicted in 
Figure~\ref{fig:oauth-acg-att}. Just
as in a regular flow, the attack starts when the user selects that she
wants to log in using HIdP (Step~\refprotostep{acg-att-start-req} in
Figure~\ref{fig:oauth-acg-att}). Now, the attacker intercepts the
request intended for the RP and modifies the content of this request
by replacing HIdP by AIdP.\footnote{At this point, the attacker could
  also read the session id for the user's session at RP. Our attack,
  however, is not based on this possibility and works even if the RP
  changes this session id as soon as the user is logged in and the
  connection is protected by HTTPS (a best practice for session
  management).} The response of the
RP~\refprotostep{acg-att-start-resp} (containing a redirect to AIdP)
is then again intercepted and modified by the attacker such that it
redirects the user to
HIdP~\refprotostep{acg-att-start-resp-manipulated}.  The attacker also
replaces the OAuth client id of the RP at AIdP with the client id of
the RP at HIdP (which is public information). (Note that we assume
that from this point on, in accordance with the OAuth security
recommendations, the communication between the user's browser and HIdP
and the RP is encrypted by using HTTPS, and thus, cannot be inspected
or altered by the attacker.) The user then authenticates to HIdP and
is redirected back to the RP~\refprotostep{acg-att-idp-auth-resp-2}.
The RP thinks, due to Step~\refprotostep{acg-start-req-manipulated} of the
attack, that the nonce $\mi{code}$ contained in this redirect was
issued by AIdP, rather than HIdP. The RP therefore now tries to redeem
this nonce for an access token at
AIdP~\refprotostep{acg-att-token-req}, rather than HIdP. This leaks
$\mi{code}$ to the attacker.

\paragraph{Breaking Authorization} If HIdP has not issued an OAuth
client secret to RP during registration, the attacker can now redeem
$\mi{code}$
for an access token at HIdP (in~\refprotostep{acg-att-token-aidp-req}
and~\refprotostep{acg-att-token-resp}).\footnote{In the
  case that RP has to provide a client secret, this would not work in
  this mode (see also Figure~\ref{fig:oauth-att-table}). Recall that
  in this mode, client secrets are optional.} This access token allows
the attacker to access protected resources of the user at HIdP. This
breaks the authorization property (see
Section~\ref{sec:secur-auth-prop}). We note that at this point, the
attacker might even provide false information about the user or her
protected resources to the RP: he could issue a self-created access
token which RP would then use to access such information at the
attacker.

\paragraph{Breaking Authentication}
To break the authentication property (see
Section~\ref{sec:secur-auth-prop}) and impersonate the honest user,
the attacker, after obtaining $\mi{code}$
in Step~\refprotostep{acg-att-token-req}, starts a new login process
(using his own browser) at the RP. He selects HIdP as the IdP for this
login process and receives a redirect to HIdP, which he ignores. This
redirect contains a cookie for a new login session and a fresh state
parameter. The attacker now sends $\mi{code}$
to the RP imitating a real login (using the cookie and fresh state
value from the previous response). The RP then retrieves an access
token at HIdP using $\mi{code}$
and uses this access token to fetch the (honest) user's id. Being
convinced that the attacker owns the honest user's account, the RP
issues a session cookie for this account to the attacker. As a result,
the attacker is logged in at the RP under the honest user's id. (Note
that the attacker does not learn an access token in this case.)

\subsubsection{Variant} There is also a variant of the IdP mix-up
attack that only requires a web attacker (which does not intercept and
manipulate network messages). In this variant, the user wants to log
in with AIdP, but is redirected by AIdP to log in at HIdP; a fact a
vigilant user might detect. 

In detail, the first four steps in Figure~\ref{fig:oauth-acg-att} are
replaced by the following steps: First, the user starts a new OAuth
flow with RP using AIdP. She is then redirected by RP to AIdP's
authorization endpoint. Now, instead of prompting the user for her
password, AIdP redirects the user to HIdP's authorization endpoint.
(Note that, as above, in this step, the attacker uses the state value
he received from the browser plus the client id of RP at HIdP.) From
here on, the attack proceeds exactly as in
Step~\refprotostep{acg-att-idp-auth-req-1} in
Figure~\ref{fig:oauth-acg-att}.

\subsubsection{Related Attacks}  
An attack in the same class, \emph{cross social-network request forgery},
was outlined by Bansal, Bhargavan, Delignat-La\-vaud, and Maffeis in
\cite{BansalBhargavanetal-JCS-2014}. It applies to RPs with na\"ive
user intention tracking (rather than explicit user intention tracking
assumed in our IdP mix-up attack above) in combination with IdPs, such as
Facebook, that only loosely check the redirect URI.\footnote{Facebook,
  by default, only checks the origin of redirect URIs.} Our IdP mix-up
attack works even if an IdP strictly checks redirect URIs. While the
attack in \cite{BansalBhargavanetal-JCS-2014} is described in the
context of concrete social network implementations, our findings show
that this class of attacks is not merely an implementation error, but
a more general problem in the OAuth standard. This was confirmed by
the IETF OAuth Working Group, who, as mentioned, are in the process of
amending the OAuth standard according to our fixes (see Section~\ref{sec:verification}).

Another attack with a similar outcome,
called \emph{Malicious Endpoints Attack}, leveraging the OpenID Connect
Discovery mechanism and therefore limited to OpenID Connect, was
described
in~\cite{MladenovMainkaKrautwaldFeldmannSchwenk-OpenIDConnect-arXiv-2016}. This
attack assumes a CSRF vulnerability on the RP's side.

\subsubsection{Fix} A fundamental problem in the authorization code
and implicit modes of the OAuth standard is a lack of reliable
information in the redirect in
Steps~\refprotostep{acg-idp-auth-resp-2} and
\refprotostep{acg-redir-ep-req} in
Figure~\ref{fig:oauth-auth-code-grant} (even if HTTPS is used). The RP
does not receive information from where the redirect was initiated
(when explicit user intention tracking is used) or receives
information that can easily be spoofed (when na\"ive user intention
tracking is used with IdPs such as Facebook). Hence, the RP cannot
check whether the information contained in the redirect stems from the
IdP that was indicated in Step~\refprotostep{acg-start-req}.

Our fix therefore is to include the identity of
the IdP in the redirect URI in some form that cannot be influenced by
the attacker, e.g., using a new URI parameter. Each IdP should add
such a parameter to the redirect URI.\footnote{The OAuth Working Group
  indeed created a draft for an
  RFC~\cite{rfc-draft-ietf-oauth-mix-up-mitigation-01} that includes
  this fix, where this parameter is called \emph{iss} (issuer).} The
RP can then check that the parameter contains the identity of the IdP
it expects to receive the response from. (This could be used with
either na\"ive or explicit user intention tracking, but to mitigate
the \emph{na\"ive RP session integrity attack} described below, we
advise to use explicit user intention tracking only, see below.)

We show in Section~\ref{sec:analysis} that this fix is indeed
sufficient to mitigate the IdP mix-up attack (as well as the
attacks pointed out in  
\cite{BansalBhargavanetal-JCS-2014,MladenovMainkaKrautwaldFeldmannSchwenk-OpenIDConnect-arXiv-2016}).

\subsection{State Leak Attack}
\label{sec:attack-state-leak}

Using the state leak attack, an attacker can force a browser to be
logged in under the attacker's name at an RP or force an RP to use a
resource of the attacker instead of a resource of the user. This
attack, which breaks our session integrity property (see
Section~\ref{sec:secur-auth-prop}), enables what is often called
session swapping or login
CSRF~\cite{BarthJacksonMitchell-robust-defenses-2008}. 

\subsubsection{Attack} After the user has authenticated to the IdP in
the authorization code mode, the user is redirected to RP
(Step~\refprotostep{acg-redir-ep-req} in
Figure~\ref{fig:oauth-auth-code-grant}). This request contains state
and code as parameters. The response to this request
(Step~\refprotostep{acg-redir-ep-resp}) can be a page containing a
link to the attacker's website or some resource located at the
attacker's website. When the user clicks the link or the resource is
loaded, the user's browser sends a request to the attacker. This
request contains a Referer header with the full URI of the page the
user was redirected to, which in this case contains state and code.

As the state value is supposed to protect the browser's session
against CSRF attacks, the attacker can now use the leaked state value
to perform a CSRF attack against the victim. For example, he can
redirect the victim's browser to the RP's redirection endpoint (again) and
by this, overwrite the previously performed authorization. The user
will then be logged in as the attacker.

Given the history of OAuth, leaks of sensitive data through the
referrer header are not surprising. For example, the fact that the
authorization code can leak through the Referer header was described
as an attack (in a similar setting)
in~\cite{homakov-github-oauth-attack-2014}. Since the authorization
code is single-use only~\cite{rfc6749-oauth2}, it might already
be redeemed by the time it is received by the attacker. State,
however, is not limited to single use, making this attack easier to
exploit in practice. Stealing the state value through the Referer
header to break session integrity has not been reported as an attack
before, as was confirmed by the IETF OAuth Working Group.

\subsubsection{State Leak at IdPs} 
A variant of this attack exists if the login page at an IdP contains
links to external resources. If the user visits this page to
authenticate at the IdP and the browser follows links to external
resources, the state is transferred in the Referer header. This
variant is applicable to the authorization code mode and the implicit mode.

\subsubsection{Fix} We suggest to limit state to a single use and to
use the recently introduced \emph{referrer policies}
\cite{w3c-draft-referrer-policy} to avoid leakage of the state (or
code) to the attacker. Using referrer policies, a web server can
instruct a web browser to (partially or completely) suppress the
Referer header when the browser follows links in or loads resources
for some web page. The Referer header can be blocked entirely, or it
can, for example, be stripped down to the origin of the URI of the web page.
Referrer policies are supported by all modern browsers. 

Our OAuth model includes this fix (such that only the origin is
permitted in the Referer header for links on  web pages of RPs/IdPs) and our
security proof shows its effectiveness (see
Section~\ref{sec:analysis}). The fix also protects the
authorization code from leaking as in the attack described
in~\cite{homakov-github-oauth-attack-2014}.

\subsection{Na\"ive RP Session Integrity Attack}\label{sec:attack-naive-rp}

This attack again breaks the session integrity property for RPs, where here we assume an RP that uses \emph{na\"ive user
intention tracking}.\footnote{Recall the meaning of ``na\"ive user
intention tracking'' from Section~\ref{sec:oauth}.} (Note that we may still assume that the OAuth state parameter is used,
i.e., RP is not necessarily stateless.) 

\subsubsection{Attack} First, an attacker starts a session with HIdP (an
honest IdP) to obtain an authorization code or access token for his
own account.  Next, when a user wants to log in at some RP using AIdP (an IdP
controlled by the attacker), AIdP redirects the user back to the
redirection URI of HIdP at RP. AIdP attaches to this redirection URI
the state issued by RP, and the code or token obtained from HIdP. Now,
since RP performs na\"ive user intention tracking only, the RP then
believes that the user logged in at HIdP.  Hence, the user is logged
in at RP using the attacker's identity at HIdP or the RP accesses the
attacker's resources at HIdP believing that these resources are owned
by the user.

\paragraph{Fix} The fix against the IdP mix-up attack (described
above) does not work in this case: Since RP does not track where the
user wanted to log in, it has to rely on parameters in the redirection
URI which the attacker can easily spoof. Instead, we propose to always
use explicit user intention tracking.

\subsection{Implications to OpenID Connect}
\label{sec:impl-open-conn}

OpenID Connect~\cite{openid-connect-core-1.0} is a standard for
authentication built on top of the OAuth protocol. Among others,
OpenID Connect is used by PayPal, Google, and Microsoft.

All four attacks can be applied to OpenID Connect as well. We here
outline OpenID Connect and how the attacks apply to this protocol. A
detailed description can be found
in Appendix~\ref{app:openid-connect}.

OpenID Connect extends OAuth in several ways, e.g., by additional
security measures. OpenID Connect defines an \emph{authorization code
  mode}, an \emph{implicit mode}, and a \emph{hybrid mode}. The former
two are based on the corresponding OAuth modes and the latter is a
combination of the two modes.

\paragraph{307 Redirect, State Leak, Na\"ive RP Session Integrity Attacks}
All three attacks apply to OpenID Connect in exactly the same way as
described above. The vulnerable steps are identical. 

\paragraph{IdP Mix-Up Attack}
In OpenID Connect, the mix-up attack applies to the authorization code
mode and the hybrid mode. In the authorization code mode, the attack
is very similar to the one on the OAuth authorization code mode. In
the hybrid mode, the attack is more complicated as additional security
measures have to be circumvented by the attacker. In particular, it
must be ensured that the RP does not detect that the issuer of the id
token, a signed cryptographic document used in OpenID Connect, is not the honest IdP.
Interestingly, in the hybrid mode, depending on an implementation
detail of the RP, either authorization or authentication is broken (or
both if no client secret is used).

\subsection{Verification and Disclosure}
\label{sec:verification}
We verified the IdP mix-up and 307 redirect attacks on the Apache
web server module \emph{mod\_auth\_openidc}, an implementation of an
OpenID Connect (and therefore also OAuth) RP. We also verified the IdP
mix-up attack on the python implementation \emph{pyoidc}. We verified
the state leak attack on the current version of the Facebook PHP SDK
and the na\"ive RP session integrity attack on
\nolinkurl{nytimes.com}.\footnote{\emph{mod\_auth\_openidc} and
  \nolinkurl{nytimes.com} are not susceptible to the state leak attack
  since after the login/authorization, the user is immediately redirected to
  another web page at the same RP.}

We reported all attacks to the OAuth and OpenID Connect working groups
who confirmed the attacks. The OAuth working group invited us to
present our findings to them and prepared a draft for an RFC that
mitigates the IdP mix-up attack (using the fix described in
Section~\ref{sec:malicious-idp-mitm})~\cite{rfc-draft-ietf-oauth-mix-up-mitigation-01}.
Fixes regarding the other attacks are currently under discussion. We
also notified \nolinkurl{nytimes.com}, Facebook, and the developers of
\emph{mod\_auth\_openidc} and \emph{pyoidc}.

%% file: figure-oauth-auth-code-grant-attack.tex
 \scriptsize{ \newlength\blockExtraHeightAECBHdHGAJFafEeEJJaJGbcHAHDGdEfGB
\settototalheight\blockExtraHeightAECBHdHGAJFafEeEJJaJGbcHAHDGdEfGB{\parbox{0.4\linewidth}{$\mi{idp}$}}
\setlength\blockExtraHeightAECBHdHGAJFafEeEJJaJGbcHAHDGdEfGB{\dimexpr \blockExtraHeightAECBHdHGAJFafEeEJJaJGbcHAHDGdEfGB - 4ex/4}
\newlength\blockExtraHeightBECBHdHGAJFafEeEJJaJGbcHAHDGdEfGB
\settototalheight\blockExtraHeightBECBHdHGAJFafEeEJJaJGbcHAHDGdEfGB{\parbox{0.4\linewidth}{$\mi{attacker}$}}
\setlength\blockExtraHeightBECBHdHGAJFafEeEJJaJGbcHAHDGdEfGB{\dimexpr \blockExtraHeightBECBHdHGAJFafEeEJJaJGbcHAHDGdEfGB - 4ex/4}
\newlength\blockExtraHeightCECBHdHGAJFafEeEJJaJGbcHAHDGdEfGB
\settototalheight\blockExtraHeightCECBHdHGAJFafEeEJJaJGbcHAHDGdEfGB{\parbox{0.4\linewidth}{Redirect to Attacker /authEP with $\mi{client\_id}'$, $\mi{redirect\_uri}$, $\mi{state}$}}
\setlength\blockExtraHeightCECBHdHGAJFafEeEJJaJGbcHAHDGdEfGB{\dimexpr \blockExtraHeightCECBHdHGAJFafEeEJJaJGbcHAHDGdEfGB - 4ex/4}
\newlength\blockExtraHeightDECBHdHGAJFafEeEJJaJGbcHAHDGdEfGB
\settototalheight\blockExtraHeightDECBHdHGAJFafEeEJJaJGbcHAHDGdEfGB{\parbox{0.4\linewidth}{Redirect to HIdP /authEP with $\mi{client\_id}$, $\mi{redirect\_uri}$, $\mi{state}$}}
\setlength\blockExtraHeightDECBHdHGAJFafEeEJJaJGbcHAHDGdEfGB{\dimexpr \blockExtraHeightDECBHdHGAJFafEeEJJaJGbcHAHDGdEfGB - 4ex/4}
\newlength\blockExtraHeightEECBHdHGAJFafEeEJJaJGbcHAHDGdEfGB
\settototalheight\blockExtraHeightEECBHdHGAJFafEeEJJaJGbcHAHDGdEfGB{\parbox{0.4\linewidth}{$\mi{client\_id}$, $\mi{redirect\_uri}$, $\mi{state}$}}
\setlength\blockExtraHeightEECBHdHGAJFafEeEJJaJGbcHAHDGdEfGB{\dimexpr \blockExtraHeightEECBHdHGAJFafEeEJJaJGbcHAHDGdEfGB - 4ex/4}
\newlength\blockExtraHeightFECBHdHGAJFafEeEJJaJGbcHAHDGdEfGB
\settototalheight\blockExtraHeightFECBHdHGAJFafEeEJJaJGbcHAHDGdEfGB{\parbox{0.4\linewidth}{}}
\setlength\blockExtraHeightFECBHdHGAJFafEeEJJaJGbcHAHDGdEfGB{\dimexpr \blockExtraHeightFECBHdHGAJFafEeEJJaJGbcHAHDGdEfGB - 4ex/4}
\newlength\blockExtraHeightGECBHdHGAJFafEeEJJaJGbcHAHDGdEfGB
\settototalheight\blockExtraHeightGECBHdHGAJFafEeEJJaJGbcHAHDGdEfGB{\parbox{0.4\linewidth}{$\mi{username}$, $\mi{password}$}}
\setlength\blockExtraHeightGECBHdHGAJFafEeEJJaJGbcHAHDGdEfGB{\dimexpr \blockExtraHeightGECBHdHGAJFafEeEJJaJGbcHAHDGdEfGB - 4ex/4}
\newlength\blockExtraHeightHECBHdHGAJFafEeEJJaJGbcHAHDGdEfGB
\settototalheight\blockExtraHeightHECBHdHGAJFafEeEJJaJGbcHAHDGdEfGB{\parbox{0.4\linewidth}{Redirect to RP $\mi{redirect\_uri}$ with $\mi{code}$, $\mi{state}$}}
\setlength\blockExtraHeightHECBHdHGAJFafEeEJJaJGbcHAHDGdEfGB{\dimexpr \blockExtraHeightHECBHdHGAJFafEeEJJaJGbcHAHDGdEfGB - 4ex/4}
\newlength\blockExtraHeightIECBHdHGAJFafEeEJJaJGbcHAHDGdEfGB
\settototalheight\blockExtraHeightIECBHdHGAJFafEeEJJaJGbcHAHDGdEfGB{\parbox{0.4\linewidth}{$\mi{code}$, $\mi{state}$}}
\setlength\blockExtraHeightIECBHdHGAJFafEeEJJaJGbcHAHDGdEfGB{\dimexpr \blockExtraHeightIECBHdHGAJFafEeEJJaJGbcHAHDGdEfGB - 4ex/4}
\newlength\blockExtraHeightJECBHdHGAJFafEeEJJaJGbcHAHDGdEfGB
\settototalheight\blockExtraHeightJECBHdHGAJFafEeEJJaJGbcHAHDGdEfGB{\parbox{0.4\linewidth}{$\mi{code}$, $\mi{client\_id}'$, $\mi{redirect\_uri}$, $\mi{client\_secret}'$}}
\setlength\blockExtraHeightJECBHdHGAJFafEeEJJaJGbcHAHDGdEfGB{\dimexpr \blockExtraHeightJECBHdHGAJFafEeEJJaJGbcHAHDGdEfGB - 4ex/4}
\newlength\blockExtraHeightBAECBHdHGAJFafEeEJJaJGbcHAHDGdEfGB
\settototalheight\blockExtraHeightBAECBHdHGAJFafEeEJJaJGbcHAHDGdEfGB{\parbox{0.4\linewidth}{None}}
\setlength\blockExtraHeightBAECBHdHGAJFafEeEJJaJGbcHAHDGdEfGB{\dimexpr \blockExtraHeightBAECBHdHGAJFafEeEJJaJGbcHAHDGdEfGB - 4ex/4}
\newlength\blockExtraHeightBBECBHdHGAJFafEeEJJaJGbcHAHDGdEfGB
\settototalheight\blockExtraHeightBBECBHdHGAJFafEeEJJaJGbcHAHDGdEfGB{\parbox{0.4\linewidth}{$\mi{code}$, $\mi{client\_id}$, $\mi{redirect\_uri}$}}
\setlength\blockExtraHeightBBECBHdHGAJFafEeEJJaJGbcHAHDGdEfGB{\dimexpr \blockExtraHeightBBECBHdHGAJFafEeEJJaJGbcHAHDGdEfGB - 4ex/4}
\newlength\blockExtraHeightBCECBHdHGAJFafEeEJJaJGbcHAHDGdEfGB
\settototalheight\blockExtraHeightBCECBHdHGAJFafEeEJJaJGbcHAHDGdEfGB{\parbox{0.4\linewidth}{$\mi{access\_token}$}}
\setlength\blockExtraHeightBCECBHdHGAJFafEeEJJaJGbcHAHDGdEfGB{\dimexpr \blockExtraHeightBCECBHdHGAJFafEeEJJaJGbcHAHDGdEfGB - 4ex/4}
\newlength\blockExtraHeightBDECBHdHGAJFafEeEJJaJGbcHAHDGdEfGB
\settototalheight\blockExtraHeightBDECBHdHGAJFafEeEJJaJGbcHAHDGdEfGB{\parbox{0.4\linewidth}{$\mi{access\_token}$}}
\setlength\blockExtraHeightBDECBHdHGAJFafEeEJJaJGbcHAHDGdEfGB{\dimexpr \blockExtraHeightBDECBHdHGAJFafEeEJJaJGbcHAHDGdEfGB - 4ex/4}
\newlength\blockExtraHeightBEECBHdHGAJFafEeEJJaJGbcHAHDGdEfGB
\settototalheight\blockExtraHeightBEECBHdHGAJFafEeEJJaJGbcHAHDGdEfGB{\parbox{0.4\linewidth}{protected resource}}
\setlength\blockExtraHeightBEECBHdHGAJFafEeEJJaJGbcHAHDGdEfGB{\dimexpr \blockExtraHeightBEECBHdHGAJFafEeEJJaJGbcHAHDGdEfGB - 4ex/4}

 \begin{tikzpicture}
   \tikzstyle{xhrArrow} = [color=blue,decoration={markings, mark=at
    position 1 with {\arrow[color=blue]{triangle 45}}}, preaction
  = {decorate}]

    \matrix [column sep={3.5cm,between origins}, row sep=4.5ex]
  {

    \node[draw,anchor=base](Browser-start-0){Browser}; & \node[draw,anchor=base](RP-start-0){RP}; & \node[draw,anchor=base](Attacker-start-0){Attacker (AIdP)}; & \node[draw,anchor=base](HIdP-start-0){HIdP};\\
\node(Browser-0){}; & \node(RP-0){}; & \node(Attacker-0){}; & \node(HIdP-0){};\\[\blockExtraHeightAECBHdHGAJFafEeEJJaJGbcHAHDGdEfGB]
\node(Browser-1){}; & \node(RP-1){}; & \node(Attacker-1){}; & \node(HIdP-1){};\\[\blockExtraHeightBECBHdHGAJFafEeEJJaJGbcHAHDGdEfGB]
\node(Browser-2){}; & \node(RP-2){}; & \node(Attacker-2){}; & \node(HIdP-2){};\\[\blockExtraHeightCECBHdHGAJFafEeEJJaJGbcHAHDGdEfGB]
\node(Browser-3){}; & \node(RP-3){}; & \node(Attacker-3){}; & \node(HIdP-3){};\\[\blockExtraHeightDECBHdHGAJFafEeEJJaJGbcHAHDGdEfGB]
\node(Browser-4){}; & \node(RP-4){}; & \node(Attacker-4){}; & \node(HIdP-4){};\\[\blockExtraHeightEECBHdHGAJFafEeEJJaJGbcHAHDGdEfGB]
\node(Browser-5){}; & \node(RP-5){}; & \node(Attacker-5){}; & \node(HIdP-5){};\\[\blockExtraHeightFECBHdHGAJFafEeEJJaJGbcHAHDGdEfGB]
\node(Browser-6){}; & \node(RP-6){}; & \node(Attacker-6){}; & \node(HIdP-6){};\\[\blockExtraHeightGECBHdHGAJFafEeEJJaJGbcHAHDGdEfGB]
\node(Browser-7){}; & \node(RP-7){}; & \node(Attacker-7){}; & \node(HIdP-7){};\\[\blockExtraHeightHECBHdHGAJFafEeEJJaJGbcHAHDGdEfGB]
\node(Browser-8){}; & \node(RP-8){}; & \node(Attacker-8){}; & \node(HIdP-8){};\\[\blockExtraHeightIECBHdHGAJFafEeEJJaJGbcHAHDGdEfGB]
\node(Browser-9){}; & \node(RP-9){}; & \node(Attacker-9){}; & \node(HIdP-9){};\\[\blockExtraHeightJECBHdHGAJFafEeEJJaJGbcHAHDGdEfGB]
\node(Browser-10){}; & \node(RP-10){}; & \node(Attacker-10){}; & \node(HIdP-10){};\\[\blockExtraHeightBAECBHdHGAJFafEeEJJaJGbcHAHDGdEfGB]
\node(Browser-11){}; & \node(RP-11){}; & \node(Attacker-11){}; & \node(HIdP-11){};\\[\blockExtraHeightBBECBHdHGAJFafEeEJJaJGbcHAHDGdEfGB]
\node(Browser-12){}; & \node(RP-12){}; & \node(Attacker-12){}; & \node(HIdP-12){};\\[\blockExtraHeightBCECBHdHGAJFafEeEJJaJGbcHAHDGdEfGB]
\node(Browser-13){}; & \node(RP-13){}; & \node(Attacker-13){}; & \node(HIdP-13){};\\[\blockExtraHeightBDECBHdHGAJFafEeEJJaJGbcHAHDGdEfGB]
\node(Browser-14){}; & \node(RP-14){}; & \node(Attacker-14){}; & \node(HIdP-14){};\\[\blockExtraHeightBEECBHdHGAJFafEeEJJaJGbcHAHDGdEfGB]
\node[draw,anchor=base](Browser-end-1){/Browser}; & \node[draw,anchor=base](RP-end-1){/RP}; & \node[draw,anchor=base](Attacker-end-1){/Attacker (AIdP)}; & \node[draw,anchor=base](HIdP-end-1){/HIdP};\\
};
\draw[->] (Browser-0) to node [above=2.6pt, anchor=base]{\protostep{acg-att-start-req} \textbf{POST /start}} node [below=-8pt, text width=0.5\linewidth, anchor=base]{\begin{center} $\mi{idp}$\end{center}} (Attacker-0); 

\draw[->] (Attacker-1) to node [above=2.6pt, anchor=base]{\protostep{acg-start-req-manipulated} \textbf{POST /start}} node [below=-8pt, text width=0.5\linewidth, anchor=base]{\begin{center} $\mi{attacker}$\end{center}} (RP-1); 

\draw[->] (RP-2) to node [above=2.6pt, anchor=base]{\protostep{acg-att-start-resp} \textbf{Response}} node [below=-8pt, text width=0.5\linewidth, anchor=base]{\begin{center} Redirect to Attacker /authEP with $\mi{client\_id}'$, $\mi{redirect\_uri}$, $\mi{state}$\end{center}} (Attacker-2); 

\draw[->] (Attacker-3) to node [above=2.6pt, anchor=base]{\protostep{acg-att-start-resp-manipulated} \textbf{Response}} node [below=-8pt, text width=0.5\linewidth, anchor=base]{\begin{center} Redirect to HIdP /authEP with $\mi{client\_id}$, $\mi{redirect\_uri}$, $\mi{state}$\end{center}} (Browser-3); 

\draw[->] (Browser-4) to node [above=2.6pt, anchor=base]{\protostep{acg-att-idp-auth-req-1} \textbf{GET /authEP}} node [below=-8pt, text width=0.5\linewidth, anchor=base]{\begin{center} $\mi{client\_id}$, $\mi{redirect\_uri}$, $\mi{state}$\end{center}} (HIdP-4); 

\draw[->] (HIdP-5) to node [above=2.6pt, anchor=base]{\protostep{acg-att-idp-auth-resp-1} \textbf{Response}} node [below=-8pt, text width=0.5\linewidth, anchor=base]{\begin{center} \end{center}} (Browser-5); 

\draw[->] (Browser-6) to node [above=2.6pt, anchor=base]{\protostep{acg-att-idp-auth-req-2} \textbf{POST /authEP}} node [below=-8pt, text width=0.5\linewidth, anchor=base]{\begin{center} $\mi{username}$, $\mi{password}$\end{center}} (HIdP-6); 

\draw[->] (HIdP-7) to node [above=2.6pt, anchor=base]{\protostep{acg-att-idp-auth-resp-2} \textbf{Response}} node [below=-8pt, text width=0.5\linewidth, anchor=base]{\begin{center} Redirect to RP $\mi{redirect\_uri}$ with $\mi{code}$, $\mi{state}$\end{center}} (Browser-7); 

\draw[->] (Browser-8) to node [above=2.6pt, anchor=base]{\protostep{acg-att-redir-ep-req} \textbf{GET $\mi{redirect\_uri}$}} node [below=-8pt, text width=0.5\linewidth, anchor=base]{\begin{center} $\mi{code}$, $\mi{state}$\end{center}} (RP-8); 

\draw[->] (RP-9) to node [above=2.6pt, anchor=base]{\protostep{acg-att-token-req} \textbf{POST /tokenEP}} node [below=-8pt, text width=0.5\linewidth, anchor=base]{\begin{center} $\mi{code}$, $\mi{client\_id}'$, $\mi{redirect\_uri}$, $\mi{client\_secret}'$\end{center}} (Attacker-9); 

\draw [dashed] (Browser-10.west) -- (HIdP-10.east);
\node[draw=none,anchor=northwest,below=2ex,right=1ex] at (Browser-10.west) {Continued attack to break authorization:};

\draw[->] (Attacker-11) to node [above=2.6pt, anchor=base]{\protostep{acg-att-token-aidp-req} \textbf{POST /tokenEP}} node [below=-8pt, text width=0.5\linewidth, anchor=base]{\begin{center} $\mi{code}$, $\mi{client\_id}$, $\mi{redirect\_uri}$\end{center}} (HIdP-11); 

\draw[->] (HIdP-12) to node [above=2.6pt, anchor=base]{\protostep{acg-att-token-resp} \textbf{Response}} node [below=-8pt, text width=0.5\linewidth, anchor=base]{\begin{center} $\mi{access\_token}$\end{center}} (Attacker-12); 

\draw[->] (Attacker-13) to node [above=2.6pt, anchor=base]{\protostep{acg-att-resource-req} \textbf{GET /resource}} node [below=-8pt, text width=0.5\linewidth, anchor=base]{\begin{center} $\mi{access\_token}$\end{center}} (HIdP-13); 

\draw[->] (HIdP-14) to node [above=2.6pt, anchor=base]{\protostep{acg-att-resource-resp} \textbf{Response}} node [below=-8pt, text width=0.5\linewidth, anchor=base]{\begin{center} protected resource\end{center}} (Attacker-14); 

\begin{pgfonlayer}{background}
\draw [color=gray] (Browser-start-0) -- (Browser-end-1);
\draw [color=gray] (RP-start-0) -- (RP-end-1);
\draw [color=gray] (Attacker-start-0) -- (Attacker-end-1);
\draw [color=gray] (HIdP-start-0) -- (HIdP-end-1);
\end{pgfonlayer}
\end{tikzpicture}}

%% file: section-model.tex
\section{FKS Model}
\label{sec:model}

Our formal security analysis of OAuth is based on a slightly  extended version
(see Section~\ref{sec:model-1}) of the FKS model, a general Dolev-Yao
(DY) style web model proposed by Fett et al.
in~\cite{FettKuestersSchmitz-SP-2014,FettKuestersSchmitz-CCS-2015}.
This model is designed independently of a specific web application and
closely mimics published (de-facto) standards and specifications for
the web, for example, the HTTP/1.1 and HTML5 standards and associated
(proposed) standards. The FKS model defines a general communication
model, and, based on it, web systems consisting of web browsers, DNS
servers, and web servers as well as web and network attackers. Here,
we only briefly recall the FKS model (see
\cite{FettKuestersSchmitz-SP-2014,FettKuestersSchmitz-CCS-2015} for a
full description, comparison with other models, and a discussion of
its limitations); see also Appendices~\ref{app:web-model}--\ref{app:deta-descr-brows}. 

\paragraph{Communication Model}
The main entities in the model are \emph{(atomic) processes}, which
are used to model browsers, servers, and attackers. Each process
listens to one or more (IP) addresses.
Processes communicate via \emph{events}, which consist of a message as
well as a receiver and a sender address. In every step of a run, one
event is chosen non-deterministically from a ``pool'' of waiting
events and is delivered to one of the processes that listens to the
event's receiver address. The process can then handle the event and
output new events, which are added to the pool of events, and so on.

As usual in DY models (see, e.g., \cite{AbadiFournet-POPL-2001}),
messages are expressed as formal terms over a signature $\Sigma$.  The
signature contains constants (for (IP) addresses, strings, nonces) as
well as sequence, projection, and function symbols (e.g., for
encryption/decryption and signatures). For example, in the web model,
an HTTP request is represented as a term $r$ containing a nonce, an
HTTP method, a domain name, a path, URI parameters, headers,
and a message body. For example, a request for the URI
\url{http://example.com/s?p=1} is represented as 
\[\mi{r} :=\! \langle
  \cHttpReq, n_1, \mGet, \str{example.com}, \str{/s},
  \an{\an{\str{p},1}}, \an{}, \an{} \rangle\] where the body and the
headers are empty. An HTTPS request for $r$
is of the form
$\ehreqWithVariable{r}{k'}{\pub(k_\text{example.com})}$,
where $k'$
is a fresh symmetric key (a nonce) generated by the sender of the
request (typically a browser); the responder is supposed to use this
key to encrypt the response.

The \emph{equational theory} associated
with $\Sigma$
is defined as usual in DY models. The theory induces a congruence
relation $\equiv$
on terms, capturing the meaning of the function symbols in $\Sigma$.
For instance, the equation in the equational theory which captures
asymmetric decryption is $\dec{\enc x{\pub(y)}}{y}=x$.
With this, we have that, for example, \[\dec{\ehreqWithVariable{r}{k'}{\pub(k_\text{example.com})}}{k_\text{example.com}}\equiv
  \an{r,k'}\]
i.e., these two terms are equivalent w.r.t.~the equational theory.

A \emph{(DY) process} consists of a set of addresses the
process listens to, a set of states (terms), an initial state, and a
relation that takes an event and a state as input and
(non-deterministically) returns a new state and a sequence of events.
The relation models a computation step of the
process.
It is required that the output can be computed (more formally, derived
in the usual DY style) from the input event and the state.

The so-called \emph{attacker process} is a DY process which
records all messages it receives and outputs all events it can
possibly derive from its recorded messages. Hence, an attacker process
carries out all attacks any DY process could possibly perform.
Attackers can corrupt other parties.

A \emph{script} models JavaScript running in a browser. Scripts are
defined similarly to DY processes. When triggered by a browser, a
script is provided with state information. The script then outputs a
term representing a new internal state and a command to be interpreted
by the browser (see also the specification of browsers
below). Similarly to an attacker process, the so-called \emph{attacker
  script} may output everything that is derivable from the input.

A \emph{system} is a set of  processes. A \emph{configuration}
of this system consists of the states of all  processes
in the system, the pool of waiting events, and a sequence of unused
nonces. Systems induce \emph{runs}, i.e., sequences of 
configurations, where each configuration is obtained by delivering one
of the waiting events of the preceding configuration to a
process, which then performs a computation step.

A \emph{web system} formalizes the web infrastructure and
web applications. It contains a system consisting of honest and attacker
processes. Honest processes can be web browsers, web servers, or DNS
servers. Attackers can be either \emph{web attackers} (who can listen
to and send messages from their own addresses only) or \emph{network
  attackers} (who may listen to and spoof all addresses and therefore
are the most powerful attackers). A web system further contains a set of
scripts (comprising honest scripts and the attacker script).

In our analysis of OAuth, we consider either one network attacker or a
set of web attackers (see Section~\ref{sec:analysis}). In our OAuth
model, we need to specify only the behavior of servers and scripts.
These are not defined by the FKS model since they depend on the
specific application, unless they are corrupt or become corrupted in
which case they behave like attacker processes and attacker scripts;
browsers are specified by the FKS model (see below). The modeling of
OAuth servers and scripts is outlined in Section~\ref{sec:model-1} and
defined in detail in Appendices~\ref{app:model-oauth-auth}
and~\ref{app:model-oauth-auth-webattackers}.

\paragraph{Web Browsers}
\label{sec:web-browsers} 
An honest browser is thought to be used by one honest user, who is
modeled as part of the browser. User actions, such as following a link, are
modeled as non-deterministic actions of the web browser. User
credentials are stored in the initial state of the browser and are
given to selected web pages when needed. Besides user credentials, the
state of a web browser contains (among others) a tree of windows and
documents, cookies, and web storage data (localStorage and
sessionStorage).

A \emph{window} inside a browser contains a set of
\emph{documents} (one being active at any time), modeling the
history of documents presented in this window. Each represents one
loaded web page and contains (among others) a script and a list of
subwindows (modeling iframes). The script, when triggered by the
browser, is provided with all data it has access to, such as a
(limited) view on other documents and windows, certain cookies, and
web storage data. Scripts then output a command and a new state. This
way, scripts can navigate or create windows, send \xhrs and
postMessages, submit forms, set/change cookies and web storage data,
and create iframes. Navigation and security rules ensure that scripts
can manipulate only specific aspects of the browser's state, according
to the web standards.

A browser can output messages on the network of different types,
namely DNS and HTTP(S) requests as well as XHRs, and it processes the responses. Several HTTP(S) headers are modeled,
including, for example, cookie, location, strict transport security
(STS), and origin headers. A browser, at any time, can also receive a
so-called trigger message upon which the browser non-de\-ter\-min\-is\-tically
chooses an action, for instance, to trigger a script in some
document. The script now outputs a command, as described above, which is then further processed by the browser. Browsers can
also become corrupted, i.e., be taken over by web and network
attackers. Once corrupted, a browser behaves like an attacker process.

%% file: section-analysis.tex
\section{Analysis}
\label{sec:analysis}

We now present our security analysis of OAuth (with the fixes
mentioned in Section~\ref{sec:attacks} applied). We first present our
model of OAuth. We then formalize the security properties and state the main theorem, namely
the security of OAuth w.r.t.~these properties. We provide full
details of the model and our proof in Appendices~\ref{app:model-oauth-auth}--\ref{app:proof-oauth}.

\subsection{Model}
\label{sec:model-1}

As mentioned above, our model for OAuth is based on the FKS model
outlined in Section~\ref{sec:model}. For the analysis, we extended the
model to include HTTP Basic Authentication~\cite{rfc2617-http-authentication} and Referrer
Policies~\cite{w3c-draft-referrer-policy} (the Referer header itself was already part of the
model). We developed the OAuth model to adhere to RFC6749, the
OAuth~2.0 standard, and follow the security considerations described
in \cite{rfc6819-oauth2-security}.

\subsubsection{Design}
Our comprehensive model of OAuth includes all configuration options of
OAuth and makes as few assumptions as possible in order to strengthen
our security results: 

\paragraph{OAuth Modes} Every RP and IdP may run any of the  
four OAuth modes, even simultaneously.

\paragraph{Corruption} RPs, IdPs, and browsers can be corrupted by the
attacker at any time.

\paragraph{Redirection URIs} RP chooses redirection URIs explicitly
or the IdP selects a redirection URI that was registered before.
Redirection URIs can contain patterns. This covers all cases specified
in the OAuth standard. We allow that IdPs do not strictly check the
redirection URIs, and instead apply loose checking, i.e., only the
origin is checked (this is the default for Facebook, for
example). This only strengthens the security guarantees we prove. 

\paragraph{Client Secrets} Just as in the OAuth standard, RPs can,
for a certain IdP, have a secret or not have a secret in our model.

\paragraph{Usage of HTTP and HTTPS} Users may visit HTTP and HTTPS
URIs (e.g., for RPs) and parties are not required to use Strict-Transport-Security
(STS), although we still recommend STS in practice (for example, to
reduce the risk of password eavesdropping). Again, this only strengthens our results. 

\paragraph{General User Interaction} As usual in the FKS model, the
user can at any time navigate backwards or forward in her browser
history, navigate to any web page, open multiple windows, start
simultaneous login flows using different or the same IdPs, etc. Web
pages at RPs can contain regular links to arbitrary external web
sites.

\paragraph{Authentication at IdP} User authentication at the
IdP, which is out of the scope of OAuth, is performed using username
and password.

\paragraph{Session Mechanism at RP} OAuth does not prescribe a
specific session mechanism to be used at an RP. Our model therefore
includes a standard cookie-based session mechanism (as suggested in
\cite{rfc-draft-oauth-jwt-encoded-state}).

\subsubsection{Attack Mitigations}
To prove the security properties of OAuth, our model includes the
fixes against the new attacks presented in Section~\ref{sec:attacks}
as well as standard mitigations against known attacks. Altogether this
offers clear implementation guidelines, without which OAuth would be
insecure:

\paragraph{Honest Parties} RPs and IdPs, as long as they are honest,
do not include (untrusted) third-party JavaScript on their websites,
do not contain open redirectors, and do not have Cross-Site Scripting
vulnerabilities. Otherwise, access tokens and authorization codes can
be stolen in various ways, as described, among others,
in~\cite{rfc6749-oauth2,rfc6819-oauth2-security,BansalBhargavanetal-JCS-2014,SantsaiBeznosov-CCS-2012-OAuth}.

\paragraph{CSRF Protection} The $\mi{state}$
parameter is used with a nonce that is bound to the user's session
(see~\cite{rfc-draft-oauth-jwt-encoded-state}) to prevent
CSRF vulnerabilities on the RP redirection
endpoint. Omitting or incorrectly using this parameter can lead to
attacks described
in~\cite{rfc6749-oauth2,rfc6819-oauth2-security,BansalBhargavanetal-JCS-2014,SantsaiBeznosov-CCS-2012-OAuth,LiMitchell-ISC-2014}.

More specifically, a new state nonce is freshly chosen for each login
attempt. Otherwise, the following attack is applicable: First, a user starts an OAuth flow at some
RP using a malicious IdP. The IdP learns the state value that is used
in the current user session. Then, as soon as the user starts a new
OAuth flow with the same RP and an honest IdP, the malicious IdP can
use the known state value to mount a CSRF attack, breaking the session
integrity property.\footnote{Note that in this attack, the state value
  does not leak unintentionally (in contrast to the state leak
  attack). Also note that this attack and the mitigation we describe here, while not
  surprising, do not seem to have been explicitly documented so far.
  For example, \nolinkurl{nytimes.com} is vulnerable also to this
  attack.}

We also model CSRF protection for some URIs as follows: For RPs, we
model origin header checking\footnote{The origin header is added to
  certain HTTP(S) requests by browsers to declare the origin of the
  document that caused the request. For example, when a user submits a
  form loaded from the URI \nolinkurl{http://a/form} and this form
  is sent to \nolinkurl{http://b/path} then the browser will add
  the origin header \nolinkurl{http://a} in the request to
  \nolinkurl{b}. All modern browsers support origin headers. See \cite{w3c/cors} for details.} (1) at the URI
where the OAuth flow is started (for the implicit and authorization
code mode), (2) at the password login for the resource owner password
credentials mode, and (3) at the URI to which the JavaScript posts the
access token in the implicit mode. For IdPs, we do the same at the URI
to which the username and password pairs are posted. The CSRF
protection of these four URIs is out of the scope of OAuth and
therefore, we follow good web development practices by checking the
origin header. Without this or similar CSRF protection, IdPs and RPs
would be vulnerable to CSRF attacks described
in~\cite{SantsaiBeznosov-CCS-2012-OAuth,BansalBhargavanetal-JCS-2014}.

\paragraph{Referrer Policy and Status Codes} RPs and IdPs use the
Referrer Policy \cite{w3c-draft-referrer-policy} to specify that
Referer headers on links from any of their web pages may not contain
more than the origin of the respective page. Otherwise, RPs or IdPs
would be vulnerable to the state leak attack described in
Section~\ref{sec:attack-state-leak} and the code leak attack described
in~\cite{homakov-github-oauth-attack-2014}. IdPs use 303 redirects
following our fix described in Section~\ref{sec:307-redirect}.

\paragraph{HTTPS Endpoints} All endpoint URIs use HTTPS to protect
against attackers eavesdropping on tokens or manipulating messages
(see, e.g.,
\cite{rfc6819-oauth2-security,SantsaiBeznosov-CCS-2012-OAuth}).
Obviously, IdPs or RPs do not register URIs that point to servers
other than their own. (Otherwise, access tokens or authorization codes
can be stolen trivially.)

\paragraph{Session Cookies} Cookies are always set with the
\emph{secure} attribute, ensuring that the cookie value is only
transmitted over HTTPS. Otherwise, a network attacker could read cookie values by eavesdropping on non-HTTPS connections
to RPs. After successful login at an RP, the RP creates a fresh
session id for that user. Otherwise, a network attacker could set a login session cookie that is
bound to a known state value into the user's browser (see
\cite{Zhengetal-cookies-usenix-2015}), lure the user into logging in
at the corresponding RP, and then use the session cookie to access the
user's data at the RP (\emph{session fixation}, see
\cite{owasp-session-fixation}).

\paragraph{Authentication to the IdP} It is assumed that the user only
ever sends her password over an encrypted channel and only to the IdP
this password was chosen for (or to trusted RPs, as mentioned above).
(The user also does not re-use her password for different IdPs.)
Otherwise, a malicious IdP would be able to use the account of the
user at an honest IdP.

\paragraph{Authentication using Access Tokens} When an RP sends an
access token to the introspection endpoint of an IdP for
authentication (Step~\refprotostep{acg-introspect-req} in
Figure~\ref{fig:oauth-auth-code-grant}), the IdP returns the user
identifier and the client id for which the access token was issued
(Step~\refprotostep{acg-introspect-resp}). The RP must check that the
returned client id is its own, otherwise a
malicious RP could impersonate an honest user at an honest RP (see
\cite{Wangetal-USENIX-Explicating-SDKs-2013,rfc6749-oauth2}). We
therefore require this check. 

\paragraph{User Intention Tracking} We use explicit user intention
tracking. Otherwise, the attack described in
Section~\ref{sec:attack-naive-rp} can be applied.

\subsubsection{Concepts Used in Our Model}
In our model and the security properties, we use the following
concepts:

\paragraph{Protected Resources}
Closely following RFC6749~\cite{rfc6749-oauth2}, OAuth protected
resources are an abstract concept for any resource an RP could use at
an IdP after successful authorization. For example, if Facebook gives
access to the friends list of a user to an RP, this would be
considered a protected resource. In our model,
there is a mapping from (IdP, RP, identity) to nonces (which model protected resources). In this mapping,
the identity part can be $\bot$,
modeling a resource that is acquired in the client credentials mode
and thus not bound to a user.

\paragraph{Service Tokens}
When OAuth is used for authentication, we assume that after successful
login, the RP sends a \emph{service token} to the browser. The
intuition is that with this service token a user can use the services
of the RP. The service token consists of a nonce, the user's
identifier, and the domain of the IdP which was used in the login
process. The service token is a generic model for any session
mechanism the RP could use to track the user's login status (e.g., a
cookie). We note that the actual session mechanism
used by the RP \emph{after} a successful login is out of the scope of
OAuth, which is why we use the generic concept of a service token. In
our model, the service token is delivered by an RP to a browser as a cookie.

\paragraph{Trusted RPs}
In our model, among others, a browser can choose to launch the
resource owner password credentials mode with any RP, causing this RP
to know the password of the user. RPs, however, can become corrupted
and thus leak the password to the attacker. Therefore, to define the
security properties, we define the concept of \emph{trusted RPs}.
Intuitively, this is a set of RPs a user entrusts with her password.
In particular, whether an RP is trusted depends on the user. In our
security properties, when we state that an adversary should not be
able to impersonate a user $u$
in a run, we would assume that all trusted RPs of $u$
have not become corrupted in this run.

\subsubsection{OAuth Web System with a Network Attacker}
\label{sec:proc-oauthw}
We model OAuth as a class of web systems (in the sense of
Section~\ref{sec:model}) that can contain an unbounded finite number
of RPs, IdPs, and browsers. We call a web system $\oauthwebsystem^n$
an \emph{OAuth web system with a network attacker} if it is of the
form described in what follows.

\paragraph{Outline}
The system consists of a network attacker, a finite set of web
browsers, a finite set of web servers for the RPs, and a finite set of
web servers for the IdPs. Recall that in $\oauthwebsystem^n$,
since we have a network attacker, we do not need to consider web
attackers (as our network attacker subsumes all web attackers). The
set of scripts consists of the three scripts $\mi{script\_rp\_index}$,
$\mi{script\_rp\_implicit}$,
and $\mi{script\_idp\_form}$.
We now briefly sketch RPs, IdPs, and the scripts, with full details
provided in Appendix~\ref{app:model-oauth-auth}.

\paragraph{Relying Parties}
Each RP is a web server modeled as an atomic DY process following the
description in Section~\ref{sec:oauth}, including all OAuth modes, as
well as the fixes and mitigations discussed before. The RP can either
(at any time) launch a client credentials mode flow or wait for users
to start any of the other flows. RP manages two kinds of sessions: The
\emph{login sessions}, which are used only during the user login
phase, and the \emph{service sessions} (modeled by a \emph{service
  token} as described above). When receiving a special message, an RP can become corrupted and then
behaves like an attacker process.

\paragraph{Identity Providers}
Each IdP is a web server modeled as an atomic DY process following the
description in Section~\ref{sec:oauth}, again including all
OAuth modes, as well as the fixes and mitigations discussed
before. Users can authenticate to an IdP with their credentials. Just
as RPs, IdPs can become corrupted at any time.

\paragraph{Scripts} The scripts which run in a user's browser are
defined as follows: The script \emph{script\_rp\_index} is loaded from
an RP into a user's browser when the user visits the RP's web site. It
starts the authorization or login process. 
The script \emph{script\_rp\_implicit} is loaded into the user's
browser from an RP during an implicit mode flow to retrieve the data
from the URI fragment. It extracts the access token and state from the
fragment part of its own URI. The script then sends this information
in the body of an HTTPS POST request to the RP.
The script \emph{script\_idp\_form} is loaded from an IdP into the
user's browser for user authentication at the IdP.

\subsubsection{OAuth Web System with Web Attackers}
In addition to $\oauthwebsystem^n$, we also consider a class of web
systems where the network attacker is replaced by an unbounded finite
set of web attackers. We denote such a system by $\oauthwebsystem^w$
and call it an \emph{OAuth web system with web attackers}, Such web
systems are used to analyze session integrity, see below.

\subsubsection{Limitations of Our OAuth Model}
While our model of OAuth is very comprehensive, a few aspects of OAuth
were not taken into consideration in our analysis:

We do not model \emph{expiration} of access tokens and session ids.
Also, IdPs may issue so-called \emph{refresh tokens} in
Step~\refprotostep{acg-token-resp} of
Figure~\ref{fig:oauth-auth-code-grant}. In practice, an RP may use
such a (long-living) refresh token to obtain a new (short-lived)
access token. In our model, we overapproximate this by not expiring
access tokens. We also do not model \emph{revocation} of access tokens
and \emph{user log out}.

OAuth IdPs support controlling the \emph{scope} of
resources made available to an RP. For example, a Facebook user can
grant a third party the right to read her user profile but deny access
to her friends list. The scope is a property of the access token, but
handled internally by the IdP with its implementation, details, and
semantics highly dependent on the IdP. We therefore model that RPs
always get full access to the user's data at the IdP.

In practice, IdPs can send \emph{error messages} (mostly static
strings) to RPs. We do not model these.

Limitations of the underlying FKS model are discussed in
\cite{FettKuestersSchmitz-SP-2014}.

\subsection{Security Properties}
\label{sec:secur-auth-prop}

Based on the formal OAuth model described above, we now formulate
central security properties of OAuth, namely authorization,
authentication, and session integrity (see
Appendix~\ref{app:form-secur-prop} for the full formal
definitions).

\subsubsection{Authorization} Intuitively, authorization for
$\oauthwebsystem^n$ means that an attacker should not be able to
obtain or use a protected resource available to some honest RP at an
IdP for some user unless, roughly speaking, the user's browser or the
IdP is corrupted. 

More formally, we say that $\oauthwebsystem^n$ is \emph{secure
  w.r.t.~authorization} if the following holds true: if at any point
in a run of $\oauthwebsystem^n$ an attacker can obtain a protected
resource available to some honest RP $r$ at an IdP $i$ for some user
$u$, then the IdP $i$ is corrupt or, if $u\not=\bot$, we have that the
browser of $u$ or at least one of the trusted RPs of $u$ must be
corrupted. Recall that if $u=\bot$, then the resource was acquired in
the client credentials mode, and hence, is not bound to a user.

\subsubsection{Authentication} Intuitively, authentication for
$\oauthwebsystem^n$ means that an attacker should not be able to login
at an (honest) RP under the identity of a user unless, roughly
speaking, the IdP involved or the user's browser is corrupted. As
explained above, being logged in at an RP under some user identity
means to have obtained a service token for this identity from the RP.

More formally, we say that $\oauthwebsystem^n$ is \emph{secure
  w.r.t.~authentication} if the following holds true: if at any point in
a run of $\oauthwebsystem^n$ an attacker can obtain the service token
that was issued by an honest RP using some IdP $i$ for a user $u$,
then the IdP $i$, the browser of $u$, or at least one of the trusted
RPs of $u$ must be corrupted.

\subsubsection{Session Integrity}
Intuitively, session integrity (for authorization) means that (a)
an RP should only be authorized to access some resources of a user
when the user actually expressed the wish to start an OAuth flow
before, and (b) if a user expressed the wish to start an OAuth
flow using some honest IdP and a specific identity, then the
OAuth flow is never completed with a different identity (in the same
session); similarly for authentication. 

More formally, we say that $\oauthwebsystem^w$ is \emph{secure
  w.r.t.~session integrity for authorization} if the following holds
true: (a) if in a run $\oauthwebsystem^w$ an OAuth login
flow is completed with a user's browser, then this user
started an OAuth flow. (b) If in addition we assume that the IdP that
is used in the completed flow is honest, then the flow was completed for the
same identity for which the OAuth flow was started by the user.  We
say that the OAuth flow was completed (for some identity $v$) iff the RP gets
access to a protected resource (of $v$).

We say that $\oauthwebsystem^w$ is \emph{secure w.r.t.~session
  integrity for authentication} if the following holds true: (a) if in
a run $\rho$ of $\oauthwebsystem^w$ a user is logged in with some
identity $v$, then the user started an OAuth flow. (b) If in addition
the IdP that is used in that flow is honest, then the user is logged
in under exactly the same identity for which the OAuth flow was
started by the user.

We note that for session integrity, as opposed to authorization and
authentication, we use the web attacker as an adversary. The rationale
behind this is that a \emph{network} attacker can always forcefully
log in a user under his own account (by setting cookies from
non-secure to secure origins \cite{Zhengetal-cookies-usenix-2015}),
thereby defeating existing CSRF defenses in OAuth (most importantly,
the state parameter). This is a common problem in the session
management of web applications, independently of OAuth. This is why we
restrict our analysis of session integrity to web attackers since
otherwise session integrity would trivially be broken. We note,
however, that more robust solutions for session integrity are conceivable
(e.g., using JavaScript and HTML5 features such as web messaging and
web storage). While some proprietary approaches exist, such approaches
are less common and typically do not conform to the OAuth standard.

\subsubsection{Main Theorem}
We prove the following theorem (see Appendix~\ref{app:proof-oauth} for the proof):
\begin{theorem}\label{thm:security}
  Let $\oauthwebsystem^n$ be an OAuth web system with a network
  attacker, then $\oauthwebsystem^n$ is secure w.r.t.~authorization
  and secure w.r.t.~authentication. Let $\oauthwebsystem^w$ be an
  OAuth web system with web attackers, then $\oauthwebsystem^w$ is
  secure w.r.t.~session integrity for authorization and authentication.
\end{theorem}

Note that this trivially implies that authentication
and authorization properties are satisfied also if web attackers are
considered.

\subsection{Discussion of Results}
\label{sec:discussion-results}

Our results show that the OAuth standard is secure, i.e., provides
strong authentication, authorization, and session integrity
properties, when (1) fixed according to our proposal and (2) when adhering to the
OAuth security recommendations and best practices, as explained in
Section~\ref{sec:model-1}. Depending on
individual implementation choices, (2) is 
potentially not satisfied in all practical scenarios. For example, RPs might run
untrusted JavaScript on their websites. Nevertheless, our
security results, for the first time, give precise implementation
guidelines for OAuth to be secure and also clearly show that if these
guidelines are not followed, then the security of OAuth cannot be
guaranteed.

%% file: section-related-work.tex
\section{Related Work}
\label{sec:related-work}

We focus on work closely related to OAuth~2.0 or formal security analysis
of web standards and web applications.

The work closest to our work is the already mentioned work by Bansal,
Bhargavan, Delignat-Lavaud, and Maffeis \cite{BansalBhargavanetal-JCS-2014}. Bansal et
al.~analyze the security of OAuth using the applied pi-calculus and
the WebSpi library, along with the protocol analysis tool ProVerif.
They model various settings of OAuth~2.0, often assuming the presence
of common web implementation flaws resulting in, for example, CSRF and open redirectors in RPs
and IdPs. They identify previously unknown attacks on the OAuth
implementations of Facebook, Yahoo, Twitter, and many other
websites. Compared to our work, the WebSpi model used in
\cite{BansalBhargavanetal-JCS-2014} is less expressive and
comprehensive (see also the discussion in
\cite{FettKuestersSchmitz-SP-2014}), and the models of OAuth they
employ are more limited.\footnote{For example, only two OAuth modes
  are considered, the model is monotonic (e.g., cookies can only be
  added, but not deleted or modified), fixed bounded number of cookies
  per request, no precise handling of windows, documents, and iframes,
  no web messaging, omission of headers, such as origin. We note that
  while OAuth does not make use of all web features, taking such
  features into account is important to make positive security results
  more meaningful.}  As pointed out by Bansal et al., the main focus of
their work is to discover attacks on OAuth, rather than
proving security. They have some positive results, which, however, are
based on their more limited model. In addition, in order to prove
these results further restrictions are assumed, e.g., they consider
only one IdP per RP and all IdPs are assumed to be honest.

Wang et~al.~\cite{Wangetal-USENIX-Explicating-SDKs-2013} present a
systematic approach to find implicit assumptions in SDKs (e.g., the
Facebook PHP SDK) used for authentication and authorization, including
SDKs that implement OAuth~2.0.

In \cite{PaiSharmaKumarPaiSingh-2011}, Pai et~al.~analyze the security
of OAuth in a very limited model that does not incorporate generic web
features. They show that using their approach, based on the Alloy
finite-state model checker, known weaknesses can be found. The same
tool is used by Kumar \cite{Kumar-OAuth-2012} in a formal analysis of
the older OAuth~1.0 protocol (which, as mentioned,
is very different to OAuth~2.0).

Chari, Jutla, and Roy \cite{ChariJutlaRoy-IACR-2011} analyze the
security of the authorization code mode in the universally
composability model, again without considering web features, such as
semantics of HTTP status codes, details of cookies, or window
structures inside a browser.

Besides these formal approaches, empirical studies were conducted on
deployed OAuth implementations. In
\cite{SantsaiBeznosov-CCS-2012-OAuth}, Sun and Beznosov analyze the
security of three IdPs and 96 RPs. In \cite{LiMitchell-ISC-2014}, Li
and Mitchell study the security of 10 IdPs and 60 RPs based in China.
In \cite{Yangetal-AsiaCCS-2016}, Yang et al.~perform an automated
analysis of 4 OAuth IdPs and 500 RPs. Shernan et
al.~\cite{Shernanetal-DIMVA-2015} evaluate the lack of CSRF protection
in various OAuth deployments. In
\cite{Chenetal-2014,ShehabMohsen-2014}, practical evaluations on the
security of OAuth implementations of mobile apps are performed.

In
\cite{MladenovMainkaKrautwaldFeldmannSchwenk-OpenIDConnect-arXiv-2016},
Mladenov et al.~perform an informal analysis of OpenID Connect. They
present several attacks related to discovery and dynamic client
registration, which are extensions of OpenID Connect; see also the
discussion in Section~\ref{sec:malicious-idp-mitm} (related attacks)
concerning their malicious endpoint attack.

Note that many of the works listed here led to improved security
recommendations for OAuth as listed in RFC6749~\cite{rfc6749-oauth2}
and RFC6819~\cite{rfc6819-oauth2-security}. These are already taken into account 
in our model and analysis of OAuth.

More generally, there have been only very few analysis efforts for web
applications and standards based on formal web models so far. Work outside of the context of OAuth 
includes~\cite{kerschbaum-SP-2007-XSRF-prevention,AkhawBarthLamMitchellSong-CSF-2010,BansalBhargavanetal-POST-2013-WebSpi,FettKuestersSchmitz-CCS-2015,FettKuestersSchmitz-SP-2014,FettKuestersSchmitz-ESORICS-BrowserID-Primary-2015,Armandoetal-SAML-CS-2013,Armandoetal-FMSE-2008}.

%% file: section-conclusion.tex
\section{Conclusion}
\label{sec:conclusion}

In this paper, we carried out the first extensive formal analysis of
OAuth~2.0 based on a comprehensive and expressive web model. Our
analysis, which aimed at the standard itself, rather than specific
OAuth implementations and deployments, comprises all modes (grant
types) of OAuth and available options and also takes malicious RPs and
IdPs as well as corrupted browsers/users into account. The generic web
model underlying our model of OAuth and its analysis is the most
comprehensive web model to date.

Our in-depth analysis revealed four attacks on OAuth as well as OpenID
connect, which builds on OAuth. We verified the attacks, proposed
fixes, and reported the attacks and our fixes to the working groups
for OAuth and OpenID Connect. The working groups confirmed the
attacks. Fixes to the standard and recommendations are currently under
discussion or already incorporated in a draft for a new
RFC~\cite{rfc-draft-ietf-oauth-mix-up-mitigation-01}.

With the fixes applied, we were able to prove strong authorization,
authentication, and session integrity properties for OAuth 2.0. Our
security analysis assumes that OAuth security recommendations and
certain best practices are followed. We show that otherwise the
security of OAuth cannot be guaranteed. By this, we also provide clear
guidelines for implementations. The fact that OAuth is one of the most
widely deployed authorization and authentication systems in the web
and the basis for other protocols makes our analysis particularly
relevant.

As for future work, our formal analysis of OAuth offers a good
starting point for the formal analysis of OpenID Connect, and hence,
such an analysis is an obvious next step for our research.

%% file: appendix-oauth-description.tex
\section{OAuth 2.0}
\label{sec:oauth-all-modes}

The OAuth authorization code mode was presented in
Section~\ref{sec:oauth}. Here, we present the three other OAuth modes
in detail.

\subsection{Preliminaries}
\label{sec:oauth-concepts}

We now first provide some preliminary information regarding OAuth.

\subsubsection{Endpoints}
In OAuth, RPs and IdPs have to provide certain URIs to each other. The
parties and services these URIs point to are called \emph{endpoints};
often the URIs themselves are called endpoints. An IdP provides an
\emph{authorization endpoint} at which the user can authenticate to
the IdP and authorize an RP to access her user data. The IdP also
provides a \emph{token endpoint} at which the RP can request access
tokens. An RP provides one or more \emph{redirection endpoints} to
which the user's browser gets redirected by an IdP after the user
authenticated to the IdP. The URIs of the endpoints are not fixed by
the standard, but are communicated when RPs register at IdPs, as
described below.

The OAuth standard~\cite{rfc6749-oauth2} and the accompanying security
recommendations~\cite{rfc6819-oauth2-security} suggest that all
endpoints use HTTPS. We follow this recommendation in our analysis of
OAuth.

\subsubsection{Registration}
Before an RP can interact with an IdP, the RP needs to be registered
at the IdP. The details of the registration process are out of the
scope of the OAuth protocol. In practice, this process is usually a
manual task. During the registration process, the IdP assigns to the
RP a fixed OAuth client id and client secret.\footnote{Recall
  that in the terminology of the OAuth standard the term ``client''
  stands for RP.} The RP may later use the client secret to
authenticate to the IdP. If the RP cannot keep the OAuth client secret
confidential, e.g., if the RP is an in-browser app or a native
application, the secret can be omitted.
Note that the OAuth client id is public information. It is, for
example, revealed to users in redirects issued by the RP.

Also, an RP registers one or more redirection endpoints at an IdP. As
we will see below, in some OAuth modes, the IdP redirects the user's
browser to one of these redirect URIs. If more than one redirect URI
is registered, the RP must specify which redirect URI is to be used in
each run of the OAuth protocol. For simplicity of presentation, we
will assume that an RP always specifies its choice, although this can
be omitted if there exits only one (fixed) redirect URI. Note that
(depending on the implementation of an IdP) an RP may also register a
pattern as a redirect URI and then specify the exact redirect URI
during the OAuth run. In this case, the IdP checks if the specified
redirect URI matches this pattern.

During the registration process, the (fixed) endpoints belonging to an
IdP are configured at an RP as well.

Our analysis presented in Section~\ref{sec:analysis} covers all the
above mentioned options: absence and presence of client secrets,
specified redirect URIs, and URI patterns.

\subsubsection{Login Sessions}
As mentioned before, in some OAuth modes, an RP redirects the user's
browser to an IdP which later redirects the browser back to the RP. In
order to prevent cross-site request forgery (CSRF) attacks, the RP
typically establishes a session with the browser before the first
redirect. The OAuth standard recommends that an RP selects the
so-called $\mi{state}$
parameter and binds this value to the session, e.g., by choosing a
fresh nonce and storing the nonce in the session state. When the user
later gets redirected back to the RP, the $\mi{state}$
value must be identical. The intention is that this value should
always be unknown to an attacker in order to prevent CSRF attacks. In
our analysis, we follow the recommendation of using the $\mi{state}$
parameter.\footnote{Note that the OAuth standard~\cite{rfc6749-oauth2}
  as well as the accompanying security
  recommendations~\cite{rfc6819-oauth2-security} do not specify the
  session mechanism for RPs. In our analysis we assume the usual
  session mechanism with session cookies following common best
  practices. For more details, see Section~\ref{sec:model-1}.}

\subsubsection{Further Recommendations and Options}
The standard and the recommendations do not specify all implementation
details. For example, the precise user interaction with an RP,
formatting details of messages, and the authentication of the user to
an IdP (e.g., user name and password or some other mechanism) are not
covered. In our security analysis of OAuth we follow all OAuth
security recommendations as well as common best practices for
state-of-the-art web applications in order to avoid known attacks.

OAuth allows RPs to specify which \emph{scope} of the user's data they
are requesting access to at an IdP. The scopes themselves are not
defined in the standard and are considered an implementation detail of
IdPs. Therefore, in our description and analysis of OAuth, we omit the
scope parameter and assume that the user always grants full access to
her data at the IdP.

\subsection{OAuth Modes}

\subsubsection{Implicit Mode}
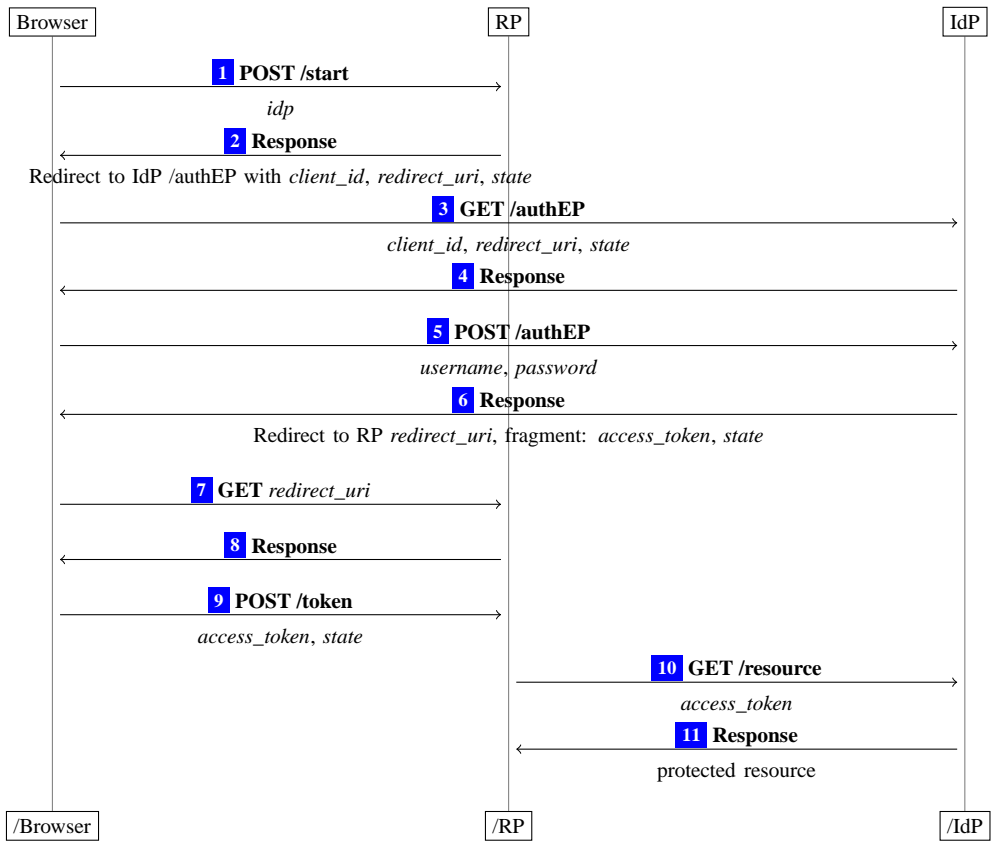
\begin{figure}
  \centering
  \input{figure-oauth-implicit-grant}
  \caption{OAuth~2.0 implicit mode}
  \label{fig:oauth-implicit-grant}
\end{figure}
This mode is a simplified version of the authorization code mode:
instead of providing an authorization code to an RP, an IdP directly
delivers an access token to the RP (via the user's browser).

\paragraph{Step-by-Step Protocol Flow} We now provide a step-by-step
description of the protocol flow (see also
Figure~\ref{fig:oauth-implicit-grant}). As in the authorization code
mode, the user starts the OAuth flow, e.g., by clicking on a button to
select an IdP, triggering the browser to send
request~\refprotostep{ig-start-req} to the RP. The RP selects the
redirect URI $\mi{redirect\_uri}$
(which will be used later in~\refprotostep{ig-redir-ep-req}) and a
value $\mi{state}$.
The RP then redirects the browser with its $\mi{client\_id}$,
$\mi{redirect\_uri}$,
and $\mi{state}$
to the authorization endpoint at the IdP\footnote{Note that also a
  fixed string ``$\str{token}$''
  indicating to the IdP that implicit mode is used is appended as a
  parameter to the URI.} in~\refprotostep{ig-start-resp}
and~\refprotostep{ig-idp-auth-req-1}. The IdP prompts the user to
enter her username and password in~\refprotostep{ig-idp-auth-resp-1}.
The user's browser sends this information to the IdP
in~\refprotostep{ig-idp-auth-req-2}. If the user's credentials are
correct, the IdP creates an access token $\mi{access\_token}$
and redirects the user's browser to the RP's redirection endpoint
$\mi{redirect\_uri}$
in~\refprotostep{ig-idp-auth-resp-2}
and~\refprotostep{ig-redir-ep-req}, where the IdP appends
$\mi{access\_token}$
and $\mi{state}$
to the fragment of the redirection URI. (Recall that a fragment is a
special part of a URI indicated by the `\#' symbol. When the browser
opens a URI, the information in the fragment is not transferred to the
server.) Hence, in Step~\refprotostep{ig-redir-ep-req}
$\mi{access\_token}$
and $\mi{state}$
are not transferred to the RP. To retrieve these values, the RP
in~\refprotostep{ig-redir-ep-resp} delivers a document containing
JavaScript code. It retrieves $\mi{access\_token}$
and $\mi{state}$
from the fragment and sends these to the RP
in~\refprotostep{ig-redir-ep-token-req}. The RP then checks if
$\mi{state}$
is the same as above. Just as in the authorization code mode, the RP
can now use $\mi{access\_token}$
for authorization (illustrated in Steps~\refprotostep{ig-resource-req}
and~\refprotostep{ig-resource-resp}); authentication is analogous to Steps~\refprotostep{acg-introspect-req},
\refprotostep{acg-introspect-resp}, and~\refprotostep{acg-redir-ep-resp} of
Figure~\ref{fig:oauth-auth-code-grant}.

For authentication, note that the response from the IdP includes the
RP's OAuth client id, which is also checked by the RP. This check
prevents re-usage of access tokens across RPs in the OAuth implicit
mode as explained in~\cite{Wangetal-USENIX-Explicating-SDKs-2013}.

We note that in the implicit mode, an IdP cannot verify the
identity of the receiver of the access token, as an RP does not
authenticate itself to the IdP (using $\mi{client\_secret}$).
Hence, this mode is more suitable for RPs that do not have access to a
secure, long-lived storage (for a $\mi{client\_secret}$)
such as in-browser applications.

\subsubsection{Resource Owner Password Credentials Mode}
\begin{figure}
  \centering
  \input{figure-oauth-resource-owner-password-credentials-grant}
  \caption{OAuth~2.0 resource owner password credentials mode}
  \label{fig:oauth-ropc-grant}
\end{figure}
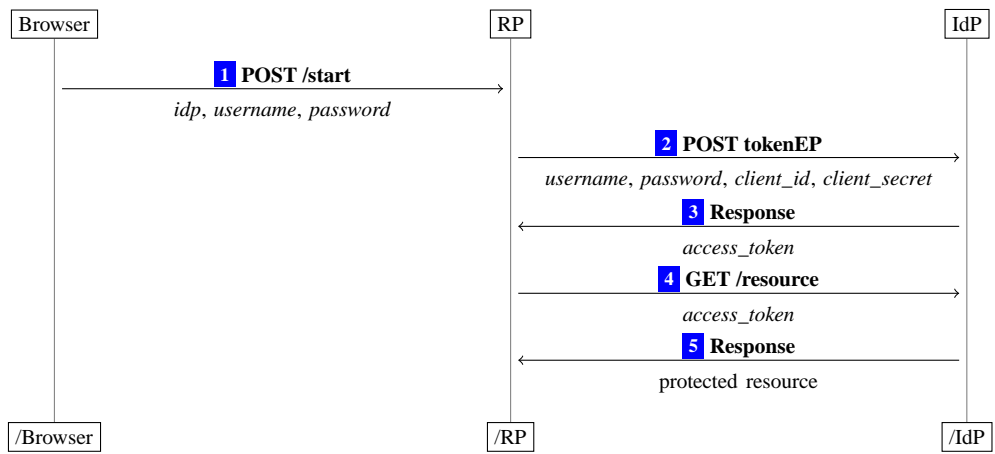
In this mode, the user gives her credentials for an IdP directly to an
RP. The RP can then authenticate to the IdP on the user's behalf and
retrieve an access token.  The resource owner password credentials
mode is intended for highly-trusted RPs, such as the operating system
of the user's device or highly-privileged applications, or if the
previous two modes are not possible to perform (e.g., for applications
without a web browser). In the following, we assume that the
authorization/login process is started by the user using a web
browser.

\paragraph{Step-by-Step Protocol Flow} We now provide a step-by-step
description of the resource owner password credentials mode (see also
Figure~\ref{fig:oauth-ropc-grant}): The user provides her username and
password for the IdP to the RP in~\refprotostep{ropcg-start-req}. Now,
the RP sends the username, the password, its $\mi{client\_id}$
and $\mi{client\_secret}$\footnote{Note that in this mode, if an RP does not have an OAuth
  client secret for an IdP, the $\mi{client\_secret}$
  and $\mi{client\_id}$
  parameters are \emph{both} omitted in this request. This option is
  also covered by our analysis.} to the IdP
in~\refprotostep{ropcg-token-req}. The IdP then issues an access token
$\mi{access\_token}$
to the RP in~\refprotostep{ropcg-token-resp}.\footnote{As
  in the authorization code mode, an IdP may also issue a refresh
  token to the RP here.} Just as in the authorization code mode, the RP
can now use $\mi{access\_token}$
for authorization (illustrated in Steps~\refprotostep{ropcg-resource-req}
and~\refprotostep{ropcg-resource-resp}) and authentication (as in
Steps~\refprotostep{acg-introspect-req},
\refprotostep{acg-introspect-resp}, and~\refprotostep{acg-redir-ep-resp} of
Figure~\ref{fig:oauth-auth-code-grant}).

\subsubsection{Client Credentials Mode}
\begin{figure}
  \centering
  \input{figure-oauth-client-credentials-grant}
  \caption{OAuth~2.0 client credentials mode}
  \label{fig:oauth-cc-grant}
\end{figure}
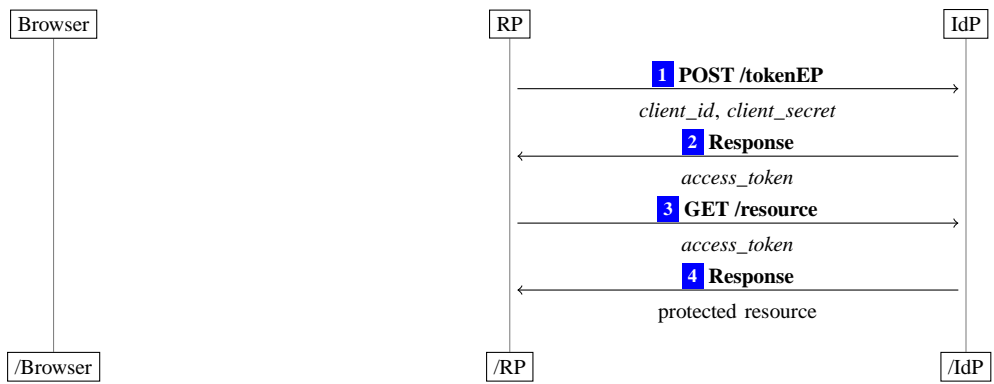
In contrast to the modes shown above, this mode works without the
user's interaction. Instead, it is started by an RP in order to fetch
an access token to access RP's own resources at an IdP or to access
resources at an IdP the RP is authorized to by other means. For
example, Facebook allows RPs to use the client credentials mode to
obtain an access token to access reports of their advertisements'
performance.

\paragraph{Step-by-Step Protocol Flow} The step-by-step description of
the client credentials mode is as follows (see also
Figure~\ref{fig:oauth-cc-grant}): First, the RP contacts the IdP with
RP's $\mi{client\_id}$
and $\mi{client\_secret}$
in~\refprotostep{ccg-token-req}. The IdP now issues an
$\mi{access\_token}$
in~\refprotostep{ccg-token-resp}. Just as in the authorization code
mode, the RP can now use $\mi{access\_token}$
for authorization (illustrated in
Steps~\refprotostep{ccg-resource-req}
and~\refprotostep{ccg-resource-resp}). In contrast to the other modes
presented above, the access token is not bound to a specific user
account, but only to the RP.

%% file: figure-oauth-implicit-grant.tex
 \scriptsize{ \newlength\blockExtraHeightAdedAdIAEBCHfEHcFaEEGJdJBCdIeaBGf
\settototalheight\blockExtraHeightAdedAdIAEBCHfEHcFaEEGJdJBCdIeaBGf{\parbox{0.4\linewidth}{$\mi{idp}$}}
\setlength\blockExtraHeightAdedAdIAEBCHfEHcFaEEGJdJBCdIeaBGf{\dimexpr \blockExtraHeightAdedAdIAEBCHfEHcFaEEGJdJBCdIeaBGf - 3ex/4}
\newlength\blockExtraHeightBdedAdIAEBCHfEHcFaEEGJdJBCdIeaBGf
\settototalheight\blockExtraHeightBdedAdIAEBCHfEHcFaEEGJdJBCdIeaBGf{\parbox{0.4\linewidth}{Redirect to IdP /authEP with $\mi{client\_id}$, $\mi{redirect\_uri}$, $\mi{state}$}}
\setlength\blockExtraHeightBdedAdIAEBCHfEHcFaEEGJdJBCdIeaBGf{\dimexpr \blockExtraHeightBdedAdIAEBCHfEHcFaEEGJdJBCdIeaBGf - 3ex/4}
\newlength\blockExtraHeightCdedAdIAEBCHfEHcFaEEGJdJBCdIeaBGf
\settototalheight\blockExtraHeightCdedAdIAEBCHfEHcFaEEGJdJBCdIeaBGf{\parbox{0.4\linewidth}{$\mi{client\_id}$, $\mi{redirect\_uri}$, $\mi{state}$}}
\setlength\blockExtraHeightCdedAdIAEBCHfEHcFaEEGJdJBCdIeaBGf{\dimexpr \blockExtraHeightCdedAdIAEBCHfEHcFaEEGJdJBCdIeaBGf - 3ex/4}
\newlength\blockExtraHeightDdedAdIAEBCHfEHcFaEEGJdJBCdIeaBGf
\settototalheight\blockExtraHeightDdedAdIAEBCHfEHcFaEEGJdJBCdIeaBGf{\parbox{0.4\linewidth}{}}
\setlength\blockExtraHeightDdedAdIAEBCHfEHcFaEEGJdJBCdIeaBGf{\dimexpr \blockExtraHeightDdedAdIAEBCHfEHcFaEEGJdJBCdIeaBGf - 3ex/4}
\newlength\blockExtraHeightEdedAdIAEBCHfEHcFaEEGJdJBCdIeaBGf
\settototalheight\blockExtraHeightEdedAdIAEBCHfEHcFaEEGJdJBCdIeaBGf{\parbox{0.4\linewidth}{$\mi{username}$, $\mi{password}$}}
\setlength\blockExtraHeightEdedAdIAEBCHfEHcFaEEGJdJBCdIeaBGf{\dimexpr \blockExtraHeightEdedAdIAEBCHfEHcFaEEGJdJBCdIeaBGf - 3ex/4}
\newlength\blockExtraHeightFdedAdIAEBCHfEHcFaEEGJdJBCdIeaBGf
\settototalheight\blockExtraHeightFdedAdIAEBCHfEHcFaEEGJdJBCdIeaBGf{\parbox{0.4\linewidth}{Redirect to RP $\mi{redirect\_uri}$, fragment: $\mi{access\_token}$, $\mi{state}$}}
\setlength\blockExtraHeightFdedAdIAEBCHfEHcFaEEGJdJBCdIeaBGf{\dimexpr \blockExtraHeightFdedAdIAEBCHfEHcFaEEGJdJBCdIeaBGf - 3ex/4}
\newlength\blockExtraHeightGdedAdIAEBCHfEHcFaEEGJdJBCdIeaBGf
\settototalheight\blockExtraHeightGdedAdIAEBCHfEHcFaEEGJdJBCdIeaBGf{\parbox{0.4\linewidth}{}}
\setlength\blockExtraHeightGdedAdIAEBCHfEHcFaEEGJdJBCdIeaBGf{\dimexpr \blockExtraHeightGdedAdIAEBCHfEHcFaEEGJdJBCdIeaBGf - 3ex/4}
\newlength\blockExtraHeightHdedAdIAEBCHfEHcFaEEGJdJBCdIeaBGf
\settototalheight\blockExtraHeightHdedAdIAEBCHfEHcFaEEGJdJBCdIeaBGf{\parbox{0.4\linewidth}{}}
\setlength\blockExtraHeightHdedAdIAEBCHfEHcFaEEGJdJBCdIeaBGf{\dimexpr \blockExtraHeightHdedAdIAEBCHfEHcFaEEGJdJBCdIeaBGf - 3ex/4}
\newlength\blockExtraHeightIdedAdIAEBCHfEHcFaEEGJdJBCdIeaBGf
\settototalheight\blockExtraHeightIdedAdIAEBCHfEHcFaEEGJdJBCdIeaBGf{\parbox{0.4\linewidth}{$\mi{access\_token}$, $\mi{state}$}}
\setlength\blockExtraHeightIdedAdIAEBCHfEHcFaEEGJdJBCdIeaBGf{\dimexpr \blockExtraHeightIdedAdIAEBCHfEHcFaEEGJdJBCdIeaBGf - 3ex/4}
\newlength\blockExtraHeightJdedAdIAEBCHfEHcFaEEGJdJBCdIeaBGf
\settototalheight\blockExtraHeightJdedAdIAEBCHfEHcFaEEGJdJBCdIeaBGf{\parbox{0.4\linewidth}{$\mi{access\_token}$}}
\setlength\blockExtraHeightJdedAdIAEBCHfEHcFaEEGJdJBCdIeaBGf{\dimexpr \blockExtraHeightJdedAdIAEBCHfEHcFaEEGJdJBCdIeaBGf - 3ex/4}
\newlength\blockExtraHeightBAdedAdIAEBCHfEHcFaEEGJdJBCdIeaBGf
\settototalheight\blockExtraHeightBAdedAdIAEBCHfEHcFaEEGJdJBCdIeaBGf{\parbox{0.4\linewidth}{protected resource}}
\setlength\blockExtraHeightBAdedAdIAEBCHfEHcFaEEGJdJBCdIeaBGf{\dimexpr \blockExtraHeightBAdedAdIAEBCHfEHcFaEEGJdJBCdIeaBGf - 3ex/4}

 \begin{tikzpicture}
   \tikzstyle{xhrArrow} = [color=blue,decoration={markings, mark=at
    position 1 with {\arrow[color=blue]{triangle 45}}}, preaction
  = {decorate}]

    \matrix [column sep={6cm,between origins}, row sep=4.5ex]
  {

    \node[draw,anchor=base](Browser-start-0){Browser}; & \node[draw,anchor=base](RP-start-0){RP}; & \node[draw,anchor=base](IdP-start-0){IdP};\\
\node(Browser-0){}; & \node(RP-0){}; & \node(IdP-0){};\\[\blockExtraHeightAdedAdIAEBCHfEHcFaEEGJdJBCdIeaBGf]
\node(Browser-1){}; & \node(RP-1){}; & \node(IdP-1){};\\[\blockExtraHeightBdedAdIAEBCHfEHcFaEEGJdJBCdIeaBGf]
\node(Browser-2){}; & \node(RP-2){}; & \node(IdP-2){};\\[\blockExtraHeightCdedAdIAEBCHfEHcFaEEGJdJBCdIeaBGf]
\node(Browser-3){}; & \node(RP-3){}; & \node(IdP-3){};\\[\blockExtraHeightDdedAdIAEBCHfEHcFaEEGJdJBCdIeaBGf]
\node(Browser-4){}; & \node(RP-4){}; & \node(IdP-4){};\\[\blockExtraHeightEdedAdIAEBCHfEHcFaEEGJdJBCdIeaBGf]
\node(Browser-5){}; & \node(RP-5){}; & \node(IdP-5){};\\[\blockExtraHeightFdedAdIAEBCHfEHcFaEEGJdJBCdIeaBGf]
\node(Browser-6){}; & \node(RP-6){}; & \node(IdP-6){};\\[\blockExtraHeightGdedAdIAEBCHfEHcFaEEGJdJBCdIeaBGf]
\node(Browser-7){}; & \node(RP-7){}; & \node(IdP-7){};\\[\blockExtraHeightHdedAdIAEBCHfEHcFaEEGJdJBCdIeaBGf]
\node(Browser-8){}; & \node(RP-8){}; & \node(IdP-8){};\\[\blockExtraHeightIdedAdIAEBCHfEHcFaEEGJdJBCdIeaBGf]
\node(Browser-9){}; & \node(RP-9){}; & \node(IdP-9){};\\[\blockExtraHeightJdedAdIAEBCHfEHcFaEEGJdJBCdIeaBGf]
\node(Browser-10){}; & \node(RP-10){}; & \node(IdP-10){};\\[\blockExtraHeightBAdedAdIAEBCHfEHcFaEEGJdJBCdIeaBGf]
\node[draw,anchor=base](Browser-end-1){/Browser}; & \node[draw,anchor=base](RP-end-1){/RP}; & \node[draw,anchor=base](IdP-end-1){/IdP};\\
};
\draw[->] (Browser-0) to node [above=2.6pt, anchor=base]{\protostep{ig-start-req} \textbf{POST /start}} node [below=-8pt, text width=0.5\linewidth, anchor=base]{\begin{center} $\mi{idp}$\end{center}} (RP-0); 

\draw[->] (RP-1) to node [above=2.6pt, anchor=base]{\protostep{ig-start-resp} \textbf{Response}} node [below=-8pt, text width=0.5\linewidth, anchor=base]{\begin{center} Redirect to IdP /authEP with $\mi{client\_id}$, $\mi{redirect\_uri}$, $\mi{state}$\end{center}} (Browser-1); 

\draw[->] (Browser-2) to node [above=2.6pt, anchor=base]{\protostep{ig-idp-auth-req-1} \textbf{GET /authEP}} node [below=-8pt, text width=0.5\linewidth, anchor=base]{\begin{center} $\mi{client\_id}$, $\mi{redirect\_uri}$, $\mi{state}$\end{center}} (IdP-2); 

\draw[->] (IdP-3) to node [above=2.6pt, anchor=base]{\protostep{ig-idp-auth-resp-1} \textbf{Response}} node [below=-8pt, text width=0.5\linewidth, anchor=base]{\begin{center} \end{center}} (Browser-3); 

\draw[->] (Browser-4) to node [above=2.6pt, anchor=base]{\protostep{ig-idp-auth-req-2} \textbf{POST /authEP}} node [below=-8pt, text width=0.5\linewidth, anchor=base]{\begin{center} $\mi{username}$, $\mi{password}$\end{center}} (IdP-4); 

\draw[->] (IdP-5) to node [above=2.6pt, anchor=base]{\protostep{ig-idp-auth-resp-2} \textbf{Response}} node [below=-8pt, text width=0.5\linewidth, anchor=base]{\begin{center} Redirect to RP $\mi{redirect\_uri}$, fragment: $\mi{access\_token}$, $\mi{state}$\end{center}} (Browser-5); 

\draw[->] (Browser-6) to node [above=2.6pt, anchor=base]{\protostep{ig-redir-ep-req} \textbf{GET $\mi{redirect\_uri}$}} node [below=-8pt, text width=0.5\linewidth, anchor=base]{\begin{center} \end{center}} (RP-6); 

\draw[->] (RP-7) to node [above=2.6pt, anchor=base]{\protostep{ig-redir-ep-resp} \textbf{Response}} node [below=-8pt, text width=0.5\linewidth, anchor=base]{\begin{center} \end{center}} (Browser-7); 

\draw[->] (Browser-8) to node [above=2.6pt, anchor=base]{\protostep{ig-redir-ep-token-req} \textbf{POST /token}} node [below=-8pt, text width=0.5\linewidth, anchor=base]{\begin{center} $\mi{access\_token}$, $\mi{state}$\end{center}} (RP-8); 

\draw[->] (RP-9) to node [above=2.6pt, anchor=base]{\protostep{ig-resource-req} \textbf{GET /resource}} node [below=-8pt, text width=0.5\linewidth, anchor=base]{\begin{center} $\mi{access\_token}$\end{center}} (IdP-9); 

\draw[->] (IdP-10) to node [above=2.6pt, anchor=base]{\protostep{ig-resource-resp} \textbf{Response}} node [below=-8pt, text width=0.5\linewidth, anchor=base]{\begin{center} protected resource\end{center}} (RP-10); 

\begin{pgfonlayer}{background}
\draw [color=gray] (Browser-start-0) -- (Browser-end-1);
\draw [color=gray] (RP-start-0) -- (RP-end-1);
\draw [color=gray] (IdP-start-0) -- (IdP-end-1);
\end{pgfonlayer}
\end{tikzpicture}}

%% file: figure-oauth-resource-owner-password-credentials-grant.tex
 \scriptsize{ \newlength\blockExtraHeightADaecDdGeJADaEcfBICfAIaeIDAHBaAEe
\settototalheight\blockExtraHeightADaecDdGeJADaEcfBICfAIaeIDAHBaAEe{\parbox{0.4\linewidth}{$\mi{idp}$, $\mi{username}$, $\mi{password}$}}
\setlength\blockExtraHeightADaecDdGeJADaEcfBICfAIaeIDAHBaAEe{\dimexpr \blockExtraHeightADaecDdGeJADaEcfBICfAIaeIDAHBaAEe - 3ex/4}
\newlength\blockExtraHeightBDaecDdGeJADaEcfBICfAIaeIDAHBaAEe
\settototalheight\blockExtraHeightBDaecDdGeJADaEcfBICfAIaeIDAHBaAEe{\parbox{0.4\linewidth}{$\mi{username}$, $\mi{password}$, $\mi{client\_id}$, $\mi{client\_secret}$}}
\setlength\blockExtraHeightBDaecDdGeJADaEcfBICfAIaeIDAHBaAEe{\dimexpr \blockExtraHeightBDaecDdGeJADaEcfBICfAIaeIDAHBaAEe - 3ex/4}
\newlength\blockExtraHeightCDaecDdGeJADaEcfBICfAIaeIDAHBaAEe
\settototalheight\blockExtraHeightCDaecDdGeJADaEcfBICfAIaeIDAHBaAEe{\parbox{0.4\linewidth}{$\mi{access\_token}$}}
\setlength\blockExtraHeightCDaecDdGeJADaEcfBICfAIaeIDAHBaAEe{\dimexpr \blockExtraHeightCDaecDdGeJADaEcfBICfAIaeIDAHBaAEe - 3ex/4}
\newlength\blockExtraHeightDDaecDdGeJADaEcfBICfAIaeIDAHBaAEe
\settototalheight\blockExtraHeightDDaecDdGeJADaEcfBICfAIaeIDAHBaAEe{\parbox{0.4\linewidth}{$\mi{access\_token}$}}
\setlength\blockExtraHeightDDaecDdGeJADaEcfBICfAIaeIDAHBaAEe{\dimexpr \blockExtraHeightDDaecDdGeJADaEcfBICfAIaeIDAHBaAEe - 3ex/4}
\newlength\blockExtraHeightEDaecDdGeJADaEcfBICfAIaeIDAHBaAEe
\settototalheight\blockExtraHeightEDaecDdGeJADaEcfBICfAIaeIDAHBaAEe{\parbox{0.4\linewidth}{protected resource}}
\setlength\blockExtraHeightEDaecDdGeJADaEcfBICfAIaeIDAHBaAEe{\dimexpr \blockExtraHeightEDaecDdGeJADaEcfBICfAIaeIDAHBaAEe - 3ex/4}

 \begin{tikzpicture}
   \tikzstyle{xhrArrow} = [color=blue,decoration={markings, mark=at
    position 1 with {\arrow[color=blue]{triangle 45}}}, preaction
  = {decorate}]

    \matrix [column sep={6cm,between origins}, row sep=4.5ex]
  {

    \node[draw,anchor=base](Browser-start-0){Browser}; & \node[draw,anchor=base](RP-start-0){RP}; & \node[draw,anchor=base](IdP-start-0){IdP};\\
\node(Browser-0){}; & \node(RP-0){}; & \node(IdP-0){};\\[\blockExtraHeightADaecDdGeJADaEcfBICfAIaeIDAHBaAEe]
\node(Browser-1){}; & \node(RP-1){}; & \node(IdP-1){};\\[\blockExtraHeightBDaecDdGeJADaEcfBICfAIaeIDAHBaAEe]
\node(Browser-2){}; & \node(RP-2){}; & \node(IdP-2){};\\[\blockExtraHeightCDaecDdGeJADaEcfBICfAIaeIDAHBaAEe]
\node(Browser-3){}; & \node(RP-3){}; & \node(IdP-3){};\\[\blockExtraHeightDDaecDdGeJADaEcfBICfAIaeIDAHBaAEe]
\node(Browser-4){}; & \node(RP-4){}; & \node(IdP-4){};\\[\blockExtraHeightEDaecDdGeJADaEcfBICfAIaeIDAHBaAEe]
\node[draw,anchor=base](Browser-end-1){/Browser}; & \node[draw,anchor=base](RP-end-1){/RP}; & \node[draw,anchor=base](IdP-end-1){/IdP};\\
};
\draw[->] (Browser-0) to node [above=2.6pt, anchor=base]{\protostep{ropcg-start-req} \textbf{POST /start}} node [below=-8pt, text width=0.5\linewidth, anchor=base]{\begin{center} $\mi{idp}$, $\mi{username}$, $\mi{password}$\end{center}} (RP-0); 

\draw[->] (RP-1) to node [above=2.6pt, anchor=base]{\protostep{ropcg-token-req} \textbf{POST tokenEP}} node [below=-8pt, text width=0.5\linewidth, anchor=base]{\begin{center} $\mi{username}$, $\mi{password}$, $\mi{client\_id}$, $\mi{client\_secret}$\end{center}} (IdP-1); 

\draw[->] (IdP-2) to node [above=2.6pt, anchor=base]{\protostep{ropcg-token-resp} \textbf{Response}} node [below=-8pt, text width=0.5\linewidth, anchor=base]{\begin{center} $\mi{access\_token}$\end{center}} (RP-2); 

\draw[->] (RP-3) to node [above=2.6pt, anchor=base]{\protostep{ropcg-resource-req} \textbf{GET /resource}} node [below=-8pt, text width=0.5\linewidth, anchor=base]{\begin{center} $\mi{access\_token}$\end{center}} (IdP-3); 

\draw[->] (IdP-4) to node [above=2.6pt, anchor=base]{\protostep{ropcg-resource-resp} \textbf{Response}} node [below=-8pt, text width=0.5\linewidth, anchor=base]{\begin{center} protected resource\end{center}} (RP-4); 

\begin{pgfonlayer}{background}
\draw [color=gray] (Browser-start-0) -- (Browser-end-1);
\draw [color=gray] (RP-start-0) -- (RP-end-1);
\draw [color=gray] (IdP-start-0) -- (IdP-end-1);
\end{pgfonlayer}
\end{tikzpicture}}

%% file: figure-oauth-client-credentials-grant.tex
 \scriptsize{ \newlength\blockExtraHeightAbaAdIfcfddEeEAEEJJcDFDIfCAaadFII
\settototalheight\blockExtraHeightAbaAdIfcfddEeEAEEJJcDFDIfCAaadFII{\parbox{0.4\linewidth}{$\mi{client\_id}$, $\mi{client\_secret}$}}
\setlength\blockExtraHeightAbaAdIfcfddEeEAEEJJcDFDIfCAaadFII{\dimexpr \blockExtraHeightAbaAdIfcfddEeEAEEJJcDFDIfCAaadFII - 3ex/4}
\newlength\blockExtraHeightBbaAdIfcfddEeEAEEJJcDFDIfCAaadFII
\settototalheight\blockExtraHeightBbaAdIfcfddEeEAEEJJcDFDIfCAaadFII{\parbox{0.4\linewidth}{$\mi{access\_token}$}}
\setlength\blockExtraHeightBbaAdIfcfddEeEAEEJJcDFDIfCAaadFII{\dimexpr \blockExtraHeightBbaAdIfcfddEeEAEEJJcDFDIfCAaadFII - 3ex/4}
\newlength\blockExtraHeightCbaAdIfcfddEeEAEEJJcDFDIfCAaadFII
\settototalheight\blockExtraHeightCbaAdIfcfddEeEAEEJJcDFDIfCAaadFII{\parbox{0.4\linewidth}{$\mi{access\_token}$}}
\setlength\blockExtraHeightCbaAdIfcfddEeEAEEJJcDFDIfCAaadFII{\dimexpr \blockExtraHeightCbaAdIfcfddEeEAEEJJcDFDIfCAaadFII - 3ex/4}
\newlength\blockExtraHeightDbaAdIfcfddEeEAEEJJcDFDIfCAaadFII
\settototalheight\blockExtraHeightDbaAdIfcfddEeEAEEJJcDFDIfCAaadFII{\parbox{0.4\linewidth}{protected resource}}
\setlength\blockExtraHeightDbaAdIfcfddEeEAEEJJcDFDIfCAaadFII{\dimexpr \blockExtraHeightDbaAdIfcfddEeEAEEJJcDFDIfCAaadFII - 3ex/4}

 \begin{tikzpicture}
   \tikzstyle{xhrArrow} = [color=blue,decoration={markings, mark=at
    position 1 with {\arrow[color=blue]{triangle 45}}}, preaction
  = {decorate}]

    \matrix [column sep={6cm,between origins}, row sep=4.5ex]
  {

    \node[draw,anchor=base](Browser-start-0){Browser}; & \node[draw,anchor=base](RP-start-0){RP}; & \node[draw,anchor=base](IdP-start-0){IdP};\\
\node(Browser-0){}; & \node(RP-0){}; & \node(IdP-0){};\\[\blockExtraHeightAbaAdIfcfddEeEAEEJJcDFDIfCAaadFII]
\node(Browser-1){}; & \node(RP-1){}; & \node(IdP-1){};\\[\blockExtraHeightBbaAdIfcfddEeEAEEJJcDFDIfCAaadFII]
\node(Browser-2){}; & \node(RP-2){}; & \node(IdP-2){};\\[\blockExtraHeightCbaAdIfcfddEeEAEEJJcDFDIfCAaadFII]
\node(Browser-3){}; & \node(RP-3){}; & \node(IdP-3){};\\[\blockExtraHeightDbaAdIfcfddEeEAEEJJcDFDIfCAaadFII]
\node[draw,anchor=base](Browser-end-1){/Browser}; & \node[draw,anchor=base](RP-end-1){/RP}; & \node[draw,anchor=base](IdP-end-1){/IdP};\\
};
\draw[->] (RP-0) to node [above=2.6pt, anchor=base]{\protostep{ccg-token-req} \textbf{POST /tokenEP}} node [below=-8pt, text width=0.5\linewidth, anchor=base]{\begin{center} $\mi{client\_id}$, $\mi{client\_secret}$\end{center}} (IdP-0); 

\draw[->] (IdP-1) to node [above=2.6pt, anchor=base]{\protostep{ccg-token-resp} \textbf{Response}} node [below=-8pt, text width=0.5\linewidth, anchor=base]{\begin{center} $\mi{access\_token}$\end{center}} (RP-1); 

\draw[->] (RP-2) to node [above=2.6pt, anchor=base]{\protostep{ccg-resource-req} \textbf{GET /resource}} node [below=-8pt, text width=0.5\linewidth, anchor=base]{\begin{center} $\mi{access\_token}$\end{center}} (IdP-2); 

\draw[->] (IdP-3) to node [above=2.6pt, anchor=base]{\protostep{ccg-resource-resp} \textbf{Response}} node [below=-8pt, text width=0.5\linewidth, anchor=base]{\begin{center} protected resource\end{center}} (RP-3); 

\begin{pgfonlayer}{background}
\draw [color=gray] (Browser-start-0) -- (Browser-end-1);
\draw [color=gray] (RP-start-0) -- (RP-end-1);
\draw [color=gray] (IdP-start-0) -- (IdP-end-1);
\end{pgfonlayer}
\end{tikzpicture}}

%% file: appendix-oauth-attack.tex
\section{IdP Mix-Up Attack in the OAuth Implicit Mode}

\label{sec:attack-mixup-in-implicit}

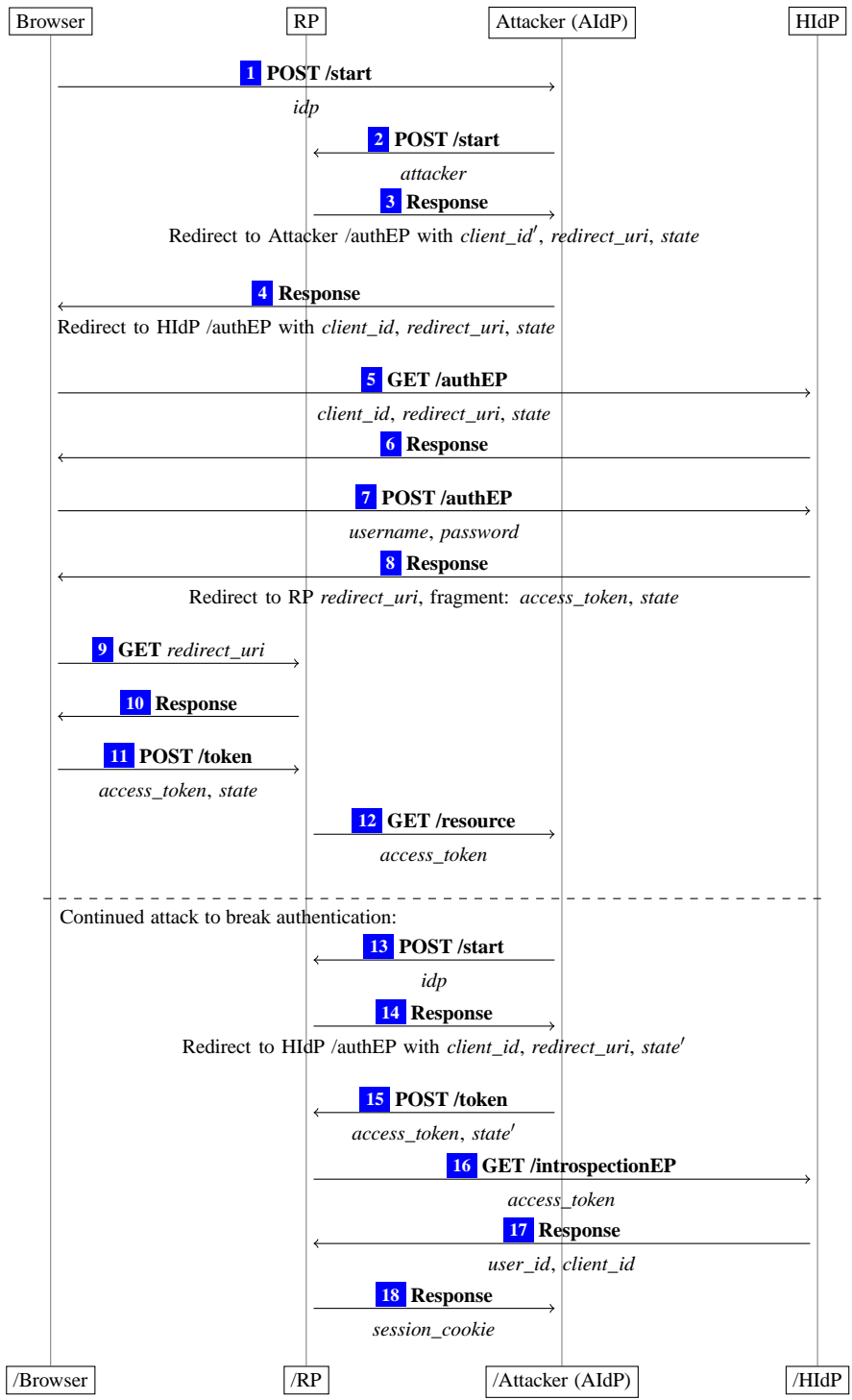
\begin{figure}[p!]
  \input{figure-oauth-implicit-grant-attack}
  \caption{IdP Mix-Up Attack on OAuth 2.0 implicit mode}
  \label{fig:oauth-ig-att}
\end{figure}

Here, we present the IdP Mix-Up attack in the implicit mode. 
It is depicted in
Figure~\ref{fig:oauth-ig-att}. 

Just as in the authorization code mode, the
attack starts when the user selects that she wants to log in using
HIdP (Step~\refprotostep{ig-att-start-req} in
Figure~\ref{fig:oauth-ig-att}). Now, the attacker intercepts the
request intended for the RP and modifies the content of this request
by replacing HIdP by AIdP. The response of the
RP~\refprotostep{ig-att-start-resp} (containing a redirect to AIdP) is
then again intercepted and modified by the attacker such that it
redirects the user to
HIdP~\refprotostep{ig-att-start-resp-manipulated}. The attacker also
replaces the OAuth client id of the RP at AIdP with the client id of
the RP at HIdP.\footnote{As mentioned above, OAuth client ids are
  public information.} (Note that we assume that from this point on,
in accordance with the OAuth security recommendations, the
communication between the user's browser and HIdP and the RP is
encrypted by using HTTPS, and thus, cannot be inspected or altered by
the attacker.) The user then authenticates to HIdP and is redirected
back to the RP~\refprotostep{ig-att-idp-auth-resp-2}. The RP, however,
still assumes that the access token contained in this redirect is an
access token issued by AIdP, rather than HIdP. The RP therefore now
uses this access token to retrieve protected resources of the user (or
the user id) at AIdP~\refprotostep{ig-att-introspect-req}, rather than
HIdP. This leaks the access token to the attacker who can now access
protected resources of the user at IdP. This breaks the authorization
property (see Section~\ref{sec:secur-auth-prop} below). (We note that
at this point, the attacker might even provide false information about
the user or her protected resources to the RP.)

To break authentication and impersonate the honest user, the attacker
now starts a new login process (using his own browser) at the RP.
In~\refprotostep{ig-att-start-attacker-req} he selects HIdP as the IdP
for this login process. He receives a redirect to HIdP, which he
skips.\footnote{Note that this redirect contains (besides a cookie for
  a new login session) a fresh state parameter, say $\mi{state}'$.
  The attacker will use this information in subsequent requests to the
  RP.} The attacker now sends the access token $\mi{access\_token}$
captured in Step~\refprotostep{ig-att-introspect-req} to the RP imitating
a real login~\refprotostep{ig-att-redir-ep-token-attacker-req}. The RP
now uses this access token to retrieve the user id at
HIdP~\refprotostep{ig-att-introspect-attacker-req} and receives the
(honest) user's id as well as its own OAuth client
id~\refprotostep{ig-att-introspect-attacker-resp}. Being convinced that
the attacker owns the honest user's account, the RP issues a session
cookie for this account to the
attacker~\refprotostep{ig-att-redir-ep-token-resp}. As a result, the
attacker is logged in at the RP under the honest user's id. This
breaks the authentication property of OAuth (see
Section~\ref{sec:secur-auth-prop} below).

%% file: figure-oauth-implicit-grant-attack.tex
 \scriptsize{ \newlength\blockExtraHeightAFdAFBeBbFFBFEEJebcDJJGbDdFGdJEGd
\settototalheight\blockExtraHeightAFdAFBeBbFFBFEEJebcDJJGbDdFGdJEGd{\parbox{0.4\linewidth}{$\mi{idp}$}}
\setlength\blockExtraHeightAFdAFBeBbFFBFEEJebcDJJGbDdFGdJEGd{\dimexpr \blockExtraHeightAFdAFBeBbFFBFEEJebcDJJGbDdFGdJEGd - 3ex/4}
\newlength\blockExtraHeightBFdAFBeBbFFBFEEJebcDJJGbDdFGdJEGd
\settototalheight\blockExtraHeightBFdAFBeBbFFBFEEJebcDJJGbDdFGdJEGd{\parbox{0.4\linewidth}{$\mi{attacker}$}}
\setlength\blockExtraHeightBFdAFBeBbFFBFEEJebcDJJGbDdFGdJEGd{\dimexpr \blockExtraHeightBFdAFBeBbFFBFEEJebcDJJGbDdFGdJEGd - 3ex/4}
\newlength\blockExtraHeightCFdAFBeBbFFBFEEJebcDJJGbDdFGdJEGd
\settototalheight\blockExtraHeightCFdAFBeBbFFBFEEJebcDJJGbDdFGdJEGd{\parbox{0.4\linewidth}{Redirect to Attacker /authEP with $\mi{client\_id}'$, $\mi{redirect\_uri}$, $\mi{state}$}}
\setlength\blockExtraHeightCFdAFBeBbFFBFEEJebcDJJGbDdFGdJEGd{\dimexpr \blockExtraHeightCFdAFBeBbFFBFEEJebcDJJGbDdFGdJEGd - 3ex/4}
\newlength\blockExtraHeightDFdAFBeBbFFBFEEJebcDJJGbDdFGdJEGd
\settototalheight\blockExtraHeightDFdAFBeBbFFBFEEJebcDJJGbDdFGdJEGd{\parbox{0.4\linewidth}{Redirect to HIdP /authEP with $\mi{client\_id}$, $\mi{redirect\_uri}$, $\mi{state}$}}
\setlength\blockExtraHeightDFdAFBeBbFFBFEEJebcDJJGbDdFGdJEGd{\dimexpr \blockExtraHeightDFdAFBeBbFFBFEEJebcDJJGbDdFGdJEGd - 3ex/4}
\newlength\blockExtraHeightEFdAFBeBbFFBFEEJebcDJJGbDdFGdJEGd
\settototalheight\blockExtraHeightEFdAFBeBbFFBFEEJebcDJJGbDdFGdJEGd{\parbox{0.4\linewidth}{$\mi{client\_id}$, $\mi{redirect\_uri}$, $\mi{state}$}}
\setlength\blockExtraHeightEFdAFBeBbFFBFEEJebcDJJGbDdFGdJEGd{\dimexpr \blockExtraHeightEFdAFBeBbFFBFEEJebcDJJGbDdFGdJEGd - 3ex/4}
\newlength\blockExtraHeightFFdAFBeBbFFBFEEJebcDJJGbDdFGdJEGd
\settototalheight\blockExtraHeightFFdAFBeBbFFBFEEJebcDJJGbDdFGdJEGd{\parbox{0.4\linewidth}{}}
\setlength\blockExtraHeightFFdAFBeBbFFBFEEJebcDJJGbDdFGdJEGd{\dimexpr \blockExtraHeightFFdAFBeBbFFBFEEJebcDJJGbDdFGdJEGd - 3ex/4}
\newlength\blockExtraHeightGFdAFBeBbFFBFEEJebcDJJGbDdFGdJEGd
\settototalheight\blockExtraHeightGFdAFBeBbFFBFEEJebcDJJGbDdFGdJEGd{\parbox{0.4\linewidth}{$\mi{username}$, $\mi{password}$}}
\setlength\blockExtraHeightGFdAFBeBbFFBFEEJebcDJJGbDdFGdJEGd{\dimexpr \blockExtraHeightGFdAFBeBbFFBFEEJebcDJJGbDdFGdJEGd - 3ex/4}
\newlength\blockExtraHeightHFdAFBeBbFFBFEEJebcDJJGbDdFGdJEGd
\settototalheight\blockExtraHeightHFdAFBeBbFFBFEEJebcDJJGbDdFGdJEGd{\parbox{0.4\linewidth}{Redirect to RP $\mi{redirect\_uri}$, fragment: $\mi{access\_token}$, $\mi{state}$}}
\setlength\blockExtraHeightHFdAFBeBbFFBFEEJebcDJJGbDdFGdJEGd{\dimexpr \blockExtraHeightHFdAFBeBbFFBFEEJebcDJJGbDdFGdJEGd - 3ex/4}
\newlength\blockExtraHeightIFdAFBeBbFFBFEEJebcDJJGbDdFGdJEGd
\settototalheight\blockExtraHeightIFdAFBeBbFFBFEEJebcDJJGbDdFGdJEGd{\parbox{0.4\linewidth}{}}
\setlength\blockExtraHeightIFdAFBeBbFFBFEEJebcDJJGbDdFGdJEGd{\dimexpr \blockExtraHeightIFdAFBeBbFFBFEEJebcDJJGbDdFGdJEGd - 3ex/4}
\newlength\blockExtraHeightJFdAFBeBbFFBFEEJebcDJJGbDdFGdJEGd
\settototalheight\blockExtraHeightJFdAFBeBbFFBFEEJebcDJJGbDdFGdJEGd{\parbox{0.4\linewidth}{}}
\setlength\blockExtraHeightJFdAFBeBbFFBFEEJebcDJJGbDdFGdJEGd{\dimexpr \blockExtraHeightJFdAFBeBbFFBFEEJebcDJJGbDdFGdJEGd - 3ex/4}
\newlength\blockExtraHeightBAFdAFBeBbFFBFEEJebcDJJGbDdFGdJEGd
\settototalheight\blockExtraHeightBAFdAFBeBbFFBFEEJebcDJJGbDdFGdJEGd{\parbox{0.4\linewidth}{$\mi{access\_token}$, $\mi{state}$}}
\setlength\blockExtraHeightBAFdAFBeBbFFBFEEJebcDJJGbDdFGdJEGd{\dimexpr \blockExtraHeightBAFdAFBeBbFFBFEEJebcDJJGbDdFGdJEGd - 3ex/4}
\newlength\blockExtraHeightBBFdAFBeBbFFBFEEJebcDJJGbDdFGdJEGd
\settototalheight\blockExtraHeightBBFdAFBeBbFFBFEEJebcDJJGbDdFGdJEGd{\parbox{0.4\linewidth}{$\mi{access\_token}$}}
\setlength\blockExtraHeightBBFdAFBeBbFFBFEEJebcDJJGbDdFGdJEGd{\dimexpr \blockExtraHeightBBFdAFBeBbFFBFEEJebcDJJGbDdFGdJEGd - 3ex/4}
\newlength\blockExtraHeightBCFdAFBeBbFFBFEEJebcDJJGbDdFGdJEGd
\settototalheight\blockExtraHeightBCFdAFBeBbFFBFEEJebcDJJGbDdFGdJEGd{\parbox{0.4\linewidth}{None}}
\setlength\blockExtraHeightBCFdAFBeBbFFBFEEJebcDJJGbDdFGdJEGd{\dimexpr \blockExtraHeightBCFdAFBeBbFFBFEEJebcDJJGbDdFGdJEGd - 3ex/4}
\newlength\blockExtraHeightBDFdAFBeBbFFBFEEJebcDJJGbDdFGdJEGd
\settototalheight\blockExtraHeightBDFdAFBeBbFFBFEEJebcDJJGbDdFGdJEGd{\parbox{0.4\linewidth}{$\mi{idp}$}}
\setlength\blockExtraHeightBDFdAFBeBbFFBFEEJebcDJJGbDdFGdJEGd{\dimexpr \blockExtraHeightBDFdAFBeBbFFBFEEJebcDJJGbDdFGdJEGd - 3ex/4}
\newlength\blockExtraHeightBEFdAFBeBbFFBFEEJebcDJJGbDdFGdJEGd
\settototalheight\blockExtraHeightBEFdAFBeBbFFBFEEJebcDJJGbDdFGdJEGd{\parbox{0.4\linewidth}{Redirect to HIdP /authEP with $\mi{client\_id}$, $\mi{redirect\_uri}$, $\mi{state}'$}}
\setlength\blockExtraHeightBEFdAFBeBbFFBFEEJebcDJJGbDdFGdJEGd{\dimexpr \blockExtraHeightBEFdAFBeBbFFBFEEJebcDJJGbDdFGdJEGd - 3ex/4}
\newlength\blockExtraHeightBFFdAFBeBbFFBFEEJebcDJJGbDdFGdJEGd
\settototalheight\blockExtraHeightBFFdAFBeBbFFBFEEJebcDJJGbDdFGdJEGd{\parbox{0.4\linewidth}{$\mi{access\_token}$, $\mi{state}'$}}
\setlength\blockExtraHeightBFFdAFBeBbFFBFEEJebcDJJGbDdFGdJEGd{\dimexpr \blockExtraHeightBFFdAFBeBbFFBFEEJebcDJJGbDdFGdJEGd - 3ex/4}
\newlength\blockExtraHeightBGFdAFBeBbFFBFEEJebcDJJGbDdFGdJEGd
\settototalheight\blockExtraHeightBGFdAFBeBbFFBFEEJebcDJJGbDdFGdJEGd{\parbox{0.4\linewidth}{$\mi{access\_token}$}}
\setlength\blockExtraHeightBGFdAFBeBbFFBFEEJebcDJJGbDdFGdJEGd{\dimexpr \blockExtraHeightBGFdAFBeBbFFBFEEJebcDJJGbDdFGdJEGd - 3ex/4}
\newlength\blockExtraHeightBHFdAFBeBbFFBFEEJebcDJJGbDdFGdJEGd
\settototalheight\blockExtraHeightBHFdAFBeBbFFBFEEJebcDJJGbDdFGdJEGd{\parbox{0.4\linewidth}{$\mi{user\_id}$, $\mi{client\_id}$}}
\setlength\blockExtraHeightBHFdAFBeBbFFBFEEJebcDJJGbDdFGdJEGd{\dimexpr \blockExtraHeightBHFdAFBeBbFFBFEEJebcDJJGbDdFGdJEGd - 3ex/4}
\newlength\blockExtraHeightBIFdAFBeBbFFBFEEJebcDJJGbDdFGdJEGd
\settototalheight\blockExtraHeightBIFdAFBeBbFFBFEEJebcDJJGbDdFGdJEGd{\parbox{0.4\linewidth}{$\mi{session\_cookie}$}}
\setlength\blockExtraHeightBIFdAFBeBbFFBFEEJebcDJJGbDdFGdJEGd{\dimexpr \blockExtraHeightBIFdAFBeBbFFBFEEJebcDJJGbDdFGdJEGd - 3ex/4}

 \begin{tikzpicture}
   \tikzstyle{xhrArrow} = [color=blue,decoration={markings, mark=at
    position 1 with {\arrow[color=blue]{triangle 45}}}, preaction
  = {decorate}]

    \matrix [column sep={3.5cm,between origins}, row sep=4.5ex]
  {

    \node[draw,anchor=base](Browser-start-0){Browser}; & \node[draw,anchor=base](RP-start-0){RP}; & \node[draw,anchor=base](Attacker-start-0){Attacker (AIdP)}; & \node[draw,anchor=base](HIdP-start-0){HIdP};\\
\node(Browser-0){}; & \node(RP-0){}; & \node(Attacker-0){}; & \node(HIdP-0){};\\[\blockExtraHeightAFdAFBeBbFFBFEEJebcDJJGbDdFGdJEGd]
\node(Browser-1){}; & \node(RP-1){}; & \node(Attacker-1){}; & \node(HIdP-1){};\\[\blockExtraHeightBFdAFBeBbFFBFEEJebcDJJGbDdFGdJEGd]
\node(Browser-2){}; & \node(RP-2){}; & \node(Attacker-2){}; & \node(HIdP-2){};\\[\blockExtraHeightCFdAFBeBbFFBFEEJebcDJJGbDdFGdJEGd]
\node(Browser-3){}; & \node(RP-3){}; & \node(Attacker-3){}; & \node(HIdP-3){};\\[\blockExtraHeightDFdAFBeBbFFBFEEJebcDJJGbDdFGdJEGd]
\node(Browser-4){}; & \node(RP-4){}; & \node(Attacker-4){}; & \node(HIdP-4){};\\[\blockExtraHeightEFdAFBeBbFFBFEEJebcDJJGbDdFGdJEGd]
\node(Browser-5){}; & \node(RP-5){}; & \node(Attacker-5){}; & \node(HIdP-5){};\\[\blockExtraHeightFFdAFBeBbFFBFEEJebcDJJGbDdFGdJEGd]
\node(Browser-6){}; & \node(RP-6){}; & \node(Attacker-6){}; & \node(HIdP-6){};\\[\blockExtraHeightGFdAFBeBbFFBFEEJebcDJJGbDdFGdJEGd]
\node(Browser-7){}; & \node(RP-7){}; & \node(Attacker-7){}; & \node(HIdP-7){};\\[\blockExtraHeightHFdAFBeBbFFBFEEJebcDJJGbDdFGdJEGd]
\node(Browser-8){}; & \node(RP-8){}; & \node(Attacker-8){}; & \node(HIdP-8){};\\[\blockExtraHeightIFdAFBeBbFFBFEEJebcDJJGbDdFGdJEGd]
\node(Browser-9){}; & \node(RP-9){}; & \node(Attacker-9){}; & \node(HIdP-9){};\\[\blockExtraHeightJFdAFBeBbFFBFEEJebcDJJGbDdFGdJEGd]
\node(Browser-10){}; & \node(RP-10){}; & \node(Attacker-10){}; & \node(HIdP-10){};\\[\blockExtraHeightBAFdAFBeBbFFBFEEJebcDJJGbDdFGdJEGd]
\node(Browser-11){}; & \node(RP-11){}; & \node(Attacker-11){}; & \node(HIdP-11){};\\[\blockExtraHeightBBFdAFBeBbFFBFEEJebcDJJGbDdFGdJEGd]
\node(Browser-12){}; & \node(RP-12){}; & \node(Attacker-12){}; & \node(HIdP-12){};\\[\blockExtraHeightBCFdAFBeBbFFBFEEJebcDJJGbDdFGdJEGd]
\node(Browser-13){}; & \node(RP-13){}; & \node(Attacker-13){}; & \node(HIdP-13){};\\[\blockExtraHeightBDFdAFBeBbFFBFEEJebcDJJGbDdFGdJEGd]
\node(Browser-14){}; & \node(RP-14){}; & \node(Attacker-14){}; & \node(HIdP-14){};\\[\blockExtraHeightBEFdAFBeBbFFBFEEJebcDJJGbDdFGdJEGd]
\node(Browser-15){}; & \node(RP-15){}; & \node(Attacker-15){}; & \node(HIdP-15){};\\[\blockExtraHeightBFFdAFBeBbFFBFEEJebcDJJGbDdFGdJEGd]
\node(Browser-16){}; & \node(RP-16){}; & \node(Attacker-16){}; & \node(HIdP-16){};\\[\blockExtraHeightBGFdAFBeBbFFBFEEJebcDJJGbDdFGdJEGd]
\node(Browser-17){}; & \node(RP-17){}; & \node(Attacker-17){}; & \node(HIdP-17){};\\[\blockExtraHeightBHFdAFBeBbFFBFEEJebcDJJGbDdFGdJEGd]
\node(Browser-18){}; & \node(RP-18){}; & \node(Attacker-18){}; & \node(HIdP-18){};\\[\blockExtraHeightBIFdAFBeBbFFBFEEJebcDJJGbDdFGdJEGd]
\node[draw,anchor=base](Browser-end-1){/Browser}; & \node[draw,anchor=base](RP-end-1){/RP}; & \node[draw,anchor=base](Attacker-end-1){/Attacker (AIdP)}; & \node[draw,anchor=base](HIdP-end-1){/HIdP};\\
};
\draw[->] (Browser-0) to node [above=2.6pt, anchor=base]{\protostep{ig-att-start-req} \textbf{POST /start}} node [below=-8pt, text width=0.5\linewidth, anchor=base]{\begin{center} $\mi{idp}$\end{center}} (Attacker-0); 

\draw[->] (Attacker-1) to node [above=2.6pt, anchor=base]{\protostep{ig-att-start-req-manipulated} \textbf{POST /start}} node [below=-8pt, text width=0.5\linewidth, anchor=base]{\begin{center} $\mi{attacker}$\end{center}} (RP-1); 

\draw[->] (RP-2) to node [above=2.6pt, anchor=base]{\protostep{ig-att-start-resp} \textbf{Response}} node [below=-8pt, text width=0.5\linewidth, anchor=base]{\begin{center} Redirect to Attacker /authEP with $\mi{client\_id}'$, $\mi{redirect\_uri}$, $\mi{state}$\end{center}} (Attacker-2); 

\draw[->] (Attacker-3) to node [above=2.6pt, anchor=base]{\protostep{ig-att-start-resp-manipulated} \textbf{Response}} node [below=-8pt, text width=0.5\linewidth, anchor=base]{\begin{center} Redirect to HIdP /authEP with $\mi{client\_id}$, $\mi{redirect\_uri}$, $\mi{state}$\end{center}} (Browser-3); 

\draw[->] (Browser-4) to node [above=2.6pt, anchor=base]{\protostep{ig-att-idp-auth-req-1} \textbf{GET /authEP}} node [below=-8pt, text width=0.5\linewidth, anchor=base]{\begin{center} $\mi{client\_id}$, $\mi{redirect\_uri}$, $\mi{state}$\end{center}} (HIdP-4); 

\draw[->] (HIdP-5) to node [above=2.6pt, anchor=base]{\protostep{ig-att-idp-auth-resp-1} \textbf{Response}} node [below=-8pt, text width=0.5\linewidth, anchor=base]{\begin{center} \end{center}} (Browser-5); 

\draw[->] (Browser-6) to node [above=2.6pt, anchor=base]{\protostep{ig-att-idp-auth-req-2} \textbf{POST /authEP}} node [below=-8pt, text width=0.5\linewidth, anchor=base]{\begin{center} $\mi{username}$, $\mi{password}$\end{center}} (HIdP-6); 

\draw[->] (HIdP-7) to node [above=2.6pt, anchor=base]{\protostep{ig-att-idp-auth-resp-2} \textbf{Response}} node [below=-8pt, text width=0.5\linewidth, anchor=base]{\begin{center} Redirect to RP $\mi{redirect\_uri}$, fragment: $\mi{access\_token}$, $\mi{state}$\end{center}} (Browser-7); 

\draw[->] (Browser-8) to node [above=2.6pt, anchor=base]{\protostep{ig-att-redir-ep-req} \textbf{GET $\mi{redirect\_uri}$}} node [below=-8pt, text width=0.5\linewidth, anchor=base]{\begin{center} \end{center}} (RP-8); 

\draw[->] (RP-9) to node [above=2.6pt, anchor=base]{\protostep{ig-att-redir-ep-resp} \textbf{Response}} node [below=-8pt, text width=0.5\linewidth, anchor=base]{\begin{center} \end{center}} (Browser-9); 

\draw[->] (Browser-10) to node [above=2.6pt, anchor=base]{\protostep{ig-att-redir-ep-token-req} \textbf{POST /token}} node [below=-8pt, text width=0.5\linewidth, anchor=base]{\begin{center} $\mi{access\_token}$, $\mi{state}$\end{center}} (RP-10); 

\draw[->] (RP-11) to node [above=2.6pt, anchor=base]{\protostep{ig-att-introspect-req} \textbf{GET /resource}} node [below=-8pt, text width=0.5\linewidth, anchor=base]{\begin{center} $\mi{access\_token}$\end{center}} (Attacker-11); 

\draw [dashed] (Browser-12.west) -- (HIdP-12.east);
\node[draw=none,anchor=northwest,below=2ex,right=1ex] at (Browser-12.west) {Continued attack to break authentication:};

\draw[->] (Attacker-13) to node [above=2.6pt, anchor=base]{\protostep{ig-att-start-attacker-req} \textbf{POST /start}} node [below=-8pt, text width=0.5\linewidth, anchor=base]{\begin{center} $\mi{idp}$\end{center}} (RP-13); 

\draw[->] (RP-14) to node [above=2.6pt, anchor=base]{\protostep{ig-att-start-attacker-resp} \textbf{Response}} node [below=-8pt, text width=0.5\linewidth, anchor=base]{\begin{center} Redirect to HIdP /authEP with $\mi{client\_id}$, $\mi{redirect\_uri}$, $\mi{state}'$\end{center}} (Attacker-14); 

\draw[->] (Attacker-15) to node [above=2.6pt, anchor=base]{\protostep{ig-att-redir-ep-token-attacker-req} \textbf{POST /token}} node [below=-8pt, text width=0.5\linewidth, anchor=base]{\begin{center} $\mi{access\_token}$, $\mi{state}'$\end{center}} (RP-15); 

\draw[->] (RP-16) to node [above=2.6pt, anchor=base]{\protostep{ig-att-introspect-attacker-req} \textbf{GET /introspectionEP}} node [below=-8pt, text width=0.5\linewidth, anchor=base]{\begin{center} $\mi{access\_token}$\end{center}} (HIdP-16); 

\draw[->] (HIdP-17) to node [above=2.6pt, anchor=base]{\protostep{ig-att-introspect-attacker-resp} \textbf{Response}} node [below=-8pt, text width=0.5\linewidth, anchor=base]{\begin{center} $\mi{user\_id}$, $\mi{client\_id}$\end{center}} (RP-17); 

\draw[->] (RP-18) to node [above=2.6pt, anchor=base]{\protostep{ig-att-redir-ep-token-resp} \textbf{Response}} node [below=-8pt, text width=0.5\linewidth, anchor=base]{\begin{center} $\mi{session\_cookie}$\end{center}} (Attacker-18); 

\begin{pgfonlayer}{background}
\draw [color=gray] (Browser-start-0) -- (Browser-end-1);
\draw [color=gray] (RP-start-0) -- (RP-end-1);
\draw [color=gray] (Attacker-start-0) -- (Attacker-end-1);
\draw [color=gray] (HIdP-start-0) -- (HIdP-end-1);
\end{pgfonlayer}
\end{tikzpicture}}

%% file: appendix-openid-connect.tex
\subsection{OpenID Connect and the Attacks on this Standard}
\label{app:openid-connect}

\begin{figure}[p!]
  \centering
  \input{figure-openid-connect-hybrid-flow}
  \caption{OpenID Connect 1.0 hybrid mode with discovery and dynamic client registration}
  \label{fig:oic-hf}
\end{figure}

We here provide a more detailed description of the OpenID Connect
standard as well as on the two attacks, 307 redirect and IdP mix-up,
on it.

\subsubsection{Modes and Protocol Flow}
OpenID Connect makes use of the OAuth authorization code mode and the
implicit mode (both OAuth modes constitute an OpenID Connect mode),
but also introduces a new \emph{hybrid} mode, which combines both
modes.

\paragraph{Overview}
From a high-level perspective, first, the RP retrieves meta data about
the IdP, such as the URLs of the IdP used in the protocol. This is the
information that is ``hard-wired'' in the manual, out-of-band
registration in a classic OAuth setup. Next, the RP automatically
registers itself as an OAuth client at the IdP (using OpenID Connect
dynamic client registration). Then, the OAuth protocol is started
(using one of the modes mentioned above). In addition to an access
token this (extended) run delivers a so-called \emph{id token} to
RP. The id token is issued by the IdP and contains a unique user
identifier along with several meta data, such as the intended receiver
(the RP) of the id token and the issuer of the id token (the IdP). The
id token is (optionally) signed by the IdP. Finally, the RP can
retrieve more meta data about the user at the \emph{userinfo} endpoint
at the IdP using the access token and consider the user to be logged
in.

\paragraph{Step-by-Step Protocol Flow}
In the step-by-step description below (see also
Figure~\ref{fig:oic-hf}), we focus on the hybrid mode only. First, the
user starts the login process by entering her email
address\footnote{Note that OpenID Connect also allows other types of
  user identifiers, such as a personal URL.} in her browser (at some
web page of an RP), which sends the email address to the RP
in~\refprotostep{oichf-start-req}.

Now, the RP uses the OpenID Connect Discovery protocol
\cite{openid-connect-discovery-1.0} to gain information about the IdP:
The RP uses the WebFinger \cite{rfc7033-webfinger} mechanism to discover
information about which IdP is responsible for this user. For this
discovery, the RP contacts the server of the user's email domain
(depicted as the same party as the IdP in the figure)
in~\refprotostep{oichf-wf-req}. The result of the WebFinger request
in~\refprotostep{oichf-wf-resp} contains the domain of the server
responsible for the OpenID Connect configuration (the IdP). The
configuration is requested from the IdP
in~\refprotostep{oichf-conf-req} and returned
in~\refprotostep{oichf-conf-resp}. The configuration contains meta
data about the IdP, including all endpoints at the IdP. This concludes
the OpenID Discovery in this login flow.

Next, if the RP is not registered at the IdP, the RP starts the OpenID
Connect dynamic client registration
\cite{openid-connect-dynamic-client-registration-1.0} protocol: the RP
contacts the IdP in~\refprotostep{oichf-reg-req} providing its
redirect URIs. Now, the IdP issues an (OAuth) client id and
(optionally) an (OAuth) client secret to the RP
in~\refprotostep{oichf-reg-resp}. This concludes the OpenID Connect
dynamic client registration.

Now, the core part of the OpenID Connect protocol (based on OAuth)
starts: the RP redirects the user's browser to the IdP
in~\refprotostep{oichf-start-resp}. This redirect contains information
that the hybrid mode is used and which tokens are requested. In this
description, we assume that an authorization code and an access token
are requested.\footnote{The Hybrid Flow allows to request several
  different combinations of authorization code, access token, and id
  token.} Also, this redirect contains the (OAuth) client id of the
RP, a redirect URI and a state value. As in the OAuth flows, this data
is sent to the IdP~\refprotostep{oichf-idp-auth-req-1}, the user
authenticates to the IdP~\refprotostep{oichf-idp-auth-resp-1},
\refprotostep{oichf-idp-auth-req-2}, and the IdP redirects the user's
browser back to the RP in~\refprotostep{oichf-idp-auth-resp-2}
and~\refprotostep{oichf-redir-ep-req} (using the redirect URI from the
request in~\refprotostep{oichf-idp-auth-req-1}). This redirect
contains an authorization code, an access token, and the state value
in the fragment part of the URL.\footnote{Note that depending on the
  parameters in step~\refprotostep{oichf-idp-auth-req-1}, also an id
  token may be contained in the fragment part of the URL.} Now, the RP
in~\refprotostep{oichf-redir-ep-resp} sends a document containing
JavaScript code which sends the parameters contained in the fragment back
to the RP (in~\refprotostep{oichf-redir-ep-token-req}). If the state
value matches, the RP contacts the IdP
in~\refprotostep{oichf-token-req} with the received authorization
code, its (OAuth) client id, its (OAuth) client secret, and the
redirect URI used to obtain the authorization code. The IdP sends a
response with the same or a fresh access token and an id token to the
RP in~\refprotostep{oichf-token-resp}. Now, the RP retrieves the key
that was used to sign the id token from the IdP
in~\refprotostep{oichf-jwks-req} and~\refprotostep{oichf-jwks-resp}
and verifies the id token's signature. As the id token typically 
contains only a unique user identifier, but no other meta data about the
user, RP requests this meta data (such as nickname, birthday, or
address) from the IdP in~\refprotostep{oichf-userinfo-req}
and~\refprotostep{oichf-userinfo-resp} using one of the authorization
tokens received before. Finally, the RP considers the user to be
logged in and may set a session cookie at the user's browser
in~\refprotostep{oichf-redir-ep-token-resp}.

Note that the authorization code mode and the implicit mode are
similar to the hybrid mode: Roughly speaking, the
Steps~\refprotostep{oichf-idp-auth-resp-2}--\refprotostep{oichf-token-resp}
of the OpenID Connect hybrid mode are replaced by the corresponding
steps of the OAuth authorization code or implicit mode, respectively.
These OAuth modes are then extended with the transfer of an id token.
In the authorization code mode, the id token is appended to the
response~\refprotostep{acg-token-resp} of
Figure~\ref{fig:oauth-auth-code-grant} and in the implicit mode, the
id token is appended to the fragment of the redirect URI
in~\refprotostep{ig-idp-auth-resp-2} of
Figure~\ref{fig:oauth-implicit-grant} (and later sent to the RP in
Step~\refprotostep{ig-redir-ep-token-req}).

\subsubsection{The 307 Redirect Attack}
The 307 redirect attack presented in Section~\ref{sec:307-redirect}
can also be applied to OpenID Connect. Note that the critical part of
OAuth, namely the redirect of the user's browser from the IdP to the
RP after the authentication of the user to the IdP, is also present in
OpenID Connect. Hence, if the IdP uses an HTTP status 307 redirect
immediately after the user's browser has transferred the user's
credentials to IdP in a POST request, the RP receives these
credentials.

\subsubsection{The IdP Mix-Up Attack}
\begin{figure}[p!]
  \input{figure-openid-connect-hybrid-flow-attack}
  \caption{Attack on OpenID Connect 1.0 hybrid mode with discovery and
    dynamic client registration}
  \label{fig:oic-hf-att}
\end{figure}
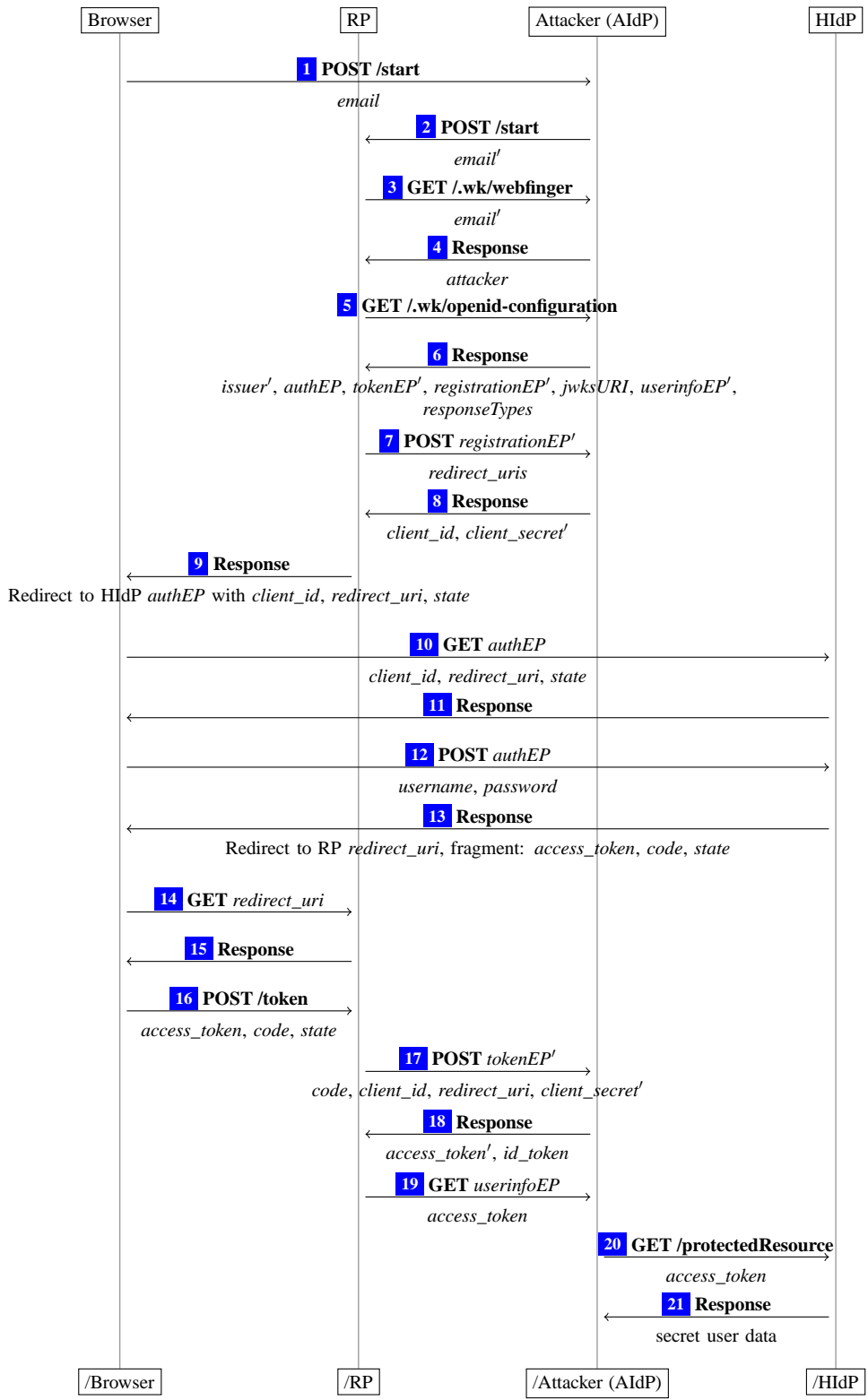
When applying the attack presented in
Section~\ref{sec:malicious-idp-mitm} to OpenID Connect, the attacker needs to circumvent \df{better: deal with}
some additional security measures: In the implicit mode of OpenID
Connect, an $\mi{id\_token}$ (as described above) is sent along with
$\mi{access\_token}$ in the redirect from HIdP to the RP. As this
redirect might use HTTPS, the attacker cannot inspect or modify the
corresponding network messages. As mentioned above, the id token
contains the domain of the issuer of both, the access token and the id
token. Therefore, the RP can detect that the user did not use AIdP
(which the RP redirected to).

An attacker could try to use the authorization code mode of OpenID
Connect to mount a similar attack as described above. In this case,
however, the attacker does not learn a valid access token for the
user's account at HIdP if a client secret is used.

In the hybrid mode, however, an attacker can learn an access token and
mount the attack as follows (see also Figure~\ref{fig:oic-hf-att}):

As above, the user first visits the RP. When the user sends her email
address to the RP in order to
login~\refprotostep{oichf-att-start-req}, the attacker manipulates the
domain part of the email address, to be the domain of
AIdP~\refprotostep{oichf-att-start-req-manipulated}. The RP then looks
up the IdP to be used (which is now AIdP) using the WebFinger protocol
in Steps~\refprotostep{oichf-att-wf-req}
and~\refprotostep{oichf-att-wf-resp}. The RP fetches the OpenID
Connect configuration from the
attacker~(\refprotostep{oichf-att-conf-req}
and~\refprotostep{oichf-att-conf-resp}). In this document, the
attacker states that the authorization endpoint is located at HIdP
while all other endpoints are located at the attacker. Using
parameters not shown in the figures, the attacker can also state that
this IdP does not support delivering an id token in the redirect and
can state that no signatures are supported. Since no signatures need
to be checked, also the key retrieval is skipped in the protocol.

After retrieving the OpenID configuration, the RP registers at AIdP, as
the attacker uses a domain previously unknown to the RP. (If the domain
was known to the RP, this step would be skipped.) The attacker issues the
same $\mi{client\_id}$
with which the RP is registered at HIdP~(\refprotostep{oichf-att-reg-req}
and~\refprotostep{oichf-att-reg-resp}). Now, the RP redirects the user's
browser to HIdP in order to log in. After the user authenticated to
HIdP, HIdP redirects the user's browser back to the RP. The fragment part
of the URL contains an authorization code and an access
token~\refprotostep{oichf-att-redir-ep-req}. The RP then sends the
authorization code to the attacker
in~\refprotostep{oichf-att-token-req}. 

If the RP does not have a client secret registered at HIdP, the
attacker can redeem this authorization code at HIdP in order to
receive an access token to access the honest user's protected
resources at HIdP. This breaks the authorization of OpenID Connect
(compare the OAuth authorization property in
Section~\ref{sec:secur-auth-prop}).

Alternatively, the attacker responds to the RP
with a faked access token and a faked id
token~\refprotostep{oichf-att-token-resp} (which the attacker can
create, because he controls all security settings for this id token,
see Step~\refprotostep{oichf-att-conf-resp}).

Next, the RP retrieves other meta information about the user from
AIdP. The RP is now in possession of two access tokens. The OpenID
Connect standard explicitly allows this situation, but fails to state
which access token has to be used in subsequent requests. The RP can
now chose either of the access tokens for the next steps, with
different outcomes for the attacker:

\paragraph{First Access Token is Selected}
In this case the access token originating from HIdP is selected by the
RP and sent to the attacker~\refprotostep{oichf-att-userinfo-req}.
(This behavior was observed by us in the real-world implementation
mod\_auth\_openidc.)

Now the attacker can use this access token to access other protected
resources of the user at HIdP. This breaks authorization for OpenID
Connect (compare our OAuth authorization property in
Section~\ref{sec:secur-auth-prop}).

\paragraph{Second Access Token is Selected}
In this case the access token originating from AIdP is selected. This
means that the attacker does not learn a valid access token for
HIdP. The attacker can, however, reuse the authorization code for
HIdP, which he learned in~\refprotostep{oichf-att-token-req} and which
is still valid as it has not been redeemed at HIdP, yet. Using this
method, the attacker can impersonate the honest user at the RP. To
accomplish this, the attacker starts a new login flow at the RP with
the user's email address. In
Step~\refprotostep{oichf-redir-ep-token-req} of
Figure~\ref{fig:oic-hf}, he provides the authorization code he has
learned along with some (invalid) access token and the state from his
(new) login flow to the RP. The RP then requests an access token and
an id token from HIdP with this (still valid) authorization code. The
RP receives a valid access token and a valid id token (for the honest
user) from HIdP. As the RP uses this valid access token in this case,
all subsequent requests from the RP to HIdP are successful and the RP
receives the user id of the honest user, the RP considers the attacker
to be logged in as the honest user. This breaks the authentication of
OpenID Connect (compare our OAuth authentication property in
Section~\ref{sec:secur-auth-prop}).

%% file: figure-openid-connect-hybrid-flow.tex
 \scriptsize{ \newlength\blockExtraHeightAHdCbGGaBdFCaEDAFbEGbfEaJIFFGFGfe
\settototalheight\blockExtraHeightAHdCbGGaBdFCaEDAFbEGbfEaJIFFGFGfe{\parbox{0.4\linewidth}{$\mi{email}$}}
\setlength\blockExtraHeightAHdCbGGaBdFCaEDAFbEGbfEaJIFFGFGfe{\dimexpr \blockExtraHeightAHdCbGGaBdFCaEDAFbEGbfEaJIFFGFGfe - 3ex/4}
\newlength\blockExtraHeightBHdCbGGaBdFCaEDAFbEGbfEaJIFFGFGfe
\settototalheight\blockExtraHeightBHdCbGGaBdFCaEDAFbEGbfEaJIFFGFGfe{\parbox{0.4\linewidth}{$\mi{email}$}}
\setlength\blockExtraHeightBHdCbGGaBdFCaEDAFbEGbfEaJIFFGFGfe{\dimexpr \blockExtraHeightBHdCbGGaBdFCaEDAFbEGbfEaJIFFGFGfe - 3ex/4}
\newlength\blockExtraHeightCHdCbGGaBdFCaEDAFbEGbfEaJIFFGFGfe
\settototalheight\blockExtraHeightCHdCbGGaBdFCaEDAFbEGbfEaJIFFGFGfe{\parbox{0.4\linewidth}{$\mi{idp}$}}
\setlength\blockExtraHeightCHdCbGGaBdFCaEDAFbEGbfEaJIFFGFGfe{\dimexpr \blockExtraHeightCHdCbGGaBdFCaEDAFbEGbfEaJIFFGFGfe - 3ex/4}
\newlength\blockExtraHeightDHdCbGGaBdFCaEDAFbEGbfEaJIFFGFGfe
\settototalheight\blockExtraHeightDHdCbGGaBdFCaEDAFbEGbfEaJIFFGFGfe{\parbox{0.4\linewidth}{}}
\setlength\blockExtraHeightDHdCbGGaBdFCaEDAFbEGbfEaJIFFGFGfe{\dimexpr \blockExtraHeightDHdCbGGaBdFCaEDAFbEGbfEaJIFFGFGfe - 3ex/4}
\newlength\blockExtraHeightEHdCbGGaBdFCaEDAFbEGbfEaJIFFGFGfe
\settototalheight\blockExtraHeightEHdCbGGaBdFCaEDAFbEGbfEaJIFFGFGfe{\parbox{0.4\linewidth}{$\mi{issuer}$, $\mi{authEP}$, $\mi{tokenEP}$, $\mi{registrationEP}$, $\mi{jwksURI}$, $\mi{userinfoEP}$}}
\setlength\blockExtraHeightEHdCbGGaBdFCaEDAFbEGbfEaJIFFGFGfe{\dimexpr \blockExtraHeightEHdCbGGaBdFCaEDAFbEGbfEaJIFFGFGfe - 3ex/4}
\newlength\blockExtraHeightFHdCbGGaBdFCaEDAFbEGbfEaJIFFGFGfe
\settototalheight\blockExtraHeightFHdCbGGaBdFCaEDAFbEGbfEaJIFFGFGfe{\parbox{0.4\linewidth}{$\mi{redirect\_uris}$}}
\setlength\blockExtraHeightFHdCbGGaBdFCaEDAFbEGbfEaJIFFGFGfe{\dimexpr \blockExtraHeightFHdCbGGaBdFCaEDAFbEGbfEaJIFFGFGfe - 3ex/4}
\newlength\blockExtraHeightGHdCbGGaBdFCaEDAFbEGbfEaJIFFGFGfe
\settototalheight\blockExtraHeightGHdCbGGaBdFCaEDAFbEGbfEaJIFFGFGfe{\parbox{0.4\linewidth}{$\mi{client\_id}$, $\mi{client\_secret}$}}
\setlength\blockExtraHeightGHdCbGGaBdFCaEDAFbEGbfEaJIFFGFGfe{\dimexpr \blockExtraHeightGHdCbGGaBdFCaEDAFbEGbfEaJIFFGFGfe - 3ex/4}
\newlength\blockExtraHeightHHdCbGGaBdFCaEDAFbEGbfEaJIFFGFGfe
\settototalheight\blockExtraHeightHHdCbGGaBdFCaEDAFbEGbfEaJIFFGFGfe{\parbox{0.4\linewidth}{Redirect to IdP $\mi{authEP}$ with $\mi{client\_id}$, $\mi{redirect\_uri}$, $\mi{state}$}}
\setlength\blockExtraHeightHHdCbGGaBdFCaEDAFbEGbfEaJIFFGFGfe{\dimexpr \blockExtraHeightHHdCbGGaBdFCaEDAFbEGbfEaJIFFGFGfe - 3ex/4}
\newlength\blockExtraHeightIHdCbGGaBdFCaEDAFbEGbfEaJIFFGFGfe
\settototalheight\blockExtraHeightIHdCbGGaBdFCaEDAFbEGbfEaJIFFGFGfe{\parbox{0.4\linewidth}{$\mi{client\_id}$, $\mi{redirect\_uri}$, $\mi{state}$}}
\setlength\blockExtraHeightIHdCbGGaBdFCaEDAFbEGbfEaJIFFGFGfe{\dimexpr \blockExtraHeightIHdCbGGaBdFCaEDAFbEGbfEaJIFFGFGfe - 3ex/4}
\newlength\blockExtraHeightJHdCbGGaBdFCaEDAFbEGbfEaJIFFGFGfe
\settototalheight\blockExtraHeightJHdCbGGaBdFCaEDAFbEGbfEaJIFFGFGfe{\parbox{0.4\linewidth}{}}
\setlength\blockExtraHeightJHdCbGGaBdFCaEDAFbEGbfEaJIFFGFGfe{\dimexpr \blockExtraHeightJHdCbGGaBdFCaEDAFbEGbfEaJIFFGFGfe - 3ex/4}
\newlength\blockExtraHeightBAHdCbGGaBdFCaEDAFbEGbfEaJIFFGFGfe
\settototalheight\blockExtraHeightBAHdCbGGaBdFCaEDAFbEGbfEaJIFFGFGfe{\parbox{0.4\linewidth}{$\mi{username}$, $\mi{password}$}}
\setlength\blockExtraHeightBAHdCbGGaBdFCaEDAFbEGbfEaJIFFGFGfe{\dimexpr \blockExtraHeightBAHdCbGGaBdFCaEDAFbEGbfEaJIFFGFGfe - 3ex/4}
\newlength\blockExtraHeightBBHdCbGGaBdFCaEDAFbEGbfEaJIFFGFGfe
\settototalheight\blockExtraHeightBBHdCbGGaBdFCaEDAFbEGbfEaJIFFGFGfe{\parbox{0.4\linewidth}{Redirect to RP $\mi{redirect\_uri}$, fragment: $\mi{access\_token}$, $\mi{code}$, $\mi{state}$}}
\setlength\blockExtraHeightBBHdCbGGaBdFCaEDAFbEGbfEaJIFFGFGfe{\dimexpr \blockExtraHeightBBHdCbGGaBdFCaEDAFbEGbfEaJIFFGFGfe - 3ex/4}
\newlength\blockExtraHeightBCHdCbGGaBdFCaEDAFbEGbfEaJIFFGFGfe
\settototalheight\blockExtraHeightBCHdCbGGaBdFCaEDAFbEGbfEaJIFFGFGfe{\parbox{0.4\linewidth}{}}
\setlength\blockExtraHeightBCHdCbGGaBdFCaEDAFbEGbfEaJIFFGFGfe{\dimexpr \blockExtraHeightBCHdCbGGaBdFCaEDAFbEGbfEaJIFFGFGfe - 3ex/4}
\newlength\blockExtraHeightBDHdCbGGaBdFCaEDAFbEGbfEaJIFFGFGfe
\settototalheight\blockExtraHeightBDHdCbGGaBdFCaEDAFbEGbfEaJIFFGFGfe{\parbox{0.4\linewidth}{}}
\setlength\blockExtraHeightBDHdCbGGaBdFCaEDAFbEGbfEaJIFFGFGfe{\dimexpr \blockExtraHeightBDHdCbGGaBdFCaEDAFbEGbfEaJIFFGFGfe - 3ex/4}
\newlength\blockExtraHeightBEHdCbGGaBdFCaEDAFbEGbfEaJIFFGFGfe
\settototalheight\blockExtraHeightBEHdCbGGaBdFCaEDAFbEGbfEaJIFFGFGfe{\parbox{0.4\linewidth}{$\mi{access\_token}$, $\mi{code}$, $\mi{state}$}}
\setlength\blockExtraHeightBEHdCbGGaBdFCaEDAFbEGbfEaJIFFGFGfe{\dimexpr \blockExtraHeightBEHdCbGGaBdFCaEDAFbEGbfEaJIFFGFGfe - 3ex/4}
\newlength\blockExtraHeightBFHdCbGGaBdFCaEDAFbEGbfEaJIFFGFGfe
\settototalheight\blockExtraHeightBFHdCbGGaBdFCaEDAFbEGbfEaJIFFGFGfe{\parbox{0.4\linewidth}{$\mi{code}$, $\mi{client\_id}$, $\mi{redirect\_uri}$, $\mi{client\_secret}$}}
\setlength\blockExtraHeightBFHdCbGGaBdFCaEDAFbEGbfEaJIFFGFGfe{\dimexpr \blockExtraHeightBFHdCbGGaBdFCaEDAFbEGbfEaJIFFGFGfe - 3ex/4}
\newlength\blockExtraHeightBGHdCbGGaBdFCaEDAFbEGbfEaJIFFGFGfe
\settototalheight\blockExtraHeightBGHdCbGGaBdFCaEDAFbEGbfEaJIFFGFGfe{\parbox{0.4\linewidth}{$\mi{access\_token}'$, $\mi{id\_token}$}}
\setlength\blockExtraHeightBGHdCbGGaBdFCaEDAFbEGbfEaJIFFGFGfe{\dimexpr \blockExtraHeightBGHdCbGGaBdFCaEDAFbEGbfEaJIFFGFGfe - 3ex/4}
\newlength\blockExtraHeightBHHdCbGGaBdFCaEDAFbEGbfEaJIFFGFGfe
\settototalheight\blockExtraHeightBHHdCbGGaBdFCaEDAFbEGbfEaJIFFGFGfe{\parbox{0.4\linewidth}{}}
\setlength\blockExtraHeightBHHdCbGGaBdFCaEDAFbEGbfEaJIFFGFGfe{\dimexpr \blockExtraHeightBHHdCbGGaBdFCaEDAFbEGbfEaJIFFGFGfe - 3ex/4}
\newlength\blockExtraHeightBIHdCbGGaBdFCaEDAFbEGbfEaJIFFGFGfe
\settototalheight\blockExtraHeightBIHdCbGGaBdFCaEDAFbEGbfEaJIFFGFGfe{\parbox{0.4\linewidth}{$\mi{pubSignKey}$}}
\setlength\blockExtraHeightBIHdCbGGaBdFCaEDAFbEGbfEaJIFFGFGfe{\dimexpr \blockExtraHeightBIHdCbGGaBdFCaEDAFbEGbfEaJIFFGFGfe - 3ex/4}
\newlength\blockExtraHeightBJHdCbGGaBdFCaEDAFbEGbfEaJIFFGFGfe
\settototalheight\blockExtraHeightBJHdCbGGaBdFCaEDAFbEGbfEaJIFFGFGfe{\parbox{0.4\linewidth}{$\mi{access\_token}$}}
\setlength\blockExtraHeightBJHdCbGGaBdFCaEDAFbEGbfEaJIFFGFGfe{\dimexpr \blockExtraHeightBJHdCbGGaBdFCaEDAFbEGbfEaJIFFGFGfe - 3ex/4}
\newlength\blockExtraHeightCAHdCbGGaBdFCaEDAFbEGbfEaJIFFGFGfe
\settototalheight\blockExtraHeightCAHdCbGGaBdFCaEDAFbEGbfEaJIFFGFGfe{\parbox{0.4\linewidth}{$\mi{user\_id}$}}
\setlength\blockExtraHeightCAHdCbGGaBdFCaEDAFbEGbfEaJIFFGFGfe{\dimexpr \blockExtraHeightCAHdCbGGaBdFCaEDAFbEGbfEaJIFFGFGfe - 3ex/4}
\newlength\blockExtraHeightCBHdCbGGaBdFCaEDAFbEGbfEaJIFFGFGfe
\settototalheight\blockExtraHeightCBHdCbGGaBdFCaEDAFbEGbfEaJIFFGFGfe{\parbox{0.4\linewidth}{$\mi{session\_cookie}$}}
\setlength\blockExtraHeightCBHdCbGGaBdFCaEDAFbEGbfEaJIFFGFGfe{\dimexpr \blockExtraHeightCBHdCbGGaBdFCaEDAFbEGbfEaJIFFGFGfe - 3ex/4}

 \begin{tikzpicture}
   \tikzstyle{xhrArrow} = [color=blue,decoration={markings, mark=at
    position 1 with {\arrow[color=blue]{triangle 45}}}, preaction
  = {decorate}]

    \matrix [column sep={6cm,between origins}, row sep=4.5ex]
  {

    \node[draw,anchor=base](Browser-start-0){Browser}; & \node[draw,anchor=base](RP-start-0){RP}; & \node[draw,anchor=base](IdP-start-0){IdP};\\
\node(Browser-0){}; & \node(RP-0){}; & \node(IdP-0){};\\[\blockExtraHeightAHdCbGGaBdFCaEDAFbEGbfEaJIFFGFGfe]
\node(Browser-1){}; & \node(RP-1){}; & \node(IdP-1){};\\[\blockExtraHeightBHdCbGGaBdFCaEDAFbEGbfEaJIFFGFGfe]
\node(Browser-2){}; & \node(RP-2){}; & \node(IdP-2){};\\[\blockExtraHeightCHdCbGGaBdFCaEDAFbEGbfEaJIFFGFGfe]
\node(Browser-3){}; & \node(RP-3){}; & \node(IdP-3){};\\[\blockExtraHeightDHdCbGGaBdFCaEDAFbEGbfEaJIFFGFGfe]
\node(Browser-4){}; & \node(RP-4){}; & \node(IdP-4){};\\[\blockExtraHeightEHdCbGGaBdFCaEDAFbEGbfEaJIFFGFGfe]
\node(Browser-5){}; & \node(RP-5){}; & \node(IdP-5){};\\[\blockExtraHeightFHdCbGGaBdFCaEDAFbEGbfEaJIFFGFGfe]
\node(Browser-6){}; & \node(RP-6){}; & \node(IdP-6){};\\[\blockExtraHeightGHdCbGGaBdFCaEDAFbEGbfEaJIFFGFGfe]
\node(Browser-7){}; & \node(RP-7){}; & \node(IdP-7){};\\[\blockExtraHeightHHdCbGGaBdFCaEDAFbEGbfEaJIFFGFGfe]
\node(Browser-8){}; & \node(RP-8){}; & \node(IdP-8){};\\[\blockExtraHeightIHdCbGGaBdFCaEDAFbEGbfEaJIFFGFGfe]
\node(Browser-9){}; & \node(RP-9){}; & \node(IdP-9){};\\[\blockExtraHeightJHdCbGGaBdFCaEDAFbEGbfEaJIFFGFGfe]
\node(Browser-10){}; & \node(RP-10){}; & \node(IdP-10){};\\[\blockExtraHeightBAHdCbGGaBdFCaEDAFbEGbfEaJIFFGFGfe]
\node(Browser-11){}; & \node(RP-11){}; & \node(IdP-11){};\\[\blockExtraHeightBBHdCbGGaBdFCaEDAFbEGbfEaJIFFGFGfe]
\node(Browser-12){}; & \node(RP-12){}; & \node(IdP-12){};\\[\blockExtraHeightBCHdCbGGaBdFCaEDAFbEGbfEaJIFFGFGfe]
\node(Browser-13){}; & \node(RP-13){}; & \node(IdP-13){};\\[\blockExtraHeightBDHdCbGGaBdFCaEDAFbEGbfEaJIFFGFGfe]
\node(Browser-14){}; & \node(RP-14){}; & \node(IdP-14){};\\[\blockExtraHeightBEHdCbGGaBdFCaEDAFbEGbfEaJIFFGFGfe]
\node(Browser-15){}; & \node(RP-15){}; & \node(IdP-15){};\\[\blockExtraHeightBFHdCbGGaBdFCaEDAFbEGbfEaJIFFGFGfe]
\node(Browser-16){}; & \node(RP-16){}; & \node(IdP-16){};\\[\blockExtraHeightBGHdCbGGaBdFCaEDAFbEGbfEaJIFFGFGfe]
\node(Browser-17){}; & \node(RP-17){}; & \node(IdP-17){};\\[\blockExtraHeightBHHdCbGGaBdFCaEDAFbEGbfEaJIFFGFGfe]
\node(Browser-18){}; & \node(RP-18){}; & \node(IdP-18){};\\[\blockExtraHeightBIHdCbGGaBdFCaEDAFbEGbfEaJIFFGFGfe]
\node(Browser-19){}; & \node(RP-19){}; & \node(IdP-19){};\\[\blockExtraHeightBJHdCbGGaBdFCaEDAFbEGbfEaJIFFGFGfe]
\node(Browser-20){}; & \node(RP-20){}; & \node(IdP-20){};\\[\blockExtraHeightCAHdCbGGaBdFCaEDAFbEGbfEaJIFFGFGfe]
\node(Browser-21){}; & \node(RP-21){}; & \node(IdP-21){};\\[\blockExtraHeightCBHdCbGGaBdFCaEDAFbEGbfEaJIFFGFGfe]
\node[draw,anchor=base](Browser-end-1){/Browser}; & \node[draw,anchor=base](RP-end-1){/RP}; & \node[draw,anchor=base](IdP-end-1){/IdP};\\
};
\draw[->] (Browser-0) to node [above=2.6pt, anchor=base]{\protostep{oichf-start-req} \textbf{POST /start}} node [below=-8pt, text width=0.5\linewidth, anchor=base]{\begin{center} $\mi{email}$\end{center}} (RP-0); 

\draw[->] (RP-1) to node [above=2.6pt, anchor=base]{\protostep{oichf-wf-req} \textbf{GET /.wk/webfinger}} node [below=-8pt, text width=0.5\linewidth, anchor=base]{\begin{center} $\mi{email}$\end{center}} (IdP-1); 

\draw[->] (IdP-2) to node [above=2.6pt, anchor=base]{\protostep{oichf-wf-resp} \textbf{Response}} node [below=-8pt, text width=0.5\linewidth, anchor=base]{\begin{center} $\mi{idp}$\end{center}} (RP-2); 

\draw[->] (RP-3) to node [above=2.6pt, anchor=base]{\protostep{oichf-conf-req} \textbf{GET /.wk/openid-configuration}} node [below=-8pt, text width=0.5\linewidth, anchor=base]{\begin{center} \end{center}} (IdP-3); 

\draw[->] (IdP-4) to node [above=2.6pt, anchor=base]{\protostep{oichf-conf-resp} \textbf{Response}} node [below=-8pt, text width=0.5\linewidth, anchor=base]{\begin{center} $\mi{issuer}$, $\mi{authEP}$, $\mi{tokenEP}$, $\mi{registrationEP}$, $\mi{jwksURI}$, $\mi{userinfoEP}$\end{center}} (RP-4); 

\draw[->] (RP-5) to node [above=2.6pt, anchor=base]{\protostep{oichf-reg-req} \textbf{POST $\mi{registrationEP}$}} node [below=-8pt, text width=0.5\linewidth, anchor=base]{\begin{center} $\mi{redirect\_uris}$\end{center}} (IdP-5); 

\draw[->] (IdP-6) to node [above=2.6pt, anchor=base]{\protostep{oichf-reg-resp} \textbf{Response}} node [below=-8pt, text width=0.5\linewidth, anchor=base]{\begin{center} $\mi{client\_id}$, $\mi{client\_secret}$\end{center}} (RP-6); 

\draw[->] (RP-7) to node [above=2.6pt, anchor=base]{\protostep{oichf-start-resp} \textbf{Response}} node [below=-8pt, text width=0.5\linewidth, anchor=base]{\begin{center} Redirect to IdP $\mi{authEP}$ with $\mi{client\_id}$, $\mi{redirect\_uri}$, $\mi{state}$\end{center}} (Browser-7); 

\draw[->] (Browser-8) to node [above=2.6pt, anchor=base]{\protostep{oichf-idp-auth-req-1} \textbf{GET $\mi{authEP}$}} node [below=-8pt, text width=0.5\linewidth, anchor=base]{\begin{center} $\mi{client\_id}$, $\mi{redirect\_uri}$, $\mi{state}$\end{center}} (IdP-8); 

\draw[->] (IdP-9) to node [above=2.6pt, anchor=base]{\protostep{oichf-idp-auth-resp-1} \textbf{Response}} node [below=-8pt, text width=0.5\linewidth, anchor=base]{\begin{center} \end{center}} (Browser-9); 

\draw[->] (Browser-10) to node [above=2.6pt, anchor=base]{\protostep{oichf-idp-auth-req-2} \textbf{POST $\mi{authEP}$}} node [below=-8pt, text width=0.5\linewidth, anchor=base]{\begin{center} $\mi{username}$, $\mi{password}$\end{center}} (IdP-10); 

\draw[->] (IdP-11) to node [above=2.6pt, anchor=base]{\protostep{oichf-idp-auth-resp-2} \textbf{Response}} node [below=-8pt, text width=0.5\linewidth, anchor=base]{\begin{center} Redirect to RP $\mi{redirect\_uri}$, fragment: $\mi{access\_token}$, $\mi{code}$, $\mi{state}$\end{center}} (Browser-11); 

\draw[->] (Browser-12) to node [above=2.6pt, anchor=base]{\protostep{oichf-redir-ep-req} \textbf{GET $\mi{redirect\_uri}$}} node [below=-8pt, text width=0.5\linewidth, anchor=base]{\begin{center} \end{center}} (RP-12); 

\draw[->] (RP-13) to node [above=2.6pt, anchor=base]{\protostep{oichf-redir-ep-resp} \textbf{Response}} node [below=-8pt, text width=0.5\linewidth, anchor=base]{\begin{center} \end{center}} (Browser-13); 

\draw[->] (Browser-14) to node [above=2.6pt, anchor=base]{\protostep{oichf-redir-ep-token-req} \textbf{POST /token}} node [below=-8pt, text width=0.5\linewidth, anchor=base]{\begin{center} $\mi{access\_token}$, $\mi{code}$, $\mi{state}$\end{center}} (RP-14); 

\draw[->] (RP-15) to node [above=2.6pt, anchor=base]{\protostep{oichf-token-req} \textbf{POST $\mi{tokenEP}$}} node [below=-8pt, text width=0.5\linewidth, anchor=base]{\begin{center} $\mi{code}$, $\mi{client\_id}$, $\mi{redirect\_uri}$, $\mi{client\_secret}$\end{center}} (IdP-15); 

\draw[->] (IdP-16) to node [above=2.6pt, anchor=base]{\protostep{oichf-token-resp} \textbf{Response}} node [below=-8pt, text width=0.5\linewidth, anchor=base]{\begin{center} $\mi{access\_token}'$, $\mi{id\_token}$\end{center}} (RP-16); 

\draw[->] (RP-17) to node [above=2.6pt, anchor=base]{\protostep{oichf-jwks-req} \textbf{GET $\mi{jwksURI}$}} node [below=-8pt, text width=0.5\linewidth, anchor=base]{\begin{center} \end{center}} (IdP-17); 

\draw[->] (IdP-18) to node [above=2.6pt, anchor=base]{\protostep{oichf-jwks-resp} \textbf{Response}} node [below=-8pt, text width=0.5\linewidth, anchor=base]{\begin{center} $\mi{pubSignKey}$\end{center}} (RP-18); 

\draw[->] (RP-19) to node [above=2.6pt, anchor=base]{\protostep{oichf-userinfo-req} \textbf{GET $\mi{userinfoEP}$}} node [below=-8pt, text width=0.5\linewidth, anchor=base]{\begin{center} $\mi{access\_token}$\end{center}} (IdP-19); 

\draw[->] (IdP-20) to node [above=2.6pt, anchor=base]{\protostep{oichf-userinfo-resp} \textbf{Response}} node [below=-8pt, text width=0.5\linewidth, anchor=base]{\begin{center} $\mi{user\_id}$\end{center}} (RP-20); 

\draw[->] (RP-21) to node [above=2.6pt, anchor=base]{\protostep{oichf-redir-ep-token-resp} \textbf{Response}} node [below=-8pt, text width=0.5\linewidth, anchor=base]{\begin{center} $\mi{session\_cookie}$\end{center}} (Browser-21); 

\begin{pgfonlayer}{background}
\draw [color=gray] (Browser-start-0) -- (Browser-end-1);
\draw [color=gray] (RP-start-0) -- (RP-end-1);
\draw [color=gray] (IdP-start-0) -- (IdP-end-1);
\end{pgfonlayer}
\end{tikzpicture}}

%% file: figure-openid-connect-hybrid-flow-attack.tex
 \scriptsize{ \newlength\blockExtraHeightABEacDecfJbAIEEHBaIAJBfJFdHbAJJGb
\settototalheight\blockExtraHeightABEacDecfJbAIEEHBaIAJBfJFdHbAJJGb{\parbox{0.4\linewidth}{$\mi{email}$}}
\setlength\blockExtraHeightABEacDecfJbAIEEHBaIAJBfJFdHbAJJGb{\dimexpr \blockExtraHeightABEacDecfJbAIEEHBaIAJBfJFdHbAJJGb - 3ex/4}
\newlength\blockExtraHeightBBEacDecfJbAIEEHBaIAJBfJFdHbAJJGb
\settototalheight\blockExtraHeightBBEacDecfJbAIEEHBaIAJBfJFdHbAJJGb{\parbox{0.4\linewidth}{$\mi{email}'$}}
\setlength\blockExtraHeightBBEacDecfJbAIEEHBaIAJBfJFdHbAJJGb{\dimexpr \blockExtraHeightBBEacDecfJbAIEEHBaIAJBfJFdHbAJJGb - 3ex/4}
\newlength\blockExtraHeightCBEacDecfJbAIEEHBaIAJBfJFdHbAJJGb
\settototalheight\blockExtraHeightCBEacDecfJbAIEEHBaIAJBfJFdHbAJJGb{\parbox{0.4\linewidth}{$\mi{email}'$}}
\setlength\blockExtraHeightCBEacDecfJbAIEEHBaIAJBfJFdHbAJJGb{\dimexpr \blockExtraHeightCBEacDecfJbAIEEHBaIAJBfJFdHbAJJGb - 3ex/4}
\newlength\blockExtraHeightDBEacDecfJbAIEEHBaIAJBfJFdHbAJJGb
\settototalheight\blockExtraHeightDBEacDecfJbAIEEHBaIAJBfJFdHbAJJGb{\parbox{0.4\linewidth}{$\mi{attacker}$}}
\setlength\blockExtraHeightDBEacDecfJbAIEEHBaIAJBfJFdHbAJJGb{\dimexpr \blockExtraHeightDBEacDecfJbAIEEHBaIAJBfJFdHbAJJGb - 3ex/4}
\newlength\blockExtraHeightEBEacDecfJbAIEEHBaIAJBfJFdHbAJJGb
\settototalheight\blockExtraHeightEBEacDecfJbAIEEHBaIAJBfJFdHbAJJGb{\parbox{0.4\linewidth}{}}
\setlength\blockExtraHeightEBEacDecfJbAIEEHBaIAJBfJFdHbAJJGb{\dimexpr \blockExtraHeightEBEacDecfJbAIEEHBaIAJBfJFdHbAJJGb - 3ex/4}
\newlength\blockExtraHeightFBEacDecfJbAIEEHBaIAJBfJFdHbAJJGb
\settototalheight\blockExtraHeightFBEacDecfJbAIEEHBaIAJBfJFdHbAJJGb{\parbox{0.4\linewidth}{$\mi{issuer}'$, $\mi{authEP}$, $\mi{tokenEP}'$, $\mi{registrationEP}'$, $\mi{jwksURI}$, $\mi{userinfoEP}'$, $\mi{responseTypes}$}}
\setlength\blockExtraHeightFBEacDecfJbAIEEHBaIAJBfJFdHbAJJGb{\dimexpr \blockExtraHeightFBEacDecfJbAIEEHBaIAJBfJFdHbAJJGb - 3ex/4}
\newlength\blockExtraHeightGBEacDecfJbAIEEHBaIAJBfJFdHbAJJGb
\settototalheight\blockExtraHeightGBEacDecfJbAIEEHBaIAJBfJFdHbAJJGb{\parbox{0.4\linewidth}{$\mi{redirect\_uris}$}}
\setlength\blockExtraHeightGBEacDecfJbAIEEHBaIAJBfJFdHbAJJGb{\dimexpr \blockExtraHeightGBEacDecfJbAIEEHBaIAJBfJFdHbAJJGb - 3ex/4}
\newlength\blockExtraHeightHBEacDecfJbAIEEHBaIAJBfJFdHbAJJGb
\settototalheight\blockExtraHeightHBEacDecfJbAIEEHBaIAJBfJFdHbAJJGb{\parbox{0.4\linewidth}{$\mi{client\_id}$, $\mi{client\_secret}'$}}
\setlength\blockExtraHeightHBEacDecfJbAIEEHBaIAJBfJFdHbAJJGb{\dimexpr \blockExtraHeightHBEacDecfJbAIEEHBaIAJBfJFdHbAJJGb - 3ex/4}
\newlength\blockExtraHeightIBEacDecfJbAIEEHBaIAJBfJFdHbAJJGb
\settototalheight\blockExtraHeightIBEacDecfJbAIEEHBaIAJBfJFdHbAJJGb{\parbox{0.4\linewidth}{Redirect to HIdP $\mi{authEP}$ with $\mi{client\_id}$, $\mi{redirect\_uri}$, $\mi{state}$}}
\setlength\blockExtraHeightIBEacDecfJbAIEEHBaIAJBfJFdHbAJJGb{\dimexpr \blockExtraHeightIBEacDecfJbAIEEHBaIAJBfJFdHbAJJGb - 3ex/4}
\newlength\blockExtraHeightJBEacDecfJbAIEEHBaIAJBfJFdHbAJJGb
\settototalheight\blockExtraHeightJBEacDecfJbAIEEHBaIAJBfJFdHbAJJGb{\parbox{0.4\linewidth}{$\mi{client\_id}$, $\mi{redirect\_uri}$, $\mi{state}$}}
\setlength\blockExtraHeightJBEacDecfJbAIEEHBaIAJBfJFdHbAJJGb{\dimexpr \blockExtraHeightJBEacDecfJbAIEEHBaIAJBfJFdHbAJJGb - 3ex/4}
\newlength\blockExtraHeightBABEacDecfJbAIEEHBaIAJBfJFdHbAJJGb
\settototalheight\blockExtraHeightBABEacDecfJbAIEEHBaIAJBfJFdHbAJJGb{\parbox{0.4\linewidth}{}}
\setlength\blockExtraHeightBABEacDecfJbAIEEHBaIAJBfJFdHbAJJGb{\dimexpr \blockExtraHeightBABEacDecfJbAIEEHBaIAJBfJFdHbAJJGb - 3ex/4}
\newlength\blockExtraHeightBBBEacDecfJbAIEEHBaIAJBfJFdHbAJJGb
\settototalheight\blockExtraHeightBBBEacDecfJbAIEEHBaIAJBfJFdHbAJJGb{\parbox{0.4\linewidth}{$\mi{username}$, $\mi{password}$}}
\setlength\blockExtraHeightBBBEacDecfJbAIEEHBaIAJBfJFdHbAJJGb{\dimexpr \blockExtraHeightBBBEacDecfJbAIEEHBaIAJBfJFdHbAJJGb - 3ex/4}
\newlength\blockExtraHeightBCBEacDecfJbAIEEHBaIAJBfJFdHbAJJGb
\settototalheight\blockExtraHeightBCBEacDecfJbAIEEHBaIAJBfJFdHbAJJGb{\parbox{0.4\linewidth}{Redirect to RP $\mi{redirect\_uri}$, fragment: $\mi{access\_token}$, $\mi{code}$, $\mi{state}$}}
\setlength\blockExtraHeightBCBEacDecfJbAIEEHBaIAJBfJFdHbAJJGb{\dimexpr \blockExtraHeightBCBEacDecfJbAIEEHBaIAJBfJFdHbAJJGb - 3ex/4}
\newlength\blockExtraHeightBDBEacDecfJbAIEEHBaIAJBfJFdHbAJJGb
\settototalheight\blockExtraHeightBDBEacDecfJbAIEEHBaIAJBfJFdHbAJJGb{\parbox{0.4\linewidth}{}}
\setlength\blockExtraHeightBDBEacDecfJbAIEEHBaIAJBfJFdHbAJJGb{\dimexpr \blockExtraHeightBDBEacDecfJbAIEEHBaIAJBfJFdHbAJJGb - 3ex/4}
\newlength\blockExtraHeightBEBEacDecfJbAIEEHBaIAJBfJFdHbAJJGb
\settototalheight\blockExtraHeightBEBEacDecfJbAIEEHBaIAJBfJFdHbAJJGb{\parbox{0.4\linewidth}{}}
\setlength\blockExtraHeightBEBEacDecfJbAIEEHBaIAJBfJFdHbAJJGb{\dimexpr \blockExtraHeightBEBEacDecfJbAIEEHBaIAJBfJFdHbAJJGb - 3ex/4}
\newlength\blockExtraHeightBFBEacDecfJbAIEEHBaIAJBfJFdHbAJJGb
\settototalheight\blockExtraHeightBFBEacDecfJbAIEEHBaIAJBfJFdHbAJJGb{\parbox{0.4\linewidth}{$\mi{access\_token}$, $\mi{code}$, $\mi{state}$}}
\setlength\blockExtraHeightBFBEacDecfJbAIEEHBaIAJBfJFdHbAJJGb{\dimexpr \blockExtraHeightBFBEacDecfJbAIEEHBaIAJBfJFdHbAJJGb - 3ex/4}
\newlength\blockExtraHeightBGBEacDecfJbAIEEHBaIAJBfJFdHbAJJGb
\settototalheight\blockExtraHeightBGBEacDecfJbAIEEHBaIAJBfJFdHbAJJGb{\parbox{0.4\linewidth}{$\mi{code}$, $\mi{client\_id}$, $\mi{redirect\_uri}$, $\mi{client\_secret}'$}}
\setlength\blockExtraHeightBGBEacDecfJbAIEEHBaIAJBfJFdHbAJJGb{\dimexpr \blockExtraHeightBGBEacDecfJbAIEEHBaIAJBfJFdHbAJJGb - 3ex/4}
\newlength\blockExtraHeightBHBEacDecfJbAIEEHBaIAJBfJFdHbAJJGb
\settototalheight\blockExtraHeightBHBEacDecfJbAIEEHBaIAJBfJFdHbAJJGb{\parbox{0.4\linewidth}{$\mi{access\_token}'$, $\mi{id\_token}$}}
\setlength\blockExtraHeightBHBEacDecfJbAIEEHBaIAJBfJFdHbAJJGb{\dimexpr \blockExtraHeightBHBEacDecfJbAIEEHBaIAJBfJFdHbAJJGb - 3ex/4}
\newlength\blockExtraHeightBIBEacDecfJbAIEEHBaIAJBfJFdHbAJJGb
\settototalheight\blockExtraHeightBIBEacDecfJbAIEEHBaIAJBfJFdHbAJJGb{\parbox{0.4\linewidth}{$\mi{access\_token}$}}
\setlength\blockExtraHeightBIBEacDecfJbAIEEHBaIAJBfJFdHbAJJGb{\dimexpr \blockExtraHeightBIBEacDecfJbAIEEHBaIAJBfJFdHbAJJGb - 3ex/4}
\newlength\blockExtraHeightBJBEacDecfJbAIEEHBaIAJBfJFdHbAJJGb
\settototalheight\blockExtraHeightBJBEacDecfJbAIEEHBaIAJBfJFdHbAJJGb{\parbox{0.4\linewidth}{$\mi{access\_token}$}}
\setlength\blockExtraHeightBJBEacDecfJbAIEEHBaIAJBfJFdHbAJJGb{\dimexpr \blockExtraHeightBJBEacDecfJbAIEEHBaIAJBfJFdHbAJJGb - 3ex/4}
\newlength\blockExtraHeightCABEacDecfJbAIEEHBaIAJBfJFdHbAJJGb
\settototalheight\blockExtraHeightCABEacDecfJbAIEEHBaIAJBfJFdHbAJJGb{\parbox{0.4\linewidth}{secret user data}}
\setlength\blockExtraHeightCABEacDecfJbAIEEHBaIAJBfJFdHbAJJGb{\dimexpr \blockExtraHeightCABEacDecfJbAIEEHBaIAJBfJFdHbAJJGb - 3ex/4}

 \begin{tikzpicture}
   \tikzstyle{xhrArrow} = [color=blue,decoration={markings, mark=at
    position 1 with {\arrow[color=blue]{triangle 45}}}, preaction
  = {decorate}]

    \matrix [column sep={3.5cm,between origins}, row sep=4.5ex]
  {

    \node[draw,anchor=base](Browser-start-0){Browser}; & \node[draw,anchor=base](RP-start-0){RP}; & \node[draw,anchor=base](Attacker-start-0){Attacker (AIdP)}; & \node[draw,anchor=base](HIdP-start-0){HIdP};\\
\node(Browser-0){}; & \node(RP-0){}; & \node(Attacker-0){}; & \node(HIdP-0){};\\[\blockExtraHeightABEacDecfJbAIEEHBaIAJBfJFdHbAJJGb]
\node(Browser-1){}; & \node(RP-1){}; & \node(Attacker-1){}; & \node(HIdP-1){};\\[\blockExtraHeightBBEacDecfJbAIEEHBaIAJBfJFdHbAJJGb]
\node(Browser-2){}; & \node(RP-2){}; & \node(Attacker-2){}; & \node(HIdP-2){};\\[\blockExtraHeightCBEacDecfJbAIEEHBaIAJBfJFdHbAJJGb]
\node(Browser-3){}; & \node(RP-3){}; & \node(Attacker-3){}; & \node(HIdP-3){};\\[\blockExtraHeightDBEacDecfJbAIEEHBaIAJBfJFdHbAJJGb]
\node(Browser-4){}; & \node(RP-4){}; & \node(Attacker-4){}; & \node(HIdP-4){};\\[\blockExtraHeightEBEacDecfJbAIEEHBaIAJBfJFdHbAJJGb]
\node(Browser-5){}; & \node(RP-5){}; & \node(Attacker-5){}; & \node(HIdP-5){};\\[\blockExtraHeightFBEacDecfJbAIEEHBaIAJBfJFdHbAJJGb]
\node(Browser-6){}; & \node(RP-6){}; & \node(Attacker-6){}; & \node(HIdP-6){};\\[\blockExtraHeightGBEacDecfJbAIEEHBaIAJBfJFdHbAJJGb]
\node(Browser-7){}; & \node(RP-7){}; & \node(Attacker-7){}; & \node(HIdP-7){};\\[\blockExtraHeightHBEacDecfJbAIEEHBaIAJBfJFdHbAJJGb]
\node(Browser-8){}; & \node(RP-8){}; & \node(Attacker-8){}; & \node(HIdP-8){};\\[\blockExtraHeightIBEacDecfJbAIEEHBaIAJBfJFdHbAJJGb]
\node(Browser-9){}; & \node(RP-9){}; & \node(Attacker-9){}; & \node(HIdP-9){};\\[\blockExtraHeightJBEacDecfJbAIEEHBaIAJBfJFdHbAJJGb]
\node(Browser-10){}; & \node(RP-10){}; & \node(Attacker-10){}; & \node(HIdP-10){};\\[\blockExtraHeightBABEacDecfJbAIEEHBaIAJBfJFdHbAJJGb]
\node(Browser-11){}; & \node(RP-11){}; & \node(Attacker-11){}; & \node(HIdP-11){};\\[\blockExtraHeightBBBEacDecfJbAIEEHBaIAJBfJFdHbAJJGb]
\node(Browser-12){}; & \node(RP-12){}; & \node(Attacker-12){}; & \node(HIdP-12){};\\[\blockExtraHeightBCBEacDecfJbAIEEHBaIAJBfJFdHbAJJGb]
\node(Browser-13){}; & \node(RP-13){}; & \node(Attacker-13){}; & \node(HIdP-13){};\\[\blockExtraHeightBDBEacDecfJbAIEEHBaIAJBfJFdHbAJJGb]
\node(Browser-14){}; & \node(RP-14){}; & \node(Attacker-14){}; & \node(HIdP-14){};\\[\blockExtraHeightBEBEacDecfJbAIEEHBaIAJBfJFdHbAJJGb]
\node(Browser-15){}; & \node(RP-15){}; & \node(Attacker-15){}; & \node(HIdP-15){};\\[\blockExtraHeightBFBEacDecfJbAIEEHBaIAJBfJFdHbAJJGb]
\node(Browser-16){}; & \node(RP-16){}; & \node(Attacker-16){}; & \node(HIdP-16){};\\[\blockExtraHeightBGBEacDecfJbAIEEHBaIAJBfJFdHbAJJGb]
\node(Browser-17){}; & \node(RP-17){}; & \node(Attacker-17){}; & \node(HIdP-17){};\\[\blockExtraHeightBHBEacDecfJbAIEEHBaIAJBfJFdHbAJJGb]
\node(Browser-18){}; & \node(RP-18){}; & \node(Attacker-18){}; & \node(HIdP-18){};\\[\blockExtraHeightBIBEacDecfJbAIEEHBaIAJBfJFdHbAJJGb]
\node(Browser-19){}; & \node(RP-19){}; & \node(Attacker-19){}; & \node(HIdP-19){};\\[\blockExtraHeightBJBEacDecfJbAIEEHBaIAJBfJFdHbAJJGb]
\node(Browser-20){}; & \node(RP-20){}; & \node(Attacker-20){}; & \node(HIdP-20){};\\[\blockExtraHeightCABEacDecfJbAIEEHBaIAJBfJFdHbAJJGb]
\node[draw,anchor=base](Browser-end-1){/Browser}; & \node[draw,anchor=base](RP-end-1){/RP}; & \node[draw,anchor=base](Attacker-end-1){/Attacker (AIdP)}; & \node[draw,anchor=base](HIdP-end-1){/HIdP};\\
};
\draw[->] (Browser-0) to node [above=2.6pt, anchor=base]{\protostep{oichf-att-start-req} \textbf{POST /start}} node [below=-8pt, text width=0.5\linewidth, anchor=base]{\begin{center} $\mi{email}$\end{center}} (Attacker-0); 

\draw[->] (Attacker-1) to node [above=2.6pt, anchor=base]{\protostep{oichf-att-start-req-manipulated} \textbf{POST /start}} node [below=-8pt, text width=0.5\linewidth, anchor=base]{\begin{center} $\mi{email}'$\end{center}} (RP-1); 

\draw[->] (RP-2) to node [above=2.6pt, anchor=base]{\protostep{oichf-att-wf-req} \textbf{GET /.wk/webfinger}} node [below=-8pt, text width=0.5\linewidth, anchor=base]{\begin{center} $\mi{email}'$\end{center}} (Attacker-2); 

\draw[->] (Attacker-3) to node [above=2.6pt, anchor=base]{\protostep{oichf-att-wf-resp} \textbf{Response}} node [below=-8pt, text width=0.5\linewidth, anchor=base]{\begin{center} $\mi{attacker}$\end{center}} (RP-3); 

\draw[->] (RP-4) to node [above=2.6pt, anchor=base]{\protostep{oichf-att-conf-req} \textbf{GET /.wk/openid-configuration}} node [below=-8pt, text width=0.5\linewidth, anchor=base]{\begin{center} \end{center}} (Attacker-4); 

\draw[->] (Attacker-5) to node [above=2.6pt, anchor=base]{\protostep{oichf-att-conf-resp} \textbf{Response}} node [below=-8pt, text width=0.5\linewidth, anchor=base]{\begin{center} $\mi{issuer}'$, $\mi{authEP}$, $\mi{tokenEP}'$, $\mi{registrationEP}'$, $\mi{jwksURI}$, $\mi{userinfoEP}'$, $\mi{responseTypes}$\end{center}} (RP-5); 

\draw[->] (RP-6) to node [above=2.6pt, anchor=base]{\protostep{oichf-att-reg-req} \textbf{POST $\mi{registrationEP}'$}} node [below=-8pt, text width=0.5\linewidth, anchor=base]{\begin{center} $\mi{redirect\_uris}$\end{center}} (Attacker-6); 

\draw[->] (Attacker-7) to node [above=2.6pt, anchor=base]{\protostep{oichf-att-reg-resp} \textbf{Response}} node [below=-8pt, text width=0.5\linewidth, anchor=base]{\begin{center} $\mi{client\_id}$, $\mi{client\_secret}'$\end{center}} (RP-7); 

\draw[->] (RP-8) to node [above=2.6pt, anchor=base]{\protostep{oichf-att-start-resp} \textbf{Response}} node [below=-8pt, text width=0.5\linewidth, anchor=base]{\begin{center} Redirect to HIdP $\mi{authEP}$ with $\mi{client\_id}$, $\mi{redirect\_uri}$, $\mi{state}$\end{center}} (Browser-8); 

\draw[->] (Browser-9) to node [above=2.6pt, anchor=base]{\protostep{oichf-att-idp-auth-req-1} \textbf{GET $\mi{authEP}$}} node [below=-8pt, text width=0.5\linewidth, anchor=base]{\begin{center} $\mi{client\_id}$, $\mi{redirect\_uri}$, $\mi{state}$\end{center}} (HIdP-9); 

\draw[->] (HIdP-10) to node [above=2.6pt, anchor=base]{\protostep{oichf-att-idp-auth-resp-1} \textbf{Response}} node [below=-8pt, text width=0.5\linewidth, anchor=base]{\begin{center} \end{center}} (Browser-10); 

\draw[->] (Browser-11) to node [above=2.6pt, anchor=base]{\protostep{oichf-att-idp-auth-req-2} \textbf{POST $\mi{authEP}$}} node [below=-8pt, text width=0.5\linewidth, anchor=base]{\begin{center} $\mi{username}$, $\mi{password}$\end{center}} (HIdP-11); 

\draw[->] (HIdP-12) to node [above=2.6pt, anchor=base]{\protostep{oichf-att-idp-auth-resp-2} \textbf{Response}} node [below=-8pt, text width=0.5\linewidth, anchor=base]{\begin{center} Redirect to RP $\mi{redirect\_uri}$, fragment: $\mi{access\_token}$, $\mi{code}$, $\mi{state}$\end{center}} (Browser-12); 

\draw[->] (Browser-13) to node [above=2.6pt, anchor=base]{\protostep{oichf-att-redir-ep-req} \textbf{GET $\mi{redirect\_uri}$}} node [below=-8pt, text width=0.5\linewidth, anchor=base]{\begin{center} \end{center}} (RP-13); 

\draw[->] (RP-14) to node [above=2.6pt, anchor=base]{\protostep{oichf-att-redir-ep-resp} \textbf{Response}} node [below=-8pt, text width=0.5\linewidth, anchor=base]{\begin{center} \end{center}} (Browser-14); 

\draw[->] (Browser-15) to node [above=2.6pt, anchor=base]{\protostep{oichf-att-redir-ep-token-req} \textbf{POST /token}} node [below=-8pt, text width=0.5\linewidth, anchor=base]{\begin{center} $\mi{access\_token}$, $\mi{code}$, $\mi{state}$\end{center}} (RP-15); 

\draw[->] (RP-16) to node [above=2.6pt, anchor=base]{\protostep{oichf-att-token-req} \textbf{POST $\mi{tokenEP}'$}} node [below=-8pt, text width=0.5\linewidth, anchor=base]{\begin{center} $\mi{code}$, $\mi{client\_id}$, $\mi{redirect\_uri}$, $\mi{client\_secret}'$\end{center}} (Attacker-16); 

\draw[->] (Attacker-17) to node [above=2.6pt, anchor=base]{\protostep{oichf-att-token-resp} \textbf{Response}} node [below=-8pt, text width=0.5\linewidth, anchor=base]{\begin{center} $\mi{access\_token}'$, $\mi{id\_token}$\end{center}} (RP-17); 

\draw[->] (RP-18) to node [above=2.6pt, anchor=base]{\protostep{oichf-att-userinfo-req} \textbf{GET $\mi{userinfoEP}$}} node [below=-8pt, text width=0.5\linewidth, anchor=base]{\begin{center} $\mi{access\_token}$\end{center}} (Attacker-18); 

\draw[->] (Attacker-19) to node [above=2.6pt, anchor=base]{\protostep{oichf-att-access-protected-req} \textbf{GET /protectedResource}} node [below=-8pt, text width=0.5\linewidth, anchor=base]{\begin{center} $\mi{access\_token}$\end{center}} (HIdP-19); 

\draw[->] (HIdP-20) to node [above=2.6pt, anchor=base]{\protostep{oichf-att-access-protected-resp} \textbf{Response}} node [below=-8pt, text width=0.5\linewidth, anchor=base]{\begin{center} secret user data\end{center}} (Attacker-20); 

\begin{pgfonlayer}{background}
\draw [color=gray] (Browser-start-0) -- (Browser-end-1);
\draw [color=gray] (RP-start-0) -- (RP-end-1);
\draw [color=gray] (Attacker-start-0) -- (Attacker-end-1);
\draw [color=gray] (HIdP-start-0) -- (HIdP-end-1);
\end{pgfonlayer}
\end{tikzpicture}}

%% file: appendix-webmodel.tex
\section{The FKS Web Model}\label{app:web-model}

In this and the following two sections, we present the FKS model for
the web infrastructure as proposed
in~\cite{FettKuestersSchmitz-SP-2014}
and~\cite{FettKuestersSchmitz-TR-BrowserID-Primary-2015}, along with
the following changes and additions:

\begin{itemize}
\item We introduce a new header, $\str{Authorization}$,
  as a model for HTTP Basic Authentication.\footnote{Note that
    although the header is called ``Authorization'' (following
    RFC2617), this is a mechanism for authentication.}
\item Browsers now may have multiple passwords stored for a single
  origin; before, there was only one password for each origin.
\item We introduce the header $\str{ReferrerPolicy}$
  as a model for a referrer policy delivered in an HTTP response
  header. 
\end{itemize}

\subsection{Communication Model}\label{app:communication-model}

We here present details and definitions on the basic concepts of the
communication model.

\subsubsection{Terms, Messages and Events} 
The signature $\Sigma$ for the terms and
messages considered in this work is the union of the following
pairwise disjoint sets of function symbols:
\begin{itemize}
\item constants $C = \addresses\,\cup\, \mathbb{S}\cup
  \{\True,\bot,\notdef\}$ where the three sets are pairwise disjoint,
  $\mathbb{S}$ is interpreted to be the set of ASCII strings
  (including the empty string $\varepsilon$), and $\addresses$ is
  interpreted to be a set of (IP) addresses,
\item function symbols for public keys, (a)symmetric
  en\-cryp\-tion/de\-cryp\-tion, and signatures: $\mathsf{pub}(\cdot)$,
  $\enc{\cdot}{\cdot}$, $\dec{\cdot}{\cdot}$, $\encs{\cdot}{\cdot}$,
  $\decs{\cdot}{\cdot}$, $\sig{\cdot}{\cdot}$,
  $\checksigThree{\cdot}{\cdot}{\cdot}$, and $\unsig{\cdot}$,
\item $n$-ary sequences $\an{}, \an{\cdot}, \an{\cdot,\cdot},
  \an{\cdot,\cdot,\cdot},$ etc., and
\item projection symbols $\pi_i(\cdot)$ for all $i \in \mathbb{N}$.
\end{itemize}
For strings (elements in $\mathbb{S}$), we use a
specific font. For example, $\cHttpReq$ and $\cHttpResp$
are strings. We denote by $\dns\subseteq \mathbb{S}$ the
set of domains, e.g., $\str{example.com}\in \dns$.  We
denote by $\methods\subseteq \mathbb{S}$ the set of methods
used in HTTP requests, e.g., $\mGet$, $\mPost\in \methods$.

The equational theory associated with the signature
$\Sigma$ is given in Figure~\ref{fig:equational-theory}.

\begin{figure}
\begin{align}
\dec{\enc x{\pub(y)}}{y} &= x\\
\decs{\encs x{y}}{y} &= x\\
\checksigThree{\sig{x}{y}}{x}{\pub(y)} &= \True\\
\unsig{\sig{x}{y}} &= x\\
\pi_i(\an{x_1,\dots,x_n}) &= x_i \text{\;\;if\ } 1 \leq i \leq n \\
\proj{j}{\an{x_1,\dots,x_n}} &= \notdef \text{\;\;if\ } j
\not\in \{1,\dots,n\}
\end{align}
\caption{Equational theory for $\Sigma$}\label{fig:equational-theory}
\end{figure}

\begin{definition}[Nonces and Terms]\label{def:terms}
  By $X=\{x_0,x_1,\dots\}$ we denote a set of variables and by
  $\nonces$ we denote an infinite set of constants (\emph{nonces})
  such that $\Sigma$, $X$, and $\nonces$ are pairwise disjoint. For
  $N\subseteq\nonces$, we define the set $\gterms_N(X)$ of
  \emph{terms} over $\Sigma\cup N\cup X$ inductively as usual: (1) If
  $t\in N\cup X$, then $t$ is a term. (2) If $f\in \Sigma$ is an
  $n$-ary function symbol in $\Sigma$ for some $n\ge 0$ and
  $t_1,\ldots,t_n$ are terms, then $f(t_1,\ldots,t_n)$ is a term.
\end{definition}

By $\equiv$ we denote the congruence relation on $\terms(X)$ induced
by the theory associated with $\Sigma$. For example, we have that
$\pi_1(\dec{\enc{\an{\str{a},\str{b}}}{\pub(k)}}{k})\equiv \str{a}$.

\begin{definition}[Ground Terms, Messages, Placeholders, Protomessages]\label{def:groundterms-messages-placeholders-protomessages}
  By $\gterms_N=\gterms_N(\emptyset)$, we denote the set of all terms
  over $\Sigma\cup N$ without variables, called \emph{ground terms}.
  The set $\messages$ of messages (over $\nonces$) is defined to be
  the set of ground terms $\gterms_{\nonces}$. 
  
  We define the set $V_{\text{process}} = \{\nu_1, \nu_2, \dots\}$ of
  variables (called placeholders). The set $\messages^\nu :=
  \gterms_{\nonces}(V_{\text{process}})$ is called the set of \emph{protomessages},
  i.e., messages that can contain placeholders.
\end{definition}

\begin{example}
  For example, $k\in \nonces$ and $\pub(k)$ are messages, where $k$
  typically models a private key and $\pub(k)$ the corresponding
  public key. For constants $a$, $b$, $c$ and the nonce $k\in
  \nonces$, the message $\enc{\an{a,b,c}}{\pub(k)}$ is interpreted to
  be the message $\an{a,b,c}$ (the sequence of constants $a$, $b$,
  $c$) encrypted by the public key $\pub(k)$.
\end{example}

\begin{definition}[Normal Form]
  Let $t$ be a term. The \emph{normal form} of $t$ is acquired by
  reducing the function symbols from left to right as far as possible
  using the equational theory shown in
  Figure~\ref{fig:equational-theory}. For a term $t$, we denote its
  normal form as $t\nf$.
\end{definition}

\begin{definition}[Pattern Matching]\label{def:pattern-matching}
  Let $\mi{pattern} \in \terms(\{*\})$ be a term containing the
  wildcard (variable $*$). We say that a term $t$ \emph{matches}
  $\mi{pattern}$ iff $t$ can be acquired from $\mi{pattern}$ by
  replacing each occurrence of the wildcard with an arbitrary term
  (which may be different for each instance of the wildcard). We write
  $t \sim \mi{pattern}$. For a sequence of patterns $\mi{patterns}$ we
  write $t \dot{\sim} \mi{patterns}$ to denote that $t$ matches at
  least one pattern in $\mi{patterns}$.

  For a term $t'$ we write $t'|\, \mi{pattern}$ to denote the term
  that is acquired from $t'$ by removing all immediate subterms of
  $t'$ that do not match $\mi{pattern}$. 
\end{definition}

\begin{example}
  For example, for a pattern $p = \an{\top,*}$ we have that $\an{\top,42} \sim p$, $\an{\bot,42} \not\sim p$, and \[\an{\an{\bot,\top},\an{\top,23},\an{\str{a},\str{b}},\an{\top,\bot}} |\, p = \an{\an{\top,23},\an{\top,\bot}}\ .\]
\end{example}

\begin{definition}[Variable Replacement]
  Let $N\subseteq \nonces$, $\tau \in \gterms_N(\{x_1,\ldots,x_n\})$,
  and $t_1,\ldots,t_n\in \gterms_N$. 

  By
  $\tau[t_1\!/\!x_1,\ldots,t_n\!/\!x_n]$ we denote the (ground) term obtained
  from $\tau$ by replacing all occurrences of $x_i$ in $\tau$ by
  $t_i$, for all $i\in \{1,\ldots,n\}$.
\end{definition}

\begin{definition}[Events and Protoevents]
  An \emph{event (over $\addresses$ and $\messages$)} is a term of the
  form $\an{a, f, m}$, for $a$, $f\in \addresses$ and $m \in
  \messages$, where $a$ is interpreted to be the receiver address and
  $f$ is the sender address. We denote by $\events$ the set of all
  events. Events over $\addresses$ and $\messages^\nu$ are called
  \emph{protoevents} and are denoted $\events^\nu$. By
  $2^{\events\an{}}$ (or $2^{\events^\nu\an{}}$, respectively) we
  denote the set of all sequences of (proto)events, including the
  empty sequence (e.g., $\an{}$, $\an{\an{a, f, m}, \an{a', f', m'},
    \dots}$, etc.). 
\end{definition}

\subsubsection{Atomic Processes, Systems and Runs} 

An atomic process takes its current state and an
event as input, and then (non-deterministi\-cally) outputs a new state
and a set of events.
\begin{definition}[Generic Atomic Processes and Systems]\label{def:atomic-process-and-process}
  A \emph{(generic) \ap} is a tuple $$p = (I^p, Z^p, R^p, s^p_0)$$ where
  $I^p \subseteq \addresses$, $Z^p \in \terms$ is a set of states,
  $R^p\subseteq (\events \times Z^p) \times (2^{\events^\nu\an{}}
  \times \terms(V_{\text{process}}))$ (input event and old state map to sequence of
  output events and new state), and $s^p_0\in Z^p$ is the initial
  state of $p$. For any new state $s$ and any sequence of nonces
  $(\eta_1, \eta_2, \dots)$ we demand that $s[\eta_1/\nu_1,
  \eta_2/\nu_2, \dots] \in Z^p$. A \emph{system} $\process$ is a
  (possibly infinite) set of \aps.
\end{definition}

\begin{definition}[Configurations]
  A \emph{configuration of a system $\process$} is a tuple $(S, E, N)$
  where the state of the system $S$ maps every atomic process
  $p\in \process$ to its current state $S(p)\in Z^p$, the sequence of
  waiting events $E$ is an infinite sequence\footnote{Here: Not in the
    sense of terms as defined earlier.} $(e_1, e_2, \dots)$ of events
  waiting to be delivered, and $N$ is an infinite sequence of nonces
  $(n_1, n_2, \dots)$.
\end{definition}

\begin{definition}[Concatenating sequences]
  For a term $a = \an{a_1, \dots, a_i}$ and a sequence $b = (b_1, b_2,
  \dots)$, we define the \emph{concatenation} as $a \cdot b := (a_1,
  \dots, a_i, b_1, b_2, \dots)$.
  
\end{definition}

\begin{definition}[Subtracting from Sequences]
  For a sequence $X$ and a set or sequence $Y$ we define $X \setminus
  Y$ to be the sequence $X$ where for each element in $Y$, a
  non-deterministically chosen occurence of that element in $X$ is
  removed.
\end{definition}

\begin{definition}[Processing Steps]\label{def:processing-step}
  A \emph{processing step of the system $\process$} is of the form
  \[(S,E,N) \xrightarrow[p \rightarrow E_{\text{out}}]{e_\text{in}
    \rightarrow p} (S', E', N')\]
  where
  \begin{enumerate}
  \item $(S,E,N)$ and $(S',E',N')$ are configurations of $\process$,
  \item $e_\text{in} = \an{a, f, m} \in E$ is an event,
  \item $p \in \process$ is a process,
  \item $E_{\text{out}}$ is a sequence (term) of events
  \end{enumerate}
  such that there exists 
  \begin{enumerate}
  \item a sequence (term)
    $E^\nu_{\text{out}} \subseteq 2^{\events^\nu\an{}}$ of protoevents,
  \item a term $s^\nu \in \gterms_{\nonces}(V_{\text{process}})$, 
  \item a sequence $(v_1, v_2, \dots, v_i)$ of all placeholders appearing in $E^\nu_{\text{out}}$ (ordered lexicographically),
  \item a sequence $N^\nu = (\eta_1, \eta_2, \dots, \eta_i)$ of the first $i$ elements in $N$ 
  \end{enumerate}
  with
  \begin{enumerate}
  \item $((e_{\text{in}}, S(p)), (E^\nu_{\text{out}}, s^\nu)) \in R^p$
    and $a \in I^p$,
  \item $E_{\text{out}} = E^\nu_{\text{out}}[m_1/v_1, \dots, m_i/v_i]$
  \item $S'(p) = s^\nu[m_1/v_1, \dots, m_i/v_i]$ and $S'(p') = S(p')$ for all $p' \neq p$
  \item $E' = E_{\text{out}} \cdot (E \setminus \{e_{\text{in}}\})$ 
  \item $N' = N \setminus N^\nu$ 
  \end{enumerate}
  We may omit the superscript and/or subscript of the arrow.
\end{definition} 
Intuitively, for a processing step, we select one of the processes in
$\process$, and call it with one of the events in the list of waiting
events $E$. In its output (new state and output events), we replace
any occurences of placeholders $\nu_x$ by ``fresh'' nonces from $N$
(which we then remove from $N$). The output events are then prepended
to the list of waiting events, and the state of the process is
reflected in the new configuration.

\begin{definition}[Runs]
  Let $\process$ be a system, $E^0$ be sequence of events, and $N^0$ be
  a sequence of nonces. A \emph{run $\rho$ of a system $\process$
    initiated by $E^0$ with nonces $N^0$} is a finite sequence of
  configurations $((S^0, E^0, N^0),\dots,$ $(S^n, E^n, N^n))$ or an infinite sequence
  of configurations $((S^0, E^0, N^0),\dots)$ such that $S^0(p) = s_0^p$ for
  all $p \in \process$ and $(S^i, E^i, N^i) \xrightarrow{} (S^{i+1},
  E^{i+1}, N^{i+1})$ for all $0 \leq i < n$ (finite run) or for all $i \geq 0$  
  (infinite run).

  We denote the state $S^n(p)$ of a process $p$ at the end of a run $\rho$ by $\rho(p)$.
\end{definition}

Usually, we will initiate runs with a set $E^0$ containing infinite
trigger events of the form $\an{a, a, \str{TRIGGER}}$ for each $a \in
\addresses$, interleaved by address.

\subsubsection{Atomic Dolev-Yao Processes}  We next define
atomic Dolev-Yao processes, for which we require that the
messages and states that they output can be computed (more
formally, derived) from the current input event and
state. For this purpose, we first define what it means to
derive a message from given messages.

\begin{definition}[Deriving Terms]
  Let $M$ be a set of ground terms. We say that \emph{a
    term $m$ can be derived from $M$ with placeholders $V$} if there
  exist $n\ge 0$, $m_1,\ldots,m_n\in M$, and $\tau\in
  \gterms_{\emptyset}(\{x_1,\ldots,x_n\} \cup V)$ such that $m\equiv
  \tau[m_1/x_1,\ldots,m_n/x_n]$. We denote by $d_V(M)$ the set of all
  messages that can be derived from $M$ with variables $V$.
\end{definition}
For example, $a\in d_{\{\}}(\{\enc{\an{a,b,c}}{\pub(k)}, k\})$.

\begin{definition}[Atomic Dolev-Yao Process] \label{def:adyp} An \emph{atomic Dolev-Yao process
    (or simply, a DY process)} is a tuple $p = (I^p, Z^p,$ $R^p,
  s^p_0)$ such that $(I^p, Z^p, R^p, s^p_0)$ is an atomic process and
  (1) $Z^p \subseteq \gterms_{\nonces}$ (and hence, $s^p_0\in
  \gterms_{\nonces}$), and (2) for all events $e \in \events$,
  sequences of protoevents $E$, $s\in \gterms_{\nonces}$, $s'\in
  \gterms_{\nonces}(V_{\text{process}})$, with $((e, s), (E, s')) \in R^p$ it holds
  true that $E$, $s' \in d_{V_{\text{process}}}(\{e,s\})$.
\end{definition}

\begin{definition}[Atomic Attacker Process]\label{def:atomicattacker}
  An \emph{(atomic) attacker process for a set of sender addresses
    $A\subseteq \addresses$} is an atomic DY process $p = (I, Z, R,
  s_0)$ such that for all events $e$, and $s\in \gterms_{\nonces}$ we
  have that $((e, s), (E,s')) \in R$ iff $s'=\an{e, E, s}$ and
  $E=\an{\an{a_1, f_1, m_1}, \dots, \an{a_n, f_n, m_n}}$ with $n \in
  \mathbb{N}$, $a_1,\dots,a_n\in \addresses$, $f_0,\dots,f_n\in A$,
  $m_1,\dots,m_n\in d_{V_{\text{process}}}(\{e,s\})$.
\end{definition}

\subsection{Scripts}
We define scripts, which model client-side scripting
technologies, such as JavaScript. Scripts are defined
similarly to DY processes.
\begin{definition}[Placeholders for Scripts]\label{def:placeholder-sp}
  By $V_{\text{script}} = \{\lambda_1, \dots\}$ we denote an infinite set of variables
  used in scripts.
\end{definition}

\begin{definition}[Scripts]\label{def:sp} A \emph{script} is a relation $R\subseteq \terms \times
  \terms(V_{\text{script}})$ such that for all $s \in \terms$, $s' \in \terms(V_{\text{script}})$ with
  $(s, s') \in R$ it follows that $s'\in d_{V_{\text{script}}}(s)$.
\end{definition}
A script is called by the browser which provides it with state
information (such as the script's last state and limited information
about the browser's state) $s$. The script then outputs a term $s'$,
which represents the new internal state and some command which is
interpreted by the browser. The term $s'$ may contain variables
$\lambda_1, \dots$ which the browser will replace by (otherwise
unused) placeholders $\nu_1,\dots$ which will be replaced by nonces
once the browser DY process finishes (effectively providing the script
with a way to get ``fresh'' nonces).

Similarly to an attacker process, we define the
\emph{attacker script} $\Rasp$: 
\begin{definition}[Attacker Script]
  The attacker script $\Rasp$ outputs everything that is derivable
  from the input, i.e., $\Rasp=\{(s, s')\mid s\in \terms, s'\in
  d_{V_{\text{script}}}(s)\}$.
\end{definition}

\subsection{Web System}\label{app:websystem}

The web infrastructure and web applications are formalized by what is
called a web system. A web system contains, among others, a (possibly
infinite) set of DY processes, modeling web browsers, web servers, DNS
servers, and attackers (which may corrupt other entities, such as
browsers).

\begin{definition}\label{def:websystem}
  A \emph{web system $\completewebsystem=(\websystem,
    \scriptset,\mathsf{script}, E^0)$} is a tuple with its
  components defined as follows:

  The first component, $\websystem$, denotes a system
  (a set of DY processes) and is partitioned into the
  sets $\mathsf{Hon}$, $\mathsf{Web}$, and $\mathsf{Net}$
  of honest, web attacker, and network attacker processes,
  respectively.  

  Every $p\in \mathsf{Web} \cup \mathsf{Net}$ is an
  attacker process for some set of sender addresses
  $A\subseteq \addresses$. For a web attacker $p\in
  \mathsf{Web}$, we require its set of addresses $I^p$ to
  be disjoint from the set of addresses of all other web
  attackers and honest processes, i.e., $I^p\cap I^{p'} =
  \emptyset$ for all $p' \in \mathsf{Hon} \cup
  \mathsf{Web}$. Hence, a web attacker cannot listen to
  traffic intended for other processes. Also, we require
  that $A=I^p$, i.e., a web attacker can only use sender
  addresses it owns. Conversely, a network attacker may
  listen to all addresses (i.e., no restrictions on $I^p$)
  and may spoof all addresses (i.e., the set $A$ may be
  $\addresses$).

  Every $p \in \mathsf{Hon}$ is a DY process which
  models either a \emph{web server}, a \emph{web browser},
  or a \emph{DNS server}, as further described in the
  following subsections. Just as for web attackers, we
  require that $p$ does not spoof sender addresses and that
  its set of addresses $I^p$ is disjoint from those of
  other honest processes and the web attackers. 

  The second component, $\scriptset$, is a finite set of
  scripts such that $\Rasp\in \scriptset$. The third
  component, $\mathsf{script}$, is an injective mapping
  from $\scriptset$ to $\mathbb{S}$, i.e., by
  $\mathsf{script}$ every $s\in \scriptset$ is assigned its
  string representation $\mathsf{script}(s)$. 

  Finally, $E^0$ is an  (infinite) sequence of events, containing an
  infinite number of events of the form $\an{a,a,\trigger}$
  for every $a \in \bigcup_{p\in \websystem} I^p$.

  A \emph{run} of $\completewebsystem$ is a run of
  $\websystem$ initiated by $E^0$.
\end{definition}

%% file: appendix-message-formats.tex
\section{Message and Data
  Formats}\label{app:message-data-formats}

We now provide some more details about data and message
formats that are needed for the formal treatment of the web
model and the analysis of BrowserID presented in the rest
of the appendix.

\subsection{Notations}\label{app:notation}

\begin{definition}[Sequence Notations]
  For a sequence $t = \an{t_1,\dots,t_n}$ and a set $s$ we
  use $t \subsetPairing s$ to say that $t_1,\dots,t_n \in
  s$.  We define $\left. x \inPairing t\right. \iff \exists
  i: \left. t_i = x\right.$.
  We write $t \plusPairing y$ to denote the sequence
  $\an{t_1,\dots,t_n,y}$.
  For a finite set $M$ with $M = \{m_1, \dots,m_n\}$ we use
  $\an{M}$ to denote the term of the form
  $\an{m_1,\dots,m_n}$. (The order of the elements does not
  matter; one is chosen arbitrarily.) 
\end{definition}

\begin{definition}\label{def:dictionaries}
  A \emph{dictionary over $X$ and $Y$} is a term of the
  form \[\an{\an{k_1, v_1}, \dots, \an{k_n,v_n}}\] where
  $k_1, \dots,k_n \in X$, $v_1,\dots,v_n \in Y$, and the
  keys $k_1, \dots,k_n$ are unique, i.e., $\forall i\neq j:
  k_i \neq k_j$. We call every term $\an{k_i,v_i}$, $i\in
  \{1,\ldots,n\}$, an \emph{element} of the dictionary with
  key $k_i$ and value $v_i$.  We often write $\left[k_1:
    v_1, \dots, k_i:v_i,\dots,k_n:v_n\right]$ instead of
  $\an{\an{k_1, v_1}, \dots, \an{k_n,v_n}}$. We denote the
  set of all dictionaries over $X$ and $Y$ by $\left[X
    \times Y\right]$.
\end{definition}
We note that the empty dictionary is equivalent to the
empty sequence, i.e.,  $[] = \an{}$.  Figure
\ref{fig:dictionaries} shows the short notation for
dictionary operations that will be used when describing the
browser atomic process. For a dictionary $z = \left[k_1:
  v_1, k_2: v_2,\dots, k_n:v_n\right]$ we write $k \in z$ to
say that there exists $i$ such that $k=k_i$. We write
$z[k_j] := v_j$ to extract elements. If $k \not\in z$, we
set $z[k] := \an{}$.

\begin{figure}[htb!]\centering
  \begin{align}
    \left[k_1: v_1, \dots, k_i:v_i,\dots,k_n:v_n\right][k_i] = v_i%
  \end{align}\vspace{-2.5em}
  \begin{align}
    \nonumber \left[k_1: v_1, \dots, k_{i-1}:v_{i-1},k_i: v_i, k_{i+1}:v_{i+1}\dots,k_n: v_n\right]-k_i =\\
         \left[k_1: v_1, \dots, k_{i-1}:v_{i-1},k_{i+1}:v_{i+1}\dots,k_n: v_n\right]
  \end{align}
  \caption{Dictionary operators with $1\le i\le n$}\label{fig:dictionaries}
\end{figure}

Given a term $t = \an{t_1,\dots,t_n}$, we can refer to any
subterm using a sequence of integers. The subterm is
determined by repeated application of the projection
$\pi_i$ for the integers $i$ in the sequence. We call such
a sequence a \emph{pointer}:

\begin{definition}\label{def:pointer}
  A \emph{pointer} is a sequence of non-negative
  integers. We write $\tau.\ptr{p}$ for the application of
  the pointer $\ptr{p}$ to the term $\tau$. This operator
  is applied from left to right. For pointers consisting of
  a single integer, we may omit the sequence braces for
  brevity.
\end{definition}

\begin{example}
  For the term $\tau = \an{a,b,\an{c,d,\an{e,f}}}$ and the
  pointer $\ptr{p} = \an{3,1}$, the subterm of $\tau$ at
  the position $\ptr{p}$ is $c =
  \proj{1}{\proj{3}{\tau}}$. Also, $\tau.3.\an{3,1} =
  \tau.3.\ptr{p} = \tau.3.3.1 = e$.
\end{example}

To improve readability, we try to avoid writing, e.g.,
$\compn{o}{2}$ or $\proj{2}{o}$ in this document. Instead,
we will use the names of the components of a sequence that
is of a defined form as pointers that point to the
corresponding subterms. E.g., if an \emph{Origin} term is
defined as $\an{\mi{host}, \mi{protocol}}$ and $o$ is an
Origin term, then we can write $\comp{o}{protocol}$ instead
of $\proj{2}{o}$ or $\compn{o}{2}$. See also
Example~\ref{ex:url-pointers}.

\subsection{URLs}\label{app:urls}

\begin{definition}\label{def:url}
  A \emph{URL} is a term of the form
  $$\an{\tUrl, \mi{protocol}, \mi{host}, \mi{path},
    \mi{parameters}, \mi{fragment}}$$ with $\mi{protocol}$
  $\in \{\http, \https\}$
  (for \textbf{p}lain (HTTP) and \textbf{s}ecure (HTTPS)),
  $\mi{host} \in \dns$,
  $\mi{path} \in \mathbb{S}$,
  $\mi{parameters} \in \dict{\mathbb{S}}{\terms}$,
  and $\mi{fragment} \in \terms$.
  The set of all valid URLs is $\urls$.
\end{definition}

The $\mi{fragment}$ part of a URL can be omitted when
writing the URL. Its value is then defined to be $\bot$.

\begin{example} \label{ex:url-pointers}
  For the URL $u = \an{\tUrl, a, b, c, d}$, $\comp{u}{protocol} =
  a$. If, in the algorithm described later, we say $\comp{u}{path} :=
  e$ then $u = \an{\tUrl, a, b, c, e}$ afterwards. 
\end{example}

\subsection{Origins}\label{app:origins}
\begin{definition} An \emph{origin} is a term of the form
  $\an{\mi{host}, \mi{protocol}}$ with $\mi{host} \in
  \dns$ and $\mi{protocol} \in \{\http, \https\}$. We write
  $\origins$ for the set of all origins.  
\end{definition}

\begin{example}
  For example, $\an{\str{FOO}, \https}$ is the HTTPS origin
  for the domain $\str{FOO}$, while $\an{\str{BAR}, \http}$
  is the HTTP origin for the domain $\str{BAR}$.
\end{example}
\subsection{Cookies}\label{app:cookies}

\begin{definition} A \emph{cookie} is a term of the form
  $\an{\mi{name}, \mi{content}}$ where $\mi{name} \in
  \terms$, and $\mi{content}$ is a term of the form
  $\an{\mi{value}, \mi{secure}, \mi{session},
    \mi{httpOnly}}$ where $\mi{value} \in \terms$,
  $\mi{secure}$, $\mi{session}$, $\mi{httpOnly} \in
  \{\True, \bot\}$. We write $\cookies$ for the set of all
  cookies and $\cookies^\nu$ for the set of all cookies
  where names and values are defined over $\terms(V)$.
\end{definition}

If the $\mi{secure}$ attribute of a cookie is set, the
browser will not transfer this cookie over unencrypted HTTP
connections. If the $\mi{session}$ flag is set, this cookie
will be deleted as soon as the browser is closed. The
$\mi{httpOnly}$ attribute controls whether JavaScript has
access to this cookie.

Note that cookies of the form described here are only
contained in HTTP(S) requests. In responses, only the
components $\mi{name}$ and $\mi{value}$ are transferred as
a pairing of the form $\an{\mi{name}, \mi{value}}$.

\subsection{HTTP Messages}\label{app:http-messages-full}
\begin{definition}
  An \emph{HTTP request} is a term of the form shown in
  (\ref{eq:default-http-request}). An \emph{HTTP response}
  is a term of the form shown in
  (\ref{eq:default-http-response}).
  \begin{align}
    \label{eq:default-http-request}
    & \hreq{ nonce=\mi{nonce}, method=\mi{method},
      xhost=\mi{host}, xpath=\mi{path},
      parameters=\mi{parameters}, headers=\mi{headers},
      xbody=\mi{body}
    } \\
    \label{eq:default-http-response}
    & \hresp{ nonce=\mi{nonce}, status=\mi{status},
      headers=\mi{headers}, xbody=\mi{body} }
  \end{align}
  The components are defined as follows:
  \begin{itemize}
  \item $\mi{nonce} \in \nonces$ serves to map each
    response to the corresponding request 
  \item $\mi{method} \in \methods$ is one of the HTTP
    methods.
  \item $\mi{host} \in \dns$ is the host name in the HOST
    header of HTTP/1.1.
  \item $\mi{path} \in \mathbb{S}$ is a string indicating
    the requested resource at the server side
  \item $\mi{status} \in \mathbb{S}$ is the HTTP status
    code (i.e., a number between 100 and 505, as defined by
    the HTTP standard)
  \item $\mi{parameters} \in
    \dict{\mathbb{S}}{\terms}$ contains URL parameters
  \item $\mi{headers} \in \dict{\mathbb{S}}{\terms}$,
    containing request/response headers. The dictionary
    elements are terms of one of the following forms: 
    \begin{itemize}
    \item $\an{\str{Origin}, o}$ where $o$ is an origin,
    \item $\an{\str{Set{\mhyphen}Cookie}, c}$ where $c$ is
      a sequence of cookies,
    \item $\an{\str{Cookie}, c}$ where $c \in,
      \dict{\mathbb{S}}{\terms}$ (note that in this header,
      only names and values of cookies are transferred),
    \item $\an{\str{Location}, l}$ where $l \in \urls$,
    \item $\an{\str{Referer}, r}$ where $r \in \urls$,
    \item $\an{\str{Strict{\mhyphen}Transport{\mhyphen}Security},\True}$,
    \item $\an{\str{Authorization}, \an{u,p}}$ where $u$, $p \in \mathbb{S}$,
    \item $\an{\str{ReferrerPolicy}, p}$ where $p \in \{\str{noreferrer}, \str{origin}\}$
    \end{itemize}
  \item $\mi{body} \in \terms$ in requests and responses. 
  \end{itemize}
  We write $\httprequests$/$\httpresponses$ for the set of
  all HTTP requests or responses, respectively.
\end{definition}

\begin{example}[HTTP Request and Response]
  \begin{align}
    \label{eq:ex-request}
    \nonumber \mi{r} := & \langle
                   \cHttpReq,
                   n_1,
                   \mPost,
                   \str{example.com},
                   \str{/show},
                   \an{\an{\str{index,1}}},\\ & \quad
                   [\str{Origin}: \an{\str{example.com, \https}}],
                   \an{\str{foo}, \str{bar}}
                \rangle \\
    \label{eq:ex-response} \mi{s} := & \hresp{ nonce=n_1,
      status=200,
      headers=\an{\an{\str{Set{\mhyphen}Cookie},\an{\an{\str{SID},\an{n_2,\bot,\bot,\True}}}}},
      xbody=\an{\str{somescript},x}}
  \end{align}
  \noindent
  An HTTP $\mGet$ request for the URL
  \url{http://example.com/show?index=1} is shown in
  (\ref{eq:ex-request}), with an Origin header and a body
  that contains $\an{\str{foo},\str{bar}}$. A possible
  response is shown in (\ref{eq:ex-response}), which
  contains an httpOnly cookie with name $\str{SID}$ and
  value $n_2$ as well as the string representation
  $\str{somescript}$ of the script
  $\mathsf{script}^{-1}(\str{somescript})$ (which should be
  an element of $\scriptset$) and its initial state
  $x$.
\end{example}

\subsubsection{Encrypted HTTP
  Messages.} \label{app:http-messages-encrypted-full}
For HTTPS, requests are encrypted using the public key of
the server.  Such a request contains an (ephemeral)
symmetric key chosen by the client that issued the
request. The server is supported to encrypt the response
using the symmetric key.

\begin{definition} An \emph{encrypted HTTP request} is of
  the form $\enc{\an{m, k'}}{k}$, where $k$, $k' \in
  \nonces$ and $m \in \httprequests$. The corresponding
  \emph{encrypted HTTP response} would be of the form
  $\encs{m'}{k'}$, where $m' \in \httpresponses$. We call
  the sets of all encrypted HTTP requests and responses
  $\httpsrequests$ or $\httpsresponses$, respectively.
\end{definition}

\begin{example}
  \begin{align}
    \label{eq:ex-enc-request} \ehreqWithVariable{r}{k'}{\pub(k_\text{example.com})} \\
    \label{eq:ex-enc-response} \ehrespWithVariable{s}{k'}
  \end{align} The term (\ref{eq:ex-enc-request}) shows an
  encrypted request (with $r$ as in
  (\ref{eq:ex-request})). It is encrypted using the public
  key $\pub(k_\text{example.com})$.  The term
  (\ref{eq:ex-enc-response}) is a response (with $s$ as in
  (\ref{eq:ex-response})). It is encrypted symmetrically
  using the (symmetric) key $k'$ that was sent in the
  request (\ref{eq:ex-enc-request}).
\end{example}

\subsection{DNS Messages}\label{app:dns-messages}
\begin{definition} A \emph{DNS request} is a term of the form
$\an{\cDNSresolve, \mi{domain}, \mi{n}}$ where $\mi{domain}$ $\in
\dns$, $\mi{n} \in \nonces$. We call the set of all DNS requests
$\dnsrequests$.
\end{definition}

\begin{definition} A \emph{DNS response} is a term of the form
$\an{\cDNSresolved, \mi{domain}, \mi{result}, \mi{n}}$ with $\mi{domain}$ $\in
\dns$, $\mi{result} \in
\addresses$, $\mi{n} \in \nonces$. We call the set of all DNS
responses $\dnsresponses$.
\end{definition}

DNS servers are supposed to include the nonce they received
in a DNS request in the DNS response that they send back so
that the party which issued the request can match it with
the request.

\subsection{DNS Servers}\label{app:DNSservers}

Here, we consider a flat DNS model in which DNS queries are
answered directly by one DNS server and always with the
same address for a domain. A full (hierarchical) DNS system
with recursive DNS resolution, DNS caches, etc.~could also
be modeled to cover certain attacks on the DNS system
itself.

\begin{definition}\label{def:dns-server}
  A \emph{DNS server} $d$ (in a flat DNS model) is modeled
  in a straightforward way as an atomic DY process
  $(I^d, \{s^d_0\}, R^d, s^d_0)$. It has a finite set of
  addresses $I^d$ and its initial (and only) state $s^d_0$
  encodes a mapping from domain names to addresses of the
  form
$$s^d_0=\langle\an{\str{domain}_1,a_1},\an{\str{domain}_2, a_2}, \ldots\rangle \ .$$ DNS
queries are answered according to this table (otherwise
ignored).
\end{definition}

%% file: appendix-browsermodel.tex
\section{Detailed Description of the Browser Model}
\label{app:deta-descr-brows}
Following the informal description of the browser model in
Section~\ref{sec:web-browsers}, we now present a formal
model. We start by introducing some notation and
terminology. 

\subsection{Notation and Terminology (Web Browser State)}

Before we can define the state of a web browser, we first
have to define windows and documents. 

\begin{sloppypar}
  \begin{definition} A \emph{window} is a term of the form
    $w = \an{\mi{nonce}, \mi{documents}, \mi{opener}}$ with
    $\mi{nonce} \in \nonces$,
    $\mi{documents} \subsetPairing \documents$ (defined
    below), $\mi{opener} \in \nonces \cup \{\bot\}$ where
    $\comp{d}{active} = \True$ for exactly one
    $d \inPairing \mi{documents}$ if $\mi{documents}$ is
    not empty (we then call $d$ the \emph{active document
      of $w$}). We write $\windows$ for the set of all
    windows. We write $\comp{w}{activedocument}$ to denote
    the active document inside window $w$ if it exists and
    $\an{}$ else.
  \end{definition}
\end{sloppypar}

We will refer to the window nonce as \emph{(window)
  reference}.

The documents contained in a window term to the left of the
active document are the previously viewed documents
(available to the user via the ``back'' button) and the
documents in the window term to the right of the currently
active document are documents available via the ``forward''
button.

A window $a$ may have opened a top-level window $b$ (i.e.,
a window term which is not a subterm of a document
term). In this case, the \emph{opener} part of the term $b$
is the nonce of $a$, i.e., $\comp{b}{opener} =
\comp{a}{nonce}$.

\begin{sloppypar}
  \begin{definition} A \emph{document} $d$ is a term of the
    form
    \begin{align*}
      \an{\mi{nonce}, \mi{location}, \mi{headers}, \mi{referrer}, \mi{script},
      \mi{scriptstate},\mi{scriptinputs}, \mi{subwindows},
      \mi{active}}  
    \end{align*}
    where $\mi{nonce} \in \nonces$,
    $\mi{location} \in \urls$,
    $\mi{headers} \in \dict{\mathbb{S}}{\terms}$,
    $\mi{referrer} \in \urls \cup \{\bot\}$,
    $\mi{script} \in \terms$,
    $\mi{scriptstate} \in \terms$,
    $\mi{scriptinputs} \in \terms$,
    $\mi{subwindows} \subsetPairing \windows$,
    $\mi{active} \in \{\True, \bot\}$. A \emph{limited
      document} is a term of the form
    $\an{\mi{nonce}, \mi{subwindows}}$ with $\mi{nonce}$,
    $\mi{subwindows}$ as above. A window
    $w \inPairing \mi{subwindows}$ is called a
    \emph{subwindow} (of $d$). We write $\documents$ for
    the set of all documents. For a document term $d$ we
    write $d.\str{origin}$ to denote the origin of the
    document, i.e., the term
    $\an{d.\str{location}.\str{host},
      d.\str{location}.\str{protocol}} \in \origins$.
  \end{definition}
\end{sloppypar}

We will refer to the document nonce as \emph{(document)
  reference}.

We can now define the set of states of web browsers. Note
that we use the dictionary notation that we introduced in
Definition~\ref{def:dictionaries}.

\begin{definition} The
  \emph{set of states $Z^p$ of a web browser atomic process}
  $p$ consists of the terms of the form
  \begin{align*} \langle\mi{windows}, \mi{ids},
    \mi{secrets}, \mi{cookies}, \mi{localStorage},
    \mi{sessionStorage}, \mi{keyMapping},& \\\mi{sts},
    \mi{DNSaddress}, \mi{pendingDNS},
    \mi{pendingRequests}, \mi{isCorrupted}&\rangle
  \end{align*} where
  \begin{itemize}
  \item $\mi{windows} \subsetPairing \windows$,
  \item $\mi{ids} \subsetPairing \terms$,
  \item $\mi{secrets} \in \dict{\origins}{\terms}$,
  \item $\mi{cookies}$ is a dictionary over $\dns$ and
    sequences of $\cookies$, 
  \item $\mi{localStorage} \in \dict{\origins}{\terms}$,
  \item $\mi{sessionStorage} \in \dict{\mi{OR}}{\terms}$ for $\mi{OR} := \left\{\an{o,r}
    \middle|\, o \in \origins,\, r \in \nonces\right\}$,
  \item $\mi{keyMapping} \in \dict{\dns}{\terms}$,
  \item $\mi{sts} \subsetPairing \dns$,
  \item $\mi{DNSaddress} \in \addresses$,
  \item $\mi{pendingDNS} \in \dict{\nonces}{\terms}$,
  \item $\mi{pendingRequests} \in$ $\terms$,
  \item and $\mi{isCorrupted} \in \{\bot, \fullcorrupt,$ $
    \closecorrupt\}$.
  \end{itemize} 
\end{definition}

\begin{definition} For two window terms $w$ and $w'$ we
  write $w \windowChildOf w'$ if \\
  \[w \inPairing \comp{\comp{w'}{activedocument}}{subwindows}\text{\ .}\]
We write
  $\windowChildOfX$ for the transitive closure.
\end{definition}

In the following description of the web browser relation
$R^p$ we use the helper functions
$\mathsf{Subwindows}$, $\mathsf{Docs}$, $\mathsf{Clean}$,
$\mathsf{CookieMerge}$ and $\mathsf{AddCookie}$. 

Given a browser state $s$, $\mathsf{Subwindows}(s)$ denotes
the set of all pointers\footnote{Recall the definition of a
  pointer in Definition~\ref{def:pointer}.} to windows in
the window list $\comp{s}{windows}$, their active
documents, and (recursively) the subwindows of these
documents. We exclude subwindows of inactive documents and
their subwindows. With $\mathsf{Docs}(s)$ we denote the set
of pointers to all active documents in the set of windows
referenced by $\mathsf{Subwindows}(s)$.
\begin{definition} 
  For a browser state $s$ we denote by
  $\mathsf{Subwindows}(s)$ the minimal set of
  pointers that satisfies the
  following conditions: (1) For all windows $w \inPairing
  \comp{s}{windows}$ there is a $\ptr{p} \in
  \mathsf{Subwindows}(s)$ such that $\compn{s}{\ptr{p}} =
  w$. (2) For all $\ptr{p} \in \mathsf{Subwindows}(s)$, the
  active document $d$ of the window $\compn{s}{\ptr{p}}$
  and every subwindow $w$ of $d$ there is a pointer
  $\ptr{p'} \in \mathsf{Subwindows}(s)$ such that
  $\compn{s}{\ptr{p'}} = w$.

  Given a browser state $s$, the set $\mathsf{Docs}(s)$ of
  pointers to active documents is the minimal set such that
  for every $\ptr{p} \in \mathsf{Subwindows}(s)$, there is
  a pointer $\ptr{p'} \in \mathsf{Docs}(s)$ with
  $\compn{s}{\ptr{p'}} =
  \comp{\compn{s}{\ptr{p}}}{activedocument}$.
\end{definition}

By $\mathsf{Subwindows}^+(s)$ and $\mathsf{Docs}^+(s)$ we
denote the respective sets that also include the inactive
documents and their subwindows.

The function $\mathsf{Clean}$ will be used to determine
which information about windows and documents the script
running in the document $d$ has access to.
\begin{definition} Let $s$
  be a browser state and $d$
  a document. By $\mathsf{Clean}(s, d)$
  we denote the term that equals $\comp{s}{windows}$
  but with (1) all inactive documents removed (including
  their subwindows etc.), (2) all subterms that represent
  non-same-origin documents w.r.t.~$d$
  replaced by a limited document $d'$
  with the same nonce and the same subwindow list, and (3)
  the values of the subterms $\str{headers}$
  for all documents set to $\an{}$.
  (Note that non-same-origin documents on all levels are
  replaced by their corresponding limited document.)
\end{definition}

The function $\mathsf{CookieMerge}$ merges two sequences of
cookies together: When used in the browser,
$\mi{oldcookies}$ is the sequence of existing cookies for
some origin, $\mi{newcookies}$ is a sequence of new cookies
that was output by some script. The sequences are merged
into a set of cookies using an algorithm that is based on
the \emph{Storage Mechanism} algorithm described in
RFC6265.
\begin{definition} \label{def:cookiemerge} For a sequence
  of cookies (with pairwise different names)
  $\mi{oldcookies}$ and a sequence of cookies
  $\mi{newcookies}$, the set
  $\mathsf{CookieMerge}(\mi{oldcookies}, \mi{newcookies})$
  is defined by the following algorithm: From
  $\mi{newcookies}$ remove all cookies $c$ that have
  $c.\str{content}.\str{httpOnly} \equiv \True$. For any
  $c$, $c' \inPairing \mi{newcookies}$, $\comp{c}{name}
  \equiv \comp{c'}{name}$, remove the cookie that appears
  left of the other in $\mi{newcookies}$. Let $m$ be the
  set of cookies that have a name that either appears in
  $\mi{oldcookies}$ or in $\mi{newcookies}$, but not in
  both. For all pairs of cookies $(c_\text{old},
  c_\text{new})$ with $c_\text{old} \inPairing
  \mi{oldcookies}$, $c_\text{new} \inPairing
  \mi{newcookies}$, $\comp{c_\text{old}}{name} \equiv
  \comp{c_\text{new}}{name}$, add $c_\text{new}$ to $m$ if
  $\comp{\comp{c_\text{old}}{content}}{httpOnly} \equiv
  \bot$ and add $c_\text{old}$ to $m$ otherwise. The result
  of $\mathsf{CookieMerge}(\mi{oldcookies},
  \mi{newcookies})$ is $m$.
\end{definition}

The function $\mathsf{AddCookie}$ adds a cookie $c$
received in an HTTP response to the sequence of cookies
contained in the sequence $\mi{oldcookies}$. It is again
based on the algorithm described in RFC6265 but simplified
for the use in the browser model.
\begin{definition} \label{def:addcookie} For a sequence of cookies (with pairwise different
  names) $\mi{oldcookies}$ and a cookie $c$, the sequence
  $\mathsf{AddCookie}(\mi{oldcookies}, c)$ is defined by the
  following algorithm: Let $m := \mi{oldcookies}$. Remove
  any $c'$ from $m$ that has $\comp{c}{name} \equiv
  \comp{c'}{name}$. Append $c$ to $m$ and return $m$.
\end{definition}

The function $\mathsf{NavigableWindows}$ returns a set of
windows that a document is allowed to navigate. We closely
follow \cite{html5}, Section~5.1.4 for this definition.
\begin{definition} The set $\mathsf{NavigableWindows}(\ptr{w}, s')$
  is the set $\ptr{W} \subseteq
  \mathsf{Subwindows}(s')$ of pointers to windows that the
  active document in $\ptr{w}$ is allowed to navigate. The
  set $\ptr{W}$ is defined to be the minimal set such that
  for every $\ptr{w'}
  \in \mathsf{Subwindows}(s')$ the following is true: %
\begin{itemize}
\item If
  $\comp{\comp{\compn{s'}{\ptr{w}'}}{activedocument}}{origin}
  \equiv
  \comp{\comp{\compn{s'}{\ptr{w}}}{activedocument}}{origin}$
  (i.e., the active documents in $\ptr{w}$ and $\ptr{w'}$ are
  same-origin), then $\ptr{w'} \in \ptr{W}$, and
\item If ${\compn{s'}{\ptr{w}} \childof
    \compn{s'}{\ptr{w'}}}$ $\wedge$ $\nexists\, \ptr{w}''
  \in \mathsf{Subwindows}(s')$ with $\compn{s'}{\ptr{w}'}
  \childof \compn{s'}{\ptr{w}''}$ ($\ptr{w'}$ is a
  top-level window and $\ptr{w}$ is an ancestor window of
  $\ptr{w'}$), then $\ptr{w'} \in \ptr{W}$, and
\item If $\exists\, \ptr{p} \in \mathsf{Subwindows}(s')$
  such that $\compn{s'}{\ptr{w}'} \windowChildOfX
  \compn{s'}{\ptr{p}}$ \\$\wedge$
  $\comp{\comp{\compn{s'}{\ptr{p}}}{activedocument}}{origin}
  =
  \comp{\comp{\compn{s'}{\ptr{w}}}{activedocument}}{origin}$
  ($\ptr{w'}$ is not a top-level window but there is an
  ancestor window $\ptr{p}$ of $\ptr{w'}$ with an active
  document that has the same origin as the active document
  in $\ptr{w}$), then $\ptr{w'} \in \ptr{W}$, and
\item If $\exists\, \ptr{p} \in \mathsf{Subwindows}(s')$ such
  that $\comp{\compn{s'}{\ptr{w'}}}{opener} =
  \comp{\compn{s'}{\ptr{p}}}{nonce}$ $\wedge$ $\ptr{p} \in
  \ptr{W}$ ($\ptr{w'}$ is a top-level window---it has an
  opener---and $\ptr{w}$ is allowed to navigate the opener
  window of $\ptr{w'}$, $\ptr{p}$), then $\ptr{w'} \in
  \ptr{W}$. 
\end{itemize}
\end{definition}

\subsection{Description of the Web Browser Atomic
  Process}\label{app:descr-web-brows}
We will now describe the relation $R^p$ of a standard HTTP
browser $p$. We define $\left(\left(\an{\an{a,f,m}},
    s\right), \left(M, s'\right)\right)$ to belong to $R^p$
if\/f the non-deterministic algorithm presented below, when
given $\left(\an{a,f,m}, s\right)$ as input, terminates
with \textbf{stop}~$M$,~$s'$, i.e., with output $M$ and
$s'$. Recall that $\an{a,f,m}$ is an (input) event and $s$
is a (browser) state, $M$ is a sequence of (output)
protoevents, and $s'$ is a new (browser) state (potentially
with placeholders for nonces).

\paragraph{Notations.} The notation $\textbf{let}\ n \leftarrow N$ is used to
describe that $n$ is chosen non-de\-ter\-mi\-nis\-tic\-ally from the
set $N$.  We write $\textbf{for each}\ s \in M\ \textbf{do}$
to denote that the following commands (until \textbf{end
  for}) are repeated for every element in $M$, where the
variable $s$ is the current element. The order in which the
elements are processed is chosen non-deterministically. %
We will write, for example, 
\begin{algorithmic}
  \LetST{$x,y$}{$\an{\str{Constant},x,y} \equiv
    t$}{doSomethingElse}
\end{algorithmic} \setlength{\parindent}{1em}
for some variables $x,y$, a string
$\str{Constant}$, and some term $t$ to express that $x :=
\proj{2}{t}$, and $y := \proj{3}{t}$ if $\str{Constant}
\equiv \proj{1}{t}$ and if $|\an{\str{Constant},x,y}| =
|t|$,  and that otherwise
$x$ and $y$ are not set and doSomethingElse is executed.

\paragraph{Placeholders.} In several places throughout the
algorithms presented next we use placeholders to generate
``fresh'' nonces as described in our communication model
(see Definition~\ref{def:terms}).
Figure~\ref{fig:browser-placeholder-list} shows a list of
all placeholders used.

\begin{figure}[htb]
  \centering
  \begin{tabular}{|@{\hspace{1ex}}l@{\hspace{1ex}}|@{\hspace{1ex}}l@{\hspace{1ex}}|}\hline 
    \hfill Placeholder\hfill  &\hfill  Usage\hfill  \\\hline\hline
    $\nu_1$ & Algorithm~\ref{alg:browsermain}, new window nonces  \\\hline
    $\nu_2$ & Algorithm~\ref{alg:browsermain}, new HTTP request nonce   \\\hline
    $\nu_3$ & Algorithm~\ref{alg:browsermain}, lookup key for pending HTTP requests entry  \\\hline
    $\nu_4$ & Algorithm~\ref{alg:runscript}, new HTTP request nonce (multiple lines)  \\\hline
    $\nu_5$ & Algorithm~\ref{alg:runscript}, new subwindow nonce  \\\hline
    $\nu_6$ & Algorithm~\ref{alg:processresponse}, new HTTP request nonce  \\\hline
    $\nu_7$ & Algorithm~\ref{alg:processresponse}, new document nonce   \\\hline
    $\nu_8$ & Algorithm~\ref{alg:send}, lookup key for pending DNS entry  \\\hline
    $\nu_9$ & Algorithm~\ref{alg:getnavigablewindow}, new window nonce  \\\hline
    $\nu_{10}, \dots$ & Algorithm~\ref{alg:runscript}, replacement for placeholders in script output   \\\hline

  \end{tabular}
  
  \caption{List of placeholders used in browser algorithms}
  \label{fig:browser-placeholder-list}
\end{figure}

Before we describe the main browser algorithm, we first
define some functions.

\subsubsection{Functions} \label{app:proceduresbrowser} In
the description of the following functions we use $a$,
$f$, $m$, and $s$ as read-only global input
variables. All other variables are local variables or
arguments.

The following function, $\mathsf{GETNAVIGABLEWINDOW}$, is
called by the browser to determine the window that is
\emph{actually} navigated when a script in the window
$s'.\ptr{w}$ provides a window reference for navigation
(e.g., for opening a link). When it is given a window
reference (nonce) $\mi{window}$,
this function returns a pointer to a
selected window term in $s'$:
\begin{itemize}
\item If $\mi{window}$ is the string $\wBlank$, a new
  window is created and a pointer to that window is
  returned.
\item If $\mi{window}$ is a nonce (reference) and there is
  a window term with a reference of that value in the
  windows in $s'$, a pointer $\ptr{w'}$ to that window term
  is returned, as long as the window is navigable by the
  current window's document (as defined by
  $\mathsf{NavigableWindows}$ above).
\end{itemize}
In all other cases, $\ptr{w}$ is returned instead (the
script navigates its own window).
\captionof{algorithm}{\label{alg:getnavigablewindow}
  Determine window for navigation}
\begin{algorithmic}[1]
  \Function{$\mathsf{GETNAVIGABLEWINDOW}$}{$\ptr{w}$, $\mi{window}$, $\mi{noreferrer}$, $s'$}
    \If{$\mi{window} \equiv \wBlank$} \Comment{Open a new window when $\wBlank$ is used}
      \If{$\mi{noreferrer} \equiv \True$}
        \Let{$w'$}{$\an{\nu_9, \an{}, \bot}$}
      \Else
        \Let{$w'$}{$\an{\nu_9, \an{}, \comp{\compn{s'}{\ptr{w}}}{nonce} }$}
      \EndIf
      \Append{$w'$}{$\comp{s'}{windows}$}  \breakalgohook{2}\textbf{and} let
      $\ptr{w}'$ be a pointer to this new element in $s'$
      \State \Return{$\ptr{w}'$}
    \EndIf
    \LetNDST{$\ptr{w}'$}{$\mathsf{NavigableWindows}(\ptr{w},
      s')$}{$\comp{\compn{s'}{\ptr{w}'}}{nonce} \equiv
      \mi{window}$\breakalgohook{1}}{\textbf{return} $\ptr{w}$} %
    \State \Return{$\ptr{w'}$}
  \EndFunction
\end{algorithmic} \setlength{\parindent}{1em}

The following function takes a window reference as input
and returns a pointer to a window as above, but it checks
only that the active documents in both windows are
same-origin. It creates no new windows.
\captionof{algorithm}{\label{alg:getwindow} Determine same-origin window}
\begin{algorithmic}[1]
  \Function{$\mathsf{GETWINDOW}$}{$\ptr{w}$, $\mi{window}$, $s'$}
    \LetNDST{$\ptr{w}'$}{$\mathsf{Subwindows}(s')$}{$\comp{\compn{s'}{\ptr{w}'}}{nonce} \equiv \mi{window}$\breakalgohook{1}}{\textbf{return} $\ptr{w}$} %
    \If{
      $\comp{\comp{\compn{s'}{\ptr{w}'}}{activedocument}}{origin}
      \equiv
      \comp{\comp{\compn{s'}{\ptr{w}}}{activedocument}}{origin}$
    }
      \State \Return{$\ptr{w}'$}
    \EndIf
    \State \Return{$\ptr{w}$}
  \EndFunction
\end{algorithmic} \setlength{\parindent}{1em}

The next function is used to stop any pending
requests for a specific window. From the pending requests
and pending DNS requests it removes any requests with the
given window reference $n$.
\captionof{algorithm}{\label{alg:cancelnav} Cancel pending requests for given window}
\begin{algorithmic}[1]
  \Function{$\mathsf{CANCELNAV}$}{$n$, $s'$}
    \State \textbf{remove all} $\an{n, \mi{req}, \mi{key}, \mi{f}}$ \textbf{ from } $\comp{s'}{pendingRequests}$ \textbf{for any} $\mi{req}$, $\mi{key}$, $\mi{f}$
    \State \textbf{remove all} $\an{x, \an{n, \mi{message}, \mi{url}}}$ \textbf{ from } $\comp{s'}{pendingDNS}$\breakalgohook{1} \textbf{for any} $\mi{x}$, $\mi{message}$, $\mi{url}$
    \State \Return{$s'$}
  \EndFunction
\end{algorithmic} \setlength{\parindent}{1em}

The following function takes an HTTP request $\mi{message}$
as input, adds cookie and origin headers to the message,
creates a DNS request for the hostname given in the request
and stores the request in $\comp{s'}{pendingDNS}$ until the
DNS resolution finishes. For normal HTTP requests,
$\mi{reference}$ is a window reference. For \xhrs,
$\mi{reference}$ is a value of the form
$\an{\mi{document}, \mi{nonce}}$ where $\mi{document}$ is a
document reference and $\mi{nonce}$ is some nonce that was
chosen by the script that initiated the request. $\mi{url}$
contains the full URL of the request (this is mainly used
to retrieve the protocol that should be used for this
message, and to store the fragment identifier for use
after the document was loaded). $\mi{origin}$ is the origin
header value that is to be added to the HTTP request.

\captionof{algorithm}{\label{alg:send} Prepare headers, do DNS resolution, save message }
\begin{algorithmic}[1]
  \Function{$\mathsf{SEND}$}{$\mi{reference}$, $\mi{message}$, $\mi{url}$, $\mi{origin}$, $\mi{referrer}$, $\mi{referrerPolicy}$, $s'$}
    \If{$\comp{\mi{message}}{host} \inPairing \comp{s'}{sts}$}
      \Let{$\mi{url}.\str{protocol}$}{$\https$}
    \EndIf
    \Let{ $\mi{cookies}$}{$\langle\{\an{\comp{c}{name}, \comp{\comp{c}{content}}{value}} | c\inPairing \comp{s'}{cookies}\left[\comp{\mi{message}}{host}\right]$} \label{line:assemble-cookies-for-request} \breakalgohook{1} $\wedge \left(\comp{\comp{c}{content}}{secure} \implies \left(\mi{url}.\str{protocol} = \https\right)\right) \}\rangle$ \label{line:cookie-rules-http}
    \Let{$\comp{\mi{message}}{headers}[\str{Cookie}]$}{$\mi{cookies}$}
    \If{$\mi{origin} \not\equiv \bot$}
      \Let{$\comp{\mi{message}}{headers}[\str{Origin}]$}{$\mi{origin}$}
    \EndIf
    \If{$\mi{referrerPolicy} \equiv \str{noreferrer}$} 
      \Let{$\mi{referrer}$}{$\bot$}
    \EndIf
    \If{$\mi{referrer} \not\equiv \bot$}
      \If{$\mi{referrerPolicy} \equiv \str{origin}$} 
        \Let{$\mi{referrer}$}{$\an{\cUrl, \mi{referrer}.\str{protocol}, \mi{referrer}.\str{host}, \str{/}, \an{}, \bot}$} \Comment{Referrer stripped down to origin.}
      \EndIf
      \Let{$\mi{referrer}.\str{fragment}$}{$\bot$} \Comment{Browsers do not send fragment identifiers in the Referer header.}
      \Let{$\comp{\mi{message}}{headers}[\str{Referer}]$}{$\mi{referrer}$}
    \EndIf
    \Let{$\comp{s'}{pendingDNS}[\nu_8]$}{$\an{\mi{reference},
        \mi{message}, \mi{url}}$} \label{line:add-to-pendingdns}
    \State \textbf{stop} $\an{\an{\comp{s'}{DNSaddress},a,
    \an{\cDNSresolve, \mi{message}.\str{host}, \nu_8}}}$, $s'$
  \EndFunction
\end{algorithmic} \setlength{\parindent}{1em}
\noindent

The following functions navigate a window forward or
backward. More precisely, they deactivate one document and
activate that document's succeeding document or preceding
document, respectively. If no such successor/predecessor exists, the
functions do not change the state.

\captionof{algorithm}{\label{alg:navback} Navigate a window backward }
\begin{algorithmic}[1]
  \Function{$\mathsf{NAVBACK}$}{$\ptr{w}$, $s'$}
      \If{$\exists\, \ptr{j} \in
        \mathbb{N}, \ptr{j} > 1$ \textbf{such that}
        $\comp{\compn{\comp{\compn{s'}{\ptr{w'}}}{documents}}{\ptr{j}}}{active}
        \equiv \True$} %
        \Let{$\comp{\compn{\comp{\compn{s'}{\ptr{w'}}}{documents}}{\ptr{j}}}{active}$}{$\bot$}
        \Let{$\comp{\compn{\comp{\compn{s'}{\ptr{w'}}}{documents}}{(\ptr{j}-1)}}{active}$}{$\True$}
        \Let{$s'$}{$\mathsf{CANCELNAV}(\comp{\compn{s'}{\ptr{w}'}}{nonce},
        s')$}
        \EndIf
  \EndFunction
\end{algorithmic} \setlength{\parindent}{1em}
\noindent

\captionof{algorithm}{\label{alg:navforward} Navigate a window forward }
\begin{algorithmic}[1]
  \Function{$\mathsf{NAVFORWARD}$}{$\ptr{w}$, $s'$}
        \If{$\exists\, \ptr{j} \in \mathbb{N} $ \textbf{such that} $\comp{\compn{\comp{\compn{s'}{\ptr{w'}}}{documents}}{\ptr{j}}}{active} \equiv \True$ \breakalgohook{3} $\wedge$  $\compn{\comp{\compn{s'}{\ptr{w'}}}{documents}}{(\ptr{j}+1)} \in \mathsf{Documents}$} %
          \Let{$\comp{\compn{\comp{\compn{s'}{\ptr{w'}}}{documents}}{\ptr{j}}}{active}$}{$\bot$}
          \Let{$\comp{\compn{\comp{\compn{s'}{\ptr{w'}}}{documents}}{(\ptr{j}+1)}}{active}$}{$\True$}
          \Let{$s'$}{$\mathsf{CANCELNAV}(\comp{\compn{s'}{\ptr{w}'}}{nonce}, s')$}
        \EndIf
  \EndFunction
\end{algorithmic} \setlength{\parindent}{1em}
\noindent

The function $\mathsf{RUNSCRIPT}$ performs a script
execution step of the script in the document
$\compn{s'}{\ptr{d}}$ (which is part of the window
$\compn{s'}{\ptr{w}}$). A new script and document state is
chosen according to the relation defined by the script and
the new script and document state is saved. Afterwards, the
$\mi{command}$ that the script issued is interpreted. 

\captionof{algorithm}{\label{alg:runscript} Execute a script}
\begin{algorithmic}[1]
  \Function{$\mathsf{RUNSCRIPT}$}{$\ptr{w}$, $\ptr{d}$, $s'$}
    \Let{$\mi{tree}$}{$\mathsf{Clean}(s', \compn{s'}{\ptr{d}})$} \label{line:clean-tree}

    \Let{$\mi{cookies}$}{$\langle\{\an{\comp{c}{name}, \comp{\comp{c}{content}}{value}} | c \inPairing \comp{s'}{cookies}\left[  \comp{\comp{\compn{s'}{\ptr{d}}}{origin}}{host}  \right]$
     \breakalgohook{1} $\wedge\,\comp{\comp{c}{content}}{httpOnly} = \bot$ \breakalgohook{1} $\wedge\,\left(\comp{\comp{c}{content}}{secure} \implies \left(\comp{\comp{\compn{s'}{\ptr{d}}}{origin}}{protocol} \equiv \https\right)\right) \}\rangle$} \label{line:assemble-cookies-for-script}
    \LetND{$\mi{tlw}$}{$\comp{s'}{windows}$ \textbf{such that} $\mi{tlw}$ is the top-level window containing $\ptr{d}$} 
    \Let{$\mi{sessionStorage}$}{$\comp{s'}{sessionStorage}\left[\an{\comp{\compn{s'}{\ptr{d}}}{origin}, \comp{\mi{tlw}}{nonce}}\right]$} %
    \Let{$\mi{localStorage}$}{$\comp{s'}{localStorage}\left[\comp{\compn{s'}{\ptr{d}}}{origin}\right]$}
    \Let{$\mi{secrets}$}{$\comp{s'}{secrets}\left[\comp{\compn{s'}{\ptr{d}}}{origin}\right]$} \label{line:browser-secrets}
    \LetND{$R$}{$\mathsf{script}^{-1}(\comp{\compn{s'}{\ptr{d}}}{script})$} %
    \Let{$\mi{in}$}{$\langle\mi{tree}$, $\comp{\compn{s'}{\ptr{d}}}{nonce}, \comp{\compn{s'}{\ptr{d}}}{scriptstate}$, $\comp{\compn{s'}{\ptr{d}}}{scriptinputs}$, $\mi{cookies},$ \breakalgohook{1}  $\mi{localStorage}$, $\mi{sessionStorage}$, $\comp{s'}{ids}$, $\mi{secrets}\rangle$}\label{line:browser-scriptinputs}
    \LetND{$\mi{state}'$}{$\terms(V)$, \breakalgohook{1}
      $\mi{cookies}' \gets \mathsf{Cookies}^\nu$, \breakalgohook{1}
      $\mi{localStorage}' \gets \terms(V)$,\breakalgohook{1}
      $\mi{sessionStorage}' \gets \terms(V)$,\breakalgohook{1}
      $\mi{command} \gets \terms(V)$, \breakalgohook{1} 
      $\mi{out}^\lambda := \an{\mi{state}', \mi{cookies}', \mi{localStorage}',$ $\mi{sessionStorage}', \mi{command}}$
      \breakalgohook{1} \textbf{such that} $(\mi{in}, \mi{out}^\lambda) \in R$}  \label{line:trigger-script} 
    \Let{$\mi{out}$}{$\mi{out}^\lambda[\nu_{10}/\lambda_1, \nu_{11}/\lambda_2, \dots]$}

    \Let{$\comp{s'}{cookies}\left[\comp{\comp{\compn{s'}{\ptr{d}}}{origin}}{host}\right]$\breakalgohook{1}}{$\langle\mathsf{CookieMerge}(\comp{s'}{cookies}\left[\comp{\comp{\compn{s'}{\ptr{d}}}{origin}}{host}\right]$, $\mi{cookies}')\rangle$} \label{line:cookiemerge}
    \Let{$\comp{s'}{localStorage}\left[\comp{\compn{s'}{\ptr{d}}}{origin}\right]$}{$\mi{localStorage}'$}
    \Let{$\comp{s'}{sessionStorage}\left[\an{\comp{\compn{s'}{\ptr{d}}}{origin}, \comp{\mi{tlw}}{nonce}}\right]$}{$\mi{sessionStorage}'$}
    \Let{$\comp{\compn{s'}{\ptr{d}}}{scriptstate}$}{$state'$}
    \Switch{$\mi{command}$}
      \Case{$\an{\tHref, \mi{url},
          \mi{hrefwindow}, \mi{noreferrer}}$}
      \Let{$\ptr{w}'$}{$\mathsf{GETNAVIGABLEWINDOW}$($\ptr{w}$,
        $\mi{hrefwindow}$, $\mi{noreferrer}$, $s'$)} 
      \Let{$\mi{req}$}{$\hreq{ nonce=\nu_4, 
          method=\mGet, host=\comp{\mi{url}}{host},
          path=\comp{\mi{url}}{path},
          headers=\an{},
          parameters=\comp{\mi{url}}{parameters}, body=\an{}
        }$}
      \If{$\mi{noreferrer} \equiv \True$}
        \Let{$\mi{referrerPolicy}$}{$\str{noreferrer}$}
      \Else
        \Let{$\mi{referrerPolicy}$}{$\compn{s'}{\ptr{d}}.\str{headers}[\str{ReferrerPolicy}]$}
      \EndIf
      \Let{$s'$}{$\mathsf{CANCELNAV}(\comp{\compn{s'}{\ptr{w}'}}{nonce}, s')$}
      \State \textsf{SEND}($\comp{\compn{s'}{\ptr{w}'}}{nonce}$, $\mi{req}$, $\mi{url}$, $\bot$, $\mi{referrer}$, $\mi{referrerPolicy}$, $s'$) \label{line:send-href}
      \EndCase
      \Case{$\an{\tIframe, \mi{url}, \mi{window}}$}
        \Let{$\ptr{w}'$}{$\mathsf{GETWINDOW}(\ptr{w}, \mi{window}, s')$}
        \Let{$\mi{req}$}{$\hreq{
            nonce=\nu_4,
            method=\mGet,
            host=\comp{\mi{url}}{host},
            path=\comp{\mi{url}}{path},
            headers=\an{},
            parameters=\comp{\mi{url}}{parameters},
            body=\an{}
          }$}
        \Let{$\mi{referrer}$}{$s'.\ptr{w}'.\str{activedocument}.\str{location}$}
        \Let{$\mi{referrerPolicy}$}{$\compn{s'}{\ptr{d}}.\str{headers}[\str{ReferrerPolicy}]$}
        \Let{$w'$}{$\an{\nu_5, \an{}, \bot}$}
        \Let{$\comp{\comp{\compn{s'}{\ptr{w}'}}{activedocument}}{subwindows}$\breakalgohook{3}}{ $\comp{\comp{\compn{s'}{\ptr{w}'}}{activedocument}}{subwindows} \plusPairing w'$}
        \State \textsf{SEND}($\nu_5$, $\mi{req}$, $\mi{url}$, $\bot$, $\mi{referrer}$, $\mi{referrerPolicy}$, $s'$) \label{line:send-iframe}
      \EndCase
      \Case{$\an{\tForm, \mi{url}, \mi{method}, \mi{data}, \mi{hrefwindow}}$}
        \If{$\mi{method} \not\in \{\mGet, \mPost\}$} \footnote{The working draft for HTML5 allowed for DELETE and PUT methods in HTML5 forms. However, these have since been removed. See \url{http://www.w3.org/TR/2010/WD-html5-diff-20101019/\#changes-2010-06-24}.}
          \State \textbf{stop} $\an{}$, $s'$
        \EndIf
        \Let{$\ptr{w}'$}{$\mathsf{GETNAVIGABLEWINDOW}$($\ptr{w}$, $\mi{hrefwindow}$, $\bot$, $s'$)}
        \If{$\mi{method} = \mGet$}
          \Let{$\mi{body}$}{$\an{}$}
          \Let{$\mi{parameters}$}{$\mi{data}$}
          \Let{$\mi{origin}$}{$\bot$}
        \Else
          \Let{$\mi{body}$}{$\mi{data}$}
          \Let{$\mi{parameters}$}{$\comp{\mi{url}}{parameters}$}
          \Let{$\mi{origin}$}{$\comp{\compn{s'}{\ptr{d}}}{origin}$}
        \EndIf
        \Let{$\mi{req}$}{$\hreq{
            nonce=\nu_4,
            method=\mi{method},
            host=\comp{\mi{url}}{host},
            path=\comp{\mi{url}}{path},
            headers=\an{},
            parameters=\mi{parameters},
            xbody=\mi{body}
          }$}
        \Let{$\mi{referrer}$}{$\comp{\compn{s'}{\ptr{d}}}{location}$}
        \Let{$\mi{referrerPolicy}$}{$\compn{s'}{\ptr{d}}.\str{headers}[\str{ReferrerPolicy}]$}
        \Let{$s'$}{$\mathsf{CANCELNAV}(\comp{\compn{s'}{\ptr{w}'}}{nonce}, s')$}
        \State \textsf{SEND}($\comp{\compn{s'}{\ptr{w}'}}{nonce}$, $\mi{req}$, $\mi{url}$, $\mi{origin}$, $\mi{referrer}$, $\mi{referrerPolicy}$, $s'$) \label{line:send-form}
      \EndCase
      \Case{$\an{\tSetScript, \mi{window}, \mi{script}}$}
        \Let{$\ptr{w}'$}{$\mathsf{GETWINDOW}(\ptr{w}, \mi{window}, s')$}
        \Let{$\comp{\comp{\compn{s'}{\ptr{w}'}}{activedocument}}{script}$}{$\mi{script}$}
        \State \textbf{stop} $\an{}$, $s'$
      \EndCase
      \Case{$\an{\tSetScriptState, \mi{window}, \mi{scriptstate}}$}
        \Let{$\ptr{w}'$}{$\mathsf{GETWINDOW}(\ptr{w}, \mi{window}, s')$}
        \Let{$\comp{\comp{\compn{s'}{\ptr{w}'}}{activedocument}}{scriptstate}$}{$\mi{scriptstate}$}
        \State \textbf{stop} $\an{}$, $s'$
      \EndCase
      \Case{$\an{\tXMLHTTPRequest, \mi{url}, \mi{method}, \mi{data}, \mi{xhrreference}}$}
        \If{$\mi{method} \in \{\mConnect, \mTrace, \mTrack\} \wedge \mi{xhrreference} \not\in \{\nonces, \bot\}$} 
          \State \textbf{stop} $\an{}$, $s'$
        \EndIf
        \If{$\comp{\mi{url}}{host} \not\equiv \comp{\comp{\compn{s'}{\ptr{d}}}{origin}}{host}$ \breakalgohook{3} $\vee$ $\mi{url} \not\equiv \comp{\comp{\compn{s'}{\ptr{d}}}{origin}}{protocol}$} 
          \State \textbf{stop} $\an{}$, $s'$
        \EndIf
        \If{$\mi{method} \in \{\mGet, \mHead\}$}
          \Let{$\mi{data}$}{$\an{}$}
          \Let{$\mi{origin}$}{$\bot$}
        \Else
          \Let{$\mi{origin}$}{$\comp{\compn{s'}{\ptr{d}}}{origin}$}
        \EndIf
        \Let{$\mi{req}$}{$\hreq{
            nonce=\nu_4,
            method=\mi{method},
            host=\comp{\mi{url}}{host},
            path=\comp{\mi{url}}{path},
            headers={},
            parameters=\comp{\mi{url}}{parameters},
            xbody=\mi{data}
          }$}
        \Let{$\mi{referrer}$}{$\comp{\compn{s'}{\ptr{d}}}{location}$}
        \Let{$\mi{referrerPolicy}$}{$\compn{s'}{\ptr{d}}.\str{headers}[\str{ReferrerPolicy}]$}
        \State \textsf{SEND}($\an{\comp{\compn{s'}{\ptr{d}}}{nonce}, \mi{xhrreference}}$, $\mi{req}$, $\mi{url}$, $\mi{origin}$, $\mi{referrer}$, $\mi{referrerPolicy}$, $s'$)\label{line:send-xhr}
      \EndCase
      \Case{$\an{\tBack, \mi{window}}$} \footnote{Note that
        navigating a window using the back/forward buttons
        does not trigger a reload of the affected
        documents. While real world browser may chose to
        refresh a document in this case, we assume that the
        complete state of a previously viewed document is
        restored. A reload can be triggered
        non-deterministically at any point (in the main algorithm).}
      \Let{$\ptr{w}'$}{$\mathsf{GETNAVIGABLEWINDOW}$($\ptr{w}$,
        $\mi{window}$, $\bot$, $s'$)} 
        \State $\mathsf{NAVBACK}$($\ptr{w}$, $s'$)
        \State \textbf{stop} $\an{}$, $s'$
      \EndCase
      \Case{$\an{\tForward, \mi{window}}$}
        \Let{$\ptr{w}'$}{$\mathsf{GETNAVIGABLEWINDOW}$($\ptr{w}$, $\mi{window}$, $\bot$, $s'$)}
        \State $\mathsf{NAVFORWARD}$($\ptr{w}$, $s'$)
        \State \textbf{stop} $\an{}$, $s'$
      \EndCase
      \Case{$\an{\tClose, \mi{window}}$}
        \Let{$\ptr{w}'$}{$\mathsf{GETNAVIGABLEWINDOW}$($\ptr{w}$, $\mi{window}$, $\bot$, $s'$)}
        \State \textbf{remove} $\compn{s'}{\ptr{w'}}$ from the sequence containing it 
        \State \textbf{stop} $\an{}$, $s'$
      \EndCase

      \Case{$\an{\tPostMessage, \mi{window}, \mi{message}, \mi{origin}}$}
        \LetND{$\ptr{w}'$}{$\mathsf{Subwindows}(s')$ \textbf{such that} $\comp{\compn{s'}{\ptr{w}'}}{nonce} \equiv \mi{window}$} %
        \If{$\exists \ptr{j} \in \mathbb{N}$ \textbf{such that} $\comp{\compn{\comp{\compn{s'}{\ptr{w'}}}{documents}}{\ptr{j}}}{active} \equiv \True$ \breakalgohook{3} $\wedge  (\mi{origin} \not\equiv \bot \implies \comp{\compn{\comp{\compn{s'}{\ptr{w'}}}{documents}}{\ptr{j}}}{origin} \equiv \mi{origin})$}    \label{line:append-pm-to-scriptinputs-condition} %
        \Let{$\comp{\compn{\comp{\compn{s'}{\ptr{w'}}}{documents}}{\ptr{j}}}{scriptinputs}$\breakalgohook{4}}{ $\comp{\compn{\comp{\compn{s'}{\ptr{w'}}}{documents}}{\ptr{j}}}{scriptinputs}$ \breakalgohook{4} $\plusPairing$
         $\an{\tPostMessage, \comp{\compn{s'}{\ptr{w}}}{nonce}, \comp{\compn{s'}{\ptr{d}}}{origin}, \mi{message}}$} \label{line:append-pm-to-scriptinputs}
        \EndIf
        \State \textbf{stop} $\an{}$, $s'$
      \EndCase
      \Case{else}
        \State \textbf{stop} $\an{}$, $s'$
      \EndCase
    \EndSwitch
  \EndFunction
\end{algorithmic} \setlength{\parindent}{1em}

The function $\mathsf{PROCESSRESPONSE}$ is responsible for
processing an HTTP response ($\mi{response}$) that was
received as the response to a request ($\mi{request}$) that
was sent earlier. In $\mi{reference}$, either a window or a
document reference is given (see explanation for
Algorithm~\ref{alg:send} above). $\mi{requestUrl}$ contains
the URL used when retrieving the document.

The function first saves any cookies that were contained in
the response to the browser state, then checks whether a
redirection is requested (Location header). If that is not
the case, the function creates a new document (for normal
requests) or delivers the contents of the response to the
respective receiver (for \xhr responses).
\captionof{algorithm}{\label{alg:processresponse} Process an HTTP response}
\begin{algorithmic}[1]
\Function{$\mathsf{PROCESSRESPONSE}$}{$\mi{response}$, $\mi{reference}$, $\mi{request}$, $\mi{requestUrl}$, $s'$}
  \If{$\mathtt{Set{\mhyphen}Cookie} \in
    \comp{\mi{response}}{headers}$}
    \For{\textbf{each} $c \inPairing \comp{\mi{response}}{headers}\left[\mathtt{Set{\mhyphen}Cookie}\right]$, $c \in \mathsf{Cookies}$}
      \Let{$\comp{s'}{cookies}\left[\mi{request}.\str{host}\right]$\breakalgohook{3}}{$\mathsf{AddCookie}(\comp{s'}{cookies}\left[\mi{request}.\str{host}\right], c)$} \label{line:set-cookie}
    \EndFor
  \EndIf  
  \If{$\mathtt{Strict{\mhyphen}Transport{\mhyphen}Security} \in \comp{\mi{response}}{headers}$ $\wedge$ $\mi{requestUrl}.\str{protocol} \equiv \https$}
    \Append{$\comp{\mi{request}}{host}$}{$\comp{s'}{sts}$}
  \EndIf
  \If{$\str{Referer} \in \comp{request}{headers}$} 
    \Let{$\mi{referrer}$}{$\comp{request}{headers}[\str{Referer}]$}
  \Else
    \Let{$\mi{referrer}$}{$\bot$}
  \EndIf
  \If{$\mathtt{Location} \in \comp{\mi{response}}{headers} \wedge \comp{\mi{response}}{status} \in \{303, 307\}$} \label{line:location-header} 
    \Let{$\mi{url}$}{$\comp{\mi{response}}{headers}\left[\mathtt{Location}\right]$}
    \If{$\mi{url}.\str{fragment} \equiv \bot$}
      \Let{$\mi{url}.\str{fragment}$}{$\mi{requestUrl}.\str{fragment}$}
    \EndIf
    \Let{$\mi{method}'$}{$\comp{\mi{request}}{method}$} 
    \Let{$\mi{body}'$}{$\comp{\mi{request}}{body}$} 
    \If{$\str{Origin} \in \comp{request}{headers}$}
      \Let{$\mi{origin}$}{$\an{\comp{request}{headers}[\str{Origin}], \an{\comp{request}{host}, \mi{url}.\str{protocol}}}$}
    \Else
      \Let{$\mi{origin}$}{$\bot$}
    \EndIf
    \If{$\comp{\mi{response}}{status} \equiv 303 \wedge \comp{\mi{request}}{method} \not\in \{\mGet, \mHead\}$}
      \Let {$\mi{method}'$}{$\mGet$}
      \Let{$\mi{body}'$}{$\an{}$} \label{browser-remove-body}
    \EndIf
    \If{$\nexists\, \ptr{w} \in \mathsf{Subwindows}(s')$ \textbf{such that} $\comp{\compn{s'}{\ptr{w}}}{nonce} \equiv \mi{reference}$} \Comment{Do not redirect XHRs.}
      \State \textbf{stop} $\an{}$, $s$
    \EndIf
    \Let{$\mi{req}$}{$\hreq{
            nonce=\nu_6,
            method=\mi{method'},
            host=\comp{\mi{url}}{host},
            path=\comp{\mi{url}}{path},
            headers=\an{},
            parameters=\comp{\mi{url}}{parameters},
            xbody=\mi{body}'
          }$}
    \Let{$\mi{referrerPolicy}$}{$\mi{response}.\str{headers}[\str{ReferrerPolicy}]$}
    \State \textsf{SEND}($\mi{reference}$, $\mi{req}$, $\mi{url}$, $\mi{origin}$, $\mi{referrer}$, $\mi{referrerPolicy}$, $s'$)\label{line:send-redirect}
  \EndIf
  \If{$\exists\, \ptr{w} \in \mathsf{Subwindows}(s')$ \textbf{such that} $\comp{\compn{s'}{\ptr{w}}}{nonce} \equiv \mi{reference}$} \Comment{normal response}
    \If{$\mi{response}.\str{body} \not\sim \an{*,*}$}
      \State \textbf{stop} $\{\}$, $s'$
    \EndIf
    \Let{$\mi{script}$}{$\proj{1}{\comp{\mi{response}}{body}}$}
    \Let{$\mi{scriptstate}$}{$\proj{2}{\comp{\mi{response}}{body}}$}
    \Let{$\mi{referrer}$}{$\mi{request}.\str{headers}[\str{Referer}]$}
    \Let{$d$}{$\an{\nu_7, \mi{requestUrl}, \mi{response}.\str{headers}, \mi{referrer}, \mi{script}, \mi{scriptstate}, \an{}, \an{}, \True}$} \label{line:take-script} \label{line:set-origin-of-document}
    \If{$\comp{\compn{s'}{\ptr{w}}}{documents} \equiv \an{}$}
      \Let{$\comp{\compn{s'}{\ptr{w}}}{documents}$}{$\an{d}$}
    \Else
      \LetND{$\ptr{i}$}{$\mathbb{N}$ \textbf{such that} $\comp{\compn{\comp{\compn{s'}{\ptr{w}}}{documents}}{\ptr{i}}}{active} \equiv \True$} %
      \Let{$\comp{\compn{\comp{\compn{s'}{\ptr{w}}}{documents}}{\ptr{i}}}{active}$}{$\bot$}
      \State \textbf{remove} $\compn{\comp{\compn{s'}{\ptr{w}}}{documents}}{(\ptr{i}+1)}$ and all following documents \breakalgohook{3} from $\comp{\compn{s'}{\ptr{w}}}{documents}$
      \Append{$d$}{$\comp{\compn{s'}{\ptr{w}}}{documents}$}
    \EndIf
    \State \textbf{stop} $\{\}$, $s'$
  \ElsIf{$\exists\, \ptr{w} \in \mathsf{Subwindows}(s')$, $\ptr{d}$ \textbf{such that} $\comp{\compn{s'}{\ptr{d}}}{nonce} \equiv \proj{1}{\mi{reference}} $ \breakalgohook{1}  $\wedge$  $\compn{s'}{\ptr{d}} = \comp{\compn{s'}{\ptr{w}}}{activedocument}$} \label{line:process-xhr-response} \Comment{process XHR response}
    \Let{$\mi{headers}$}{$\mi{response}.\str{headers} - \str{Set\mhyphen{}Cookie}$}
    \Append{\breakalgo{3}$\an{\tXMLHTTPRequest, \mi{headers}, \comp{\mi{response}}{body}, \proj{2}{\mi{reference}}}$}{$\comp{\compn{s'}{\ptr{d}}}{scriptinputs}$}
  \EndIf
\EndFunction
\end{algorithmic} \setlength{\parindent}{1em}

\subsubsection{Main Algorithm.}\label{app:mainalgorithmwebbrowserprocess}
This is the main algorithm of the browser relation.
It receives the message $m$ as input, as
well as $a$, $f$ and $s$ as above.

\captionof{algorithm}{\label{alg:browsermain} Web browser main algorithm.}
\begin{algorithmic}[1]
\Statex[-1] \textbf{Input:} $\an{a,f,m},s$
  \Let{$s'$}{$s$}

  \If{$\comp{s}{isCorrupted} \not\equiv \bot$}
    \Let{$\comp{s'}{pendingRequests}$}{$\an{m, \comp{s}{pendingRequests}}$} \Comment{Collect incoming messages}
    \LetND{$m'$}{$d_{V}(s')$} %
    \LetND{$a'$}{$\addresses$} %
    \State \textbf{stop} $\an{\an{a',a,m'}}$, $s'$
  \EndIf
  \If{$m \equiv \trigger$} \Comment{A special trigger message. }
    \LetND{$\mi{switch}$}{$\{\str{script},\str{urlbar},\str{reload},\str{forward}, \str{back}\}$} \label{line:browser-switch}  %
    \LetNDST{$\ptr{w}$}{$\mathsf{Subwindows}(s')$}{$\comp{\compn{s'}{\ptr{w}}}{documents} \neq \an{}$\breakalgohook{2}}{\textbf{stop} $\an{}$, $s'$}%
    \Comment{Pointer to some window.}
    \LetNDST{$\ptr{tlw}$}{$\mathbb{N}$}{$\comp{\compn{s'}{\ptr{tlw}}}{documents} \neq \an{}$\breakalgohook{2}}{\textbf{stop} $\an{}$, $s'$}%
    \Comment{Pointer to some top-level window.}
    \If{$\mi{switch} \equiv \str{script}$} \Comment{Run some script.}
      \Let{$\ptr{d}$}{$\ptr{w} \plusPairing \str{activedocument}$} \label{line:browser-trigger-document}  
      \State \textsf{RUNSCRIPT}($\ptr{w}$, $\ptr{d}$, $s'$)
    \ElsIf{$\mi{switch} \equiv \str{urlbar}$} \Comment{Create some new request.}
      \LetND{$\mi{newwindow}$}{$\{\True, \bot \}$}
      \If{$\mi{newwindow} \equiv \True$} \Comment{Create a new window.}
        \Let{$\mi{windownonce}$}{$\nu_1$}
        \Let{$w'$}{$\an{\mi{windownonce}, \an{}, \bot}$}
        \Append{$w'$}{$\comp{s'}{windows}$}
      \Else \Comment{Use existing top-level window.}
        \Let{$\mi{windownonce}$}{$s'.\ptr{tlw}.nonce$}
      \EndIf
      \LetND{$\mi{protocol}$}{$\{\http, \https\}$} \label{line:browser-choose-url} %
      \LetND{$\mi{host}$}{$\dns$} %
      \LetND{$\mi{path}$}{$\mathbb{S}$} %
      \LetND{$\mi{fragment}$}{$\mathbb{S}$} %
      \LetND{$\mi{parameters}$}{$\dict{\mathbb{S}}{\mathbb{S}}$} %
      \Let{$\mi{url}$}{$\an{\cUrl, \mi{protocol}, \mi{host}, \mi{path}, \mi{parameters}, \mi{fragment}}$}
      \Let{$\mi{req}$}{$\hreq{
          nonce=\nu_2,
          method=\mGet,
          host=\mi{host},
          path=\mi{path},
          headers=\an{},
          parameters=\mi{parameters},
          body=\an{}
        }$}
      \State \textsf{SEND}($\mi{windownonce}$, $\mi{req}$, $\mi{url}$, $\bot$, $\bot$, $\bot$, $s'$)\label{line:send-random}
    \ElsIf{$\mi{switch} \equiv \str{reload}$} \Comment{Reload some document.}
      \LetNDST{$\ptr{w}$}{$\mathsf{Subwindows}(s')$}{$\comp{\compn{s'}{\ptr{w}}}{documents} \neq \an{}$\breakalgohook{2}}{\textbf{stop} $\an{}$, $s'$} \label{line:browser-reload-window}%
      \Let{$\mi{url}$}{$s'.\ptr{w}.\str{activedocument}.\str{location}$}
      \Let{$\mi{req}$}{$\hreq{ nonce=\nu_2, 
          method=\mGet, host=\comp{\mi{url}}{host},
          path=\comp{\mi{url}}{path},
          headers=\an{},
          parameters=\comp{\mi{url}}{parameters}, body=\an{}
        }$}
      \Let{$\mi{referrer}$}{$s'.\ptr{w}.\str{activedocument}.\str{referrer}$}
      \Let{$s'$}{$\mathsf{CANCELNAV}(\comp{\compn{s'}{\ptr{w}}}{nonce}, s')$}
      \State \textsf{SEND}($\comp{\compn{s'}{\ptr{w}}}{nonce}$, $\mi{req}$, $\mi{url}$, $\bot$, $\mi{referrer}$, $\bot$, $s'$)
    \ElsIf{$\mi{switch} \equiv \str{forward}$}
      \State $\mathsf{NAVFORWARD}$($\ptr{w}$, $s'$)
    \ElsIf{$\mi{switch} \equiv \str{back}$}
      \State $\mathsf{NAVBACK}$($\ptr{w}$, $s'$)
    \EndIf
  \ElsIf{$m \equiv \fullcorrupt$} \Comment{Request to corrupt browser}
    \Let{$\comp{s'}{isCorrupted}$}{$\fullcorrupt$}
    \State \textbf{stop} $\an{}$, $s'$
  \ElsIf{$m \equiv \closecorrupt$} \Comment{Close the browser}
    \Let{$\comp{s'}{secrets}$}{$\an{}$}  
    \Let{$\comp{s'}{windows}$}{$\an{}$}
    \Let{$\comp{s'}{pendingDNS}$}{$\an{}$}
    \Let{$\comp{s'}{pendingRequests}$}{$\an{}$}
    \Let{$\comp{s'}{sessionStorage}$}{$\an{}$}
    \State \textbf{let} $\comp{s'}{cookies} \subsetPairing \cookies$ \textbf{such that} \breakalgohook{1} $(c \inPairing \comp{s'}{cookies}) {\iff} (c \inPairing \comp{s}{cookies} \wedge \comp{\comp{c}{content}}{session} \equiv \bot$)
    \Let{$\comp{s'}{isCorrupted}$}{$\closecorrupt$}
    \State \textbf{stop} $\an{}$, $s'$
  \ElsIf{$\exists\, \an{\mi{reference}, \mi{request}, \mi{url}, \mi{key}, f}$
      $\inPairing \comp{s'}{pendingRequests}$ \breakalgohook{0}
      \textbf{such that} $\proj{1}{\decs{m}{\mi{key}}} \equiv \cHttpResp$ } %
    \Comment{Encrypted HTTP response}
    \Let{$m'$}{$\decs{m}{\mi{key}}$}
    \If{$\comp{m'}{nonce} \not\equiv \comp{\mi{request}}{nonce}$}
      \State \textbf{stop} $\an{}$, $s$
    \EndIf
    \State \textbf{remove} $\an{\mi{reference}, \mi{request}, \mi{url}, \mi{key}, f}$ \textbf{from} $\comp{s'}{pendingRequests}$
    \State \textsf{PROCESSRESPONSE}($m'$, $\mi{reference}$, $\mi{request}$, $\mi{url}$, $s'$)
  \ElsIf{$\proj{1}{m} \equiv \cHttpResp$ $\wedge$ $\exists\, \an{\mi{reference}, \mi{request}, \mi{url}, \bot, f}$ $\inPairing \comp{s'}{pendingRequests}$ \breakalgohook{0}\textbf{such that} $\comp{m'}{nonce} \equiv \comp{\mi{request}}{key}$ } %
    \State \textbf{remove} $\an{\mi{reference}, \mi{request}, \mi{url}, \bot, f}$ \textbf{from} $\comp{s'}{pendingRequests}$
    \State \textsf{PROCESSRESPONSE}($m$, $\mi{reference}$, $\mi{request}$, $\mi{url}$, $s'$)
  \ElsIf{$m \in \dnsresponses$} \Comment{Successful DNS response}
      \If{$\comp{m}{nonce} \not\in \comp{s}{pendingDNS} \vee \comp{m}{result} \not\in \addresses \vee \comp{m}{domain} \not\equiv \comp{\proj{2}{\comp{s}{pendingDNS}}}{host}$}
        \State \textbf{stop} $\an{}$, $s$ \label{line:browser-dns-response-stop}
      \EndIf
      \Let{$\an{\mi{reference}, \mi{message}, \mi{url}}$}{$\comp{s}{pendingDNS}[\comp{m}{nonce}]$}
      \If{$\mi{url}.\str{protocol} \equiv \https$}
        \AppendBreak{2}{$\langle\mi{reference}$, $\mi{message}$, $\mi{url}$, $\nu_3$, $\comp{m}{result}\rangle$}{$\comp{s'}{pendingRequests}$} \label{line:add-to-pendingrequests-https}
        \Let{$\mi{message}$}{$\enc{\an{\mi{message},\nu_3}}{\comp{s'}{keyMapping}\left[\comp{\mi{message}}{host}\right]}$} \label{line:select-enc-key}
      \Else
        \AppendBreak{2}{$\langle\mi{reference}$, $\mi{message}$, $\mi{url}$, $\bot$, $\comp{m}{result}\rangle$}{$\comp{s'}{pendingRequests}$} \label{line:add-to-pendingrequests}
      \EndIf
      \Let{$\comp{s'}{pendingDNS}$}{$\comp{s'}{pendingDNS} - \comp{m}{nonce}$}
      \State \textbf{stop} $\an{\an{\comp{m}{result}, a, \mi{message}}}$, $s'$
  \EndIf
  \State \textbf{stop} $\an{}$, $s$

\end{algorithmic} \setlength{\parindent}{1em}

%% file: appendix-oauth-model.tex
\section{Formal Model of OAuth with a Network Attacker}
\label{app:model-oauth-auth}

We here present the full details of our formal model of OAuth which we
use to analyze all but one of the authentication and authorization
properties. This model contains a network attacker. We will later
derive from this model a model where the network attacker is replaced
by a web attacker. 

We model OAuth as a web system (in the sense of
Appendix~\ref{app:websystem}). We call a web system
$\oauthwebsystem^n=(\bidsystem, \scriptset, \mathsf{script}, E^0)$
an \emph{OAuth web system with a network attacker} if it is of the
form described in what follows.

\subsection{Outline}\label{app:outlineoauthmodel}
The system $\bidsystem=\mathsf{Hon} \cup \mathsf{Net}$
consists of a network attacker process (in $\mathsf{Net}$),
a finite set $\fAP{B}$
of web browsers, a finite set $\fAP{RP}$
of web servers for the relying parties, a finite set $\fAP{IDP}$
of web servers for the identity providers, with
$\mathsf{Hon} := \fAP{B} \cup \fAP{RP} \cup \fAP{IDP}$.
More details on the processes in $\bidsystem$
are provided below. We do not model DNS servers, as they are subsumed
by the network attacker.
Figure~\ref{fig:scripts-in-w} shows the set of scripts $\scriptset$
and their respective string representations that are defined by the
mapping $\mathsf{script}$.
The set $E^0$ contains only the trigger events as specified in
Appendix~\ref{app:websystem}. 

\begin{figure}[htb]
  \centering
  \begin{tabular}{|@{\hspace{1ex}}l@{\hspace{1ex}}|@{\hspace{1ex}}l@{\hspace{1ex}}|}\hline 
    \hfill $s \in \scriptset$\hfill  &\hfill  $\mathsf{script}(s)$\hfill  \\\hline\hline
    $\Rasp$ & $\str{att\_script}$  \\\hline
    $\mi{script\_rp\_index}$ & $\str{script\_rp\_index}$  \\\hline
    $\mi{script\_rp\_implicit}$ & $\str{script\_rp\_implicit}$  \\\hline
    $\mi{script\_idp\_form}$ &  $\str{script\_idp\_form}$  \\\hline
  \end{tabular}
  
  \caption{List of scripts in $\scriptset$ and their respective string
    representations.}
  \label{fig:scripts-in-w}
\end{figure}

This outlines $\oauthwebsystem^n$. We will now define the DY processes in
$\oauthwebsystem^n$ and their addresses, domain names, and secrets in more
detail. 

\subsection{Addresses and Domain Names}\label{app:addresses-and-domain-names}
The set $\addresses$
contains for the network attacker in $\fAP{Net}$,
every relying party in $\fAP{RP}$,
every identity provider in $\fAP{IDP}$,
and every browser in $\fAP{B}$
a finite set of addresses each. By $\mapAddresstoAP$
we denote the corresponding assignment from a process to its address.
The set $\dns$
contains a finite set of domains for every relying party in
$\fAP{RP}$,
every identity provider in $\fAP{IDP}$,
and the network attacker in $\fAP{Net}$.
Browsers (in $\fAP{B})$ do not have a domain.

By $\mapAddresstoAP$ and $\mapDomain$ we denote the assignments from
atomic processes to sets of $\addresses$ and $\dns$, respectively.

\subsection{Keys and Secrets} The set $\nonces$ of nonces is
partitioned into five sets, an infinite sequence $N$, an infinite set
$K_\text{SSL}$, an infinite set $K_\text{sign}$, and finite sets
$\mathsf{Passwords}$, $\mathsf{RPSecrets}'$ and $\SecRes$. We thus have
\begin{align*}
\def\hereMaxHeightPhantom{\vphantom{K_{\text{p}}^\bidsystem}}
\nonces = 
\underbrace{N\hereMaxHeightPhantom}_{\text{infinite sequence}} 
\dot\cup \underbrace{K_{\text{SSL}}\hereMaxHeightPhantom}_{\text{finite}} 
\dot\cup \underbrace{\mathsf{Passwords}\hereMaxHeightPhantom}_{\text{finite}}
\dot\cup \underbrace{\mathsf{RPSecrets}'\hereMaxHeightPhantom}_{\text{finite}}
\dot\cup \underbrace{\SecRes\hereMaxHeightPhantom}_{\text{finite}}\ .
\end{align*}
We then define $\mathsf{RPSecrets} := \mathsf{RPSecrets} \cup \{\bot\}$.
These sets are used as follows:
\begin{itemize}
\item The set $N$ contains the nonces that are available for each DY
  process in $\bidsystem$ (it can be used to create a run of
  $\bidsystem$).

\item The set $K_\text{SSL}$
  contains the keys that will be used for SSL encryption. Let
  $\mapSSLKey\colon \dns \to K_\text{SSL}$
  be an injective mapping that assigns a (different) private key to
  every domain. For an atomic DY process $p$
  we define
  $\mi{sslkeys}^p = \an{\left\{\an{d, \mapSSLKey(d)} \mid d \in
      \mapDomain(p)\right\}}$.

\item The set $\mathsf{Passwords}$
  is the set of passwords (secrets) the browsers share with the
  identity providers. These are the passwords the users use to log in
  at the IdPs.

\item The set $\mathsf{RPSecrets}$
  is the set of passwords (secrets) the relying parties share with the
  identity providers. These are the passwords the relying parties use
  to log in at the IdPs. The passwords can also be blank ($\bot$).

\item   The set $\SecRes$ contains a secret for each combination of IdP,
  client, and user. These are thought of as protected resources that
  only the owner of the resource (i.e., the user) should be able to
  read. (See also Definition~\ref{def:property-authz-a}.)
\end{itemize}

\subsection{Identities, Passwords, and Protected Resources}\label{app:oauth-pidp-identities}
Identites consist, similar to email addresses, of a user name and a
domain part. For our model, this is defined as follows:
\begin{definition}
  An \emph{identity} (email address) $i$ is a term of the form
  $\an{\mi{name},\mi{domain}}$ with $\mi{name}\in \mathbb{S}$ and
  $\mi{domain} \in \dns$.

  Let $\IDs$ be the finite set of identities. By $\IDs^y$ we denote
  the set $\{ \an{\mi{name}, \mi{domain}} \in \IDs\,|\, \mi{domain}
  \in \mapDomain(y) \}$.

  We say that an ID is \emph{governed} by the DY process to which the
  domain of the ID belongs. Formally, we define the mapping $\mapGovernor:
  \IDs \to \bidsystem$, $\an{\mi{name}, \mi{domain}} \mapsto
  \mapDomain^{-1}(\mi{domain})$.
\end{definition}%
The governor of an ID will usually be an IdP, but could also be the
attacker. Besides $\mapGovernor$, we define the following mappings:

\begin{itemize}
\item By $\mapIDtoPLI:\IDs \to \mathsf{Passwords}$ we denote the bijective
  mapping that assigns secrets to all identities.

\item   Let $\mapPLItoOwner: \mathsf{Passwords} \to \fAP{B}$ denote the mapping that
  assigns to each secret a browser that \emph{owns} this secret. Now,
  we define the mapping $\mapIDtoOwner: \IDs \to \fAP{B}$, $i \mapsto
  \mapPLItoOwner(\mapIDtoPLI(i))$, which assigns to each identity the
  browser that owns this identity (we say that the identity belongs to
  the browser).

\item Let $\mathsf{trustedRPs}: \mathsf{Passwords} \to 2^{\fAP{RP}}$
  denote a mapping that assigns a set of \emph{trusted relying
    parties} to each password. Intuitively a trusted relying party is
  a relying party the user entrusts with her password (in the resource
  owner password credentials grant mode of OAuth).

\item Let
  $\mathsf{clientIDOfRP}: (\mathsf{RP} \cup \{\bot\}) \times
  \mathsf{IDP} \to \mathbb{S} \cup \{\bot\}$ denote a mapping that
  assigns an OAuth client id for an relying party to each combination
  of a relying party and an identity provider. We require that
  $\mathsf{clientIDOfRP}(\cdot,i)$
  is bijective for all $i \in \mathsf{IDP}$
  and that $\mathsf{clientIDOfRP}(r, i) = \bot$
  iff $r$ = $\bot$ for all $i \in \mathsf{IDP}$.

\item Let
  $\mathsf{secretOfRP}: \mathsf{RP} \times \mathsf{IDP} \to
  \mathsf{RPSecrets}$ denote a bijective mapping that
  assigns a relying party password (or the empty password $\bot$)
  to each combination of a relying party and an identity provider.

\item As a shortcut, we define the mapping
  $\mathsf{secretOfClientID}: \mathbb{S} \times \mathsf{IDP} \to
  \mathsf{RPSecrets}$ to return the relying party password to a
  relying party identified by an OAuth client id (at some specific
  identity provider), i.e., $\mathsf{secretOfClientID}(s,i)$
  maps to $\mathsf{secretOfRP}(r, i)$
  with $r$ such that $s = \mathsf{clientIDOfRP}(r,i)$.

\item By
  $\mathsf{resourceOf}: \fAP{IDP} \times (\fAP{RP} \cup \{\bot\})
  \times (\IDs \cup \{\bot\}) \to \SecRes$ we denote the injective
  mapping that assigns a protected resource to each combination of
  user identity, IdP and client (RP). We also include protected
  resources that are not assigned to a specific user (in this case,
  the user is $\bot$)
  and those that are not assigned to a specific RP (the RP then is
  $\bot$).
  Note that a protected resource depends not only on the IdP and user
  ID but also the RP. This is motivated by the fact that different RPs
  may get access to different protected resources at one IdP, even if
  they access the resources of the same user. In the resource owner
  password credentials mode, RPs can also access resources that do not
  depend on the RP, we then have that RP is $\bot$.\footnote{In
    the resource owner password credentials mode, the RP gets the
    user's credentials and thus has full access to the user's account
    at IdP. This access is not bound to potential limitations that
    depend on the RP's identity.}
\end{itemize}

\subsection{Corruption}
RPs and IdPs can become corrupted: If they receive the message
$\corrupt$, they start collecting all incoming messages in their state
and (upon triggering) send out all messages that are derivable from
their state and collected input messages, just like the attacker
process. We say that an RP or an IdP is \emph{honest} if the according
part of their state ($s.\str{corrupt}$) is $\bot$, and that they are
corrupted otherwise.

We are now ready to define the processes in $\websystem$ as well as
the scripts in $\scriptset$ in more detail.

\subsection{Processes in $\bidsystem$ (Overview)}

We first provide an overview of the processes in $\bidsystem$. All
processes in $\websystem$ contain in their initial states all public
keys and the private keys of their respective domains (if any). We
define $I^p=\mapAddresstoAP(p)$ for all $p\in \mathsf{Hon}$.

\subsubsection{Network Attacker}
There is one atomic DY process $\mi{na} \in \mathsf{Net}$
which is a network attacker (see Appendix~\ref{app:websystem}), who
uses all addresses for sending and listening.

\subsubsection{Browsers} Each $b \in \fAP{B}$
is a web browser as defined in Appendix~\ref{app:deta-descr-brows}.
The initial state contains all secrets owned by $b$,
stored under the origins of the respective IdP and of all trusted RPs
for the respective secret. See Appendix~\ref{app:browsers-oauth} for
details.

\subsubsection{Relying Parties} 
Each relying party is a web server modeled as an atomic DY process
following the description in Section~\ref{sec:oauth} and the fixes
discussed in Section~\ref{sec:attacks}. The RP can either (at any
time) launch a client credentials mode flow or wait for users to start
any of the other flows. RP manages two kinds of sessions: The
\emph{login sessions}, which are only used during the login phase of a
user, and the \emph{service sessions} (modeled by a \emph{service
  token} as described above).

When receiving a special message ($\corrupt$)
RPs can become corrupted. Similar to the definition of corruption for
the browser, RPs then start sending out all messages that are
derivable from their state.

\subsubsection{Identity Providers} Each IdP is a web server modeled
as an atomic DY process following the description in
Section~\ref{sec:oauth} and the fixes discussed in
Section~\ref{sec:attacks}. In particular, users can authenticate to
the IdP with their credentials. Authenticated users can interact with
the authorization endpoint of the IdP (e.g., to acquire an
authorization code). Just as RPs, IdPs can become corrupted.

\subsection{Network Attackers}\label{app:networkattackers-oauth} As mentioned, the network attacker
$\mi{na}$
is modeled to be a network attacker as specified in
Appendix~\ref{app:websystem}. We allow it to listen to/spoof all
available IP addresses, and hence, define $I^\mi{na} = \addresses$.
The initial state is
$s_0^\mi{na} = \an{\mi{attdoms}, \mi{sslkeys}, \mi{signkeys}}$,
where $\mi{attdoms}$
is a sequence of all domains along with the corresponding private keys
owned by the attacker $\mi{na}$,
$\mi{sslkeys}$
is a sequence of all domains and the corresponding public keys, and
$\mi{signkeys}$
is a sequence containing all public signing keys for all IdPs.

\subsection{Browsers}\label{app:browsers-oauth} 

Each $b \in \fAP{B}$ is a web browser as defined in
Appendix~\ref{app:deta-descr-brows}, with $I^b := \mapAddresstoAP(b)$
being its addresses.

To define the inital state, first let $\IDs_b :=
\mapIDtoOwner^{-1}(b)$ be
 the set of all IDs of $b$. We then define the set of passwords that a browser $b$ gives to an origin $o$ to consist of two parts: (1) If the origin belongs to an IdP, then the user's passwords of this IdP are contained in the set. (2) If the origin belongs to an RP, then those passwords with which the user entrusts this RP are contained in the set. To define this mapping in the initial state, we first define for some process~$p$

\begin{align*}
 \mathsf{Secrets}^{b,p} = \Big\{ s \Bigm| b = \mathsf{ownerOfSecret}(s)  \wedge  \big( & (\exists\, i : s = \mathsf{secretOfID}(i) \wedge i \in \mathsf{governor}^{-1}(p))\\ &\vee  (\exists\, R : p \in R \wedge s \in \mathsf{trustedRPs}^{-1}(R)) \big) \Big\} \,.
\end{align*}

Then, the initial state $s_0^b$
is defined as follows: the key mapping maps every domain to its public
(ssl) key, according to the mapping $\mapSSLKey$;
the DNS address is an address of the network attacker; the list of
secrets contains an entry
$\an{\an{d,\https}, \an{\mathsf{Secrets}^{b,p}}}$
for each $p \in \fAP{RP} \cup \fAP{IDP}$
and $d \in \mathsf{dom}(p)$;
$\mi{ids}$ is $\an{\IDs_b}$; $\mi{sts}$ is empty.

\subsection{Relying Parties} \label{app:relying-parties-oauth}

A relying party $r \in \fAP{RP}$
is a web server modeled as an atomic DY process
$(I^r, Z^r, R^r, s^r_0)$
with the addresses $I^r := \mapAddresstoAP(r)$.
Its initial state $s^r_0$
contains its domains, the private keys associated with its domains,
the DNS server address, and information about IdPs RP is registered
at. The full state additionally contains the sets of service tokens
and login session identifiers the RP has issued as well as information
about pending DNS and pending HTTPS requests (similar to browsers). RP
only accepts HTTPS requests.

RP manages two kinds of sessions: The \emph{login sessions}, which are
only used during the login phase of a user, and the \emph{service
  sessions} (we call the session identifier of a service session a
\emph{service token}). Service sessions allow a user to use RP's
services. The ultimate goal of a login flow is to establish such a
service session.

We now first describe how $r$
can become corrupted, then we describe the handling of DNS and HTTPS
requests and responses, before we describe the behaviour of $r$
during a login flow.

\paragraph{Corruption}
When $r$
receives a corrupt message, it becomes corrupt and acts like the
attacker from then on (i.e., it collects all incoming messages and
non-deterministically sends out all messages derivable from its
state).

\paragraph{Pending DNS Requests and Pending HTTPS Requests}
Since the RP $r$
also acts as an HTTPS client, it manages two kinds of records for
messages that have been sent out into the network and are waiting for
corresponding responses. When an HTTPS message is to be sent, the RP
first needs to resolve the hostname into an IP address. To this end,
the RP first stores the HTTPS request (together with some state
information) in a subterm of its state called $\mi{pendingDNS}$
and (instead of sending the HTTPS request immediately) sends out a DNS
request to the DNS server. When a DNS response arrives that matches
one of the entries in this subterm, the HTTPS request is sent out over
the network (to the resolved IP address) and stored in the subterm
$\mi{pendingRequests}$
of the RP's state. Note that this mechanism is very similar to
(generic) browsers (see Appendix~\ref{app:deta-descr-brows}).

\paragraph{Initial Request}
In a typical flow, $r$
will first receive an HTTP GET request from a browser for the path
$\str{/}$.
In this case, $r$
returns the script $\str{script\_rp\_index}$. Besides providing arbitrary links, this script allows users to start an OAuth flow in the browser.
If an OAuth flow is started, this script non-deterministically chooses an identity
of the user, i.e., a combination of a username and a domain of an IdP.
Further this script non-deterministically decides whether an
interactive login (i.e., authorization code mode or implicit mode) or
a non-interactive login (i.e., resource owner password credentials
mode) is used. If an interactive login is chosen, the script instructs
the browser to send an HTTPS POST request to $r$
for the path $\str{/startInteractiveLogin}$.
This POST request contains in its body the domain of the
IdP.\footnote{Note that while the script has selected an identity of
  the user, only the domain of the IdP is used in this case and during
  the authentication to the IdP, a different username may be chosen.}
If the script chooses a non-interactive login, the domain of the IdP,
the username, and the user's password are sent to $r$
in an HTTPS POST request for the path $\str{/passwordLogin}$.

As the flow now forks into different branches, we will explain (the
first part of) each of these branches separately: If the script has
chosen to run an interactive login, we continue our description in the
paragraph \emph{Interactive Login} below. Else, if the script has
chosen to run a non-interactive login, we continue our description of this in
the paragraph \emph{Non-Interactive Login}.

\paragraph{Interactive Login}
In this case, $\str{script\_rp\_index}$
has sent an HTTPS POST request for the path
$\str{/startInteractiveLogin}$
to $r$
containing the name of an IdP in its body. When $r$
receives such a request, $r$
non-deterministically decides whether the OAuth authorization code
mode or the OAuth implicit mode is used. Also, $r$
non-deterministically selects a redirect URI $\mi{redirect\_uri}$
of its redirection endpoints (and appends the domain of the IdP to
this redirect URI) or selects no redirect URI. Further, $r$
non-deterministically selects a (fresh) nonce $\mi{state}$
and a (fresh) nonce as login session id. Then, $r$
saves all the chosen information in its state. Now, $r$
constructs and sends an HTTPS response containing an HTTP 303 location
redirect or an HTTP 307 location redirect\footnote{Note that while in
  this paper we present an attack against OAuth based on an HTTP 307
  location redirect, our analysis shows that an HTTP 307 location
  redirect is safe at this point in the protocol flow.} (chosen
non-deterministically) which points to the corresponding authorization
endpoint at the IdP along with $r$'s
OAuth client id for this IdP, $\mi{state}$
and information which OAuth mode $r$
has chosen. Additionally, this response also contains a Set-Cookie
header, which sets a cookie containing the login session id. $r$
also stores a record in the subterm $\mi{loginSessions}$
of its state. This record contains the login session id, the chosen
OAuth mode, and the domain of the IdP.

Later, when IdP redirects the user's browser to $r$'s
redirection endpoint, $r$
will receive an HTTPS GET request for the path
$\str{/redirectionEndpoint}$.
This request must contain a login session id cookie, which refers to
the information stored in the subterm $\mi{loginSessions}$ in $r$'s
state.  The request must also contain a parameter with the domain of the
IdP and this domain must match the domain stored for this login session.

If $r$
has stored that for this login session the OAuth authorization code
mode is used, $r$
checks if the $\mi{state}$
value contained in a parameter is correct (i.e., the value of this
parameter is congruent to the value recorded in
$r$'s
state). Then, $r$
extracts the authorization code $\mi{code}$
from the parameters of the incoming request and prepares an HTTPS POST
request to the IdP's token endpoint to obtain an access token as
follows: $r$
adds the authorization code to the request's body. If a redirect URI
has been set by $r$
before (according to $r$'s
state for this login session), the redirect URI is included in the
request's body. If $r$
knows an OAuth client secret for the IdP, $r$
adds its OAuth client id and its OAuth client secret for the IdP to
the header of the request, else $r$
adds its OAuth client id for the IdP to the request's body. Now, $r$
sends a DNS request for the domain of the IdP's token endpoint to the
DNS server (according to $r$'s
state), saves this (prepared) request and all information belonging to
the (incoming) HTTPS request $r$
received from the browser (such as IP addresses, temporary HTTPS keys)
in $\mi{pendingDNS}$
in its state. We will
continue our description of which requests $r$
will process next in the OAuth authorization code mode in the
paragraph \emph{Token Response} below.

If the (incoming) HTTPS request's login session at $r$
states that implicit mode is used, $r$
instead sends an HTTPS response to the sender of the incoming
message. This HTTPS response contains the script
$\mi{script\_rp\_implicit}$
and the initial state for this script in this response contains the
domain of the IdP.

In a browser, this script extracts $\mi{access\_token}$
and $\mi{state}$
from the fragment part of its URL and extracts the domain of the IdP
from its initial state. The script then sends this information in the
body of an HTTPS POST request for the path
$\str{/receiveTokenFromImplicitGrant}$ to $r$.

When $r$
receives such an HTTPS POST request (for the path
$\str{/receiveTokenFromImplicitGrant}$),
$r$
checks if this request contains a login session id cookie, which
refers to the information stored in its
state and if the values of $\mi{state}$
and $\mi{idp}$
(contained in the request) match the information there. Next, $r$
prepares an HTTPS request to IdP's introspection endpoint containing the
access token just received. $r$
saves all information belonging to this new request and the (incoming)
request it had just received in $\mi{pendingDNS}$ in its
state and sends out a DNS request for the domain of the IdP's
introspection endpoint to the DNS server.

We describe what happens when $r$
later receives the response from IdP in the paragraph \emph{Introspection
  Response} below.

\paragraph{Non-Interactive Login} 
In this case, $\str{script\_rp\_index}$
has sent an HTTPS POST request for the path $\str{/passwordLogin}$
to $r$
containing a domain of an IdP, a username and a user's password in
its body.  Next, $r$
constructs an HTTPS POST request to the token endpoint of the
IdP. This request contains the username and the user's password in its
body and if $r$
knows an OAuth client secret for the IdP, the request contains an HTTP
header with $r$'s
OAuth client id and OAuth client secret.  $r$
saves all information belonging to this new request and the (incoming) request
$r$
has just received in the subterm $\mi{pendingDNS}$ in $r$'s state
and sends out a DNS request for the domain of the IdP's token endpoint to the DNS server.

We describe what happens when $r$
later receives the response from the IdP in the paragraph \emph{Token
  Response} below.

\paragraph{Client Credentials Mode} 
When $r$
receives a $\str{TRIGGER}$
message (which models that $r$
non-deterministically starts an OAuth flow in the client credentials
mode), $r$
first non-deterministically selects a domain of an IdP. Then, $r$
constructs an HTTPS POST request to the token endpoint of the
IdP. This request contains an HTTP header with $r$'s
OAuth client id and OAuth client secret.\footnote{Note that in our
  model, $r$
  may even construct such a request if $r$
  does not have an OAuth client secret for the IdP. In this case, the
  symbol $\bot$
  is placed in this header instead of an OAuth client secret. The IdP,
  however, will drop such a request, as it is not authenticated.} $r$
saves all information belonging to this (prepared) request in
$\mi{pendingDNS}$
and sends out a DNS request for the domain of the IdP's token endpoint
to the DNS server.

We describe what happens when $r$
later receives the response from IdP in the paragraph \emph{Token
  Response} below.

\paragraph{Token Response} 
 When $r$
receives an encrypted HTTP response that matches a record in the
subterm $\mi{pendingRequests}$
of its state and belongs to a request for an access token from an IdP
(according to the information recorded in $\mi{pendingRequests}$),
then $r$
extracts the access token and prepares an HTTPS request to the IdP's
introspection endpoint containing the access token. $r$
saves all information belonging to this new request in
$\mi{pendingDNS}$.
Further, $r$
also stores selected information, which is passed along in $r$'s
state in the corresponding record of the incoming request, such as the
IP address of the sender and the HTTPS response key of the request
which initiated $r$'s
request for the access token before. Then, $r$
sends out a DNS request for the domain of the IdP's introspection
endpoint to the DNS server.

\paragraph{Introspection Response}
When $r$
receives an encrypted HTTP response that matches a record in the
subterm $\mi{pendingRequests}$
in its state and this record belongs to a request to an IdP's
introspection endpoint, $r$
checks whether the response belongs to a flow in client credentials
mode (according to the record). If that is the case, $r$
stops. Otherwise, $r$
non-deterministically proceeds with either an authorization flow or an
authentication flow:
\begin{itemize}
\item If authorization is selected, $r$
  retrieves the protected resource from the IdP's response and sends
  out an HTTPS response to the IP address recorded in the record in
  $\mi{pendingRequests}$
  (which contains the IP address of the browser, which initially sent
  either user credentials, an authorization code, or an access token).
\item Else, authentication is selected. Now, if the response does not
  contain $r$'s
  OAuth client id, $r$
  stops. Otherwise, $r$
  retrieves the user id from the response and non-deterministically
  chooses a fresh nonce as a service token. $r$
  records in its state that the service token belongs to the user
  identified by the user id at the IdP. Now, $r$
  sends out a response (as above) which contains the service token in
  a cookie.
\end{itemize}
In both cases, $r$ replies with the script $\mi{script\_rp\_index}$,
which provides arbitrary links and the possibility to start a new
OAuth flow (see above).

This concludes the description of the behaviour of an RP.

\subsubsection{Formal description} We now provide the formal
definition of $r$
as an atomic DY process $(I^r, Z^r, R^r, s^r_0)$.
As mentioned, we define $I^r = \mapAddresstoAP(r)$.
Next, we define the set $Z^r$
of states of $r$ and the initial state $s^r_0$ of~$r$.

\begin{definition}\label{def:idp-registration-record}
  An \emph{IdP registration record} is a term of the form
  \[
    \an{\mi{tokenEndpoint},\mi{authorizationEndpoint},\mi{introspectionEndpoint},\mi{clientId},\mi{clientPassword}}
  \]
  with $\mi{tokenEndpoint}$,
  $\mi{authorizationEndpoint}$,
  $\mi{introspectionEndpoint} \in \urls$,
  $\mi{clientId} \in \mathbb{S}$,
  and $\mi{clientPassword} \in \nonces$.

  An \emph{IdP registration record for an identity provider $i$
    at a relying party $r$}
  is an IdP registration record with $\mi{tokenEndpoint}.\str{host}$,
  $\mi{authorizationEndpoint}.\str{host}$,
  $\mi{introspectionEndpoint}.\str{host} \in \mathsf{dom}(i)$,
  $\mi{clientId} = \mathsf{clientIDOfRP}(r,i)$,
  and $\mi{clientPassword} = \mathsf{secretOfRP}(r,i)$.
\end{definition}

\begin{definition}\label{def:relying-parties}
  A \emph{state $s\in Z^r$ of an RP $r$} is a term of the form
  $\langle\mi{DNSAddress}$, $\mi{idps}$, $\mi{serviceTokens}$,
  $\mi{loginSessions}$, $\mi{keyMapping}$, $\mi{sslkeys}$,
  $\mi{pendingDNS}$, $\mi{pendingRequests}$, $\mi{corrupt}\rangle$
  where $\mi{DNSAddress} \in \addresses$,
  $\mi{idps} \in \dict{\dns}{\terms}$ is a dictionary of IdP
  registration records,
  $\mi{serviceTokens}\in\dict{\nonces}{\terms}$,
  $\mi{loginSessions} \in \dict{\nonces}{\terms}$ is a dictionary of
  login session records,
  $\mi{keyMapping} \in \dict{\mathbb{S}}{\nonces}$,
  $\mi{sslkeys}=\mi{sslkeys}^r$,
  $\mi{pendingDNS} \in \dict{\nonces}{\terms}$,
  $\mi{pendingRequests} \in \dict{\nonces}{\terms}$,
  $\mi{corrupt}\in\terms$.

  An \emph{initial state $s^r_0$ of $r$} is a state of $r$ with
  $s^r_0.\str{idps}$ being a dictionary that maps each domain of all
  identity providers $i$ to an IdP registration record for $i$ at $r$,
  $s^r_0.\str{serviceTokens} = s^r_0.\str{loginSessions} = \an{}$,
  $s^r_0.\str{corrupt} = \bot$, and $s^r_0.\str{keyMapping}$ is the
  same as the keymapping for browsers above.
\end{definition}

We now specify the relation $R^r$. Just like
in Appendix~\ref{app:deta-descr-brows}, we describe this
relation by a non-deterministic algorithm. In several places throughout this
algorithm  we use placeholders to generate
``fresh'' nonces as described in our communication model
(see Definition~\ref{def:terms}).
Figure~\ref{fig:rp-placeholder-list} shows a list of
all placeholders used.

\begin{figure}[htb]
  \centering
  \begin{tabular}{|@{\hspace{1ex}}l@{\hspace{1ex}}|@{\hspace{1ex}}l@{\hspace{1ex}}|}\hline 
    \hfill Placeholder\hfill  &\hfill  Usage\hfill  \\\hline\hline
    $\nu_1$ & new HTTP request nonce  \\\hline
    $\nu_2$ & lookup key for pending DNS entry   \\\hline
    $\nu_3$ & new service token  \\\hline
    $\nu_4$ & fresh HTTPS response key  \\\hline
    $\nu_5$ & new HTTP request nonce  \\\hline
    $\nu_6$ & lookup key for pending DNS entry  \\\hline
    $\nu_7$ & new CSRF token  \\\hline
    $\nu_8$ & new login session cookie  \\\hline
    $\nu_9$ & new HTTP request nonce  \\\hline
    $\nu_{10}$ & lookup key for pending DNS entry   \\\hline
    $\nu_{11}$ & new HTTP request nonce  \\\hline
    $\nu_{12}$ & lookup key for pending DNS entry   \\\hline
    $\nu_{13}$ & new HTTP request nonce  \\\hline
    $\nu_{14}$ & lookup key for pending DNS entry   \\\hline
    
  \end{tabular}
  
  \caption{List of placeholders used in the relying party algorithm}
  \label{fig:rp-placeholder-list}
\end{figure}

\input{appendix-oauth-model-rp-code}
In the following scripts, to extract the current
URL of a document, the function
$\mathsf{GETURL}(\mi{tree},\mi{docnonce})$
is used. We define this function as follows: It searches for the
document with the identifier $\mi{docnonce}$
in the (cleaned) tree $\mi{tree}$
of the browser's windows and documents. It then returns the URL $u$
of that document. If no document with nonce $\mi{docnonce}$
is found in the tree $\mi{tree}$, $\notdef$ is returned.

We use the helper function $\mathsf{GETDOCWINDOW}(\mi{tree},\mi{docnonce})$. It returns the nonce of the window in $\mi{tree}$ that contains the document identified by $\mi{docnonce}$.

\input{appendix-oauth-model-script-rp-index-code}

\input{appendix-oauth-model-script-rp-implicit-code}

\subsection{Identity Providers}  \label{app:idps}

An identity provider $i \in \mathsf{IdPs}$ is a web server modeled as
an atomic process $(I^i, Z^i, R^i, s_0^i)$ with the addresses $I^i :=
\mapAddresstoAP(i)$. Its initial state $s^i_0$ contains a list of its
domains and (private) SSL keys, the paths for the endpoints
(authorization and token), a list of users, a list of clients, and
information about the corruption status (initially, the IdP is not
corrupted). Besides this, the full state of $i$ further contains a
list of issued authorization codes and access tokens.

Once the IdP becomes corrupted (when it receives the message
$\str{corrupt}$), it starts collecting all input messages and
non-deterministically sending out whatever messages are derivable from
its state.

Otherwise, IdPs react to three types of requests:

\textbf{Requests to the authorization endpoint path:} In this case, the
IdP expects a POST request containing valid user credentials. If the
user credentials are not supplied, or the request is not a POST
request, the answer contains a script which
shows a form to the user to enter her user credentials. In our model,
the script just extracts the user credentials from the browser and
sends a request to the IdP containing the user credentials and any
OAuth parameters contained in the original request (e.g., the intended
redirect URI).

If the IdP received a POST request with valid user credentials, it
checks the contained client identifier against its own list of
clients. If the client identifier is unknown, the IdP aborts.
Otherwise, it ensures that the redirect URI, if contained in the
request, is valid. For this, it checks the list of redirect URIs
stored along with the client identifier. If none of the redirect URIs
match the redirect URI presented in the request (see ``Matching
Redirect URIs'' below), the IdP aborts. If no redirect URI is provided
in the request, the first URI in the list of redirect URIs is chosen
as the redirect URI. 

Now the IdP creates a new authorization code and 
saves this code together with the client identifier and the redirect
URI (if provided in the request) to the list of authorization codes.

Now, if the response type parameter in the request is ``code'', the
IdP issues a Location redirect header to the redirect URI, appending
(as parameters) the newly created authorization code and the state (if
provided in the request).

If the reponse type is ``token'', the IdP redirects the browser to the
redirect URI, but appends the authorization code, the state (if
provided) and a fixed string (containing the token type, which is
``bearer'') to the hash of the redirect URI.

\textbf{Requests to the token endpoint path:} Requests to the token
endpoint path are only accepted by the IdP if they are POST requests.
The IdP then checks that the request either contains a valid client
ID, provided as a parameter, or a pair of client ID and client
password in a basic authentication header. 

If the grant type parameter is \emph{authorization code}, then the IdP
checks that the authorization code delivered to it is contained in the
list of codes. It checks that the client ID and redirect URI are the
same as those stored in the list of codes. It then creates an access
token and returns it in the HTTPS response (with token type
``bearer'').

If the grant type is \emph{password}, the IdP checks the provided
username and password and creates an access token as above.

If the grant type is \emph{client credentials}, the IdP checks that
the client was authorized with client ID and client password above. If
so, it creates an access token as above.

\textbf{Requests to the introspection endpoint path:} In this case, the
IdP expects an access token in the parameters of the request. If the
access token is valid, the IdP returns the client and user id for
which the access token was issued along with the protected resource for
this client, user, and IdP. 

\subsubsection{Formal description} In the following, we will first
define the (initial) state of $i$ formally and afterwards present the
definition of the relation $R^i$.

To define the initial state, we will need to add a list of all
protected resources that this IdP manages. We therefore define
$\mi{srlist}^i := \an{\{\mathsf{resourceOf}(i, c, u)\,|\, c \in
  \fAP{RP} \cup \{\bot\}, u \in \IDs\}}$ for some IdP $i$.
(Note that we do not use this term for term manipulations in the
algorithm. Instead, this term ensures that the output of the atomic
process is derivable from the input.)

\begin{definition}\label{def:initial-state-idp}
  A \emph{state $s\in Z^i$ of an IdP $i$} is a term of the form
  $\langle\mi{sslkeys}$, $\mi{srlist}$, $\mi{authEndpoint}$,
  $\mi{tokenEndpoint}$, $\mi{introspectEndpoint}$, $\mi{clients}$,
  $\mi{codes}$, $\mi{corrupt}\rangle$ where $\mi{sslkeys} =
  \mi{sslkeys}^i$, $\mi{srlist} = \mi{srlist}^i$, $\mi{authEndpoint}$,
  $\mi{tokenEndpoint}$, $\mi{introspectEndpoint}$ $\in \mathbb{S}$,
  $\mi{clients} \in \dict{\mathbb{S}}{\terms}$, $\mi{codes} \in \terms$, $\mi{atokens} \in
  \dict{\nonces}{\mathbb{S}}$.

  An \emph{initial state $s^i_0$
    of $i$}
  is a state of the form
  $\an{ \mi{sslkeys}^i, \mi{srlist}^i, w, x, y, \mi{clients}^i, \an{},
    \an{}, \bot}$ for some strings $w$,
  $x$
  and $y$
  and a dictionary $\mi{clients}^i$
  that for each relying party $r$
  contains an entry of the form $\an{\mathsf{clientIDOfRP}(r,i), z}$
  where $z$
  is a sequence of URL terms that may contain the wildcard $*$
  (see Definition~\ref{def:pattern-matching}) where for every
  $u \inPairing z$
  we have that $u.\str{protocol} \equiv \https$,
  $u.\str{host} \in \mathsf{dom}(r)$,
  $u.\str{parameters}[\str{iss}] \equiv d$
  for some $d \in \mathsf{dom}(i)$,
  $u.\str{parameters}[\str{client\_id}] \equiv
  \mathsf{clientIDOfRP}(r,i)$, $u.\str{fragment} \equiv \an{}$,
  and $u.\str{path} \equiv \str{/redirectionEndpoint}$.
  (Note that this includes the changes proposed by \df{todo})
\end{definition}

The relation $R^i$ that defines the behavior of the IdP $i$ is defined as follows:

\captionof{algorithm}{\label{alg:idp-oauth} Relation of IdP $R^i$}
\begin{algorithmic}[1]
\Statex[-1] \textbf{Input:} $\an{a,f,m},s$
  \If{$s'.\str{corrupt} \not\equiv \bot \vee m \equiv \corrupt$}
    \Let{$s'.\str{corrupt}$}{$\an{\an{a, f, m}, s'.\str{corrupt}}$}
    \LetND{$m'$}{$d_{V}(s')$}
    \LetND{$a'$}{$\addresses$}
    \Stop{$\an{\an{a',a,m'}}$, $s'$}
  \EndIf
  \Let{$s'$}{$s$}
  \LetST{$m_{\text{dec}}$, $k$, $k'$, $\mi{inDomain}$}{\breakalgohook{0}$\an{m_{\text{dec}}, k} \equiv \dec{m}{k'} \wedge \an{inDomain,k'} \in s.\str{sslkeys}$\breakalgohook{0}}{\textbf{stop} \DefStop} %
  \LetST{$n$, $\mi{method}$, $\mi{path}$, $\mi{parameters}$, $\mi{headers}$, $\mi{body}$}{\breakalgohook{0}$\an{\cHttpReq, n, \mi{method}, \mi{inDomain}, \mi{path}, \mi{parameters}, \mi{headers}, \mi{body}} \equiv m_{\text{dec}}$\breakalgohook{0}}{\textbf{stop} \DefStop} %
  \If{$\mi{path} \equiv s.\str{authEndpoint}$} \Comment{Authorization Endpoint.}
    \If{$\mi{method} \equiv \mGet \vee (\mi{method} \equiv \mPost \wedge (\mi{body}[\str{username}] \equiv \an{} \vee \mi{body}[\str{password}] \equiv \an{}))$}
      \Let{$\mi{data}$}{$\mi{parameters}$}
      \Let{$m'$}{$\encs{\an{\cHttpResp, n, 200, \an{\an{\str{ReferrerPolicy}, \str{origin}}}, \an{\str{script\_idp\_form}, \mi{data}}}}{k}$} \label{line:idp-send-form}
      \Stop{\StopWithMPrime}
    \ElsIf{$\mi{method} \equiv \mPost$} \label{line:idp-auth-endpoint-post}
      \If{$\mi{headers}[\str{Origin}]  \not\equiv \an{\mi{inDomain}, \https}$}  \Comment{CSRF protection.}
        \Stop{\DefStop}
      \EndIf
      \Let{$\mi{username}$}{$\mi{body}[\str{username}]$}
      \Let{$\mi{password}$}{$\mi{body}[\str{password}]$}
      \Let{$\mi{clientid}$}{$\mi{body}[\str{client\_id}]$} 
      \Let{$\mi{allowedredirects}$}{$s.\str{clients}[\mi{clientid}]$} \label{line:check-redir-uris}
      \If{$\mi{password} \not\equiv \mapIDtoPLI(\mi{username})$}
        \Stop{\DefStop}
      \EndIf
      \If{$\mi{allowedredirects} \equiv \an{}$}
        \Stop{\DefStop}
      \EndIf
      \Let{$\mi{redirecturi}$}{$\mi{body}[\str{redirect\_uri}]$}
      \If{$\mi{redirecturi} \not\equiv \an{}$}
        \If{\textbf{not}\ $\mi{redirecturi}\,\dot{\sim}\,\mi{allowedredirects}$}
          \Stop{\DefStop}
        \EndIf
      \Else
        \LetND{$\mi{redirecturi}$}{$\mi{allowedredirects}$} \Comment{Take one from list of redir URIs.}
      \EndIf
      \If{$\mi{body}[\str{response\_type}] \equiv \str{code}$}
        \Append{$\an{\nu_1,\an{\mi{clientid}, \mi{body}[\str{redirect\_uri}], \mi{username}}}$}{$s'.\str{codes}$} \label{line:idp-create-auth-code}\Comment{Create authorization code.} 
        \Append{$\an{\str{code}, \nu_1}$}{$\mi{redirecturi}.\str{parameters}$}
        \Append{$\an{\str{state}, \mi{body}[\str{state}]}$}{$\mi{redirecturi}.\str{parameters}$}
        \Let{$m'$}{$\encs{\an{\cHttpResp, n, 303, \an{\an{\str{Location}, \mi{redirecturi}}}, \an{}}}{k}$} \label{line:idp-redir-with-auth-code}
        \Stop{\StopWithMPrime}
      \Else \Comment{Assume response type token.}
        \Append{$\an{\nu_1, \mi{clientid}, \mi{username}}$}{$s'.\str{atokens}$} \label{line:idp-create-atoken-implicit}
        \Append{$\an{\str{access\_token}, \nu_1}$}{$\mi{redirecturi}.\str{fragment}$}
        \Append{$\an{\str{token\_type}, \str{bearer}}$}{$\mi{redirecturi}.\str{fragment}$}
        \Append{$\an{\str{state}, \mi{body}[\str{state}]}$}{$\mi{redirecturi}.\str{fragment}$}
        \Let{$m'$}{$\encs{\an{\cHttpResp, n, 303, \an{\an{\str{Location}, \mi{redirecturi}}}, \an{}}}{k}$}\label{line:idp-redir-with-access-token}\label{line:idp-redir-with-token}
        \Stop{\StopWithMPrime}
      \EndIf
    \EndIf \label{line:idp-auth-endpoint-post-end}
  \ElsIf{$\mi{path} \equiv s.\str{tokenEndpoint}$} \label{line:idp-token-endpoint} \Comment{Token Endpoint.}
    \If{$\mi{method} \not\equiv \mPost$}
      \Stop{\DefStop}
    \EndIf
    \Let{$\mi{auth}$}{$\bot$}
    \Let{$\mi{clientid}$}{$\bot$}
    \If{$\mi{body}[\str{client\_id}] \not\equiv \an{}$} \Comment{Only client ID is provided, no password.}
      \Let{$\mi{clientid}$}{$\mi{body}[\str{client\_id}]$}
      \Let{$\mi{clientinfo}$}{$s.\str{clients}[\mi{clientid}]$}
      \If{$\mi{clientinfo} \equiv \an{} \vee \mathsf{secretOfClientID}(\mi{clientid}, i) \not\equiv \bot$} \Comment{Empty client secret allowed?}
        \Stop{\DefStop}
      \EndIf
    \ElsIf{$\mi{headers}[\str{Authorization}].1 \not\equiv \an{}$} \label{line:idp-check-auth}
      \Let{$\mi{clientid}$}{$\mi{headers}[\str{Authorization}].1$}
      \Let{$\mi{clientpw}$}{$\mi{headers}[\str{Authorization}].2$}
      \If{$\mathsf{secretOfClientID}(\mi{clientid}, i) \not\equiv \mi{clientpw} \vee \mi{clientpw} \equiv \bot$}
        \Stop{\DefStop}
      \EndIf
      \Let{$\mi{auth}$}{$\mi{clientid}$} \Comment{Authentication with client credentials.}
    \EndIf
    \If{$\mi{body}[\str{grant\_type}] \equiv \str{authorization\_code}$}
      \If{$\mi{clientid} \equiv \bot$}
        \Stop{\DefStop}
      \EndIf
      \Let{$\mi{codeinfo}$}{$s.\str{codes}[\mi{body}[\str{code}]]$}
      \If{$\mi{codeinfo} \equiv \an{} \vee \mi{codeinfo}.1 \not\equiv \mi{clientid} \vee \mi{codeinfo}.2 \not\equiv \mi{body}[\str{redirect\_uri}]$}
        \Stop{\DefStop}
      \EndIf
      \Remove{$s'.\str{codes}$}{$\mi{body}[\str{code}]$}
      \Append{$\an{\nu_1, \mi{clientid}, \mi{codeinfo}.3}$}{$s'.\str{atokens}$} \label{line:idp-create-atoken-code}\Comment{Add nonce, client ID and user ID to list of tokens.}
      \Let{$m'$}{$\encs{\an{\cHttpResp, n, 200, \an{}, \an{\an{\str{access\_token}, \nu_1}, \an{\str{token\_type}, \str{bearer}}}}}{k}$}      
      \Stop{\StopWithMPrime}
    \ElsIf{$\mi{body}[\str{grant\_type}] \equiv \str{password}$} \label{line:idp-grant-type-password}
      \Let{$\mi{username}$}{$\mi{body}[\str{username}]$}
      \Let{$\mi{password}$}{$\mi{body}[\str{password}]$}
      \If{$\mi{password} \not\equiv \mapIDtoPLI(\mi{username})$}
        \Stop{\DefStop}
      \EndIf
      \Append{$\an{\nu_1, \mi{clientid}, \mi{username}}$}{$s'.\str{atokens}$} \label{line:idp-create-atoken-password}
      \Let{$m'$}{$\encs{\an{\cHttpResp, n, 200, \an{}, \an{\an{\str{access\_token}, \nu_1}, \an{\str{token\_type}, \str{bearer}}}}}{k}$}      
      \Stop{\StopWithMPrime}
    \ElsIf{$\mi{body}[\str{grant\_type}] \equiv \str{client\_credentials}$}
      \If{$\mi{auth} \equiv \bot$}
        \Stop{\DefStop}
      \EndIf
      \Append{$\an{\nu_1, \mi{clientid}, \bot}$}{$s'.\str{atokens}$} \label{line:idp-create-atoken-clientcred}
      \Let{$m'$}{$\encs{\an{\cHttpResp, n, 200, \an{}, \an{\an{\str{access\_token}, \nu_1}, \an{\str{token\_type}, \str{bearer}}}}}{k}$}      
      \Stop{\StopWithMPrime}
    \EndIf \label{line:idp-token-endpoint-end}
  \ElsIf{$\mi{path} \equiv s.\str{introspectEndpoint}$}\label{line:idp-introspection-endpoint} \Comment{Introspection Endpoint.}
    \If{$\mi{method} \not\equiv \mGet$}
      \Stop{\DefStop}
    \EndIf
    \Let{$\mi{atoken}$}{$\mi{parameters}[\str{token}]$}
    \LetST{$\mi{clientid}$, $\mi{userid}$}{$\an{\mi{atoken}, \mi{clientid}, \mi{userid}} \inPairing s'.\str{atokens}$}{\textbf{stop}\ \DefStop}
    \Let{$\mi{secret}$}{$\mathsf{resourceOf}(i, \mi{clientid}, \mi{userid})$}
    \Let{$\mi{body}'$}{$\an{\an{\str{protected\_resource}, \mi{secret}}, \an{\str{client\_id}, \mi{clientid}}, \an{\str{user}, \mi{userid}}}$}
    \Let{$m'$}{$\encs{\an{\cHttpResp, n, 200, \an{}, \mi{body}'}}{k}$} \label{line:idp-send-secret-resource}
    \Stop{\StopWithMPrime}
  \EndIf
  \Stop{\DefStop}
\end{algorithmic} \setlength{\parindent}{1em}

\captionof{algorithm}{\label{alg:oauth-script-idp-form} Relation of $\mi{script\_idp\_form}$ }
\begin{algorithmic}[1]
\Statex[-1] \textbf{Input:} $\langle\mi{tree}$, $\mi{docnonce}$, $\mi{scriptstate}$, $\mi{scriptinputs}$, $\mi{cookies}$, $\mi{localStorage}$, $\mi{sessionStorage}$, $\mi{ids}$, $\mi{secrets}\rangle$
\Let{$\mi{url}$}{$\mathsf{GETURL}(\mi{tree},\mi{docnonce})$} 
\LetND{$\mi{url}.\str{path}$}{$\mathbb{S}$}
\Let{$\mi{formdata}$}{$\mi{scriptstate}$}
\LetND{$\mi{id}$}{$\mi{ids}$} \label{line:script-idp-form-select-id}
\LetND{$\mi{secret}$}{$\mi{secrets}$}
\Append{$\an{\str{username}, \mi{id}}$}{$\mi{formdata}$}
\Append{$\an{\str{password}, \mi{secret}}$}{$\mi{formdata}$} 
\Let{$\mi{command}$}{$\an{\tForm, \mi{url}, \mi{\mPost}, \mi{formdata}, \bot}$}
\State \textbf{stop} $\an{s,\mi{cookies},\mi{localStorage},\mi{sessionStorage},\mi{command}}$

\end{algorithmic} \setlength{\parindent}{1em}

%% file: appendix-oauth-model-rp-code.tex
\captionof{algorithm}{\label{alg:rp-oauth} Relation of a Relying
  Party $R^r$}
\begin{algorithmic}[1]
\Statex[-1] \textbf{Input:} $\an{a,f,m},s$
  \If{$s'.\str{corrupt} \not\equiv \bot \vee m \equiv \corrupt$}
    \Let{$s'.\str{corrupt}$}{$\an{\an{a, f, m}, s'.\str{corrupt}}$}
    \LetND{$m'$}{$d_{V}(s')$}
    \LetND{$a'$}{$\addresses$}
    \Stop{$\an{\an{a',a,m'}}$, $s'$}
  \EndIf
  \If{$\exists\, \an{\mi{reference}, \mi{request}, \mi{key}, f}$
      $\inPairing \comp{s'}{pendingRequests}$ \breakalgohook{0}
      \textbf{such that} $\proj{1}{\decs{m}{\mi{key}}} \equiv \cHttpResp$ } \label{line:rp-https-response}
    \Comment{Encrypted HTTP response}

    \Let{$m'$}{$\decs{m}{\mi{key}}$}
    \If{$\comp{m'}{nonce} \not\equiv \comp{\mi{request}}{nonce}$}
      \Stop{\DefStop}
    \EndIf
    \State \textbf{remove} $\an{\mi{reference}, \mi{request}, \mi{key}, f}$ \textbf{from} $\comp{s'}{pendingRequests}$\label{line:rp-remove-pending-request}

    \Let{$\mi{mode}$}{$\proj{1}{\mi{reference}}$}
    \If{$\mi{mode} \equiv \str{code} \vee \mi{mode} \equiv \str{password} \vee \mi{mode} \equiv \str{client\_credentials}$}
      \LetST{$\mi{idp}$, $a'$, $f'$, $n'$, $k'$}{$\an{\mi{mode},\mi{idp},a',f',n',k'} \equiv \mi{reference}$}{\textbf{stop} \DefStop}\label{line:rp-https-response-code-password-ccred}
      \Let{$\mi{token}$}{$m'.\str{body}[\str{access\_token}]$}\label{line:rp-https-response-atoken}
      \Let{$\mi{introspectionEndpoint}$}{$s'.\str{idps}[\mi{idp}].\str{introspectionEndpoint}$}
      \Let{$\mi{parameters}$}{$\mi{introspectionEndpoint}.\str{parameters}$}
      \Append{$\an{\str{token},\mi{token}}$}{$\mi{parameters}$}
      \Let{$\mi{host}$}{$\mi{introspectionEndpoint}.\str{domain}$}
      \Let{$\mi{path}$}{$\mi{introspectionEndpoint}.\str{path}$}
      \Let{$\mi{message}$}{$\hreq{ nonce=\nu_1,
             method=\mGet,
             xhost=\mi{host},
             path=\mi{path},
             parameters=\mi{parameters},
             headers=\an{},
             xbody=\an{}}$}  \label{line:rp-send-something-1}
      \Let{$\comp{s'}{pendingDNS}[\nu_2]$}{$\an{\an{\str{introspect},\mi{mode},\mi{idp},a',f',n',k'}, \mi{message}}$} \label{line:rp-add-introspect-1}
      \Stop{$\an{\an{s'.\str{DNSaddress}, a, \an{\cDNSresolve, \mi{introspectionEndpoint}.\str{domain}, \nu_2}}}$, $s'$}

    \ElsIf{$\mi{mode} \equiv \str{introspect}$} \label{line:rp-introspect-result}
      \LetST{$\mi{resource}$, $\mi{clientId}$, $\mi{user}$}{\breakalgohook{2}$\an{\an{\str{protected\_resource}, \mi{resource}}, \an{\str{client\_id}, \mi{clientId}}, \an{\str{user}, \mi{user}}} \equiv m'.\str{body}$\breakalgohook{2}}{\textbf{stop} \DefStop} \label{line:rp-get-resource}
      \LetST{$\mi{mode}'$, $\mi{idp}$, $a'$, $f'$, $n'$, $k'$}{$\an{\str{introspect},\mi{mode}',\mi{idp},a',f',n',k'} \equiv \mi{reference}$\breakalgohook{2}}{\textbf{stop} \DefStop}\label{line:rp-check-reference-for-introspect}
      \If{$\mi{mode}' \equiv \str{client\_credentials}$} \label{line:rp-check-mode-cc}
        \Stop{\DefStop} \Comment{In client credential grant mode, no service token is issued.}
      \EndIf
      \LetND{$\mi{goal}$}{$\{\str{authz},\str{authn}\}$} \Comment{Proceed with authorization or authentication.}
      \If{$\mi{goal} \equiv \str{authz}$}
        \Let{$\mi{headers}$}{$\an{}$}
      \Else
        \If{$\mi{clientId} \equiv s'.\str{idps}[\mi{idp}].\str{clientId} \vee (\mi{clientId} \equiv \an{} \wedge \mi{mode} \equiv \str{password}\, \wedge$ \breakalgohook{3}$s'.\str{idps}[\mi{idp}].\str{clientPassword} \equiv \bot)$}
          \If{$\mi{user} \equiv \an{}$} \label{line:rp-check-user-empty}
           \Stop{\DefStop}
          \EndIf
        \Else
          \Stop{\DefStop}
        \EndIf
        \Let{$\mi{serviceToken}$}{$\nu_3$} \label{line:rp-create-service-token}
        \Let{$s'.\str{serviceTokens}[\mi{serviceToken}]$}{$\an{\mi{user},\mi{idp}}$} \label{line:rp-store-service-token}
        \Let{$\mi{headers}$}{$\an{\an{\str{Set{\mhyphen}Cookie},\an{\an{\str{serviceToken},\an{\mi{serviceToken},\bot,\bot,\True}}}}}$} \label{line:rp-set-cookie}
      \EndIf
      \Append{$\an{\str{ReferrerPolicy}, \str{origin}}$}{$\mi{headers}$}
      \Let{$m'$}{$\encs{\an{\cHttpResp, n', 200, \mi{headers}, \an{\str{script\_rp\_index}, \an{}}}}{k'}$} 
      \Stop{$\an{\an{f',a',m'}},s'$} \label{line:rp-send-auth-response}

    \EndIf
    \Stop{\DefStop}

  \ElsIf{$m \in \dnsresponses$} \Comment{Successful DNS response}
      \If{$\comp{m}{nonce} \not\in \comp{s}{pendingDNS} \vee \comp{m}{result} \not\in \addresses \vee \comp{m}{domain} \not\equiv s.\str{pendingDNS}[m.\str{domain}].2.\str{host}$}
      \Stop{\DefStop}
      \EndIf
      \Let{$\an{\mi{reference}, \mi{request}}$}{$\comp{s}{pendingDNS}[\comp{m}{nonce}]$}
      \AppendBreak{2}{$\langle\mi{reference}$, $\mi{request}$, $\nu_4$, $\comp{m}{result}\rangle$}{$\comp{s'}{pendingRequests}$} \label{line:move-reference-to-pending-request} 
      \Let{$\mi{message}$}{$\enc{\an{\mi{request},\nu_4}}{\comp{s'}{keyMapping}\left[\comp{\mi{request}}{host}\right]}$} \label{line:select-enc-key-oauth}
      \Let{$\comp{s'}{pendingDNS}$}{$\comp{s'}{pendingDNS} - \comp{m}{nonce}$}\label{line:rp-remove-pendingdns}
      \Stop{$\an{\an{\comp{m}{result}, a, \mi{message}}}$, $s'$} \label{line:rp-send-https-request}
  \ElsIf{$m \equiv \str{TRIGGER}$} \Comment{Start Client Credentials Grant}
    \LetND{$\mi{idpEntry}$}{$\comp{s'}{idps}$} \label{line:rp-send-client-password-1}
    \Let{$\mi{idp}$}{$\proj{1}{\mi{idpEntry}}$}
    \Let{$\mi{tokenEndpoint}$}{$s'.\str{idps}[\mi{idp}].\str{tokenEndpoint}$} \Comment{$\mi{tokenEndpoint}$ is a URL}
    \Let{$\mi{host}$}{$\mi{tokenEndpoint}.\str{domain}$}
    \Let{$\mi{path}$}{$\mi{tokenEndpoint}.\str{path}$}
    \Let{$\mi{parameters}$}{$\mi{tokenEndpoint}.\str{parameters}$}
 \Let{$\mi{headers}$}{$\an{\an{\str{Authorization},\an{s'.\str{idps}[\mi{idp}].\str{clientId},s'.\str{idps}[\mi{idp}].\str{clientPassword}}}}$}
    \Let{$\mi{message}$}{\breakalgohook{1}$\hreq{ nonce=\nu_5,
        method=\mPost,
        xhost=\mi{host},
        path=\mi{path},
        parameters=\mi{parameters},
        headers=\mi{headers},
        xbody=\an{\an{\str{grant\_type}, \str{client\_credentials}}}}$}\label{line:rp-send-something-2}  \Let{$\comp{s'}{pendingDNS}[\nu_6]$}{$\an{\an{\str{client\_credentials},\mi{idp},\bot,\bot,\bot,\bot}, \mi{message}}$} \label{line:rp-pendingdns-trigger}
    \Stop{$\an{\an{\comp{s'}{DNSaddress}, a, \an{\cDNSresolve, \comp{\comp{\mi{idp}}{tokenEndpoint}}{domain}, \nu_6}}}$, $s'$}

   \Else \Comment{Handle HTTP requests}
     \LetST{$m_{\text{dec}}$, $k$, $k'$, $\mi{inDomain}$}{\breakalgohook{0}$\an{m_{\text{dec}}, k} \equiv \dec{m}{k'} \wedge \an{inDomain,k'} \in s.\str{sslkeys}$\breakalgohook{0}}{\textbf{stop} \DefStop}
     \LetST{$n$, $\mi{method}$, $\mi{path}$, $\mi{parameters}$, $\mi{headers}$, $\mi{body}$}{\breakalgohook{0}$\an{\cHttpReq, n, \mi{method}, \mi{inDomain}, \mi{path}, \mi{parameters}, \mi{headers}, \mi{body}} \equiv m_{\text{dec}}$\breakalgohook{0}}{\textbf{stop} \DefStop}
     \If{$\mi{path} \equiv \str{/}$} \label{line:rp-serve-index} \Comment{Serve index page.}
      \Let{$\mi{headers}$}{$\an{\an{\str{ReferrerPolicy}, \str{origin}}}$}
       \Let{$m'$}{$\encs{\an{\cHttpResp, n, 200, \mi{headers}, \an{\str{script\_rp\_index}, \an{}}}}{k}$}
       \Stop{$\an{\an{f,a,m'}}$, $s'$}
     \ElsIf{$\mi{path} \equiv \str{/startInteractiveLogin} \wedge \mi{method} \equiv \mPost$} \Comment{Serve start interactive login request.}
       \If{$\mi{headers}[\str{Origin}]  \not\equiv \an{\mi{inDomain}, \https}$}  \Comment{CSRF protection.}
         \Stop{\DefStop}
       \EndIf
       \Let{$\mi{idp}$}{$\mi{body}$}
       \If{$\mi{idp} \not\in \comp{s'}{idps}$}
         \Stop{\DefStop}
       \EndIf
       \Let{$\mi{state}$}{$\nu_7$}
       \LetND{$\mi{mode}$}{$\{\str{code},\str{token}\}$}

       \LetND{$\mi{responseStatus}$}{$\{303,307\}$}
       \Let{$\mi{authEndpoint}$}{$s'.\str{idps}[\mi{idp}].\str{authorizationEndpoint}$}\Comment{$\mi{authEndpoint}$ is a URL}
       \Append{$\an{\str{response\_type},\mi{mode}}$}{$\mi{authEndpoint}.\str{parameters}$}
       \Append{\breakalgohook{2} $\an{\str{client\_id},s'.\str{idps}[\mi{idp}].\str{clientId}}$}{$\mi{authEndpoint}.\str{parameters}$}
       \Append{$\an{\str{state},\mi{state}}$}{$\mi{authEndpoint}.\str{parameters}$}
       \LetND{$\mi{redirectUri}$}{$\{\bot,\top\}$}
       \If{$\mi{redirectUri} \equiv \top$}
         \LetND{$\mi{sslkey'}$}{$s'.\str{sslkeys}$} \Comment{Choose one of RP's domains non-deterministically}
         \Let{$\mi{host'}$}{$\proj{1}{\mi{sslkey'}}$}
         \Let{$\mi{redirectUri}$}{$\an{\tUrl, \https, \mi{host'}, \str{/redirectionEndpoint},
    \an{\an{\str{idp}, \mi{idp}}}, \an{}}$}
       \EndIf
       \Let{$\mi{loginSessionId}$}{$\nu_8$}
       \Append{$\an{\mi{loginSessionId},\an{\mi{idp},\mi{state},\mi{mode},\mi{redirectUri}}}$}{$s'.\str{loginSessions}$} \label{line:rp-create-session}
       \Let{$\mi{headers}$}{$\an{\an{\str{Location},\mi{authEndpoint}}}$}
       \Append{$\an{\cSetCookie,\an{\an{\str{loginSessionId},\an{\mi{loginSessionId},\top,\top,\top}}}}$}{$\mi{headers}$}
       \Append{$\an{\str{ReferrerPolicy}, \str{origin}}$}{$\mi{headers}$}
       \Let{$m'$}{$\mathsf{enc_s}(\langle\cHttpResp, n, \mi{responseStatus}, \mi{headers}, \bot\rangle, k)$} \label{line:rp-send-redirect}
       \Stop{$\an{\an{f,a,m'}}$, $s'$}

     \ElsIf{$\mi{path} \equiv \str{/redirectionEndpoint}$} \label{line:rp-redir-endpoint}
       \Let{$\mi{loginSessionId}$}{$\mi{headers}[\str{Cookie}][\str{loginSessionId}]$}
       \LetST{$\mi{idp}$, $\mi{state}$, $\mi{mode}$, $\mi{redirectUri}$}{$\an{\mi{idp},\mi{state},\mi{mode},\mi{redirectUri}} \equiv$ \breakalgohook{2} $s'.\str{loginSessions}[\mi{loginSessionId}]$}{\textbf{stop} \DefStop}
       \Let{$\mi{clientId}$}{$s'.\str{idps}[\mi{idp}].\str{clientId}$}
       \If{$\mi{idp} \not\equiv \mi{parameters}[\str{iss}] \vee \mi{clientId} \not\equiv \mi{parameters}[\str{client\_id}]$} \label{line:rp-check-idp-param}
         \Stop{\DefStop}
       \EndIf
       \If{$\mi{mode} \equiv \str{code}$}
         \Comment{Continue Authorization Code Grant}
         \If{$\mi{parameters}[\str{state}] \not\equiv \mi{state}$} \label{line:rp-check-state}
           \Stop{\DefStop}
         \EndIf
         \Let{$\mi{code}$}{$\mi{parameters}[\str{code}]$}

         \Let{$\mi{tokenRequestHeaders}$}{$\an{}$}
         \Let{$\mi{tokenRequestBody}$}{$\an{\an{\str{grant\_type}, \str{authorization\_code}},\an{\str{code},\mi{code}}}$} \label{line:rp-prepare-atoken-from-code-req}
         \If{$\mi{redirectUri} \not\equiv \bot$}
           \Append{$\an{\str{redirect\_uri},\mi{redirectUri}}$}{$\mi{tokenRequestBody}$}
         \EndIf

         \Let{$\mi{clientPassword}$}{$s'.\str{idps}[\mi{idp}].\str{clientPassword}$} \label{line:rp-send-client-password-2}
         \If{$\mi{clientPassword} \equiv \bot$}
           \Append{$\an{\str{client\_id},\mi{clientId}}$}{$\mi{tokenRequestBody}$}
         \Else
           \Append{\breakalgohook{4} $\an{\str{Authorization},\an{\mi{clientId},\mi{clientPassword}}}$}{$\mi{tokenRequestHeaders}$}
         \EndIf
         
         \Let{$\mi{tokenEndpoint}$}{$s'.\str{idps}[\mi{idp}].\str{tokenEndpoint}$}
         \Let{$\mi{message}$}{$\hreq{ nonce=\nu_9,
             method=\mPost,
             xhost=\mi{tokenEndpoint}.\str{domain},
             path=\mi{tokenEndpoint}.\str{path},
             parameters=\mi{tokenEndpoint}.\str{parameters},
             headers=\mi{tokenRequestHeaders},
             xbody=\mi{tokenRequestBody}}$}  \label{line:rp-send-something-3}
         \Let{$\comp{s'}{pendingDNS}[\nu_{10}]$}{$\an{\an{\str{code},\mi{idp},a,f,n,k}, \mi{message}}$} \label{line:rp-start-retrieve-code}

         \Stop{$\an{\an{\comp{s'}{DNSaddress}, a, \an{\cDNSresolve, \mi{tokenEndpoint}.\str{domain}, \nu_{10}}}}$, $s'$}

       \ElsIf{$\mi{mode} \equiv \str{token}$}
         \Comment{Continue Implicit Grant}
         \Let{$\mi{headers}$}{$\an{\an{\str{ReferrerPolicy}, \str{origin}}}$}
         \Let{$m'$}{$\encs{\an{\cHttpResp, n, 200, \mi{headers}, \an{\str{script\_rp\_implicit}, \mi{idp}}}}{k}$} \label{line:rp-send-script-implicit}
         \Stop{$\an{\an{f,a,m'}}$, $s'$}
       \EndIf
       \Stop{\DefStop}

     \ElsIf{$\mi{path} \equiv \str{/passwordLogin} \wedge \mi{method} \equiv \mPost$}\label{line:rp-password-login} 
       \If{$\mi{headers}[\str{Origin}]  \not\equiv \an{\mi{inDomain}, \https}$}  \Comment{CSRF protection.}
         \Stop{\DefStop}
       \EndIf
       \LetST{$\mi{idp}$, $\mi{username}$, $\mi{password}$}{$\an{\an{\mi{username},\mi{idp}},\mi{password}} \equiv \mi{body}$}{\breakalgohook{2} \textbf{stop} \DefStop}

       \Let{$\mi{tokenRequestHeaders}$}{$\an{}$}
       \Let{$\mi{tokenRequestBody}$}{$\langle \an{\str{grant\_type}, \str{password}},\an{\str{username},\an{\mi{username},\mi{idp}}},$ \breakalgohook{2} $\an{\str{password},\mi{password}}\rangle$}\label{line:rp-prepare-atoken-from-password-req}

       \Let{$\mi{clientId}$}{$s'.\str{idps}[\mi{idp}].\str{clientId}$}\label{line:rp-send-client-password-3}
       \Let{$\mi{clientPassword}$}{$s'.\str{idps}[\mi{idp}].\str{clientPassword}$}
       \If{$\mi{clientPassword} \not\equiv \bot$}
     \Append{\breakalgohook{3} $\an{\str{Authorization},\an{\mi{clientId},\mi{clientPassword}}}$}{$\mi{tokenRequestHeaders}$}
       \EndIf

       \Let{$\mi{tokenEndpoint}$}{$s'.\str{idps}[\mi{idp}].\str{tokenEndpoint}$}
       \Let{$\mi{message}$}{$\hreq{ nonce=\nu_{11},
           method=\mPost,
           xhost=\mi{tokenEndpoint}.\str{domain},
           path=\mi{tokenEndpoint}.\str{path},
           parameters=\mi{tokenEndpoint}.\str{parameters},
           headers=\mi{tokenRequestHeaders},
           xbody=\mi{tokenRequestBody}}$}\label{line:rp-send-something-4}
       \Let{$\comp{s'}{pendingDNS}[\nu_{12}]$}{$\an{\an{\str{password},\mi{idp},a,f,n,k}, \mi{message}}$}\label{line:rp-pendingdns-password}\label{line:rp-start-retrieve-code-from-password}

       \Stop{$\an{\an{\comp{s'}{DNSaddress}, a, \an{\cDNSresolve, \mi{tokenEndpoint}.\str{domain}, \nu_{12}}}}$, $s'$}

     \ElsIf{$\mi{path} \equiv \str{/receiveTokenFromImplicitGrant} \wedge \mi{method} \equiv \mPost$}\label{line:rp-receive-implicit-token}
       \If{$\mi{headers}[\str{Origin}]  \not\equiv \an{\mi{inDomain}, \https}$}  \Comment{CSRF protection.}
         \Stop{\DefStop}
       \EndIf
       \Let{$\mi{loginSessionId}$}{$\mi{headers}[\str{Cookie}][\str{loginSessionId}]$}
       \LetST{$\mi{idp}$, $\mi{state}$, $\mi{mode}$, $\mi{redirectUri}$}{$\an{\mi{idp},\mi{state},\mi{mode},\mi{redirectUri}} \equiv$ \breakalgohook{2} $s'.\str{loginSessions}[\mi{loginSessionId}]$}{\textbf{stop} \DefStop} \label{line:rp-take-value-idp}
       \LetST{$\mi{token}$}{$\an{\mi{token},\mi{state},\mi{idp}} \equiv \mi{body}$}{\textbf{stop} \DefStop} 
       \Let{$\mi{introspectionEndpoint}$}{$s'.\str{idps}[\mi{idp}].\str{introspectionEndpoint}$}\Comment{$\mi{introspectionEndpoint}$ is a URL} \label{line:rp-prepare-introspect-2}
       \Let{$\mi{parameters'}$}{$\mi{introspectionEndpoint}.\str{parameters}$}
       \Append{$\an{\str{token},\mi{token}}$}{$\mi{parameters'}$}
       \Let{$\mi{host}$}{$\mi{introspectionEndpoint}.\str{domain}$}
       \Let{$\mi{path'}$}{$\mi{introspectionEndpoint}.\str{path}$}
       \Let{$\mi{message}$}{$\hreq{ nonce=\nu_{13},
           method=\mGet,
           xhost=\mi{host},
           path=\mi{path'},
           parameters=\mi{parameters'},
           headers=\an{},
           xbody=\an{}}$} \label{line:rp-send-something-5}
       \Let{$\comp{s'}{pendingDNS}[\nu_{14}]$}{$\an{\an{\str{introspect},\str{implicit},\mi{idp},a,f,n,k}, \mi{message}}$} \label{line:rp-add-introspect-2}
       \Stop{$\an{\an{\comp{s'}{DNSaddress}, a, \an{\cDNSresolve, \mi{introspectionEndpoint}.\str{domain}, \nu_{14}}}}$, $s'$}
     \EndIf
  \EndIf
  \Stop{\DefStop}
  
\end{algorithmic} \setlength{\parindent}{1em}

%% file: appendix-oauth-model-script-rp-index-code.tex
\captionof{algorithm}{\label{alg:script-rp-index}Relation of $\mi{script\_rp\_index}$ }
\begin{algorithmic}[1]
\Statex[-1] \textbf{Input:} $\langle\mi{tree}$, $\mi{docnonce}$, $\mi{scriptstate}$, $\mi{scriptinputs}$, $\mi{cookies}$, $\mi{localStorage}$, $\mi{sessionStorage}$, $\mi{ids}$, $\mi{secrets}\rangle$

\LetND{$\mi{switch}$}{$\{\str{auth},\str{link}\}$}

\If{$\mi{switch} \equiv \str{auth}$}

\Let{$\mi{url}$}{$\mathsf{GETURL}(\mi{tree},\mi{docnonce})$}
\LetND{$\mi{id}$}{$\mi{ids}$}\label{line:script-rp-index-select-id}

\Let{$\mi{username}$}{$\proj{1}{\mi{id}}$}
\Let{$\mi{domain}$}{$\proj{2}{\mi{id}}$} \label{line:script-rp-index-selected-domain}
\LetND{$\mi{interactive}$}{$\{\bot,\top\}$} \label{line:script-rp-index-select-interactive}

\If{$\mi{interactive} \equiv \top$}
  \Let{$\mi{url'}$}{$\an{\tUrl, \https, \mi{url}.\str{host}, \str{/startInteractiveLogin},
    \an{}, \an{}}$}
  \Let{$\mi{command}$}{$\an{\tForm, \mi{url'}, \mi{\mPost}, \mi{domain}, \bot}$}
\Else
  \Let{$\mi{url'}$}{$\an{\tUrl, \https, \mi{url}.\str{host}, \str{/passwordLogin},
    \an{}, \an{}}$}
  \LetST{$\mi{secret}$}{$\mi{secret} = \mathsf{secretOfID}(\mi{id}) \wedge \mi{secret} \in \mi{secrets}$}{\breakalgohook{1}\textbf{stop} $\an{s,\mi{cookies},\mi{localStorage},\mi{sessionStorage},\an{}}$}\label{line:script-rp-index-select-secret}
  \Let{$\mi{command}$}{$\an{\tForm, \mi{url'}, \mi{\mPost}, \an{\mi{id},\mi{secret}}, \bot}$}\label{line:script-rp-index-password-login-form}
\EndIf

\State \textbf{stop} $\an{s,\mi{cookies},\mi{localStorage},\mi{sessionStorage},\mi{command}}$ \label{line:script-rp-index-start-oauth-session}

\Else

  \LetND{$\mi{protocol}$}{$\{\http, \https\}$} 
  \LetND{$\mi{host}$}{$\dns$}
  \LetND{$\mi{path}$}{$\mathbb{S}$}
  \LetND{$\mi{fragment}$}{$\mathbb{S}$}
  \LetND{$\mi{parameters}$}{$\dict{\mathbb{S}}{\mathbb{S}}$} 
  \Let{$\mi{url}$}{$\an{\cUrl, \mi{protocol}, \mi{host}, \mi{path}, \mi{parameters}, \mi{fragment}}$}

  \State \textbf{stop} $\an{\tHref, \mi{url}, \mathsf{GETDOCWINDOW}(\mi{tree},\mi{docnonce})), \bot}$

\EndIf

\end{algorithmic} \setlength{\parindent}{1em}

%% file: appendix-oauth-model-script-rp-implicit-code.tex
\captionof{algorithm}{Relation of $\mi{script\_rp\_implicit}$ }
\begin{algorithmic}[1]
\Statex[-1] \textbf{Input:} $\langle\mi{tree}$, $\mi{docnonce}$, $\mi{scriptstate}$, $\mi{scriptinputs}$, $\mi{cookies}$, $\mi{localStorage}$, $\mi{sessionStorage}$, $\mi{ids}$, $\mi{secrets}\rangle$
\Let{$\mi{url}$}{$\mathsf{GETURL}(\mi{tree},\mi{docnonce})$}
\Let{$\mi{url'}$}{$\an{\tUrl, \https, \mi{url}.\str{host}, \str{/receiveTokenFromImplicitGrant},
    \an{}, \an{}}$}
\Let{$\mi{body}$}{$\an{\mi{url}.\str{fragment}[\str{access\_token}],\mi{url}.\str{fragment}[\str{state}], \mi{scriptstate}}$}
\Let{$\mi{command}$}{$\an{\tForm, \mi{url'}, \mi{\mPost}, \mi{body}, \bot}$}

\State \textbf{stop} $\an{s,\mi{cookies},\mi{localStorage},\mi{sessionStorage},\mi{command}}$

\end{algorithmic} \setlength{\parindent}{1em}

%% file: appendix-oauth-model-webattacker.tex
\section{Formal Model of OAuth with Web Attackers}
\label{app:model-oauth-auth-webattackers}

We now derive $\oauthwebsystem^w$
(an \emph{OAuth web system with web attackers}) from
$\oauthwebsystem^n$
by replacing the network attacker with a finite set of web attackers.

\begin{definition}
  An \emph{OAuth web system with web attackers}, $\oauthwebsystem^w$,
  is an OAuth web system
  $\oauthwebsystem^n=(\bidsystem, \scriptset, \mathsf{script}, E^0)$
  with the following changes:
  \begin{itemize}
  \item We have $\bidsystem = \mathsf{Hon} \cup \mathsf{Web}$,
    in particular, there is no network attacker. The set
    $\mathsf{Web}$
    contains a finite number of web attacker processes. The set
    $\mathsf{Hon}$
    is as described above, and additionally contains a DNS server $d$
    as defined below.
  \item The set of IP addresses $\mathsf{IPs}$ contains no IP addresses for the network
    attacker, but instead a finite set of IP addresses for each web
    attacker.
  \item The set of Domains $\mathsf{Doms}$
    contains no domains for the network attacker, but instead a finite
    set of domains for each web attacker.

  \item All honest parties use the DNS server $d$ as their DNS server.
  \end{itemize}
\end{definition}

\subsection{DNS Server}
\label{sec:dns-server-1}

The DNS server $d$
is a DNS server as defined in Definition~\ref{def:dns-server}. Its
initial state $s_0^d$
contains only pairings $\an{D, i}$
such that $i \in \mathsf{addr}(\mathsf{dom}^{-1}(D))$,
i.e., any domain is resolved to an IP address belonging to the owner
of that domain (as defined in Appendix~\ref{app:addresses-and-domain-names}).

\subsection{Web Attackers}
\label{sec:web-attackers}

Web attackers, as opposed to network attackers, can only use their own
IP addresses for listening to and sending messages. Therefore, for any
web attacker process $w$
we have that $I^w = \mathsf{addr}(w)$.
The inital states of web attackers are defined parallel to those of
network attackers, i.e., the initial state for a web attacker process
$w$
is $s_0^w = \an{\mi{attdoms}^w, \mi{sslkeys}, \mi{signkeys}}$,
where $\mi{attdoms}^w$
is a sequence of all domains along with the corresponding private keys
owned by the attacker $w$,
$\mi{sslkeys}$
is a sequence of all domains and the corresponding public keys, and
$\mi{signkeys}$
is a sequence containing all public signing keys for all IdPs.

%% file: appendix-oauth-secproperties.tex
\section{Formal Security Properties}\label{app:form-secur-prop}

The security properties for OAuth are formally defined as follows.

\subsection{Authorization}
\label{sec:fprop-authorization}

Intuitively, authorization for $\oauthwebsystem^n$
means that an attacker should not be able to obtain or use a protected
resource available to some honest RP at an IdP for some user unless
certain parties involved in the authorization process are corrupted.

\begin{definition}[Authorization Property]\label{def:property-authz-a} Let $\oauthwebsystem^n$ be an OAuth web system with a network attacker. We say that \emph{$\oauthwebsystem^n$
    is secure w.r.t.~authorization} iff for every run $\rho$
  of $\oauthwebsystem^n$,
  every state $(S^j, E^j, N^j)$
  in $\rho$,
  every IdP $i \in \fAP{IDP}$,
  every $r \in \fAP{RP} \cup \{\bot\}$
  with $r$ being honest in $S^j$ unless $r = \bot$,
  every $u \in \IDs \cup \{\bot\}$,
  for $n = \mathsf{resourceOf}(i,r,u)$,
  $n$
  is derivable from the attackers knowledge in $S^j$
  (i.e., $n \in d_{\emptyset}(S^j(\fAP{attacker}))$), it follows that
  \begin{enumerate}
  \item $i$ is corrupted in $S^j$, or
  \item $u \neq \bot$
    and (i) the browser $b$
    owning $u$
    is fully corrupted in $S^j$
    or (ii) some $r' \in \mathsf{trustedRPs}(\mathsf{secretOfID}(u))$
    is corrupted in $S^j$.
  \end{enumerate}
\end{definition}

Note that the protected resource $n$
being available to the attacker also models that the attacker can use
a service of the IdP $i$
under the name of the user $u$
(e.g., the attacker can post to the Facebook wall of the victim).

\subsection{Authentication}
\label{sec:fprop-authentication}

Intuitively, authentication for $\oauthwebsystem^n$
means that an attacker should not be able to login at an (honest) RP
under the identity of a user unless certain parties involved in the
login process are corrupted. As explained above, being logged in at an
RP under some user identity means to have obtained a service token for
this identity from the RP.

\begin{definition}[Authentication Property]\label{def:property-authn-a} Let $\oauthwebsystem^n$ be an OAuth web
  system with a network attacker. We say that \emph{$\oauthwebsystem^n$
    is secure w.r.t.~authentication} iff for every run $\rho$
  of $\oauthwebsystem^n$,
  every state $(S^j, E^j, N^j)$
  in $\rho$,
  every $r\in \fAP{RP}$
  that is honest in $S^j$,
  every $i \in \fAP{IDP}$,
  every $g \in \mathsf{dom}(i)$,
  every $u \in \mathbb{S}$,
  every RP service token of the form $\an{n,\an{u,g}}$
  recorded in $S^j(r).\str{serviceTokens}$,
  and $n$
  being derivable from the attackers knowledge in $S^j$
  (i.e., $n \in d_{\emptyset}(S^j(\fAP{attacker}))$),
  then the browser $b$
  owning $u$
  is fully corrupted in $S^j$
  (i.e., the value of $\mi{isCorrupted}$
  is $\fullcorrupt$),
  some $r' \in \mathsf{trustedRPs}(\mathsf{secretOfID}(\an{u,g}))$
  is corrupted in $S^j$, or $i$ is corrupted in $S^j$.
\end{definition}

\subsection{Session Integrity for Authorization and Authentication}
\label{sec:fprop-session-integrity}

\df{Todo: Einfuehrung hier ueberarbeiten.}

Before we can define the session integrity property for authorization
and authentication, we need to define the notion of \emph{Sessions}
and, in particular, \emph{OAuth Sessions}. These capture series of
processing steps related to a single OAuth flow. Note that sessions
here are not the same as sessions in the web which are usually
identified by some session identifier in a cookie.

\subsubsection{Notations}
In the following, given a finite run $\rho = ((S^0, E^0, N^0),\dots,$
$(S^n, E^n, N^n))$
or an infinite run $\rho = ((S^0, E^0, N^0),\dots)$,
we denote by $Q_i$
the processing step
$(S^i, E^i, N^i) \xrightarrow{} (S^{i+1}, E^{i+1}, N^{i+1})$
(with $i \geq 0$ and, for finite runs, $i<n$). \df{are the processing steps uniquely defined?}

\begin{definition}[Emitting Events]\label{def:emitting}
  Given an atomic process $p$,
  an event $e$,
  and a finite run $\rho = ((S^0, E^0, N^0),\dots,$
$(S^n, E^n, N^n))$
or an infinite run $\rho = ((S^0, E^0, N^0),\dots)$
  we say that \emph{$p$
    emits $e$} iff there is a processing step in $\rho$ of the form
  \[ (S^i, E^i, N^i) \xrightarrow[p \rightarrow E]{} (S^{i+1},
    E^{i+1}, N^{i+1})\] for some $i \geq 0$ and 
  a set of events $E$
  with $e \in E$.
  We also say that \emph{$p$
    emits $m$} iff $e = \an{x,y,m}$ for some addresses $x$, $y$.
\end{definition}

\subsubsection{Sessions and OAuth Sessions}
We now define a relation between processing steps. Intuitively, we say
that two processing steps are connected if one processing step causes
the other. This can happen either directly (i.e., one DY process
handles an event output by another process) or indirectly (e.g., a
script that was loaded from an earlier message runs in a browser and
outputs a new message).

\begin{definition}[Connected Processing Steps]\label{def:connected-steps}
  We say that two processing steps
  \begin{eqnarray*}
    Q_x &=& (S^x, E^x, N^x) \xrightarrow[p_x \rightarrow
            E_{\text{out},x}]{e_{\text{in},x} \rightarrow p_x} (S^{x+1}, E^{x+1},
            N^{x+1})\text{ and }\\
    Q_y &=& (S^y, E^y, N^y) \xrightarrow[p_y \rightarrow
            E_{\text{out},y}]{e_{\text{in},y} \rightarrow p_y} (S^{y+1}, E^{y+1},
            N^{y+1})
  \end{eqnarray*}
  are \emph{connected} iff (1) $e_{\text{in},y} \in E_{\text{out},x}$,
  or (2) $p_y$
  is a browser, $e_{\text{in},y}$
  is a trigger event, the browser $p_y$
  selects to run a script (i.e., selects $\str{script}$
  in Line~\ref{line:browser-switch} of
  Algorithm~\ref{alg:browsermain}), and the document selected in
  Line~\ref{line:browser-trigger-document} was created as the result
  of an HTTP(S) message in $E_{\text{out},x}$.
\end{definition}

Based on the notion of connected processing steps, we now define
sessions to be sequences of connected processing steps.
\begin{definition}[Sessions]\label{def:sessions}
  A \emph{Session (in a run $\rho$
    of a web system)} is a sequence of processing steps
  $(Q_0, \ldots, Q_n)$
  or $(Q_0, Q_1, \ldots)$
  such that (1) for all $Q_i$
  with $i > 0$,
  $Q_i$
  is connected to some processing step in $(Q_0, \ldots, Q_{i-1})$,
  and (2) all processing steps appear in the same order as in $\rho$.
\end{definition}

We can now define OAuth Sessions. Intuitively, an OAuth session starts
when a user expresses her wish to use some identity at some RP. Each
session can only contain one such request. A session ends when a
authorization or log in is complete (which does not necessarily happen
in all OAuth Sessions).

\begin{definition}[Start and End Processing Steps for OAuth]\label{def:oauth-start-end}
  We write $\mathsf{startsOA}(Q, b, r, i)$
  iff in the processing step $Q$
  the browser $b$
  triggers the script $\mi{script\_rp\_index}$
  which selects some domain of $i$
  (in Line~\ref{line:script-rp-index-selected-domain} of
  Algorithm~\ref{alg:script-rp-index}) and instructs the browser $b$
  to send a message to $r$
  in Line~\ref{line:script-rp-index-start-oauth-session}.

  We write $\mathsf{endsOA}(Q, b, r, i, t)$
  iff the RP $r$
  in the processing step $Q$
  receives an HTTPS response with a body of the form $\an{\an{\str{protected\_resource}, t}, \an{\str{client\_id}, c}, \an{\str{user}, u}}$ for some terms $c$ and $u$
  from $i$
  and emits an event in Line~\ref{line:rp-send-auth-response} of
  Algorithm~\ref{alg:rp-oauth} that is addressed to $b$.
\end{definition}

\begin{definition}[OAuth Sessions]\label{def:oauth-sessions}
  Let $\oauthwebsystem^w$
  be an OAuth web system with web attackers and $\rho$
  be a run of $\oauthwebsystem^w$.
  An \emph{OAuth Session in $\rho$
    by a browser $b$
    with an RP $r$
    and an IdP $i$}
  is a infinite session $(Q_0, Q_1, \ldots)$
  or a finite session $(Q_0, \ldots, Q_n)$
  in $\rho$
  such that $\mathsf{startsOA}(Q_0, b, r, i)$,
  but there is no $j >0$,
  $i'$
  such that $\mathsf{startsOA}(Q_j, b, r, i')$.
  If there are $j>0$,
  $t$
  such that $\mathsf{endsOA}(Q_j, b, r, i, t)$,
  then the OAuth Session is finite and $n=j$.
\end{definition}
We write $\mathsf{OASessions}(\rho,b,r,i)$
for the set of all OAuth Sessions in $\rho$
by $b$ with the RP $r$ and the IdP $i$.

We now introduce a notation to associate an OAuth Session with the
identity that the browser selected during that session. This models
the user intention to log in/authorize using a specific identity. Note
that this expression of intent can take place in two places, either
during the first step of an OAuth Session (in the resource owner
password credentials mode) or at a later time when the user logs in at
the IdP (in the implicit mode and the authorization code mode).

\begin{definition}[Selected Identity in an OAuth Session]
  Given a run $\rho$
  of an an OAuth web system with a web attacker, a browser $b$,
  an RP $r$,
  some IdP $i$,
  and an OAuth Session $o \in \mathsf{OASessions}(\rho,b,r,i)$
  we write $\mathsf{selected}_\text{nia}(o, b, r, \an{u,g})$
  iff $b$
  in (the first processing step of) $o$
  selected $\mi{id} \equiv \an{u,g}$
  in Line~\ref{line:script-rp-index-select-id} of
  Algorithm~\ref{alg:script-rp-index} and selected
  $\mi{interactive} \equiv \bot$
  in Line~\ref{line:script-rp-index-select-interactive}.

  We write $\mathsf{selected}_\text{ia}(o, b, r, \an{u,g})$
  iff $b$
  in (the first processing step of) $o$
  selected $\mi{interactive} \equiv \top$
  in Line~\ref{line:script-rp-index-select-interactive} and there is
  some $Q'$
  in $o$
  such that $b$
  triggers the script $\mi{script\_idp\_form}$
  in $Q'$
  and selects $\an{u,g}$
  in Line~\ref{line:script-idp-form-select-id} of
  Algorithm~\ref{alg:oauth-script-idp-form} and sends a message out to
  $i$.
\end{definition}

\df{Theorem anpassen, auch im Paper!}

\subsubsection{Session Integrity Property for Authorization}
This security property captures that (a) an RP should only be
authorized to access some resources when the user actually expressed
the wish to start an OAuth flow before, and (b) if a user expressed
the wish to start an OAuth flow using some honest identity provider
and a specific identity, then the OAuth flow is never completed with a
different identity.

\begin{definition}[Session Integrity for Authorization]\label{def:property-authz-b}
  Let $\oauthwebsystem^w$
  be an OAuth web system with web attackers. We say that
  \emph{$\oauthwebsystem^w$
    is secure w.r.t.~session integrity for authorization} iff for
  every run $\rho$
  of $\oauthwebsystem^w$,
  every processing step $Q$
  in $\rho$,
  every browser $b$
  that is honest in $Q$,
  every $r\in \fAP{RP}$
  that is honest in $Q$,
  every $i \in \fAP{IDP}$,
  every identity $\an{u,g}$,
  some protected resource $t$,
  the following holds true: If $\mathsf{endsOA}(Q, b, r, i, t)$, then
  \begin{enumerate}[label=(\alph*)]
  \item there is an OAuth Session $o \in \mathsf{OASessions}(\rho, b, r, i)$, and
  \item if $i$
    is honest in $Q$
    then $Q$
    is in $o$
    and we have that 
    $$\mathsf{selected}_\text{ia}(o, b, r,  \an{u,g}) \iff \big(t
    \equiv \mathsf{resourceOf}(i, r, \an{u,g})\big)$$ or
    $$\mathsf{selected}_\text{nia}(o, b, r,  \an{u,g}) \iff \big(t
    \equiv \mathsf{resourceOf}(i, r', \an{u,g})\big)$$ for some $r' \in \{r, \bot\}$.
  \end{enumerate}

\end{definition}

\subsubsection{Session Integrity Property for Authentication}
This security property captures that (a) a user should only be logged
in when the user actually expressed the wish to start an OAuth flow
before, and (b) if a user expressed the wish to start an OAuth flow
using some honest identity provider and a specific identity, then user
is not logged in under a different identity.

\begin{definition} [Session Integrity for
  Authentication]\label{def:property-authn-b}
  Let $\oauthwebsystem^w$
  be an OAuth web system with web attackers. We say that
  \emph{$\oauthwebsystem^w$
    is secure w.r.t.~session integrity for authentication} iff for
  every run $\rho$
  of $\oauthwebsystem^w$,
  every processing step $Q_\text{login}$
  in $\rho$,
  every browser $b$
  that is honest in $Q_\text{login}$,
  every $r\in \fAP{RP}$
  that is honest in $Q_\text{login}$,
  every $i \in \fAP{IDP}$,
  every identity $\an{u,g}$,
  the following holds true: If in
  $Q_\text{login}$
   a service token of the form
  $\an{n, \an{\an{u',g'},m}}$ for a domain $m \in \mathsf{dom}(i)$ and some $n$, $u'$, $g'$ 
  is created in $r$
  (in Line~\ref{line:rp-store-service-token} of
  Algorithm~\ref{alg:rp-oauth}) and $n$ is sent to the browser $b$, then
  \begin{enumerate}[label=(\alph*)]
  \item there is an OAuth Session $o \in \mathsf{OASessions}(\rho, b, r, i)$, and
  \item if $i$
    is honest in $Q_\text{login}$
    then $Q_\text{login}$
    is in $o$
    and we have that 
    $$\big(\mathsf{selected}_\text{ia}(o, b, r, \an{u,g}) \vee \mathsf{selected}_\text{nia}(o, b, r, \an{u,g})\big) \iff \big(\an{u,g} \equiv \an{u',g'}\big)\ .$$
  \end{enumerate}
  
\end{definition}

%% file: appendix-oauth-proof.tex
\section{Proof of Theorem~\ref{thm:security}}
\label{app:proof-oauth}

Before we present the proof for Theorem~\ref{thm:security}, we first
provide a high-level proof outline. We then show some general
properties of OAuth web systems with a network attacker. Afterwards,
we first prove the authentication property and then the authorization
property.

\subsection{Proof Outline}

We first show three basic lemmas that apply to honest RPs and capture
specific technical details: (1) messages transferred over HTTPS
connections that were initiated by honest RPs cannot be read or
altered by other parties. In particular, honest RPs do not leak the
encryption keys to other parties. (2) HTTP(S) messages which await DNS
resolution in a state of an honest RP are later sent out over the
network without being altered in between. (3) Honest RPs never send
messages to other RPs or themselves, and they send only HTTPS messages
that other RPs cannot decrypt.

\paragraph{Authentication} We then prove the authentication property,
by contradiction. To this end, we show in three separate
lemmas building on each other that (1) the attacker does not learn
passwords of the user, (2) the attacker does not learn authorization
codes that could be used to learn a relevant access token, and (3)
that the attacker in fact does not learn an access token that could be
used to retrieve a service token as described in the authentication
property. We finally show that there is no other way for an attacker
to get hold of a service token (as described in the authentication
property), and that therefore, the authentication property holds
true.

\paragraph{Authorization} As above, we assume that the authorization
property does not hold and lead this to a contradication. The proof
then builds upon lemmas shown in the authentication proof. We show
that the attacker would need to know an access token to acquire a
protected resource. If the protected resource is bound to a user
(i.e., it was not issued in the client credentials mode), then (3)
from above applies and shows that the attacker cannot learn such an
access token, and thus cannot learn this protected resource. If the
protected resource was not assigned to a user (i.e., it was issued in
the client credentials mode), then we can show that the attacker would
need to know client secrets to get the protected resource. We show,
however, that it is not possible for the attacker to learn the
necessary client secrets (which are always required in the client
credentials mode). Therefore, whether it is a user-bound protected
resource or not, the attacker cannot learn it, leading our assumption
to a contradiction.

\paragraph{Session Integrity}
We first show session integrity for authorization. To this end, we
show that an OAuth flow (when the browser $b$
and the RP $r$
are honest) can only be completed when it was actively started by the
browser $b$,
i.e., the correct script was run under an origin of $r$
and this script started the login using some identity $v$.
This is achieved by showing the existence of certain events, starting
from the last event (where the flow is completed) and backtracing to
some starting event. We then show that if $i$
is also honest, the start and end events belong to the same flow, and
that the identity $v$
that was selected in this flow is exactly the same identity for which
$r$
accesses a resource in the last event. This is done by showing that
all events (from the event where the identity was selected to the last
event) are connected and that certain values (such as the chosen
identity) are relayed correctly and not modified in between processing
steps or messages. We then show that session integrity for
authentication follows from session integrity for authorization.

\subsection{Properties of $\oauthwebsystem^n$}

Let $\oauthwebsystem^n = (\bidsystem, \scriptset, \mathsf{script}, E^0)$
be an OAuth web system with a network attacker. Let $\rho$ be a run of $\oauthwebsystem^n$. We write
$s_x = (S^x,E^x,N^x)$ for the states in $\rho$.

\begin{definition}\label{def:contains}
  We say that a term $t$ \emph{is derivably contained in (a term) $t'$
    for (a set of DY processes) $P$ (in a processing step $s_i
    \rightarrow s_{i+1}$ of a run $\rho=(s_0,s_1,\ldots)$)} if $t$ is
  derivable from $t'$ with the knowledge available to $P$, i.e.,
  \begin{align*}
    t \in d_{\emptyset}(\{t'\} \cup \bigcup_{p\in P}S^{i+1}(p))
  \end{align*}

\end{definition}

\begin{definition}\label{def:leak}
  We say that \emph{a set of processes $P$ leaks a term $t$ (in a
    processing step $s_i \rightarrow s_{i+1}$) to a set of processes
    $P'$} if there exists a message $m$ that is emitted (in $s_i
  \rightarrow s_{i+1}$) by some $p \in P$ and $t$ is derivably
  contained in $m$ for $P'$ in the processing step $s_i \rightarrow
  s_{i+1}$. If we omit $P'$, we define $P' := \bidsystem \setminus
  P$. If $P$ is a set with a single element, we omit the set notation.
\end{definition}

\begin{definition}\label{def:creating}
  We say that an DY process $p$ \emph{created} a message $m$ (at
  some point) in a run if $m$ is derivably contained in a message
  emitted by $p$ in some processing step and if there is no earlier
  processing step where $m$ is derivably contained in a message
  emitted by some DY process $p'$.
\end{definition}

\begin{definition}\label{def:accepting}
  We say that \emph{a browser $b$ accepted} a message (as a response
  to some request) if the browser decrypted the message (if it was an
  HTTPS message) and called the function $\mathsf{PROCESSRESPONSE}$,
  passing the message and the request (see
  Algorithm~\ref{alg:processresponse}).
\end{definition}

\begin{definition}\label{def:knowing}
  We say that an atomic DY process \emph{$p$ knows a term $t$} in some
  state $s=(S,E,N)$ of a run if it can derive the term from its
  knowledge, i.e., $t \in d_{\emptyset}(S(p))$.
\end{definition}

\begin{definition}\label{def:initiating}
  We say that a \emph{script initiated a request $r$} if a browser
  triggered the script (in Line~\ref{line:trigger-script} of
  Algorithm~\ref{alg:runscript}) and the first component of the
  $\mi{command}$ output of the script relation is either $\tHref$, 
  $\tIframe$, $\tForm$, or $\tXMLHTTPRequest$ such that the browser
  issues the request $r$ in the same step as a result.
\end{definition}

The following lemma captures properties of RP when it uses HTTPS. For
example, the lemma says that other parties cannot decrypt messages
encrypted by RP.

\begin{lemma}[RP messages are protected by HTTPS]\label{lemma:k-does-not-leak-from-honest-rp} 
  If in the processing step $s_i \rightarrow s_{i+1}$ of a run $\rho$
  of $\oauthwebsystem^n$ an honest relying party $r$ (I) emits an HTTPS
  request of the form

  \[ m = \ehreqWithVariable{\mi{req}}{k}{\pub(k')} \]
  (where $\mi{req}$ is an HTTP request, $k$ is a nonce (symmetric
  key), and $k'$ is the private key of some other DY process $u$), and (II) in the
  initial state $s_0$ the private key $k'$ is only known to $u$, and
  (III) $u$ never leaks $k'$, then all of the following
  statements are true:
  \begin{enumerate}
  \item There is no state of $\oauthwebsystem^n$ where any party except
    for $u$ knows $k'$, thus no one except for $u$ can
    decrypt $\mi{req}$.
    \label{prop:attacker-cannot-decrypt-oauth}
  \item If there is a processing step $s_j \rightarrow s_{j+1}$ where
    the RP $r$ leaks $k$ to $\bidsystem \setminus \{u, r\}$ there
    is a processing step $s_h \rightarrow s_{h+1}$ with $h < j$
    where $u$ leaks the symmetric key $k$ to $\bidsystem \setminus
    \{u,r\}$ or $r$ is corrupted in
    $s_j$. \label{prop:k-doesnt-leak-oauth}
  \item The value of the host header in $\mi{req}$ is the domain that
    is assigned the public key $\pub(k')$ in RP's keymapping
    $s_0.\str{keyMapping}$ (in its initial
    state). \label{prop:host-header-matches-oauth}
  \item If $r$ accepts a response (say, $m'$) to $m$ in a processing step $s_j
    \rightarrow s_{j+1}$ and $r$ is honest in $s_j$ and $u$ did not
    leak the symmetric key $k$ to $\bidsystem \setminus \{u,r\}$ prior
    to $s_j$, then either $u$ or $r$ created the HTTPS response $m'$ to the HTTPS
    request $m$, in particular, the nonce of the HTTP request $\mi{req}$ is not known to
    any atomic process $p$, except for the atomic DY processes $r$ and
    $u$.\label{prop:only-owner-answers-oauth}
  \end{enumerate}
\end{lemma}

\begin{proof} 
  \textbf{(\ref{prop:attacker-cannot-decrypt-oauth})} follows
  immediately from the condition. If $k'$ is initially only known to
  $u$ and $u$ never leaks $k'$, i.e., even with the knowledge of all
  nonces (except for those of $u$), $k'$ can never be derived from any network
  output of $u$, $k'$ cannot be known to any other party. Thus, nobody
  except for $u$ can derive $\mi{req}$ from $m$.

  \textbf{(\ref{prop:k-doesnt-leak-oauth})} 
  We assume that $r$ leaks $k$ to $\bidsystem \setminus \{u,r\}$ in
  the processing step $s_j \rightarrow s_{j+1}$ without $u$ prior
  leaking the key $k$ to anyone except for $u$ and $r$ and that the
  RP is not fully corrupted in $s_j$, and lead this to a contradiction.

  The RP is honest in $s_i$. From the definition of the RP, we see
  that the key $k$ is always a fresh nonce that is not used anywhere
  else. Further, the key is stored in $\mi{pendingRequests}$ ($\nu_4$ in Lines~\ref{line:move-reference-to-pending-request}f. of Algorithm~\ref{alg:rp-oauth}). The
  information from $\mi{pendingRequests}$ is not extracted or used
  anywhere else, except when handling the received messages, where the key
  is only checked against and used to decrypt the message (Lines~\ref{line:rp-https-response}ff. of Algorithm~\ref{alg:rp-oauth}). Hence, $r$ does not leak $k$ to any other
  party in $s_j$ (except for $u$ and $r$). This proves
  (\ref{prop:k-doesnt-leak-oauth}).

  \textbf{(\ref{prop:host-header-matches-oauth})} Per the definition
  of RPs (Algorithm~\ref{alg:rp-oauth}), a host header is always
  contained in HTTP requests by RPs. From
  Line~\ref{line:select-enc-key-oauth} of Algorithm~\ref{alg:rp-oauth}
  we can see that the encryption key for the request $\mi{req}$ was
  chosen using the host header of the message. It is chosen from the
  $\mi{keyMapping}$ in RP's state, which is never changed during
  $\rho$. This proves (\ref{prop:host-header-matches-oauth}).

  \textbf{(\ref{prop:only-owner-answers-oauth})} An HTTPS response
  $m'$ that is accepted by $r$ as a response to $m$ has to be
  encrypted with $k$. The nonce $k$ is stored by the RP in the
  $\mi{pendingRequests}$ state information (see
  Line~\ref{line:move-reference-to-pending-request} of
  Algorithm~\ref{alg:rp-oauth}). The RP only stores freshly chosen
  nonces there (i.e., the nonces are not used twice, or for other
  purposes than sending one specific request). The information cannot
  be altered afterwards (only deleted) and cannot be read except when
  the RP checks incoming messages. The nonce $k$ is only known to $u$
  (which did not leak it to any other party prior to $s_j$) and $r$
  (which did not leak it either, as $u$ did not leak it and $r$ is
  honest, see (\ref{prop:k-doesnt-leak-oauth})). This proves
  (\ref{prop:only-owner-answers-oauth}). \qed
\end{proof}

On a high level, the following lemma shows that the contents in the
list of pending HTTP requests are immutable.

\begin{lemma}[Pending DNS messages become pending
  requests]\label{lemma:rp-pendingdns-to-pendingrequests}
  Let $r$ be some honest relying party in $\oauthwebsystem^n$, $\nu \in
  \nonces$, $l > 0$ such that $(S^l,E^l,N^l)$ is a state in $\rho$,
  and let $\mi{ref} \in \terms$, $\mi{req} \in \httprequests$ such
  that $S^l(r).\str{pendingDNS} \equiv S^{l-1}(r).\str{pendingDNS}
  \plusPairing \an{\nu,\an{\mi{ref},\mi{req}}}$. Then we have that
  $\forall l'$: if there exist $\mi{ref'}$, $\mi{req'}$, $x$, $y \in
  \terms$ with $\mi{req}.\str{nonce} \equiv \mi{req'}.\str{nonce}$ and
  $\an{\mi{ref}',\mi{req'},x,y} \inPairing
  S^{l'}(r).\str{pendingRequests}$ then $\mi{req} \equiv \mi{req'}
  \wedge \mi{ref} \equiv \mi{ref'}$.
\end{lemma}

\begin{proof}
  We first note that Algorithm~\ref{alg:rp-oauth} (of relying
  parties) modifies the subterm $\mi{pendingDNS}$ of the RP's state
  only in such a way that entries are appended to or removed from this
  subterm, but never modified. Entries are appended in
  Lines~\ref{line:rp-add-introspect-1}, \ref{line:rp-pendingdns-trigger},
  \ref{line:rp-start-retrieve-code},
  \ref{line:rp-pendingdns-password}, and
  \ref{line:rp-add-introspect-2}. At all these places in the algorithm, an
  HTTP message term, say $\mi{req}$, having a fresh (HTTP) nonce, is
  appended (together with some term $\mi{ref}$) to the subterm
  $\mi{pendingDNS}$. (A processing step executing one of these parts
  of the algorithm results in the state $(S^l,E^l,N^l)$ of $\rho$.)
  Entries are only removed in Line~\ref{line:rp-remove-pendingdns}. In
  this part of the algorithm, a sequence $\an{\mi{ref''},\mi{req''},
    x,y}$ with $x$, $y \in \terms$ and $\mi{req''} \equiv \mi{req}$
  and $\mi{ref''} \equiv \mi{ref}$ (which could not have been altered
  in any processing step) are appended to the subterm
  $\mi{pendingRequests}$ of RP's state (in
  Line~\ref{line:move-reference-to-pending-request}). Besides
  Line~\ref{line:rp-remove-pending-request}, where some entry is
  removed from this subterm, there is no other part of the algorithm
  that alters $\mi{pendingRequests}$ in any way. Hence, there we
  cannot have any state $(S^{l'},E^{l'},N^{l'})$ of $\rho$ where we
  have an request in $\mi{pendingRequests}$ with the same (HTTP) nonce
  but a different $\mi{req'}$ or a different $\mi{ref'}$. \qed
\end{proof}

\begin{lemma}[RPs never send requests to themselves]\label{lemma:rp-never-sends-requests-to-itself}
  An honest RP never sends an HTTP request to any RP  (including itself), and only sends
  HTTPS requests to RPs that the receiving RP cannot decrypt.
\end{lemma}

\begin{proof}
  Honest RPs send HTTP requests only in
  Lines~\ref{line:rp-send-something-1},
  \ref{line:rp-send-something-2}, \ref{line:rp-send-something-3},
  \ref{line:rp-send-something-4}, and \ref{line:rp-send-something-5}.
  In all of these cases, they send the HTTPS request to an endpoint
  configured in the state (in $\str{idps}$). With
  Definition~\ref{def:relying-parties}, it follows that the domains to
  which these requests are sent, are never a domain of an RP. All
  requests are sent over HTTPS, and the ``correct'' encryption keys
  (as stored in $\str{keyMapping}$) are used (i.e., even if the
  attacker changes the DNS response such that an HTTPS request is sent
  to an RP, it cannot be decrypted by the RP).
\end{proof}

\input{appendix-oauth-proof-authentication}

\input{appendix-oauth-proof-authorization}

\input{appendix-oauth-proof-session-integrity}

%% file: appendix-oauth-proof-authentication.tex
\subsection{Proof of Authentication}
\label{sec:proof-property-a}

We here want to show that every OAuth web system is secure
w.r.t.~authentication, and therefore assume that there exists an
OAuth web system that is not secure w.r.t.~authentication. We then
lead this to a contradiction, thereby showing that all OAuth web
systems are secure w.r.t.~authentication. In detail, we assume:

\begin{assumption}\label{asn:prop-a}  
  There exists an OAuth web system with a network attacker $\oauthwebsystem^n$,
  a run $\rho$
  of $\oauthwebsystem^n$,
  a state $(S^j, E^j, N^j)$
  in $\rho$,
  some $r\in \fAP{RP}$
  that is honest in $S^j$,
  some $i \in \fAP{IDP}$
  that is honest in $S^j$,
  some $g \in \mathsf{dom}(i)$,
  some $u \in \mathbb{S}$
  with the browser $b$
  owning $u$
  being not fully corrupted in $S^j$
  and all $r' \in \mathsf{trustedRPs}(\mathsf{secretOfID}(\an{u,g}))$
  being honest, some RP service token of the form $\an{n,\an{u,g}}$
  recorded in $S^j(r).\str{serviceTokens}$
  such that $n$
  is derivable from the attackers knowledge in $S^j$
  (i.e., $n \in d_{\emptyset}(S^j(\fAP{attacker}))$).
\end{assumption}

To show that this is a contradiction, we first show some lemmas:

\begin{lemma}[Attacker does not learn passwords]\label{lemma:attacker-does-not-learn-password}
  There exists no $l \le j$, $(S^l, E^l, N^l)$ being a state in $\rho$
  such that
  $\mathsf{secretOfID}(u) \in d_\emptyset(S^l(\fAP{attacker}))$.
\end{lemma}

\begin{proof}
  Let $s := \mathsf{secretOfID}(\an{u,g})$ and $R := \mathsf{trustedRPs}(s)$.
  Initially, in $S^0$, $s$ is only contained in
  $S^0(b).\str{secrets}[\an{d,\https}]$ for any
  $d \in \bigcup_{r' \in R} \mathsf{dom}(r') \cup \mathsf{dom}(i)$ and
  in no other states (or waiting events). By the definition of the
  browser, we can see that only scripts loaded from the origins
  $\an{d,\https}$ can access $s$. We know that $i$ and all $r' \in R$
  are honest (from the assumption). We therefore have that only the
  scripts $\mi{script\_rp\_index}$, $\mi{script\_rp\_implicit}$, and
  $\mi{script\_idp\_form}$ can access $s$ (if loaded from their
  respective origins) and that the browser does not use or leak $s$ in
  any other way. $\mi{script\_rp\_implicit}$ does not use any browser
  secrets. We therefore focus on the remaining two scripts:

  \begin{description}
  \item[\textbf{\textit{script\_rp\_index}.}] If this script was
    loaded and has access to $s$, it must have been loaded from origin
    $\an{d,\https}$ for a domain $d$ of some trusted relying party,
    say $t$ ($\in R$). If $\mi{script\_rp\_index}$ selects the secret
    $s$ in Line~\ref{line:script-rp-index-select-secret} of
    Algorithm~\ref{alg:script-rp-index}, we know that it must have
    selected the id $u$ in Line~\ref{line:script-rp-index-select-id}.
    We therefore know that in
    Line~\ref{line:script-rp-index-password-login-form}, the browser
    $b$ is instructed to send (using HTTPS) $\an{u,s}$ to the path
    $\str{/passwordLogin}$ at $d$. If $b$ sends such a request, $t$ is
    the only party able to decrypt this request (see the general security
    properties in~\cite{FettKuestersSchmitz-TR-BrowserID-Primary-2015}). This
    message is then processed by $t$ according to
    Lines~\ref{line:rp-password-login}ff. There, username and password
    are forwarded to some IdP, say $i'$, using an HTTPS POST
    request. More precisely, this request is sent to the domain of the
    token endpoint URL contained in the IdP registration record for
    the domain contained in $u$. From
    Definitions~\ref{def:idp-registration-record}
    and~\ref{def:relying-parties} and the fact that this part of the
    state (of relying parties) is never changed, we can see that the
    request is sent to a domain of $i$, and therefore $i' = i$. (The
    attacker can also not modify or read this request, see
    Lemma~\ref{lemma:k-does-not-leak-from-honest-rp}.) The body of the
    HTTPS POST request sent to $i$ is of the following form:
    \[\an{\an{\str{grant\_type},
        \str{password}},\an{\str{username},u},\an{\str{password},s}}
    \, .\]
    Such a request can processed by IdP only in
    Lines~\ref{line:idp-grant-type-password}ff. of
    Algorithm~\ref{alg:idp-oauth}. There, IdP checks $s$ and discards
    it. Therefore, $s$ does not leak from $i$, $t$, or $b$ to the
    attacker (or any other party).

  \item[\textbf{\textit{script\_idp\_form}.}] If this script was
    loaded and has access to $s$, it must have been loaded from origin
    $\an{d,\https}$ for a domain $d$ of $i$. This script sends $s$ to
    $d$ in an HTTPS POST request.  If $b$ sends such a request, $i$ is
    the only party able to decrypt this request (see the general
    security properties
    in~\cite{FettKuestersSchmitz-TR-BrowserID-Primary-2015}). This
    message is then processed by $i$ according to
    Lines~\ref{line:idp-auth-endpoint-post}ff. of
    Algorithm~\ref{alg:idp-oauth}. There, the IdP $i$ checks $s$ and
    discards it. Therefore, $s$ does not leak from $i$ or $b$ to the
    attacker (or any other party).

  \end{description}

  This proves Lemma~\ref{lemma:attacker-does-not-learn-password}. \qed

\end{proof}

\begin{lemma}[Attacker does not learn authorization codes]\label{lemma:attacker-does-not-learn-code}
  There exists no $l \le j$, $(S^l, E^l, N^l)$ being a state in
  $\rho$, $v \in \nonces$, $y \in \terms$ such that
  $v \in d_\emptyset(S^l(\fAP{attacker}))$ and
  $\an{v,\an{\mathsf{clientIDOfRP}(r,i),y,u}} \inPairing
  S^l(i).\str{codes}$.
\end{lemma}\df{$u$ or $\an{u,g}$ in this lemma?}

\begin{proof}
  $S^l(i).\str{codes}$ is initially empty and appended to only in
  Line~\ref{line:idp-create-auth-code} of
  Algorithm~\ref{alg:idp-oauth} (where an authorization code is
  created). From Line~\ref{line:idp-auth-endpoint-post}ff. it is easy
  to see that the request which triggers the creation of the
  authorization code must carry a valid password for the specific
  identity in the request body. With
  Lemma~\ref{lemma:attacker-does-not-learn-password}, we can see that
  such a request can not come from the attacker, as the attacker does
  not know the password needed in the request. It can also not
  originate from an IdP, as IdPs do not send requests. Further, the
  request can not originate from any corrupted party or an
  attacker-controlled origin in the honest browser (as otherwise there
  would be a flow where the attacker would learn the password by
  sending it to himself, which can be ruled out by
  Lemma~\ref{lemma:attacker-does-not-learn-password}). It is also
  impossible that the request originated from any non-attacker
  controlled origin in the honest browser: Such a request could be
  caused by either a Location redirect or a script. (We will refer to
  the following as *.) A Location redirect must have been issued by an
  honest party (otherwise, the attacker would have learned the
  password by the time he issued the response, see
  Lemma~\ref{lemma:attacker-does-not-learn-password}). There are two
  occasions where honest parties issue Location redirect headers:
  \begin{description}
  \item[IdP in
    Lines~\ref{line:idp-redir-with-auth-code}/\ref{line:idp-redir-with-token}
    of Algorithm~\ref{alg:idp-oauth}] In this case, an HTTP status
    code of 303 is sent. While this causes the browser to do a new
    request, the new request has an empty body in any
    case.\footnote{Note that at this point it is important that a 303
      redirect is performed, not a 307 redirect. See
      Line~\ref{browser-remove-body} of
      Algorithm~\ref{alg:processresponse} for details.}
  \item[RP in Line~\ref{line:rp-send-redirect} of
    Algorithm~\ref{alg:rp-oauth}] In this case, a 307 redirect could
    be issued, causing the browser to preserve the request body. We
    therefore have to check what could have caused the browser to
    issue a request that caused this Location redirect response, and
    what body could be contained in such a request. For clarity, we
    call the request causing the redirection $m$. It is clear that $m$
    cannot come from the attacker (as it contains the password). It
    must therefore come from an honest browser. If it was caused by a
    redirect in the honest browser, (*) applies recursively.
    Otherwise, there are three scripts that could send such a request
    to RP: $\mi{script\_rp\_index}$, $\mi{script\_rp\_implicit}$, and
    $\mi{script\_idp\_form}$. Of these, only $\mi{script\_rp\_index}$
    causes a request for the path $\str{/startInteractiveLogin}$
    (which triggers the redirection in
    Line~\ref{line:rp-send-redirect} of Algorithm~\ref{alg:rp-oauth}),
    which, however, does not contain any secret.
  \end{description}
  A Location redirect can therefore be ruled out as the cause of the
  request. There are three scripts that could send such a request:
  $\mi{script\_rp\_index}$, $\mi{script\_rp\_implicit}$, and
  $\mi{script\_idp\_form}$. The first two, $\mi{script\_rp\_index}$,
  $\mi{script\_rp\_implicit}$, do not send requests to any IdP
  (instead, they only send requests to the RP that sent the scripts to
  the browser, IdP does not send these scripts to the browser). The
  latter script, $\mi{script\_idp\_form}$, can send the request. In
  this (last remaining) case, the IdP responds with a Location
  redirect header in the response, which, among others, carries a URL
  containing the critical value $v$ (in
  Line~\ref{line:idp-redir-with-auth-code}). In this case, the browser
  receives the response, and immediately triggers a new request to the
  redirection URL. This URL was composed by the IdP using the list of
  valid redirection URIs from $S^l(i).\str{clients}$, a part of the
  state of $i$ that is not changed during any run.
  Definition~\ref{def:initial-state-idp} defines how
  $S^l(i).\str{clients}$ is initialized: For the client id $c:=
  \mathsf{clientIDOfRP}(r,i)$, all redirection URLs carry hosts
  (domains) of $r$, have the protocol $\https$ (HTTPS), and contain a
  query parameter component identifying the IdP $i$. In the checks in
  Lines~\ref{line:check-redir-uris}ff., it is ensured that in any
  case, this restriction on domain and protocol applies to the
  resulting redirection URI (called $\mi{redirecturi}$ in the
  algorithm) as well. Therefore, the browser's GET request which is
  triggered by the Location header and contains the value $v$ is sent
  to $r$ over HTTPS.

  The RP $r$
  can process such a GET request only in
  Lines~\ref{line:rp-serve-index} and~\ref{line:rp-redir-endpoint} of
  Algorithm~\ref{alg:rp-oauth}. It is clear, that in
  Line~\ref{line:rp-serve-index}, the value $v$
  does not leak to the attacker: An honest script is loaded into the
  browser, which does not use $v$
  in any form. If this script causes a request to the attacker (or
  causes a request which would be redirected to the attacker), the
  request does not contain $v$.
  In particular, $v$
  cannot be contained in the Referer header, because this is prevented
  by the Referrer Policy.

  In Lines~\ref{line:rp-redir-endpoint}ff., $v$
  is forwarded to the IdP for checking its validity and retrieving the
  access token (there is also code for retrieving the access code from
  the implicit flow in this part of the code, which is not of interest
  here). When sending the authorization code, it is critical to ensure
  that $v$
  is forwarded to an honest IdP (in particular, $i$),
  and not to the attacker. This is ensured by checking the redirection
  URL parameters, which, as mentioned above, contain a hint for the
  IdP in use, in this case $i$.
  In Line~\ref{line:rp-check-idp-param} it is checked that the IdP, to
  which $v$ is eventually sent, is $i$.

  Therefore, we know that $v$
  is sent via POST to the honest IdP $i$.
  There, it can only be processed in
  Lines~\ref{line:idp-token-endpoint}ff. Here, it is easy to see that
  the value $v$
  (called $\mi{body}[\str{code}]$
  in the algorithm) is checked. However, the value is never sent out
  to any other party and therefore does not leak.

  We have shown that the value $v$ cannot be known to the attacker,
  which proves Lemma~\ref{lemma:attacker-does-not-learn-code}. \qed
\end{proof}

\begin{lemma}[Attacker does not learn access tokens]\label{lemma:attacker-does-not-learn-token}
  There exists no $l \le j$, $(S^l, E^l, N^l)$ being a state in
  $\rho$, $v \in \nonces$, such that
  $v \in d_\emptyset(S^l(\fAP{attacker}))$ and
  $\an{v,\mathsf{clientIDOfRP}(r,i),u} \inPairing
  S^l(i).\str{atokens}$.
\end{lemma}\df{$u$ or $\an{u,g}$ in this lemma?}

\begin{proof}
  Initially, we have $S^0(i).\str{atokens} \equiv \an{}$.
  $S^l(i).\str{atokens}$ is appended to only in
  Lines~\ref{line:idp-create-atoken-implicit},
  \ref{line:idp-create-atoken-code},
  \ref{line:idp-create-atoken-password}, and
  \ref{line:idp-create-atoken-clientcred} (where in each an access
  token is issued) of Algorithm~\ref{alg:idp-oauth} and not altered in
  any other way.

  In Line~\ref{line:idp-create-atoken-clientcred}, a term of the form
  $\an{*,*,\bot}$ is appended, which is not of the form
  $\an{v,\mathsf{clientIDOfRP}(r,i),u}$. In what follows, we will
  distinguish between the lines of Algorithm~\ref{alg:idp-oauth} were
  $\an{v,\mathsf{clientIDOfRP}(r,i),u}$ is created:

  \begin{description}
  \item[Line~\ref{line:idp-create-atoken-implicit}.] It is easy to
    see, that $i$
    must have received an HTTPS POST request containing an Origin
    header with one of its HTTPS origins and containing (in its body)
    a dictionary with the entries $\an{\str{username},u}$,
    $\an{\str{password},\mathsf{secretOfID}(u)}$,
    and $\an{\str{client\_id},\mathsf{clientIDOfRP}(r,i)}$.
    (Note that in this case, $\mathsf{clientIDOfRP}(r,i) \neq \bot$,
    and therefore, $r \neq \bot$.)
    From Lemma~\ref{lemma:attacker-does-not-learn-password} it follows
    that such a request cannot be assembled by the attacker. Also,
    neither an IdP nor an RP sends such a request. Hence, this request
    must have be sent from a browser. In the browser, only the scripts
    $\str{script\_idp\_form}$
    and the attacker script $\Rasp$
    can instruct the browser to send such a request. From
    Lemma~\ref{lemma:attacker-does-not-learn-password} we know that
    the attacker script cannot access $\mathsf{secretOfID}(u)$
    (otherwise, there would be a run $\rho'$
    in which the attacker script would send $\mathsf{secretOfID}(u)$
    to the attacker instead). Hence, this request must originate from
    a command returned by $\str{script\_idp\_form}$
    and it must be created by the browser $b$
    (which is $\mathsf{ownerOfID}(u)$).
    This script only sends such a request to its own origin, which
    must be an HTTPS origin (it would not have access to
    $\mathsf{secretOfID}(u)$
    otherwise). The IdP responds with a Location redirect header in
    the response, which among others, carries a URL containing the
    critical value $v$
    (in Line~\ref{line:idp-redir-with-access-token}) in the fragment
    of the URL. In this case, the browser receives the response, and
    immediately triggers a new request to the redirection URL. This
    URL was composed by the IdP using the list of valid redirection
    URIs from $S^l(i).\str{clients}$,
    a part of the state of $i$
    that is not changed during any run.
    Definition~\ref{def:initial-state-idp} defines how
    $S^l(i).\str{clients}$
    is initialized: For the client id
    $c:= \mathsf{clientIDOfRP}(r,i)$,
    all redirection URLs carry hosts (domains) of $r$,
    have the protocol $\https$
    (HTTPS), and contain a query parameter component identifying the
    IdP $i$.
    In the checks in Lines~\ref{line:check-redir-uris}ff., it is
    ensured that in any case, this restriction on domain and protocol
    applies to the resulting redirection URI (called
    $\mi{redirecturi}$
    in the algorithm) as well. Therefore, the browser's GET request
    which is triggered by the Location header and contains the value
    $v$ in the fragment, is sent to $r$ over HTTPS.

    The RP $r$ can process such a GET request only in
    Lines~\ref{line:rp-serve-index} and~\ref{line:rp-redir-endpoint}
    of Algorithm~\ref{alg:rp-oauth}. It is clear, that in
    Line~\ref{line:rp-serve-index}, the value $v$ does not leak to the
    attacker: The honest script $\mi{script\_rp\_index}$ is loaded
    into the browser, which does not use $v$ in any form.\gs{etwas
      sehr kurz}

    In Lines~\ref{line:rp-redir-endpoint}ff., RP's algorithm branches
    into two different flows: (1) RP takes some value from the URL
    parameters (which do not contain $v$) and sends it to some
    process. RP defers its response to the browser and will (later)
    only send out the response in
    Lines~\ref{line:rp-create-service-token}ff. This response,
    however, does not contain a script and hence, the browser will not
    be instructed to create any new messages from the resulting
    document. Hence, $v$ does not leak in this case. (2) RP sends an
    HTTPS response containing the script $\mi{script\_rp\_implicit}$
    (and, in the script's initial state, a domain of $i$ derived from
    the redirection URL), which takes $v$ from the URL parameters and
    instructs the browser to send an HTTPS POST request containing $v$
    and the domain of $i$ to the script's (secure) origin at path
    $\str{/receiveTokenFromImplicitGrant}$. RP processes such a
    request in Lines~\ref{line:rp-receive-implicit-token}ff. where it
    forwards $v$ to the IdP for checking its validity. Here, it is
    critical to ensure that $v$ is forwarded to an honest IdP (in
    particular, $i$), and not to the attacker. This is fulfilled since
    a domain of $i$ is contained in the request's body, and, before
    forwarding, it is checked that $v$ is only forwarded to this
    domain.

    Therefore, we know that $v$ is sent via GET to the honest IdP $i$.
    There, it can only be processed in
    Lines~\ref{line:idp-introspection-endpoint}ff. Here, it is easy to
    see that the value $v$ is never sent out to any
    other party and therefore does not leak.
  \item[Line~\ref{line:idp-create-atoken-code}.] In this case, $i$
    must have received an HTTPS POST request carrying a dictionary in
    its body containing the entries
    $\an{\str{grant\_type},\str{authorization\_code}}$
    and $\an{\str{code},\mi{code}}$
    with $\mi{code} \in \nonces$
    such that
    $\an{\mi{code},\an{\mathsf{clientIDOfRP}(r,i),y,u}} \inPairing
    S^{l'}(i).\str{codes}$ for some $y \in \terms$
    and $l' \le l$.\gs{etwas
      kurz}(Note that, as above,
    $\mathsf{clientIDOfRP}(r,i) \neq \bot$,
    and therefore, $r \neq \bot$.)
    From Lemma~\ref{lemma:attacker-does-not-learn-code} it follows
    that such a request can neither be constructed by the attacker nor
    by a browser instructed by the attacker script $\Rasp$.
    In a browser, the remaining honest scripts do not instruct the
    browser to send such a request. (Honest) IdPs do not send such
    requests. Hence, such a request must have been constructed by an
    (honest) RP. An RP prepares such a request only in
    Lines~\ref{line:rp-prepare-atoken-from-code-req}ff. (of
    Algorithm~\ref{alg:rp-oauth}) and finally sends out this request
    in Line~\ref{line:rp-send-https-request} (after a DNS response).
    With Lemma~\ref{lemma:rp-pendingdns-to-pendingrequests} and
    Lemma~\ref{lemma:k-does-not-leak-from-honest-rp} we know that
    $\mi{reference}$
    contains a term of the form $\an{\str{code},\mi{idp},*,*,*,*}$
    with $\mi{idp} \in \mathsf{dom}(i)$
    (as the request was sent encrypted for and to $i$).
    When RP receives the response from $i$,
    RP processes this response in
    Lines~\ref{line:rp-https-response}ff. where RP distinguishes
    between two cases based on the first subterm in $\mi{reference}$.
    As we know that this subterm is $\str{code}$,
    we have that the response is processed only in
    Lines~\ref{line:rp-https-response-code-password-ccred}ff. RP takes
    a subterm from the response's body which might
    contain\footnote{The subterm actually is $v$.}
    $v$
    in Line~\ref{line:rp-https-response-atoken} and prepares an HTTPS
    POST request to an URL of $i$
    (which is taken from the subterm $\mi{idps}$
    of RP's state and this subterm is never altered and initially
    configured such that the URLs under the dictionary key $\mi{idp}$
    are actually belonging to $i$).
    This HTTPS POST request contains $v$
    in the parameter $\str{token}$.
    This request is finally sent out this request in
    Line~\ref{line:rp-send-https-request} (after a DNS response)
    encrypted for and to $i$.

    It is now easy to see that $i$ only accepts\gs{was heißt accept hier genau. vielleicht sollten wir besser ``processes'' sagen.} the
    request only in Lines~\ref{line:idp-introspection-endpoint}ff. (of
    Algorithm~\ref{alg:idp-oauth}). There, the IdP only checks the
    parameter $\str{token}$ against its state and discards it
    afterwards. Hence, $v$ does not leak.

  \item[Line~\ref{line:idp-create-atoken-password}.] In this case, $i$
    must have received an HTTPS POST request carrying a dictionary in
    its body containing the entries
    $\an{\str{grant\_type},\str{password}}$, $\an{\str{username},u}$,
    and $\an{\str{password},\mathsf{secretOfID}(u)}$. From
    Lemma~\ref{lemma:attacker-does-not-learn-password} it follows that
    such a request cannot be constructed by the attacker, dishonest
    scripts in browsers, or any other dishonest party. (Honest) IdPs
    do not construct such a request. All honest scripts do not
    instruct a browser to send such a request. Hence, the request must
    have been constructed by an honest RP. An RP prepares such a
    request only in
    Lines~\ref{line:rp-prepare-atoken-from-password-req}ff. (of
    Algorithm~\ref{alg:rp-oauth}) and finally sends out this request
    in Line~\ref{line:rp-send-https-request} (after a DNS response).
    With Lemma~\ref{lemma:rp-pendingdns-to-pendingrequests} and
    Lemma~\ref{lemma:k-does-not-leak-from-honest-rp} we know that
    $\mi{reference}$ contains a term of the form
    $\an{\str{password},\mi{idp},*,*,*,*}$ with
    $\mi{idp} \in \mathsf{dom}(i)$ (as the request was sent encrypted
    for and to $i$). When RP receives the response from $i$, RP
    processes this response in Lines~\ref{line:rp-https-response}ff.
    where RP distinguishes between two cases based on the first
    subterm in $\mi{reference}$. As we know that this subterm is
    $\str{code}$, we have that the response is processed only in
    Lines~\ref{line:rp-https-response-code-password-ccred}ff. RP takes
    a subterm from the response's body which might
    contain\footnote{The subterm actually is $v$.} $v$ in
    Line~\ref{line:rp-https-response-atoken} and prepares an HTTPS
    POST request to an URL of $i$ (which is taken from the subterm
    $\mi{idps}$ of RP's state and this subterm is never altered and
    initially configured such that the URLs under the dictionary key
    $\mi{idp}$ are actually belonging to $i$). This HTTPS POST request
    contains $v$ in the parameter $\str{token}$. This request is
    finally sent out this request in
    Line~\ref{line:rp-send-https-request} (after a DNS response)
    encrypted for and to $i$. 
    It is now easy to see that $i$ only accepts\gs{genauer} the request only in
    Lines~\ref{line:idp-introspection-endpoint}ff. (of
    Algorithm~\ref{alg:idp-oauth}). There, the IdP only checks the
    parameter $\str{token}$ against its state and discards it
    afterwards. Hence, $v$ does not leak.
  \end{description}

  We have shown that the value $v$ cannot be known to the attacker,
  which proves Lemma~\ref{lemma:attacker-does-not-learn-token}. \qed
\end{proof}

We can now show that Assumption~\ref{asn:prop-a} is a contradiction.

\begin{lemma}\label{lemma:proof-authentication-final}
  Assumption~\ref{asn:prop-a} is a contradiction. 
\end{lemma}
\begin{proof}
  The service token $\an{n,\an{u,g}}$ can only be created and added to
  the state $S^j(r).\str{serviceTokens}$ in
  Line~\ref{line:rp-create-service-token} of
  Algorithm~\ref{alg:rp-oauth}. To get to this point in the algorithm,
  in Line~\ref{line:rp-check-reference-for-introspect}, it is checked
  that $\mi{reference}$ is a tupel of the form
  $\an{\str{introspect}, \mi{mode}, g, a', f', n', k'}$. This is taken
  from the pending requests, where the value is transferred to from
  the pending DNS subterm (see
  Lemma~\ref{lemma:rp-pendingdns-to-pendingrequests}). Such a term
  (starting with $\str{introspect}$) is added to the $\str{pendingDNS}$
  subterm only in Lines~\ref{line:rp-add-introspect-1}
  and~\ref{line:rp-add-introspect-2}. We can now do a case distinction
  between these two possibilities to identify the request $m'$ to
  which the response containing the service token will be sent.

  \begin{description}
  \item[Subterm was added in Line~\ref{line:rp-add-introspect-1}.] In
    this case, in
    Line~\ref{line:rp-https-response-code-password-ccred}, an entry of
    the form $\an{\mi{mode}, g, a', f', n', k'}$ must have existed as
    a reference in the pending HTTP requests, where $\mi{mode}$ is
    either $\str{code}$ or $\str{password}$.\footnote{If $\mi{mode}$
      was $\str{client\_credentials}$, no service token is created. }
    Such entries are created in the following lines:

    \begin{description}
    \item[Line~\ref{line:rp-start-retrieve-code}.] Here, a request
      $m'$ must have been received which contained a valid
      authorization code for the identity $u$ at the IdP
      $i$.\footnote{Otherwise, the IdP would not have returned an
        access token for the identity $u$. As $g = \mi{idp}$ is the
        value stored in the reference, it is also clear that the
        authorization code was, in fact, sent to $i$ for retrieving
        the access token, and not to the attacker or another identity
        provider. Also, the request to $i$ was sent over HTTPS, and
        therefore, Lemma~\ref{lemma:k-does-not-leak-from-honest-rp}
        applies.} The attacker cannot know such an authorization code
      (see Lemma~\ref{lemma:attacker-does-not-learn-code}). The RP $r$
      does not send requests to itself or to other RPs (see
      Lemma~\ref{lemma:rp-never-sends-requests-to-itself}), and no
      IdPs send requests. Therefore, $m'$ must have originated from an
      honest browser.
    \item[Line~\ref{line:rp-start-retrieve-code-from-password}.] In
      this case, a request $m'$ was received which contained a valid
      username and password combination for $u$ at $i$. (As above, we
      know that $i$ was used to verify that information as $g$ is a
      domain of $i$, and $\mi{idp} = g$. \df{<- auch sehr kurz}) Only
      the honest browser $b$ and some relying parties know this
      password (see
      Lemma~\ref{lemma:attacker-does-not-learn-password}), but the RPs
      would not send such a request. The request $m'$ was therefore
      sent from the browser $b$.
    \end{description}
  \item[Subterm was added in Line~\ref{line:rp-add-introspect-2}.] If the
    subterm $\an{\str{introspect}, \mi{mode}, g, a', f', n', k'}$ was
    added in this line, the request causing this ($m'$) must have
    carried a valid access token for the identity $u$ at $i$. (As
    above, the access token was sent to $i$ for validation.) The
    attacker does not know such an access token (see
    Lemma~\ref{lemma:attacker-does-not-learn-token}), and other RPs or
    IdPs cannot send $m'$. Therefore, an honest browser must have sent
    $m'$.
  \end{description}

  We therefore have that in all cases, $m'$ was sent by an honest
  browser. Further, $m'$ must have been an HTTPS request (by the
  definition of RPs). If the request was sent as the result of an
  XMLHTTPRequest command from a script, that script must have been
  loaded from the origin $\an{g_r, \https}$ with $g_r \in
  \mathsf{dom}(r)$. This is a contradiction (there are no honest
  scripts that use XMLHTTPRequest). Otherwise, it was a ``regular''
  request. In this case, the browser tries to load the service token
  as a document (which will fail). In particular, the service token
  $\an{n, \an{u,g}}$ never leaks to the attacker.

  We therefore know that the attacker cannot know the service token,
  which is a contradiction to the assumption.
\end{proof}

%% file: appendix-oauth-proof-authorization.tex
\subsection{Proof of Authorization}
\label{sec:proof-authorization}

As above, we assume that there exists an OAuth web system that is not
secure w.r.t.~authorization and lead this to a contradiction. Note
that in the following, some of the lemmas shown in
Appendix~\ref{sec:proof-property-a} are used.

\begin{assumption}\label{asn:authorization}
  There exists a run $\rho$
  of an OAuth web system with a network attacker $\oauthwebsystem^n$,
  a state $(S^j, E^j, N^j)$
  in $\rho$,
  some IdP $i \in \fAP{IDP}$
  that is honest in $S^j$,
  some RP $r \in \fAP{RP} \cup \{ \bot \}$
  with $r$
  being honest in $S^j$
  unless $r = \bot$,
  some $u \in \IDs \cup \{\bot\}$,
  some $n = \mathsf{resourceOf}(i,r,u)$,
  $n$
  being derivable from the attackers knowledge in $S^j$
  (i.e., $n \in d_{\emptyset}(S^j(\fAP{attacker}))$),
  and $u = \bot$
  or ((i) the browser $b$
  owning $u$
  is not fully corrupted in $S^j$
  and (ii) all $r' \in \mathsf{trustedRPs}(\mathsf{secretOfID}(u))$
  are honest $S^j$).
\end{assumption}

We first show the following lemma:

\begin{lemma}[Attacker does not learn RP secrets.]\label{lemma:att-does-not-learn-rp-secrets}
  There exists no $l \leq j$,
  $(S^l, E^l, N^l)$
  being a state in $\rho$
  such that
  $\mathsf{secretOfRP}(r,i) \in d_{\emptyset}(S^l(\fAP{attacker}))$
  unless $\mathsf{secretOfRP}(r,i) \equiv \bot$.
\end{lemma}

\begin{proof}
  Following the definition of the initial states of all atomic
  processes (in particular Definition~\ref{def:relying-parties}),
  initially, $\mathsf{secretOfRP}(r,i)$ is only known to $r$.

  The secret is being used and sent out in an HTTPS message in
  Lines~\ref{line:rp-send-client-password-1}ff.~of
  Algorithm~\ref{alg:rp-oauth} The message is being sent to the token
  endpoint configured for $i$,
  which, according to Definition~\ref{def:idp-registration-record},
  bears a host name belonging to $i$.
  With the definition of $\mi{sslkeys}$
  in Definition~\ref{def:relying-parties} and
  Lemma~\ref{lemma:k-does-not-leak-from-honest-rp} it can be seen that
  this outgoing HTTP POST request can therefore only be read by the
  intended receiver, $i$.

  In $i$,
  the message cannot be processed in the authentication endpoint,
  Lines~\ref{line:idp-auth-endpoint-post}
  to~\ref{line:idp-auth-endpoint-post-end} of
  Algorithm~\ref{alg:idp-oauth}, since it does not carry an Origin
  header. It can be processed in Lines~\ref{line:idp-token-endpoint}
  to~\ref{line:idp-token-endpoint-end}. It is easy to see that the
  secret in the message is not used in any outgoing message, neither
  stored in the IdP's data structures. The message not be processed in
  Line~\ref{line:idp-introspection-endpoint}ff., since it is a POST
  request.

  The same applies when the client sends the password in
  Line~\ref{line:rp-send-client-password-2}ff. or
  Line~\ref{line:rp-send-client-password-3}ff. of
  Algorithm~\ref{alg:rp-oauth}.

  Therefore, the secret $\mathsf{secretOfRP}(r,i)$
  cannot be known to the attacker.
\end{proof}

\begin{lemma}\label{lemma:proof-authorization-final}
  Assumption~\ref{asn:authorization} is a contradiction.
\end{lemma}

\begin{proof}
  At the beginning of each run, the attacker cannot know $n$
  (as defined in the initial states). Only the IdP $i$
  can send out the protected resource $n$,
  in Line~\ref{line:idp-send-secret-resource} of
  Algorithm~\ref{alg:idp-oauth}. In a state $(S^{l'}, E^{l'}, N^{l'})$
  in $\rho$
  for some $l' < j$,
  for $i$
  to send out $n$,
  an HTTPS request must be received by $i$
  which contains, among others, an access token $a$
  such that
  $\an{a, \mathsf{clientIDOfRP}(r,i), u} \inPairing
  S^{l'}(i).\str{atokens}$. We therefore note that for the attacker to
  learn $n$,
  it has to know $a$.
  We also note that if $r$
  requests $n$
  at the IdP $i$,
  the attacker cannot read $n$
  or $a$
  from such messages (see
  Lemma~\ref{lemma:k-does-not-leak-from-honest-rp}).

  We now have to distinguish two cases:

  \begin{description}
  \item[Anonymous Resource,] i.e., $u \equiv \bot$.
    In this case, the access token $a$
    was chosen by $i$
    in Line~\ref{line:idp-create-atoken-clientcred} of
    Algorithm~\ref{alg:idp-oauth}. There, $a$
    is sent out in response to a request that must have contained the
    client credentials for $r$,
    where the client secret cannot be $\bot$
    (see Line~\ref{line:idp-check-auth}. With
    Lemma~\ref{lemma:att-does-not-learn-rp-secrets} we see that the
    attacker cannot send such a request, and therefore, cannot learn
    $a$.
    This implies that the attacker cannot send the request to learn
    $n$ from $i$.
  \item[User Resource,] i.e., $u \not\equiv \bot$.
    In this Case, Lemma~\ref{lemma:attacker-does-not-learn-token}
    shows that it is not possible for the attacker to send a request
    to learn $n$.
  \end{description}
  With this, we have shown that the attacker cannot learn $n$,
  and therefore, Assumption~\ref{asn:authorization} is a
  contradiction.
\end{proof}

%% file: appendix-oauth-proof-session-integrity.tex
\subsection{Proof of Session Integrity}
\label{sec:proof-session-integrity-all}

Before we prove this property, we highlight that in the absence of a
network attacker and with the DNS server as defined for
$\oauthwebsystem^w$,
HTTP(S) requests by (honest) parties can only be answered by the owner
of the domain the request was sent to, and neither the requests nor
the responses can be read or altered by any attacker unless he is the
intended receiver. This property is important for the following proof.

We further show the following lemma, which says that an attacker
(under the assumption above) cannot learn a $\mi{state}$
value that is used in a login session between an honest browser, an
honest IdP, and an honest RP.

\begin{lemma}[Third parties do not learn state]\label{lemma:attacker-does-not-learn-state}
  Let $\rho$
  be a run of an OAuth web system with web attackers $\oauthwebsystem^w$,
  $(S^j, E^j, N^j)$
  be a state of $\rho$,
  $r\in \fAP{RP}$
  be an RP that is honest in $S_j$,
  $i \in \fAP{IDP}$
  be an IdP that is honest in $S_j$,
  $b$ be a browser that is honest in $S_j$.

  Then there exists no $l \le j$,
  with $(S^l, E^l, N^l)$
  being a state in $\rho$,
  a nonce $\mi{loginSessionId} \in \nonces$,
  a nonce $\mi{state} \in \nonces$,
  a domain $h \in \mathsf{dom}(r)$
  of $r$,
  terms $x$,
  $y$,
  $x'$,
  $y'$,
  $z \in \terms$,
  cookie
  $c := \an{\str{loginSessionId}, \an{\mi{loginSessionId}, x', y',
      z}}$, an atomic DY process
  $p \in \websystem \setminus \{b,i,r\}$
  such that $\mi{state} \in d_\emptyset(S^l(p))$,
  $\an{\mi{loginSessionId}, \an{g, \mi{state}, x, y}} \inPairing
  S^l(r).\str{loginSessions}$ and
  $\an{h, c} \inPairing S^l(b).\str{cookies}$.
\end{lemma}

\begin{proof}
  To prove Lemma~\ref{lemma:attacker-does-not-learn-state}, we track
  where the login session identified by $\mi{loginSessionId}$
  is created and used. 
  
  We have that $\an{h, c} \inPairing S^l(b).\str{cookies}$.
  Login sessions are only created in Line~\ref{line:rp-create-session}
  of Algorithm~\ref{alg:rp-oauth} (and never altered afterwards).
  After the session identifier $\mi{loginSessionId}$
  was chosen, its value is sent over the network to the party that
  requested the login. We have that for $\mi{loginSessionId}$,
  this party must be $b$
  because only $r$
  can set the cookie $c$
  for the domain $h$
  in the state of $b$\footnote{Note
    that we have only web attackers.} and
  Line~\ref{line:rp-create-session} of Algorithm~\ref{alg:rp-oauth} is
  actually the only place where $r$ does so.

  Since $b$
  is honest, $b$
  follows the location redirect contained in the response sent by $r$.
  This location redirect contains the $\mi{state}$
  (as a URL parameter). The redirect points to some domain of
  $i$.\footnote{This
    follows from Definition~\ref{def:idp-registration-record} and
    Definition~\ref{def:relying-parties}.} The browser therefore sends
  (among others) $\mi{state}$
  to $i$.
  Of all the endpoints at $i$
  where the request can be received, the authorization endpoint is the
  only endpoint where $\mi{state}$
  could potentially leak to another party. (For all other endpoints,
  the value is dropped.) If the request is received at the
  authorization endpoint, $\mi{state}$
  is only sent back to $b$
  in the initial scriptstate of $\mi{script\_idp\_form}$.
  In this case, the script sends $\mi{state}$
  back to $i$
  in a POST request to the authorization endpoint. Note that in the
  steps outlined here, the value
  $\mi{client\_id} = \mathsf{clientIDOfRP}(r, i)$
  is transferred alongside with $\mi{state}$
  (and not altered in-between). Now, after receiving $\mi{state}$
  and $\mi{client\_id}$
  in a POST request at the authorization endpoint, $i$
  looks up some redirection URI for $\mi{client\_id}$,
  which, by Definition~\ref{def:initial-state-idp}, is some URI at a
  domain of $r$.
  The value $\mi{state}$
  is appended to this URI (either as a parameter or in the fragment).
  The redirection to the redirection URI is then sent to the browser
  $b$. Therefore, $b$ now sends a GET request to $r$.

  If $\mi{state}$
  is contained in the parameter, then $\mi{state}$
  is immediately sent to $r$
  where it is compared to the stored login session records but neither
  stored nor sent out again. In each case, a script is sent back to
  $b$.
  The scripts that $r$
  can send out are $\mi{script\_rp\_index}$
  and $\mi{script\_rp\_implicit}$,
  none of which cause requests that contain $\mi{state}$.
  Also, since both scripts are always delivered with a restrictive
  Referrer Policy header, any requests that are caused by these scripts
  (e.g., the start of a new login flow) do not contain $\mi{state}$
  in the referer header.\footnote{We note that, as discussed earlier,
    without the Referrer Policy, $\mi{state}$
    could leak to a malicious IdP or other parties.} 

  If $\mi{state}$
  is contained in the fragment, then $\mi{state}$
  is not immediately sent to $r$,
  but instead, a request without $\mi{state}$
  is sent to $r$.
  Since this is a GET request, $r$
  either answers with an empty response
  (Lines~\ref{line:rp-set-cookie}ff.~of Algorithm~\ref{alg:rp-oauth}),
  a response containing $\mi{script\_rp\_index}$
  (Lines~\ref{line:rp-serve-index}ff.), or a response containing
  $\mi{script\_rp\_implicit}$
  (Line~\ref{line:rp-send-script-implicit}). In case of the empty
  response, $\mi{state}$
  is not used anymore by the browser. In case of
  $\mi{script\_rp\_index}$,
  the fragment is not used. (As above, there is no other way in which
  $\mi{state}$
  can be sent out, also because the fragment part of an URL is
  stripped in the referer header.) In the case of
  $\mi{script\_rp\_implicit}$
  being loaded into the browser, the script sends $\mi{state}$
  in the body of an HTTPS request to $r$
  (using the path $\str{/receiveTokenFromImplicitGrant}$).
  When $r$
  receives this request, it does not send out $\mi{state}$
  to any party (see Lines~\ref{line:rp-receive-implicit-token}ff.\ of
  Algorithm~\ref{alg:rp-oauth}).

  This shows that $\mi{state}$
  cannot be known to any party except for $b$, $i$, and $r$.
\end{proof}

\begin{definition}\label{def:corresponding-steps}
  Let $e_1 = \an{a_1, f_1, m_1}$
  and $e_2 = \an{a_2, f_2, m_2}$
  be events with $m_1$
  being a DNS request and $m_2$
  being a DNS response or $m_1$
  being an HTTP(S) request and $m_2$
  being an HTTP(S) response. We say that the events \emph{correspond}
  to each other if $m_1$
  and $m_2$
  use the same DNS/HTTP(S) message nonce, $a_1 = f_2$
  and $a_2 = f_1$,
  and (for HTTP(S) messages) either both $m_1$
  and $m_2$ are encrypted or both are not encrypted.
\end{definition}

Given a run $\rho$,
and two events $e_1$
and $e_2$
where $e_1$
is emitted in a processing step $Q_1$
in $\rho$
before $e_2$
is emitted in a processing step $Q_2$
in $\rho$,
we write $e_1 \rightsquigarrow e_2$
if $e_1$
corresponds to $e_2$
and we write $e_1 \dashrightarrow e_2$ if $Q_1$ is connected to $Q_2$.

\begin{figure}
  \begin{eqnarray}
    d_\text{auth}^\text{req} \rightsquigarrow d_\text{auth}^\text{resp} \dashrightarrow e_\text{auth}^\text{req} \dashrightarrow d_\text{cred}^\text{req} \rightsquigarrow d_\text{cred}^\text{resp} \dashrightarrow e_\text{cred}^\text{req} \rightsquigarrow e_\text{cred}^\text{resp} \dashrightarrow d_\text{intr}^\text{req} \rightsquigarrow d_\text{intr}^\text{resp} \dashrightarrow e_\text{intr}^\text{req} \rightsquigarrow e_\text{intr}^\text{resp} \dashrightarrow e_\text{auth}^\text{resp} & & \label{eqn:run-auth-ropcg} \\[2ex]
    e_\text{auth}^\text{req} \dashrightarrow d_\text{tokn}^\text{req} \rightsquigarrow d_\text{tokn}^\text{resp} \dashrightarrow e_\text{tokn}^\text{req} \rightsquigarrow e_\text{tokn}^\text{resp} \dashrightarrow d_\text{intr}^\text{req} \rightsquigarrow d_\text{intr}^\text{resp} \dashrightarrow e_\text{intr}^\text{req} \rightsquigarrow e_\text{intr}^\text{resp} \dashrightarrow e_\text{auth}^\text{resp}  & &  \label{eqn:run-auth-code} \\[2ex]
    e_\text{auth}^\text{req} \dashrightarrow d_\text{intr}^\text{req} \rightsquigarrow d_\text{intr}^\text{resp} \dashrightarrow e_\text{intr}^\text{req} \rightsquigarrow e_\text{intr}^\text{resp} \dashrightarrow e_\text{auth}^\text{resp} & &  \label{eqn:run-auth-implicit}
  \end{eqnarray}
  \caption[Events as described in
  Lemma~\ref{lemma:oauth-flow-backtracking-to-auth}]{Events as
    described in Lemma~\ref{lemma:oauth-flow-backtracking-to-auth}.
    Here, $e_{\cdot}^{\cdot}$
    denotes events containing HTTP(S) messages, $d_{\cdot}^{\cdot}$
    denotes events containing DNS messages. (\ref{eqn:run-auth-ropcg})
    applies to the resource owner password credentials mode,
    (\ref{eqn:run-auth-code}) applies to the authorization code mode,
    and (\ref{eqn:run-auth-implicit}) applies to the implicit mode. }
  \label{fig:oauth-flow-backtracking-to-auth}
\end{figure}

\begin{lemma}\label{lemma:oauth-flow-backtracking-to-auth}
  Given a run $\rho$,
  an RP $r$,
  and a browser $b$,
  if $r$,
  in the run $\rho$,
  emits an event, say $e_\text{auth}^\text{resp}$,
  in Line~\ref{line:rp-send-auth-response} of
  Algorithm~\ref{alg:rp-oauth} that is addressed to $b$,
  and $b$
  and $r$
  are not corrupted at this point in the run, then all of the
  following statements hold true:
  \begin{enumerate}[label=(\alph*)]
  \item Events of one of the forms shown in
    Figure~\ref{fig:oauth-flow-backtracking-to-auth} exist in $\rho$.
  \item The event $e_\text{auth}^\text{req}$
    was emitted by $b$ and is addressed to $r$.
  \item Let
    $e_\text{intr}^\text{resp} = \an{a_\text{intr}^\text{resp},
      f_\text{intr}^\text{resp}, m_\text{intr}^\text{resp}}$ with
    $f_\text{intr}^\text{resp}$
    being an IP adress of some party, say, $i$.
    Then there is a $Q_\text{starts}$
    such that $\mathsf{startsOA}(Q_\text{starts}, b, r, i)$
    and we have that (1) $d_\text{auth}^\text{req}$
    was emitted in $Q_\text{starts}$, or (2) there are events
    $$ d_\text{strt}^\text{req} \rightsquigarrow d_\text{strt}^\text{resp} \dashrightarrow e_\text{strt}^\text{req} \rightsquigarrow e_\text{strt}^\text{resp}$$ 
    such that $d_\text{strt}^\text{req}$
    was emitted in $Q_\text{starts}$
    and $e_\text{strt}^\text{resp}$
    was received by $r$
    before $e_\text{auth}^\text{req}$ was received by $r$.
  \end{enumerate}

\end{lemma}

\begin{proof} \textbf{(a)} We have that
  $e_\text{auth}^\text{resp} = \an{a_\text{auth}^\text{resp},
    f_\text{auth}^\text{resp}, m_\text{auth}^\text{resp}}$ was emitted
  by $r$
  in Line~\ref{line:rp-send-auth-response} of
  Algorithm~\ref{alg:rp-oauth}. (Note that $a_\text{auth}^\text{resp}$
  is an address of $b$.)
  This requires that $r$
  received (and further processed) an HTTPS response in
  $e_\text{intr}^\text{resp}$.
  Also, it is required that (before receiving this event) there is an
  entry in the state of $r$
  in the subterm $\str{pendingRequests}$
  of the form
  $\mi{ref} = \an{\mi{reference}, \mi{request}, \mi{key}, \mi{f}}$
  for some terms $\mi{request}$,
  $\mi{key}$,
  and $f$.
  In this subterm, $\mi{request}.\str{nonce}$
  must be the nonce used in the HTTPS response in
  $e_\text{intr}^\text{resp}$,
  and $\mi{reference}$
  must be of the form
  $\an{\str{introspect}, \mi{mode}', \mi{idp},
    f_\text{auth}^\text{resp}, a_\text{auth}^\text{resp}, n', k'}$
  where $n'$
  is the nonce used in $m_\text{auth}^\text{resp}$,
  $k'$
  is the key used to encrypt $m_\text{auth}^\text{resp}$,
  and $\mi{idp}$ is some domain.

  A subterm of the form of $\mi{ref}$
  therefore had to be created in $\str{pendingRequests}$
  before. This term is only appended to in
  Line~\ref{line:move-reference-to-pending-request} of
  Algorithm~\ref{alg:rp-oauth}. There, the message in $\mi{request}$
  was sent out because a DNS response with some message nonce $n''$
  was received and in the state of $r$
  the following holds true:
  $\str{pendingDNS}[n''] \equiv \an{\mi{reference}, \mi{request}}$.
  Such entries in $\str{pendingDNS}$
  can only be created when a corresponding DNS request is sent out,
  which can happen in Lines~\ref{line:rp-send-something-1},
  \ref{line:rp-send-something-2}, \ref{line:rp-send-something-3},
  \ref{line:rp-send-something-4}, and~\ref{line:rp-send-something-5}.
  We therefore have that the events $d_\text{intr}^\text{req}$,
  $d_\text{intr}^\text{resp}$,
  and $e_\text{intr}^\text{req}$
  exist and have the mutual relations shown in
  (\ref{eqn:run-auth-ropcg}), (\ref{eqn:run-auth-code}), and
  (\ref{eqn:run-auth-implicit}).

  The string $\str{introspect}$
  is set as the first part of $\mi{reference}$
  in Lines~\ref{line:rp-add-introspect-2}
  and~\ref{line:rp-add-introspect-1}. We examine these cases
  separately.

  In the case that $\mi{reference}$
  was created in Line~\ref{line:rp-add-introspect-2} (where also the
  second part of $\mi{reference}$
  is set to $\str{implicit}$),
  an incoming HTTPS request from $a_\text{auth}^\text{resp}$,
  i.e., from $b$,
  must have been received. This shows the existence and mutual
  relations of all events depicted in (\ref{eqn:run-auth-implicit})
  for the implicit mode.

  Otherwise, $\mi{reference}$
  was created in Line~\ref{line:rp-add-introspect-1}. This requires
  that $r$
  must have received an HTTPS response ($e_\text{cred}^\text{resp}$
  or $e_\text{tokn}^\text{resp}$),
  that, as above, has a matching entry in $\str{pendingRequests}$,
  which, as above, was created by sending out an HTTPS request, which,
  again as above, was preceded by a DNS request and response. We
  therefore have that (in the resource owner password credentials
  mode) $d_\text{cred}^\text{req}$,
  $d_\text{cred}^\text{resp}$,
  $e_\text{cred}^\text{req}$,
  $e_\text{cred}^\text{resp}$
  or (in the authorization code mode) $d_\text{tokn}^\text{req}$,
  $d_\text{tokn}^\text{resp}$,
  $e_\text{tokn}^\text{req}$,
  $e_\text{tokn}^\text{resp}$
  exist and have the mutual relations shown in`
  (\ref{eqn:run-auth-ropcg}) and (\ref{eqn:run-auth-code}),
  respectively.

  It is further required that another reference term,
  $\mi{reference}'$
  was in $\str{pendingRequests}$
  when $e_\text{cred}^\text{resp}$
  or $e_\text{tokn}^\text{resp}$
  was received. The term $\mi{reference}'$
  must be of the following form:
  $$ \mi{reference}' = \an{w, \mi{idp}, f_\text{auth}^\text{resp},
    a_\text{auth}^\text{resp}, n', k'}$$ with
  $w \in \{\str{password}, \str{code}\}$.\footnote{Note
    that $w$
    cannot be $\str{client\_credentials}$
    because in this case, $\mi{mode}'$
    in $\mi{reference}$
    would have been $\str{client\_credentials}$,
    which contradicts that in the processing step $Q$,\df{<-?}
    an event was emitted.}
  
  Now, as above, we can check where $\mi{reference}'$
  was created as an entry in $\str{pendingDNS}$.
  This can only happen in
  Line~\ref{line:rp-start-retrieve-code-from-password}
  ($w \equiv \str{password}$)
  and~\ref{line:rp-start-retrieve-code} ($w \equiv \str{code}$).
  In both cases, an incoming HTTPS request from
  $a_\text{auth}^\text{resp}$,
  i.e., from $b$,
  must have been received. This shows the existance and mutual
  relations of all events depicted in (\ref{eqn:run-auth-code}).

  For (\ref{eqn:run-auth-ropcg}), it is easy to see (as above) that
  $d_\text{auth}^\text{req}$
  and $d_\text{auth}^\text{resp}$
  exist and have the mutual relations as shown.

  \textbf{(b)} As already shown above, in all cases,
  $e_\text{auth}^\text{req}$ was sent by $b$ to $r$.

  \textbf{(c)} We have that $e_\text{intr}^\text{resp}$
  was received from $i$.
  Therefore, $e_\text{intr}^\text{req}$
  must have been sent to $i$.
  Therefore, $r$
  requested the IP address of some domain of $i$
  in $d_\text{intr}^\text{req}$.
  This DNS request was created for the domain of a token endpoint
  which was looked up in an IdP registration record stored under the
  key $\mi{idp}$.
  From Definitions~\ref{def:relying-parties}
  and~\ref{def:idp-registration-record} it follows that $\mi{idp}$
  is a domain of $i$.
  
  As above, we now have to distinguish where the value
  $\mi{reference}$
  is created such that the first part is $\str{introspect}$.
  This can happen in Lines~\ref{line:rp-add-introspect-1}
  and~\ref{line:rp-add-introspect-2}. We examine these cases
  separately.

  \begin{itemize}
  \item From \textbf{(a)} above we have that $\mi{reference}'$
    (which contains $\mi{idp}$)
    was created as an entry in $\str{pendingDNS}$
    in Line~\ref{line:rp-start-retrieve-code-from-password}
    or~\ref{line:rp-start-retrieve-code}.

    In the case that $\mi{reference}'$
    was created in
    Line~\ref{line:rp-start-retrieve-code-from-password} we have that
    the HTTPS request $e_\text{auth}^\text{req}$
    (which was sent by $b$
    as shown above) must have been received by $r$
    and that this request was a POST request for the path
    $\str{/passwordLogin}$,
    with a message body $\mi{body}$
    such that $\proj{2}{\proj{1}{\mi{body}}} \equiv \mi{idp}$,
    and that contains an origin header for some domain of $r$.
    Such a request can only be caused by $\mi{script\_rp\_index}$
    loaded into $b$
    from some domain of $r$.
    Hence, this script selected the domain $\mi{idp}$
    in Line~\ref{line:script-rp-index-selected-domain} of
    Algorithm~\ref{alg:script-rp-index} and we have that
    $\mathsf{startsOA}(Q_\text{auth}, b, r, i)$
    where $Q_\text{auth}$
    is the processing step that emitted $d_\text{auth}^\text{req}$.

    In the case that $\mi{reference}'$
    was created in Line~\ref{line:rp-start-retrieve-code} we have that (*)
    the HTTPS request $e_\text{auth}^\text{req}$
    must have been received by $r$
    and that in this request there is a cookie $\str{loginSessionId}$
    with a value, say, $l$
    such that in the state of $r$
    (when receiving the request) in the subterm $\str{loginSessions}$
    under the key $l$
    there is a sequence with the first element being $\mi{idp}$.

    Since we have that $e_\text{auth}^\text{req}$
    was sent by $b$
    (as shown above) we have that $b$
    must have received an HTTP(S) response from $r$
    which contains a Set-Cookie header for the cookie
    $\str{loginSessionId}$
    with the value $l$.\footnote{Note
      that this cookie cannot be set by any party except for $r$
      and there are no scripts sent out by $r$
      that set cookies.} We denote the event of this message as
    $e_\text{strt}^\text{resp}$.
    This message must have been created in
    Line~\ref{line:rp-send-redirect} and, in the same processing step,
    an entry in $\str{loginSessions}$
    under the key $l$
    as described above is created in
    Line~\ref{line:rp-create-session}. (There are no other places
    where login session entries are created.) We have that the
    corresponding request $e_\text{strt}^\text{req}$
    is a POST request with an origin header for some domain of $r$,
    the path $\str{/startInteractiveLogin}$,
    and that the body must be $\mi{idp}$.
    As above, such a request can only be caused by
    $\mi{script\_rp\_index}$
    loaded into $b$
    from some domain of $r$.
    Hence, this script selected the domain $\mi{idp}$
    in Line~\ref{line:script-rp-index-selected-domain} of
    Algorithm~\ref{alg:script-rp-index}, which output an
    $\str{HREF}$-command
    to the browser to send $e_\text{strt}^\text{req}$
    to $r$.
    This request is preceded by a pair of corresponding DNS messages
    $d_\text{strt}^\text{req}$
    and $d_\text{strt}^\text{resp}$
    as defined in the browser relation. We therefore have that
    $\mathsf{startsOA}(Q_\text{strt}, b, r, i)$
    where $Q_\text{strt}$
    is the processing step that emitted $d_\text{strt}^\text{req}$.
  
  \item In the case that $\mi{reference}$
    was created in Line~\ref{line:rp-add-introspect-2} we have the
    same situation as in (*) and the proof continues exactly as in
    (*).
  
  \end{itemize}

\end{proof}

\begin{lemma}\label{lemma:owsw-secure-session-integrity-authz}\label{lemma:proof-si-authorization-final}
  Let $\oauthwebsystem^w$
  be an OAuth web system with web attackers, then $\oauthwebsystem^w$
  is secure w.r.t.~session integrity for authorization.
\end{lemma}

\begin{proof}
  We have to show that for all OAuth web system with web attackers
  $\oauthwebsystem^w$,
  for every run $\rho$
  of $\oauthwebsystem^w$,
  every processing step $Q_\text{ends}$
  in $\rho$,
  every browser $b$
  that is honest in $Q_\text{ends}$,
  every $r\in \fAP{RP}$
  that is honest in $Q_\text{ends}$,
  every $i \in \fAP{IDP}$,
  every identity $\an{u,g}$,
  some protected resource $t$,
  the following holds true: If $\mathsf{endsOA}(Q_\text{ends}, b, r, i, t)$, then
  \begin{enumerate}[label=(\alph*)]
  \item there is an OAuth Session $o \in \mathsf{OASessions}(\rho, b, r, i)$, and
  \item if $i$
    is honest in $Q_\text{ends}$
    then $Q_\text{ends}$
    is in $o$
    and we have that 
    $$\mathsf{selected}_\text{ia}(o, b, r,  \an{u,g}) \iff \big(t
    \equiv \mathsf{resourceOf}(i, r, \an{u,g})\big)$$ or
    $$\mathsf{selected}_\text{nia}(o, b, r,  \an{u,g}) \iff \big(t
    \equiv \mathsf{resourceOf}(i, r', \an{u,g})\big)$$ for some $r' \in \{r, \bot\}$.
  \end{enumerate}
  We can see that Lemma~\ref{lemma:oauth-flow-backtracking-to-auth}
  applies, since $\mathsf{endsOA}(Q_\text{ends}, b, r, i, t)$
  where $Q_\text{ends}$
  is the processing step in which $e_\text{intr}^\text{resp}$
  was received by $r$
  from $i$
  and $e_\text{auth}^\text{resp}$
  was emitted to $b$.
  With Lemma~\ref{lemma:oauth-flow-backtracking-to-auth} (c)
  and Definition~\ref{def:oauth-sessions} it immediately follows that there
  is an OAuth Session $o \in \mathsf{OASessions}(\rho, b, r, i)$.

  For part (b), we now show the connection between $Q_\text{ends}$
  and $o$
  and show that one of the logical equivalences in (b) hold true. In
  the following, we therefore have that $i$ is honest.

  In Lemma~\ref{lemma:oauth-flow-backtracking-to-auth} we have already
  shown the existence of and the relations between the events of one
  of the forms shown in
  Figure~\ref{fig:oauth-flow-backtracking-to-auth}. For any two events
  $e_1 \rightsquigarrow e_2$
  in Figure~\ref{fig:oauth-flow-backtracking-to-auth}, the processing
  steps where these events where emitted are connected (as $i$
  and DNS servers are honest).

  \begin{figure}[h]
    \begin{eqnarray}
      d_\text{strt}^\text{req} &\dashrightarrow& d_\text{strt}^\text{resp} \dashrightarrow e_\text{strt}^\text{req} \dashrightarrow e_\text{strt}^\text{resp}  \nonumber \\ 
                               &\dashrightarrow& d_\text{aep1}^\text{req} \dashrightarrow d_\text{aep1}^\text{resp} \dashrightarrow e_\text{aep1}^\text{req} \dashrightarrow e_\text{aep1}^\text{resp} \nonumber \\
                               &\dashrightarrow& d_\text{aep2}^\text{req} \dashrightarrow d_\text{aep2}^\text{resp} \dashrightarrow e_\text{aep2}^\text{req} \dashrightarrow e_\text{aep2}^\text{resp}\nonumber \\
                               &\dashrightarrow& d_\text{auth}^\text{req} \dashrightarrow d_\text{auth}^\text{resp} \dashrightarrow e_\text{auth}^\text{req} \label{eqn:run-auth-code-start}\\[2ex]
      d_\text{strt}^\text{req} &\dashrightarrow& d_\text{strt}^\text{resp} \dashrightarrow e_\text{strt}^\text{req} \dashrightarrow e_\text{strt}^\text{resp}  \nonumber \\
                               &\dashrightarrow& d_\text{aep1}^\text{req} \dashrightarrow d_\text{aep1}^\text{resp} \dashrightarrow e_\text{aep1}^\text{req} \dashrightarrow e_\text{aep1}^\text{resp} \nonumber \\
                               &\dashrightarrow& d_\text{aep2}^\text{req} \dashrightarrow d_\text{aep2}^\text{resp} \dashrightarrow e_\text{aep2}^\text{req} \dashrightarrow e_\text{aep2}^\text{resp} \nonumber \\
                               &\dashrightarrow& d_\text{impl}^\text{req} \dashrightarrow d_\text{impl}^\text{resp} \dashrightarrow e_\text{impl}^\text{req} \dashrightarrow e_\text{impl}^\text{resp}\nonumber \\
                               &\dashrightarrow& d_\text{auth}^\text{req} \dashrightarrow d_\text{auth}^\text{resp} \dashrightarrow e_\text{auth}^\text{req} \label{eqn:run-auth-implicit-start}
    \end{eqnarray}
    \caption{Structure of run from start to redirection endpoint}
    \label{fig:oauth-flow-backtracking-to-start}
  \end{figure}

  \noindent
  \textbf{Authorization Code Mode.} We now show that if the events are
  structured as shown in (\ref{eqn:run-auth-code}) in
  Figure~\ref{fig:oauth-flow-backtracking-to-auth} then there also
  exist events as shown in (\ref{eqn:run-auth-code-start}) in
  Figure~\ref{fig:oauth-flow-backtracking-to-start}. (The event
  $e_\text{auth}^\text{req}$ is the same in both figures.)
  
  Since we have that $e_\text{auth}^\text{req}$
  exists and was sent by $b$,
  the DNS messages $d_\text{auth}^\text{req}$
  and $d_\text{auth}^\text{resp}$
  (as shown) follow immediately. The request
  $e_\text{auth}^\text{req}$
  contains a session cookie containing a session id, say, $l$.
  The request also contains a URI parameter $\str{state}$
  with some value, say, $z$.\footnote{
    From the proof of
    Lemma~\ref{lemma:oauth-flow-backtracking-to-auth} we follow that
    $e_\text{auth}^\text{req}$
    must be an HTTPS request for the path $\str{/redirectionEndpoint}$
    containing the parameters $\str{code}$,
    $\str{state}$, $\str{iss}$, and $\str{client\_id}$. }

  With Lemma~\ref{lemma:attacker-does-not-learn-state}, we can see
  that the attacker (or any other party except for $i$,
  $b$,
  and $r$)
  cannot instruct the browser to send $e_\text{auth}^\text{req}$.
  Also, $r$
  does not instruct the browser to send such a request, and neither
  does any honest script. The request must therefore have been caused
  by a redirection contained in an event $e_\text{aep2}^\text{resp}$
  that was sent from $i$
  to $b$
  (see Line~\ref{line:idp-redir-with-auth-code} of
  Algorithm~\ref{alg:idp-oauth}). (The redirection must have included
  the state parameter in the URI as above.) This requires that an
  event $e_\text{aep2}^\text{resp}$
  was sent from $b$
  to $i$.
  (Which, as above, was preceded by DNS messages
  $d_\text{aep2}^\text{req}$
  and $d_\text{aep2}^\text{resp}$.)
   This event must contain an HTTP(S) POST request,
  with an origin header value of some domain of $i$,
  and in the body there must be a dictionary with an entry for the key
  $\str{client\_id}$
  containing the client id $c = \mathsf{clientIDOfRP}(r,i)$,
  and an entry for the key $\str{state}$ with the value $z$. (Note that in this case, $c \neq \bot$.)
  
  Because of the origin header value, this request can only be caused
  by the script $\mi{script\_idp\_form}$.
  This script extracted $c$
  and $z$
  from its initial scriptstate, which was a dictionary with the keys
  as above.\footnote{This initial scriptstate is never changed if the
    script runs under the origin of an honest IdP, which it does in
    this case.} The initial scriptstate must have been sent by $i$
  in an event $e_\text{aep1}^\text{resp}$.
  Such an event can only be sent out in Line~\ref{line:idp-send-form}
  of Algorithm~\ref{alg:idp-oauth}.

  The event $e_\text{aep1}^\text{resp}$,
  as above, must have been preceded by connected events
  $d_\text{aep1}^\text{req}$,
  $d_\text{aep1}^\text{resp}$,
  and $e_\text{aep1}^\text{req}$.
  In $e_\text{aep1}^\text{req}$
  the message must be an HTTP(S) request which must have two
  parameters, first, under the key $\str{state}$,
  the value $z$,
  and second, under the key $\str{client\_id}$,
  the value $l$.
  (These parameters are used as the initial scriptstate for the script
  $\mi{script\_idp\_form}$ above.)

  Similar to above, with Lemma~\ref{lemma:attacker-does-not-learn-state}, we have that the
  event $e_\text{aep1}^\text{req}$
  (and, with that, $d_\text{aep1}^\text{req}$)
  must have been caused by a redirect that was sent from $r$
  to $b$.
  Such a response is only created by $r$
  in Line~\ref{line:rp-send-redirect} of Algorithm~\ref{alg:rp-oauth}.
  Since the state value is always chosen freshly, and we have that in
  this case it is $z$,
  the event containing this redirect is $e_\text{strt}^\text{resp}$.

  It is now easy to see that the sequence of processing steps emitting
  the events in (\ref{eqn:run-auth-code-start}) and
  (\ref{eqn:run-auth-code}) is a session (as in
  Definition~\ref{def:sessions}), say, $o$.
  We already know that $\mathsf{startsOA}(Q_\text{starts}, b, r, i)$
  where $Q_\text{starts}$
  is the processing step in which $d_\text{strt}^\text{req}$
  was emitted. There is no other processing step in $o$
  in which the browser $b$
  triggers the script $\mi{script\_rp\_index}$.
  The processing step $Q_\text{ends}$
  (in which $e_\text{auth}^\text{resp}$
  is emitted) is the only processing step in which $r$
  receives a protected resource from $i$
  and emits an event in Line~\ref{line:rp-send-auth-response} of
  Algorithm~\ref{alg:rp-oauth}. Therefore, $o$
  is an OAuth session, and $Q_\text{ends}$ is in $o$.
  
  We now show that
  $$\mathsf{selected}_\text{ia}(o, b, r, \an{u,g}) \iff \big(t \equiv
  \mathsf{resourceOf}(i, r, \an{u,g})\big)\ .$$ Iff
  $\mathsf{selected}_\text{ia}(o, b, r, \an{u,g})$
  then we have that $b$
  in $Q_\text{start}$
  selected $\mi{interactive} \equiv \top$
  in Line~\ref{line:script-rp-index-select-interactive} and there is
  some $Q_\text{select}$
  in $o$
  such that $b$
  triggers the script $\mi{script\_idp\_form}$
  in $Q_\text{select}$
  and selects $\an{u,g}$
  in Line~\ref{line:script-idp-form-select-id} of
  Algorithm~\ref{alg:oauth-script-idp-form} and sends a message out to
  $i$.

  We therefore have that $Q_\text{select}$
  is the processing step where $d_\text{aep2}^\text{req}$
  was emitted. (This is the only processing step in which the browser
  triggers the script $\mi{script\_idp\_form}$.)
  We have that in this step, the browser selected $\an{u,g}$
  in Line~\ref{line:script-idp-form-select-id} of
  Algorithm~\ref{alg:oauth-script-idp-form}. Then, and only then, the
  HTTPS POST request in $e_\text{aep2}^\text{req}$
  contained, in the body, the credentials (username and password) for
  the identity $\an{u,g}$.
  From the proof of Lemma~\ref{lemma:oauth-flow-backtracking-to-auth}
  we see that in $e_\text{strt}^\text{resp}$,
  in the redirection URI, and hence in the URI in
  $e_\text{aep1}^\text{req}$,
  the parameter $\str{response\_type}$
  must be $\str{code}$.
  We therefore have that the initial scriptstate of
  $\mi{script\_idp\_form}$
  in $e_\text{aep1}^\text{resp}$
  contains the entry $\an{\str{response\_type}, \str{code}}$.
  Now, in $e_\text{aep2}^\text{req}$,
  the body also contains the same entry. Therefore, iff $i$
  receives $e_\text{aep2}^\text{req}$,
  then it creates an entry in the subterm $\str{codes}$
  of its state (in Line~\ref{line:idp-create-auth-code} of
  Algorithm~\ref{alg:idp-oauth}) of the form
  $$ \an{\mi{code}, \an{c, \mi{redirecturi}, \an{u,g}}} $$
  (where $\mi{redirecturi}$
  is some URI and $\mi{code}$ is a freshly chosen nonce).

  Then, and only then, $e_\text{aep2}^\text{resp}$
  contains $\mi{code}$
  in the parameter $\str{code}$
  of the location redirect URI (which is the URI for the HTTPS request
  in $e_\text{auth}^\text{req}$).
  RP sends (as shown in the proof of
  Lemma~\ref{lemma:oauth-flow-backtracking-to-auth}) $\mi{code}$
  to IdP in $e_\text{tokn}^\text{req}$.
  This request contains the body
  $\an{\an{\str{grant\_type}, \str{authorization\_code}},
    \an{\str{code}, \mi{code}}}$. 

  Then, and only then, IdP processes $e_\text{tokn}^\text{req}$
  (in Line~\ref{line:idp-create-atoken-code} of
  Algorithm~\ref{alg:idp-oauth}) and creates an entry in the subterm
  $\str{atokens}$ of its state of the form
  $$ \an{\mi{atoken}, \an{c, \an{u,g}}} $$
  for a freshly chosen nonce $\mi{atoken}$
  (as there exists an entry in the subterm $\str{code}$
  of the form $ \an{\mi{code}, \an{c, \mi{redirecturi}, \an{u,g}}} $).
  Then and only then, $\mi{atoken}$
  is contained in $e_\text{tokn}^\text{resp}$.
  Then and only then, $r$
  sends $\mi{atoken}$
  to $i$
  in $e_\text{intr}^\text{req}$.
  (In this request, $\mi{atoken}$
  is contained in the URI parameter $\str{token}$.) 

  Iff there is an entry of the form
  $ \an{\mi{atoken}, \an{c, \an{u,g}}} $
  in the subterm $\str{atokens}$
  in the state of $i$
  and $i$
  receives $e_\text{intr}^\text{req}$
  (containing $\mi{atoken}$
  as shown) then $i$
  processed $e_\text{intr}^\text{req}$
  in Line~\ref{line:idp-introspection-endpoint}ff. and emitted an
  event ($e_\text{intr}^\text{resp}$)
  containing $\mathsf{resourceOf}(i, r, \an{u,g})$. 

  \noindent
  \textbf{Implicit Mode.} This case is very similar to the
  authorization code mode above. We therefore only describe the
  differences between the two modes. 

  In this case, with the proof of
  Lemma~\ref{lemma:oauth-flow-backtracking-to-auth}, we have that
  $e_\text{auth}^\text{req}$
  is an HTTPS POST request to the path
  $\str{/receiveTokenFromImplicitGrant}$
  with an origin header being some domain of $r$.
  Further, as above, $e_\text{auth}^\text{req}$
  contains the state $z$.
  This request must have been created in the browser by
  $\mi{script\_rp\_implicit}$
  running under an origin of $r$.
  This script retrieves the state value from the fragment of the URI
  from which the script was loaded. Therefore, there must have been a
  request, $e_\text{impl}^\text{req}$
  containing such a fragment in the URI. This implies the presence of
  the events $d_\text{impl}^\text{req}$,
  $d_\text{impl}^\text{resp}$,
  and $e_\text{impl}^\text{resp}$.

  We can now that $Q_\text{ends}$
  is in $o$
  and
  $\mathsf{selected}_\text{ia}(o, b, r, \an{u,g}) \iff \big(t \equiv
  \mathsf{resourceOf}(i, r, \an{u,g})\big)$ by applying the same
  reasoning as above, with the following differences:
  \begin{itemize}
  \item The event $e_\text{impl}^\text{req}$
    takes the role of $e_\text{auth}^\text{req}$ in the proof above.
  \item We can show that the sequence of processing steps emitting the
    events in (\ref{eqn:run-auth-implicit}) in
    Figure~\ref{fig:oauth-flow-backtracking-to-auth} and
    (\ref{eqn:run-auth-implicit-start}) in
    Figure~\ref{fig:oauth-flow-backtracking-to-start} are the OAuth
    session $o$ and (as above) that $Q_\text{ends}$ is in $o$.
  \item Where the parameter $\str{response\_type}$
    was $\str{code}$
    above, it now is $\str{token}$.
    The same applies to the initial scriptstate of
    $\mi{script\_idp\_form}$.
  \item Instead of creating $\mi{code}$
    in the processing step that emits $e_\text{aep2}^\text{resp}$,
    this step now creates an access token $\mi{token}$
    (in the same way as the token was created in the authorization
    code mode in the processing step that emits
    $e_\text{tokn}^\text{resp}$).
    The steps $d_\text{tokn}^\text{req}$,
    $d_\text{tokn}^\text{resp}$,
    $e_\text{tokn}^\text{req}$,
    and $e_\text{tokn}^\text{resp}$ are skipped.
  \item The redirection URI contained in $e_\text{aep2}^\text{resp}$
    contains an access token instead of an authorization code, and the
    access token and the state value are contained in the fragment
    instead of in the parameters.
  \item As already discussed, $e_\text{auth}^\text{req}$
    was created by the script $\mi{script\_rp\_implicit}$
    which relays the access token from the URI fragment to $r$.
  \end{itemize}

  \noindent
  \textbf{Resource Owner Password Credentials Mode.} It is easy to see
  that the sequence of processing steps emitting the events in
  (\ref{eqn:run-auth-ropcg}) is a session (as in
  Definition~\ref{def:sessions}), say, $o$.
  In this case, $\mathsf{startsOA}(Q_\text{starts}, b, r, i)$
  holds true if $Q_\text{starts}$
  is the processing step in which $d_\text{auth}^\text{req}$
  was emitted. As above, $o$
  is also an OAuth session, and $Q_\text{ends}$ is in $o$.

  We now show that
  $$\mathsf{selected}_\text{nia}(o, b, r, \an{u,g}) \iff \big(t
  \equiv \mathsf{resourceOf}(i, r', \an{u,g})\big)$$ for some
  $r' \in \{r, \bot\}$.
  Iff $\mathsf{selected}_\text{nia}(o, b, r,  \an{u,g})$
  then we have that $b$
  in $Q_\text{start}$
  selected $\mi{id} \equiv \an{u,g}$
  in Line~\ref{line:script-rp-index-select-id} of
  Algorithm~\ref{alg:script-rp-index} and selected
  $\mi{interactive} \equiv \bot$
  in Line~\ref{line:script-rp-index-select-interactive}.

  Then and only
  then, $e_\text{auth}^\text{req}$
  is an HTTPS POST request for the path $\str{/passwordLogin}$
  with an origin header containing some domain of $r$
  and with the identity $\an{u,g}$
  and the corresponding password, say $p$,
  in the body. Then and only then, the body in
  $e_\text{cred}^\text{req}$
  is of the form
  $$ \an{\an{\str{grant\_type}, \str{password}}, \an{\str{username},
      \an{u,g}}, \an{\str{password}, p}}\ .$$ 

  Then, and only then, IdP processes $e_\text{cred}^\text{req}$
  (in Line~\ref{line:idp-grant-type-password}ff. of
  Algorithm~\ref{alg:idp-oauth}) and creates an entry in the subterm
  $\str{atokens}$ of its state of the form
  $$ \an{\mi{atoken}, \an{c', \an{u,g}}} $$
  for a freshly chosen nonce $\mi{atoken}$
  (as there exists an entry in the subterm $\str{code}$
  of the form $ \an{\mi{code}, \an{c, \mi{redirecturi}, \an{u,g}}} $) and for $c' \in \{\mathsf{clientIDOfRP}(r,i), \bot\}$.
  Then and only then, $\mi{atoken}$
  is contained in $e_\text{cred}^\text{resp}$.
  Then and only then, $r$
  sends $\mi{atoken}$
  to $i$
  in $e_\text{intr}^\text{req}$.
  (In this request, $\mi{atoken}$
  is contained in the URI parameter $\str{token}$.) 

  Iff there is an entry of the form
  $ \an{\mi{atoken}, \an{c', \an{u,g}}} $
  in the subterm $\str{atokens}$
  in the state of $i$
  and $i$
  receives $e_\text{intr}^\text{req}$
  (containing $\mi{atoken}$
  as shown) then $i$
  processed $e_\text{intr}^\text{req}$
  in Line~\ref{line:idp-introspection-endpoint}ff. and emitted an
  event ($e_\text{intr}^\text{resp}$)
  containing $\mathsf{resourceOf}(i, r, \an{u,g})$
  if $c' \neq \bot$
  and containing $\mathsf{resourceOf}(i, \bot, \an{u,g})$ otherwise.
\end{proof}

\begin{lemma}\label{lemma:owsw-secure-session-integrity-authn}\label{lemma:proof-si-authentication-final}
  Let $\oauthwebsystem^w$
  be an OAuth web system with web attackers, then $\oauthwebsystem^w$
  is secure w.r.t.~session integrity for authentication.
\end{lemma}

\begin{proof}
  We have that $r$
  sends a service token to $b$,
  and thus, $\mathsf{endsOA}(Q_\text{login}, b, r, i t)$
  for some term $t$.
  Since $\oauthwebsystem^w$
  is secure w.r.t.~session integrity for authorization, we have that
  (a) holds true. For (b), we see from
  Line~\ref{line:idp-introspection-endpoint}ff. that honest IdPs, at
  their introspection endpoint, if they send out an HTTPS response,
  the body of that response is of the form
  $$\an{\an{\str{protected\_resource}, \mathsf{resourceOf}(i'',r'',\an{u'',g''})}, \an{\str{client\_id}, c''}, \an{\str{user}, \an{u'',g''}}}$$
  for any $\an{u'', g''}$ and some $c''$, $i''$, $r''$. We therefore have that 
  $$ \big(t \equiv \mathsf{resourceOf}(i,r,\an{u,g})\big) \iff \big(\an{u,g} \equiv \an{u',g'}\big)\ .$$
  Since $\oauthwebsystem^w$
  is secure w.r.t.~session integrity for authorization, we have that
  (b) holds true. 

\end{proof}

With Lemma~\ref{lemma:proof-authentication-final},
Lemma~\ref{lemma:proof-authorization-final},
Lemma~\ref{lemma:proof-si-authorization-final} and
Lemma~\ref{lemma:proof-si-authentication-final} we have proven
Theorem~\ref{thm:security}. \qed